\documentclass[oneside,english]{amsart}
\usepackage[T1]{fontenc}
\usepackage[latin9]{inputenc}
\usepackage{geometry}
\usepackage{verbatim}
\usepackage{amsthm}
\usepackage{amstext}
\usepackage{amssymb}
\usepackage{graphicx}
\usepackage{esint}
\usepackage{url}
\usepackage{yhmath}
\usepackage[all,cmtip]{xy}
\usepackage{color}

\usepackage[percent]{overpic}

\usepackage[space]{grffile} 
\makeatletter
\numberwithin{equation}{section}
\numberwithin{figure}{section}
\theoremstyle{plain}
\newtheorem*{thm*}{\protect\theoremname}
\theoremstyle{plain}
\newtheorem{thm}{\protect\theoremname}[section]
\theoremstyle{plain}
\newtheorem{lem}[thm]{\protect\lemmaname}
\theoremstyle{remark}
\newtheorem{rem}[thm]{\protect\remarkname}
\theoremstyle{plain}
\newtheorem{prop}[thm]{\protect\propositionname}
\theoremstyle{plain}
\newtheorem{cor}[thm]{\protect\corollaryname}
\theoremstyle{plain}

\theoremstyle{plain}

\theoremstyle{plain}

\theoremstyle{plain}
\newtheorem{defn}[thm]{Definition}
\theoremstyle{plain}

\theoremstyle{plain}
\newtheorem{ass}[thm]{Assumption}

\usepackage{mathrsfs}
\usepackage{subfigure}

\newcommand{\SLE}{\mathrm{SLE}}

\newcommand{\PR}{\mathbb{P}}
\newcommand{\EX}{\mathbb{E}}

\newcommand{\sE}{\mathscr{E}}
\newcommand{\sF}{\mathscr{F}}
\newcommand{\sG}{\mathscr{G}}
\newcommand{\sB}{\mathscr{B}}
\newcommand{\sH}{\mathscr{H}}
\newcommand{\sZ}{\mathcal{Z}}

\newcommand{\bR}{\mathbb{R}}
\newcommand{\R}{\bR}

\newcommand{\bZ}{\mathbb{Z}}
\newcommand{\Z}{\bZ}
\newcommand{\bN}{\mathbb{N}}
\newcommand{\N}{\bN}

\newcommand{\bC}{\mathbb{C}}
\newcommand{\C}{\bC}

\newcommand{\bH}{\mathbb{H}}

\newcommand{\domain}{\Lambda}
\newcommand{\bdry}{\partial}

\newcommand{\Mob}{\mu}
\newcommand{\confmap}{\phi} 
\newcommand{\confmapDH}{\boldsymbol{\psi}} 

\newcommand{\ctsfcns}{\mathcal{C}} 

\newcommand{\ud}{\mathrm{d}}

\newcommand{\pder}[1]{\frac{\partial}{\partial#1}}
\newcommand{\pdder}[1]{\frac{\partial^{2}}{\partial#1^{2}}}

\newcommand{\clos}[1]{\overline{ #1 }}

\newcommand{\Arch}{\mathrm{LP}}
\newcommand{\LP}{\Arch}

\newcommand{\chamber}{\mathfrak{X}}
\newcommand{\PartF}{\sZ}

\newcommand{\HarmMeas}{\mathsf{H}}

\newcommand{\Pf}{\mathrm{Pf}}

\newcommand{\diam}{\mathrm{diam}}

\newcommand{\eps}{\varepsilon}

\newcommand{\Gr}{\mathcal{G}}
\renewcommand{\Vert}{\mathcal{V}}
\newcommand{\Edg}{\mathcal{E}}

\newcommand{\edgeof}[2]{{\langle #1 , #2 \rangle}}

\newcommand{\shrinkto}{\downarrow}
\newcommand{\incrto}{\uparrow}

\newcommand{\InfiniteGr}{\Gamma}

\newcommand{\UnitD}{\mathbb{D}}

\newcommand{\InitSegmNoInd}{\lambda}
\newcommand{\InitSegm}[1]{\lambda_{#1}} 
\newcommand{\InitSegmLatt}[2]{\lambda^{(#1)}_{#2}} 
\newcommand{\InitSegmDelta}[1]{\lambda_{({#1})}} 
\newcommand{\FinalSegmDelta}[1]{\vartheta_{({#1})}} 
\newcommand{\DrFcn}[2]{W_{#1; #2}} 
\newcommand{\DrFcnNoInd}{W}
\newcommand{\DrFcnLattNotime}[2]{W^{(#1)}_{#2}} 
\newcommand{\DrFcnLatt}[3]{W^{(#1)}_{#2; #3}} 
\newcommand{\IterDrFcnNoInd}{\tilde{W}}
\newcommand{\IterDrFcnLattNotime}[2]{\tilde{W}^{(#1)}_{#2}} 
\newcommand{\IterDrFcnLatt}[3]{\tilde{W}^{(#1)}_{#2; #3}} 
\newcommand{\RandCurve}[1]{\gamma_{#1}} 
\newcommand{\RandCurveD}[1]{\gamma_{\mathbb{D}; #1}} 

\newcommand{\DetDrFcn}{V}
\newcommand{\LoeNbhd}{\mathcal{N}}
\newcommand{\confmapSH}{\varpi}

\newcommand{\DetIterDrFcn}{\tilde{V}}

\newcommand{\Dr}{\DetDrFcn}

\newcommand{\EXNSLE}{\EX^{N\textrm{-}\SLE}}
\newcommand{\PRNSLE}{\PR^{N\textrm{-}\SLE}}
\newcommand{\EXSLEcurves}[1]{\EX^{#1\textrm{-}\SLE}}

\newcommand{\SLEcurve}{\eta}
\newcommand{\DetCurve}{\nu}

\newcommand{\Unitp}{\tilde{p}} 
\newcommand{\eqd}{\stackrel{\scriptsize{(d)}}{=}}

\newcommand{\FKsub}{\omega}

\newcommand{\Mart}{\mathscr{M}}

\makeatletter
\newcommand*{\centerfloat}{%
  \parindent \z@
  \leftskip \z@ \@plus 1fil \@minus \textwidth
  \rightskip\leftskip
  \parfillskip \z@skip}
\makeatother

\AtBeginDocument{

}

\makeatother

\usepackage{babel}
\providecommand{\corollaryname}{Corollary}
\providecommand{\lemmaname}{Lemma}
\providecommand{\propositionname}{Proposition}
\providecommand{\remarkname}{Remark}
\providecommand{\theoremname}{Theorem}
\providecommand{\conjecturename}{Conjecture}

\definecolor{kallecol}{rgb}{.75,.0,.55}

\begin{document}



\

\vspace{2.5cm}

\begin{center}
\LARGE \bf \scshape {Multiple SLE type scaling limits: from local to global
}
\end{center}

\vspace{0.75cm}

\begin{center}
{\large \scshape Alex Karrila}\\
{\footnotesize{\tt alex.karrila@aalto.fi}; \texttt{alex.karrila@gmail.com}}\\
{\small{Department of Mathematics and Systems Analysis}}\\
{\small{P.O. Box 11100, FI-00076 Aalto University, Finland}}\bigskip{}
\end{center}

\vspace{0.75cm}

\begin{center}
\begin{minipage}{0.85\textwidth} \footnotesize
{\scshape Abstract.}
We consider collections of $N$ chordal random curves obtained from a critical lattice model on a planar graph, in the limit when a fine-mesh graph approximates a simply-connected domain.
We define and study candidates for such limits in terms of conformally invariant collections of random curves, generated via iterated Loewner equations. These curves are a natural ``domain Markov extension'' of the earlier introduced local multiple SLE initial segments to global multiple SLE curves. For realizing them as scaling limits, we provide two \emph{a priori} results to guarantee the precompactness of the discrete random curves and to allow promoting a discrete domain Markov property to the scaling limit. These results essentially only take as input certain crossing conditions, very similar to those introduced by Kemppainen and Smirnov, and they allow the identification of scaling limits via the martingale strategy of classical SLE convergence proofs.
The use of these results is exemplified with convergence proofs in various lattice models.
\end{minipage}
\end{center}

\vspace{0.75cm}
\tableofcontents

\bigskip{}

\section{Introduction}
\label{sec: intro}

\addtocontents{toc}{\setcounter{tocdepth}{1}}

\subsection{Background}

The scaling limits of critical random models on lattices, as the lattice mesh size tends to zero, are studied in physics via Conformal field theory~\cite{Polyakov, BPZ-infinite_conformal_symmetry_in_2D_QFT, BPZ-infinite_conformal_symmetry_of_critical_fluctiations, Cardy-conformal_invariance_and_statistical_mechanics}. One mathematical approach to proving conformal invariance in such limits is to characterize the scaling limits of some discrete interfaces in terms of conformally invariant random curves. A breakthrough in this approach was the observation by Schramm~\cite{Schramm-LERW_and_UST} that if such conformally invariant scaling limits exist and inherit the domain Markov property --- a domain reduction property prominent in many simple lattice models --- they belong to a one-parameter family of random curve models, called Schramm-Loewner evolutions (SLEs). This has led to the identification of scaling limits in various lattice models in terms of SLE type curves; see~\cite{Smirnov-critical_percolation, LSW-LERW_and_UST, SS05, CN07, Zhan-scaling_limits_of_planar_LERW, SS09,  CDHKS-convergence_of_Ising_interfaces_to_SLE} on chordal SLEs and~\cite{LSW-LERW_and_UST, Zhan-scaling_limits_of_planar_LERW, HK-Ising_interfaces_and_free_boundary_conditions, Izyurov-Smirnovs_observable_for_free_boundary_conditions, KS-bdary_loops_FK, LV-natural_parametrization_for_SLE, Wu17, Izyurov-critical_Ising_interfaces_in_multiply_connected_domains, KS18, BPW, GW18} on other SLE type curves. 
This paper in concerned with SLE type models and convergence results for multiple simultaneous chordal curves.

All SLE convergence proofs consist of two parts: precompactness and identification. Precompactness means that any sequence of discrete curves on lattices of decreasing mesh sizes has a weakly convergent subsequence. The identification of any subsequential weak limit then proves weak convergence along the entire sequence. These two parts are in a typical proof very different in spirit: the precompactness relies on verifying certain \emph{a priori} crossing estimates that are \emph{non-specific}, in the sense that they hold in a wide range of lattice models. A machinery of precompactness results then applies for the curves~\cite{AB-regularity_and_dim_bounds_for_random_curves, KS, mie}.
The identification, following the nowadays established strategy of~\cite{Smirnov-ICM}, relies on finding an observable in the lattice model that is a martingale under growing an interface, then promoting this martingale to the subsequential limit by a strong enough convergence of the observable, and finally showing that the obtained continuous martingale characterizes the scaling limit. In contrast to the precompactness part, the identification step relies on the exact, highly \emph{model-specific} relations in the matringale and its convergence, and it seems to be the bottleneck in finding SLE convergence proofs. This paper derives non-specific \emph{a priori} results that allow the use of such martingale identification of multiple SLE type scaling limits, and provides several examples of convergence proofs.

Two SLE variants have then been proposed to describe the scaling limit of multiple simultaneous chordal interfaces: the local~\cite{BBK-multiple_SLEs, Dubedat-commutation, Graham-multiple_SLEs, KP-pure_partition_functions_of_multiple_SLEs} and global~\cite{LK-configurational_measure, Lawler-glob_NSLE, PW, BPW} multiple SLEs, and both models have their advantages and disadvantages when working with convergence results. 

Let us briefly discuss local multiple SLEs first --- see Section~\ref{sec: SLE} for a more formal introduction and Figure~\ref{fig: local and global NSLE}(left) for an illustration. Consider a simply-connected domain $\domain$ with $2N$ distinct marked boundary points $p_1, \ldots, p_{2N}$. The local multiple SLE on some disjoint neighbourhoods $U_1, \ldots, U_{2N}$ of the marked boundary points in $\domain$ yields $2N$ curve initial segments (``localizations'') starting from a point $p_i$, $1 \le i \le 2N$, up to exiting the corresponding neighbourhood $U_i$. These initial segments are described explicitly via Loewner's equation, as suitably weighted chordal SLE measures of initial segments in the localization neighbourhoods. One important motivation for studying multiple SLEs is that these weights are given by the most central objects of Conformal field theory, the correlation functions, see e.g.~\cite{BBK-multiple_SLEs, Graham-multiple_SLEs, KKP2}. Back to scaling limits, an advantage of the local multiple SLEs is its similarity to the chordal SLE, while its disadvantages are that it only describes initial segments and that a new martingale observable needs to be introduced for a convergence proof. 

\begin{figure}
\includegraphics[width=0.45\textwidth]{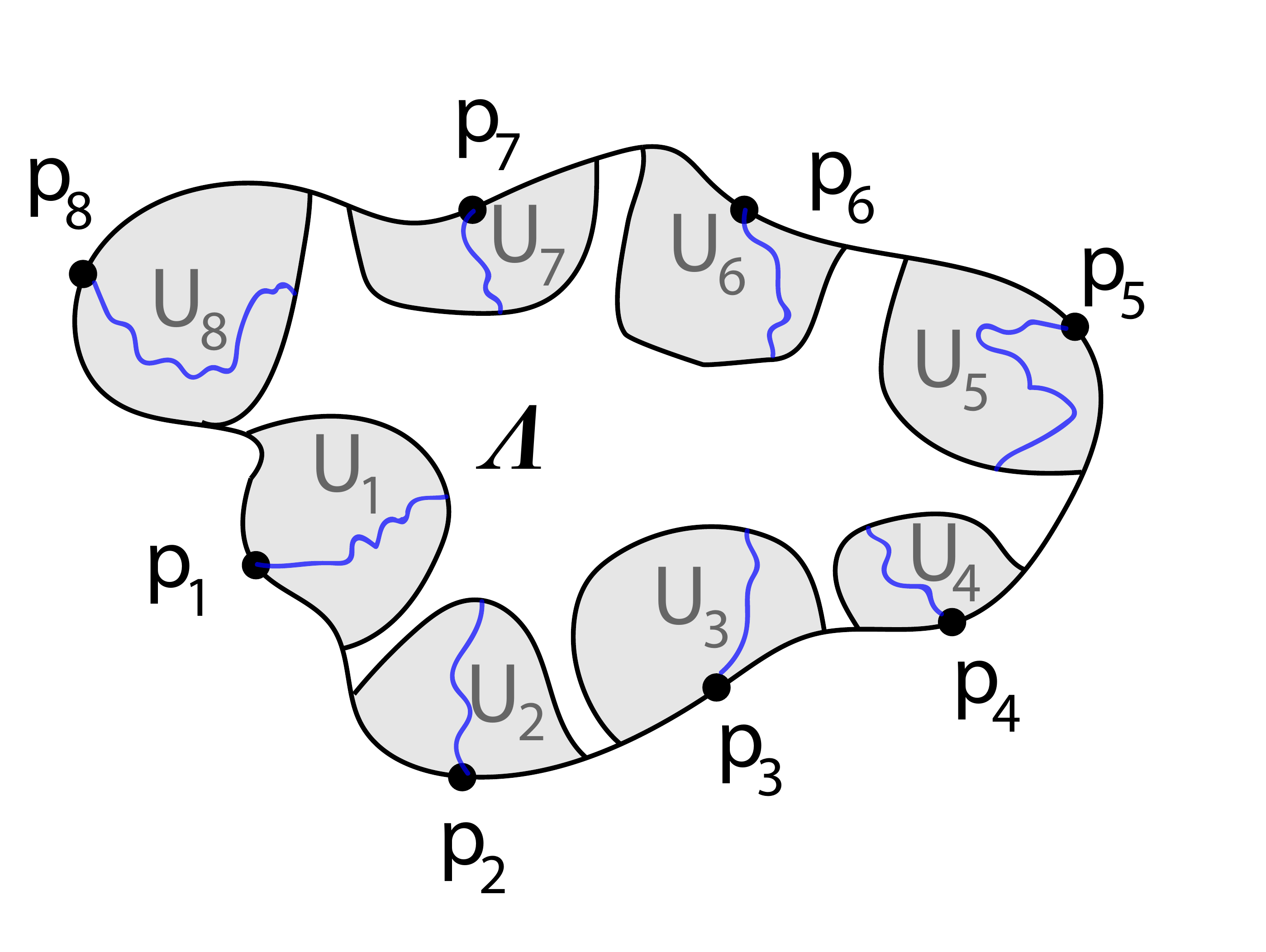} \quad
\includegraphics[width=0.45\textwidth]{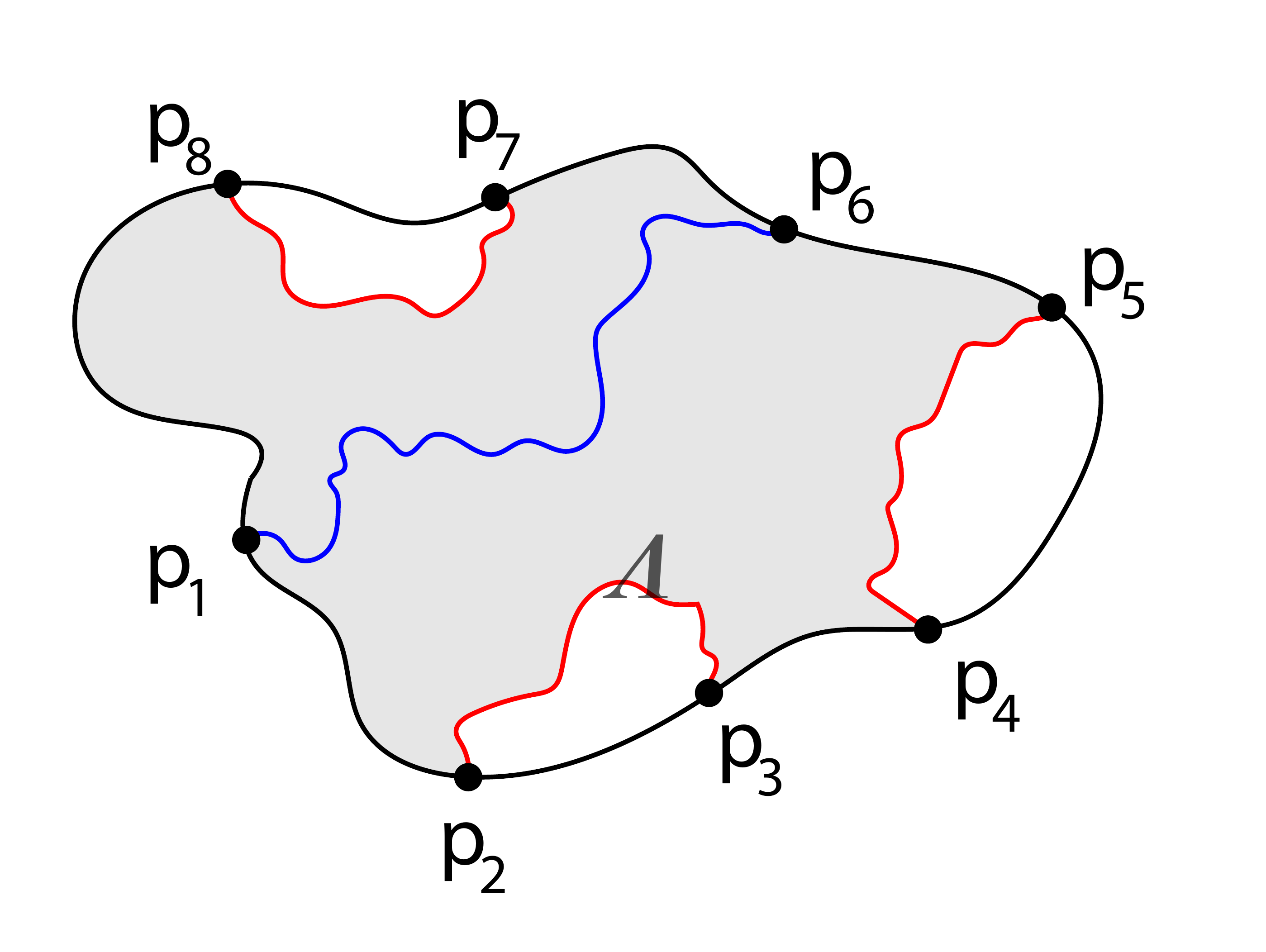}
\caption{ \label{fig: local and global NSLE}
Left: A schematic illustration of the local multiple SLE in a domain $\domain$ with eight marked boundary points and their localization neighbourhoods. Right: A schematic illustration of the global multiple SLE with the pairing $\{ \{ 1, 6\}, \{2, 3 \}, \{ 4, 5 \}, \{ 7, 8 \} \}$ of the boundary points. Conditional on the red curves, the blue one is a chordal SLE in the subdomain of $\domain$ left for it, shaded in the figure.
}
\end{figure}

The global multiple SLE on $(\domain; p_1, \ldots, p_{2N})$, in turn, describes collections of $N$ mutually non-crossing random curves $\gamma_1, \ldots, \gamma_N$, pairing the marked boundary points in some predetermined manner, see Figure~\ref{fig: local and global NSLE}(right) for an illustration. Given the SLE parameter $\kappa \in (0, 8)$ it is  defined, following~\cite{PW, BPW} (see also~\cite{IG2, MSW}), as the stationary distribution of the discrete time Markov chain on collections of $N$ curves, where at each time step one curve is resampled as a chordal $\SLE(\kappa)$ in the domain left for it by the remaining curves. Such a stationary distribution is proven to be unique~\cite{BPW} and exist~\cite{PW} for $\kappa \in (0, 4]$; also the case $\kappa \in (4, 8)$ is conjectured in~\cite{BPW}. This definition is rather implicit, but it can be shown to yield a local multiple $\SLE(\kappa)$ if $\kappa \in (0, 4]$~\cite{PW}. We will not rely on global multiple SLEs in this paper, but we will show that the obtained scaling limits satisfy the above Markov chain stationarity.

From the lattice model point of view, a great benefit of the global multiple SLE is that it often yields miraculously short convergence proofs, provided that the convergence of the corresponding one-curve lattice model to chordal $\SLE(\kappa)$ has been established. Namely, lattice models with domain Markov property satisfy a discrete version of this curve resampling stationarity, and with some \emph{a priori} estimates, it can be promoted to a subsequential scaling limit; see~\cite{BPW} for examples. In particular, no new matringale observable is needed after to the one-curve convergence. Nevertheless, convergence proofs of this type only hold for conditioned lattice models, where the pairing of the boundary points by the interfaces is predetermined. Such a conditioning may appear slightly unnatural, for instance for magnetization cluster interfaces in the Ising magnetism model. To find an unconditional scaling limit, one would thus need to solve the probabilities of the different pairings of boundary points as in~\cite{PW18} for the Ising model and~\cite{Dubedat-Euler_integrals, {KW-boundary_partitions_in_trees_and_dimers}, {KW-double_dimer_pairings_and_skew_Young_diagrams}, KKP, PW} for some other models. This seems not to be easy. Indeed, in lattice models with the discrete domain Markov property, such pairing probabilities yield, under growing an interface, conditional pairing probabilities, and are hence martingales. Proving their convergence should thus be roughly equivalent to an SLE identification step with the usual martingale strategy; see~\cite{Smirnov-critical_percolation, KS18} for examples.

Finally, we remark that the connection of these two multiple SLE type models is not completely clear. For $\kappa \in (0,4]$ the initial segments of a global multiple SLEs are local multiple SLEs~\cite{PW} (see also~\cite{Wu17} on $N=2$ curves), but for $\kappa \in (4, 8)$ such a connection remains conjectural. Furthermore, in this paper we will provide a warning example (with $\kappa = 6$ and $N \ge 3$) showing that curves whose initial segments in any localization neighbourhoods are local multiple SLEs are not necessarily global multiple SLEs.

%
%
%
%
%

\subsection{Contributions of this paper}

In this paper, we show how the convergence of multiple simultaneous chordal interfaces can be proven following the classical strategy of~\cite{Smirnov-ICM}. We characterize such limits in terms of explicit Loewner growth processes similar to local multiple SLEs, and show that such scaling limits are convex combinations of global multiple SLEs wih different pairings. Roughly speaking, this takes three ingredients.

First, we propose a natural ``domain Markov extension'' of local multiple SLEs to full curves, which we call local-to-global multiple SLEs. The well-definedness of the obtained curves follows by realizing them as scaling limits. (For $\kappa \in (0, 4]$ it could also be done based on global multiple SLEs being local, but we avoid taking this or other SLE theory as logical inputs, consistently relying only on arguments based on the underlying lattice models.)

Second, we provide two important \emph{non-specific} results related to the convergence of lattice models: a straightforward generalization of the precompactness conditions~\cite{KS, mie} for multiple curves, and a result showing that any subsequential scaling limit inherits a domain Markov type property from the discrete model. By the latter property, identifying \emph{one initial segment} of \emph{one curve} as a local multiple SLE suffices to identify the \emph{full collection} of \emph{full curves} as its domain Markov extension. The \emph{a priori} results needed for these non-specific results to hold are the discrete domain Markov property and a crossing condition, very similar to that in~\cite{KS} to guarantee precompactness. In particular, these conditions are known to be satisfied in most well-studied lattice models.
As a by-product of the domain Markov type properties, we also obtain the connection to global multiple SLEs.

Third, a convergence proof requires an identification step, in this case identifying one initial segment of one curve as a local multiple SLE. We review three priorly known convergence proofs in Ising, FK-Ising and percolation models, and two new proofs, in detail for the multiple harmonic explorer curves and a sketch for the uniform spanning tree branches. Also FK cluster model is discussed.

Except for referring to the precompactness results of~\cite{KS, mie}, the paper is self-contained. We have tried minimize the amount of logical inputs taken, as well as the \emph{a priori} estimates required from the lattice models.

%
%
%
%

\subsection{Related work}

Apart from the related work mentioned so far, let us mention some references that address similar underlying principles.

One motivation and a Conformal field theory approach to the study of multiple SLEs is their description as chordal SLEs weighted by correlation functions. Some central notions of Conformal field theory, such as fusion and conformal blocks, do not arise when studying single SLEs. This is not the perspective of this paper, but should be kept in mind, see e.g.~\cite{BBK-multiple_SLEs, Graham-multiple_SLEs, Dubedat-commutation, Dubedat-SLE_and_Virasoro_representations_localization, Dubedat-SLE_and_Virasoro_representations_fusion, Peltola-basis_for_solutions_of_BSA_PDEs, KKP2} for more.

The idea of working with non-specific results based on crossing estimates dates back to~\cite{AB-regularity_and_dim_bounds_for_random_curves, KS}. We also have to prove precompactness in different topologies and the agreement of the different weak limits, similarly to~\cite{KS, mie}. As regards the non-specific result on the domain Markov property, the non-triviality of promoting the discrete domain Markov property to a scaling limit has been addressed recently in, e.g.,~\cite{GW18, BPW}. It should be noticed the latter non-specific results in this paper take very little inputs and follow (essentially) once precompactness is verified with the stadard crossing estimates, cf.~\cite{KS}.

The idea of proposing scaling limit random models that are well-defined due to being scaling limits is present in SLE literature at least in~\cite{LSW-LERW_and_UST, Zhan-scaling_limits_of_planar_LERW, Izyurov-critical_Ising_interfaces_in_multiply_connected_domains, BPW}.

This work was initiated in attempt to answer Conjecture~4.3 in the author's earlier paper~\cite{KKP}, whose proof is now sketched in Section~\ref{subsec: LERW outline}. 
Multiple SLE type models have since then attracted quite some attention, see~\cite{Izyurov-critical_Ising_interfaces_in_multiply_connected_domains, Wu17, PW, KS18, BPW, PW18}. 


%
%

\subsection{Organization}

This paper is organized as follows. In Section~\ref{sec: SLE}, we introduce the local multiple SLE and propose its domain Markov extension, the local-to-global multiple SLE. Section~\ref{sec: preli} contains some preliminaries used throughout the paper. Section~\ref{sec: precompactness} addresses non-specific results on precompactness and contains our first main result, Theorem~\ref{thm: precompactness thm multiple curves}. Section~\ref{sec: loc-2-glob} addresses non-specific results on domain Markov property, with the main theorem~\ref{thm: loc-2-glob multiple SLE convergence, kappa le 4} guaranteeing that identification of one initial segment actually identifies the full collection of full curves. We also give a variant of that theorem, suited for the local multiple SLE collection of initial segments, as well as some consequences. For the ease of reading, Sections~\ref{sec: precompactness} and~\ref{sec: loc-2-glob} are arranged so that the statements of the main results are given first, and the technical proofs are postponed to the end of the section. Finally, in Section~\ref{sec: application examples} we give various applications of these results, addressing several lattice models.

\subsection{Acknowledgements}

The author wishes to thank Dmitry Chelkak, Konstantin Izyurov, Antti Kemppainen, Hao Wu, and especially Kalle Kyt\"{o}l\"{a} and Eveliina Peltola for useful and interesting discussions and comments on a preliminary version of this paper. The author is supported by the Vilho, Yrj\"{o} and Kalle V\"{a}is\"{a}l\"{a} Foundation.

\bigskip{}

\section{Multiple SLE type models}
\label{sec: SLE}

\addtocontents{toc}{\setcounter{tocdepth}{2}}


The purpose of this section is to give a sufficient overview of $N-\SLE$ type random curve models. We emphasize that the results of this paper, describing scaling limits as SLE type curves, \emph{do not take SLE theory as logical inputs}, but only rely properties of the converging lattice models\footnote
{To be very precise, there is one small exception to this rule, namely the \emph{a posteriori} argument in the proof of Proposition~\ref{prop: initial segment end points}, using the fact that the chordal $\SLE(\kappa)$ has no boundary visits if and only if $\kappa \le 4$~\cite{RS-basic_properties_of_SLE}. This result is used only after convergence of a lattice model has already been proven, to yield a more convenient description of the scaling limit.
}.


We begin with a brief exposition of local multiple SLEs in Section~\ref{subsec: local NSLE}. No new results are introduced there. Then, in Section~\ref{subsec: loc-2-glob}, we define local-to-global multiple SLEs. This definition is new, ans it will describe the scaling limits in our main results.
We assume that the reader is familiar with the basic properties of the most well-known SLE variant, the chordal SLE; see, e.g., the text books~\cite{Lawler-SLE_book, Berestycki-SLE_book, Kemppainen-SLE_book}.


\subsection{Local multiple SLE}
\label{subsec: local NSLE}

The local multiple SLE is a generalization of the chordal SLE to handle collections of $N$ simultaneous chordal SLE type random curves connecting $2N$ boundary points, first proposed in~\cite{Dubedat-commutation}. Similarly to the chordal SLE, it is defined via conformal invariance and the hulls of a random Loewner growth process in $\bH$. However, the definition of the local multiple SLE does not give full chordal curves, but in stead only initial segments. A familiar example of an analogous restriction the chordal $\SLE(\kappa)$ between two real points $x_1$ and $x_2$~\cite[Lemma 3]{Dubedat-commutation}.

More formally, the local multiple SLE is defined in simply-connected domains $\domain$ with $2N$ distinct boundary points (or in more general, prime ends) $p_i$, $1 \le i \le 2N$, numbered counterclockwise, and their localization neighbourhoods $U_i$ which are closed neighbourhoods of $p_i$ in $\overline{\domain}$, and pairwise disjoint, ${U_i} \cap {U_j} = \emptyset$ if $i \ne j$, and such that $\domain \setminus U_i$ is simply-connected for all $i$. See Figure~\ref{fig: local and global NSLE}(left) for an illustration. The  local multiple SLE is then a measure on collections of curve initial segments from $p_i$ up to the \emph{exit time} of $U_i$, i.e., the first hitting time of $ \overline{ (\domain \setminus U_i ) } $.

Fundamentally, local multiple SLEs arise as multiple random random curve initial segments that satisfy conformal invariance and the domain Markov property and such that the marginal law of each initial segment is absolutely continuous with respect to initial segments of the chordal SLE, see~\cite{Dubedat-commutation, KP-pure_partition_functions_of_multiple_SLEs}. 
Nevertheless,~\cite[Theorem A.4]{KP-pure_partition_functions_of_multiple_SLEs} gives an equivalent characterization in terms of a Loewner chain driven by the sum of a Brownian motion and a partition function term. We adopt the latter as a defintion of local multiple SLEs for the rest of this paper. 

In this paper, we will restrict our attention to SLEs with parameter $\kappa \in (0, 8)$.

\subsubsection{\textbf{Partition functions}}
\label{subsubsec: local multiple SLE partition functions}

The definition of the local multiple $\SLE(\kappa)$ with $2N$ boundary points relies on a partition function $\PartF$.
A function $\PartF$ defined on a chamber $\chamber_{2N} = \{ (x_1, \ldots, x_{2N}) \in \R^{2N} \; : x_1 < \ldots < x_{2N}  \}$ is called an $N$-$\SLE(\kappa)$ partition function if it is positive, $\PartF (x_1, \ldots, x_{2N})  > 0$ for all $(x_1, \ldots, x_{2N}) \in \chamber_{2N}$, and  satisfies the linear partial differential equations (PDEs)
\begin{align}
\label{eq: PDE for multiple SLEs} \tag{PDE} 
& \left[ \frac{\kappa}{2} \pdder{x_j}
    + \sum_{ \substack{ i = 1 \\ i \neq j } }^N \Big( \frac{2}{x_i-x_j} \pder{x_i} - \frac{2h}{(x_i-x_j)^2} \Big) \right] \PartF (x_1 , \ldots, x_{2N}) = 0 \qquad \text{for all } j=1,\ldots,2N,
\end{align}
where 
\begin{align*}
h = 
h ( \kappa )= \frac{6 - \kappa}{2 \kappa}
\end{align*}
and the M\"obius covariance
\begin{align}
\label{eq: COV for multiple SLEs} \tag{COV} 
& \PartF(x_1 , \ldots, x_{2N}) = 
    \prod_{i=1}^{2 N} \Mob'(x_i)^{h} \times \PartF(\Mob(x_1) , \ldots, \Mob(x_{2N}))  \\
\nonumber
& \text{for all } \Mob(z) = \frac{a z + b}{c z + d}, \; \text{ with } a,b,c,d \in \bR, \; ad-bc > 0, 
 \text{ such that } \Mob(x_1) < \cdots < \Mob(x_{2N}).
\end{align}

\subsubsection*{\textbf{Remarks}}

Characterizing the positive solutions to~\eqref{eq: PDE for multiple SLEs} and~\eqref{eq: COV for multiple SLEs}, and hence all local multiple SLEs, is a long-standing task, recently completed for $\kappa \in (0, 4]$, and still partly open for $\kappa \in (4, 8)$~\cite{FK-solution_space_for_a_system_of_null_state_PDEs_1, FK-solution_space_for_a_system_of_null_state_PDEs_2, FK-solution_space_for_a_system_of_null_state_PDEs_3, FK-solution_space_for_a_system_of_null_state_PDEs_4, KP-pure_partition_functions_of_multiple_SLEs, KKP, Wu17, PW, BPW}. 
We stress that \emph{the results in this paper do not rely on the analysis of these PDE solutions}. In stead, we assume that partition functions are obtained as a part of the identification of a scaling limit; see Section~\ref{subsubsec: HE local identification} for a conctrete example.

As a second remark, 
the conditions~\eqref{eq: PDE for multiple SLEs} and~\eqref{eq: COV for multiple SLEs} arise in the derivation of~\cite{Dubedat-commutation} by purely probabilistic arguments, but the exact same conditions are also encountered in Conformal field theory as the covariance rule and degeneracy PDEs~\cite{BPZ-infinite_conformal_symmetry_in_2D_QFT} for primary boundary fields of conformal weight $h$; see, e.g.,~\cite[Section~3.3]{KKP2}.



\subsubsection{\textbf{One-curve marginals in $\bH$}}

Let us now describe the marginal law of the initial segment from the $j$:th marked boundary point in a local multiple $\SLE(\kappa)$ in $\bH$, given the partition function $\PartF$ as above. The initial segment is described by a Loewner equation up to the hitting time $T_j$ of $\overline{ ( \bH \setminus U_j  ) }$, where $U_j \subset \overline{\bH}$ is the localization neighbourhood. Let us denote the real boundary points by $p_i = x_i$, $1 \le i \le 2N$, and assume that $-\infty < x_1 < \ldots < x_{2N}< + \infty.$ We will also need to assume that the localization neighbourhood $U_j $ is bounded (in other words, it is a compact $\bH$-hull). Then, the marginal law of the $j$:th initial segment is described by the Loewner differential equation
\begin{align}
\label{eq: loc NSLE Loewner eq}
\partial_t g_t (z) &= \frac{2}{g_t(z) - \DrFcn{j}{t}}, 
\end{align}
where the \emph{driving function} $\DrFcn{j}{t}$, for $t \in [0, T_j ]$ is determined by the system stochastic differential equations
\begin{align}
\nonumber
& (\DrFcn{1}{0}, \ldots, \DrFcn{2N}{0} ) = (x_0, \ldots, x_{2N}) \\
\label{eq: SDE definition of N-SLE}
&
\begin{cases}
\ud \DrFcn{j}{t} &= \sqrt{\kappa} \ud B_t + \kappa \partial_j \left( \log \PartF (\DrFcn{1}{t}, \ldots, \DrFcn{2N}{t} ) \right) \ud t \\
\ud  \DrFcn{i}{t} &= \frac{ 2 \ud t}{ \DrFcn{i}{t} - \DrFcn{j}{t} }, \qquad i \ne j.
\end{cases}
\end{align}
Here $\PartF$ is the partition function of the local multiple SLE, and two partition functions that are not constant multiples of each others will yield different multiple SLE measures. From basic SDE theory, the driving function stopped at $T_j$ is a measurable random variable in the topology of Section~\ref{subsubsec: space of fcns}.

\subsubsection*{\textbf{Remarks}} First, by absolute continuity with respect to the chordal SLE~(see~\cite{Dubedat-commutation} or~\cite[Section~A.3]{KP-pure_partition_functions_of_multiple_SLEs}), the one-curve marginals up to the exit time of $U_j$ enjoy many good properties of the chordal SLE. For instance, for the model in $\bH$, the local multiple SLE hulls are indeed curves~\cite{RS-basic_properties_of_SLE} and share the same fractal dimension depending on $\kappa$~\cite{Beffara-dimension}. Likewise, for $\kappa < 8$, their conformal images in a bounded domain $(\domain; p_1, \ldots, p_{2N})$, with marked prime ends where radial limits exist (see Section~\ref{subsubec: rad cont ext of conf maps}), are curves and measurable random variables the topology of Section~\ref{subsubsec: space of curves}; see~\cite[Proposition~5.2]{mie}.

Second, if there are $N=1$ curves, condition~\eqref{eq: COV for multiple SLEs} for scalings $\mu$ alone determines a solution to~\eqref{eq: PDE for multiple SLEs}, unique up to scaling, namely $\PartF(x_1, x_2) \propto (x_2 - x_1)^{-2h}$. Then, the growth process~\eqref{eq: SDE definition of N-SLE} coincides with the chordal $\SLE(\kappa)$ from $x_1$ to $x_2$, appearing in, e.g.,~\cite[Lemma 3]{Dubedat-commutation}. For general $N$, the variables $W^{(i)}_t$ , $i \ne j$ are the conformal images of the boundary points, $W^{(i)}_t = g_t (x_i)$, and $ W^{(j)}_t$ is the conformal image of the tip of the growing curve at time $t$.


%

\subsubsection{\textbf{Curve collections in $\UnitD$}}
\label{subsubsec: loc N-SLE in D}

Let us finally address the local multiple $\SLE(\kappa)$ as a collection of curves. Due to using the topology of compact curves in this paper (see Section~\ref{subsubsec: space of curves}) we will now use the unit disc $\UnitD$ as our reference domain in stead of $\bH$.
Thus, we consider the domain $(\UnitD; \Unitp_1, \ldots, \Unitp_{2N})$ with localization neighbourhoods $U_1, \ldots, U_{2N}$. We fix a point $\Unitp_\infty \in \bdry \UnitD$ on the counterclockwise arc of $\bdry \UnitD$ from $\Unitp_{2N} $ to $\Unitp_1$, and a conformal map $\confmapDH$ taking $(\UnitD; \Unitp_\infty)$ to $(\bH, \infty)$. (Hence $- \infty < \confmapDH(\Unitp_1) < \ldots < \confmapDH(\Unitp_{2N}) < + \infty$.) 

The local multiple SLE in this setup is a collection of curve initial segments, defined via the regular conditional laws of the driving function of the $j$:th initial segment, conditional on the initial segments $1, 2, \ldots, (j-1)$, for each $j$. (For basics of regular conditional laws, see Appendix~\ref{sec: abstract nonsense}.) Denote by $\InitSegm{j}$ the $j$:th initial segment, up to the hitting time $T_j$ of $\UnitD \setminus U_j$. Given a partition function $\PartF$, the local multiple SLE is now defined as follows.

\begin{defn}
Given the previous initial segments $\InitSegm{1}, \ldots, \InitSegm{{j-1}}$ and a conformal map $\confmapDH_j$ (where $\confmapDH_0 = \confmapDH$) from the connected component of $\UnitD \setminus ( \InitSegm{1} \cup \ldots \cup \InitSegm{j-1} )$ adjacent to the marked boundary points $\InitSegm{1}(T_1), \ldots, \InitSegm{j-1}(T_{j-1}), p_j, \ldots, p_{2N}$, the regular conditional law of the $j$:th initial segment $\InitSegm{j}$, is the conformal image under $\confmapDH_{j-1}^{-1}$ of the curve given by multiple SLE one-curve marginal~\eqref{eq: SDE definition of N-SLE} in the localization neighbourhood $\confmapDH_{j-1} (U_j)$ of $\bH$. The map $\confmapDH_j$ is $g_{T_j} \circ \confmapDH_{j-1}$, where $g_{T_j}$ is given by~\eqref{eq: loc NSLE Loewner eq} and $T_j$ is the exit time of $\confmapDH_{j-1} (U_j)$ by the growth process in $\bH$.
\end{defn}

\subsubsection*{\textbf{Important remarks}}

The definition above does not depend on the choice of the reference point $\Unitp_\infty$ and the conformal map $\confmapDH$ due to the conformal invariance of the local multiple SLE initial segments in $\bH$.

Even if in the above definition, the initial segments are sampled in the order from $1$ to $2N$, any order of sampling will produces the same law of the curves, see~\cite[Sampling procedure~A.3]{KP-pure_partition_functions_of_multiple_SLEs}.

An alternative and perhaps more fundamental way to state the definition above would be in terms of the regular conditional laws of the driving functions of $\confmapDH_{j-1} ( \InitSegm{j} )$ being given by~\eqref{eq: SDE definition of N-SLE}.
To see the equivalence, first by~\cite[Proposition~5.2]{mie} (stated for chordal SLEs, holding for multiple SLEs by absolute continuity) and Corollary~\ref{cor: coordinate change is a bijection}, the driving functions of $\InitSegm{j} $ and the curves $\InitSegm{j} $ are measurable functions of each other. Then, as discussed in Section~\ref{subsec: precompactness for local NSLE} of this paper, specifically commutative diagram~\eqref{dia: local multiple SLE second commutation}, the collections of driving functions of $\InitSegm{j} $ and $\confmapDH_{j-1} ( \InitSegm{j} )$ are measurable functions of each other. By these two-way measurabilities, one can see the equivalence of conditional-law descriptions via curves or driving functions. We have nevertheless chosen to postpone the further treatment in terms of driving functions to later sections to keep the notation minimal in this introductory section.

\subsubsection{\textbf{General domains}}

For a general bounded simply-connected domain $\domain$ with marked prime ends $p_1, \ldots, p_{2N}$ where radial limits exist (see Section~\ref{subsubec: rad cont ext of conf maps}) and their localization neighbourhoods $U_1, \ldots, U_{2N}$, the local multiple SLE is the conformal image of a local multiple SLE in $\UnitD$, with the boundary points and localization neighbourhoods chosen according to the conformal images.

\subsubsection{\textbf{Continuous stopping times}}
\label{subsubsec: cts stopping times}

The (capacity at the) hitting time $T_j$, introduced in the previous paragraph, is not continuous in any of the topologies that we will impose on curves.
This is illustrated in Figure~\ref{fig: discont st} in Appendix~\ref{app: continuous stopping times}. Because the main focus of this article is on weak convergence results, we will often have to use the continuous modifications $\tau_j$ of the hitting times $T_j$. These stopping times are introduced in more detail in Appendix~\ref{app: continuous stopping times} --- for a busy reader it suffices to us to know that they are conformally invariant and satisfy $\tau_j > T_j$. It then follows from the ``local commutation property'' of~\cite{Dubedat-commutation} that if a collection of initial segments $\InitSegm{1}([0, \tau_1]), \ldots, \InitSegm{j-1}([0, \tau_{j-1}])$ satisfies the regular conditional law property of the previous subsection, then also the shorter intial segments $\InitSegm{1}([0, T_1]), \ldots, \InitSegm{j-1}([0, T_{j-1}])$ satisfy the same property. Thus, treating continuous stopping times should be regarded merely as a technicality arising from weak convergence.

\subsection{Local-to-global multiple SLE}
\label{subsec: loc-2-glob}

We now define the local-to-global multiple SLE, which is a natural domain Markov extension of the local multiple SLEs in the preceding subsection, and the main object of interest in this paper.

\subsubsection{\textbf{Unconditional and conditional random curve models}}

A \emph{link pattern} of $N$ links is a partition of $\{1, 2, \ldots, 2N \}$ into $N$ disjoint pairs $\{ \{ a_1, b_1 \}, \ldots, \{ a_N, b_N \} \}$, called \emph{links}, such that the real-line points $ a_i $ and $b_i$, for all $1 \le i  \le N$, can be connected by pairwise disjoint curves in the upper half-plane. The set of all link patterns of $N$ links is denoted by $\LP_N$. We use link patterns to encode in which way some chordal curves pair $2N$ marked boundary points of a simply-connected domain. Note also that due to parity reasons, every link of a link pattern must contain one odd and one even boundary point.

We will define separately local-to-global multiple SLEs and conditional local-to-global multiple SLEs. The unconditional versions arise as scaling limits of interface models when no condition is imposed on the link pattern formed by the interfaces in the corresponding lattice model. Similarly, the conditional version will be scaling limits of $N$ interfaces conditional on each particular link pattern $\alpha \in \LP_N$.

\subsubsection{\textbf{The definitions}}

Let us begin with the unconditional version of the local-to-global multiple SLE. Similarly to local multiple SLE, the definition relies on conformal invariance, Loewner growth processes for suitable initial segments, and regular conditional laws given the initial segments. A new ingredient is induction the number $N$ of curves.

Let $(\domain; p_1, \ldots, p_{2N})$ be a bounded simply-connected planar domain  with $2N$ marked prime ends with radial limits (indexed counterclockwise). Suppose that we are given a family of local multiple $\SLE(\kappa)$ partition functions $\PartF_N$, for $N$ up to some value (possibly all $N \in \N$). We define the \emph{local-to-global multiple $\SLE(\kappa)$} as the following random curves.

\begin{itemize}
\item[1)] (Induction.) If $N=1$, we define the symmetric multiple $\SLE(\kappa)$ to be the usual chordal $\SLE(\kappa)$ on $(\domain; p_1, p_2)$. Assume now that the unconditional multiple $\SLE(\kappa)$ with partition functions $\PartF_N$, is defined for $k$ curves (in any bounded domain with degenerate prime ends), for all $1 \le k \le N-1$, and define it for $N$ curves as follows.
\item[2)] (Conformal invariance.)  Let $\confmap: \domain \to \UnitD$ be a conformal map taking our domain of interest $(\domain; p_1, \ldots, p_{2N})$ to $(\UnitD; \Unitp_1, \ldots, \Unitp_{2N})$. We will define the symmetric mutliple SLEs below in $(\UnitD; \Unitp_1, \ldots, \Unitp_{2N})$ as random curves\footnote{In our topology of random curves, the curves have a direction. We thus choose the convention that every curve flows from odd to even boundary point, $\RandCurveD{1}$ starting from $\Unitp_1$, $\RandCurveD{2}$ from $\Unitp_3$, etc.} $(\RandCurveD{1}, \ldots, \RandCurveD{N})$ and then in $(\domain; p_1, \ldots, p_{2N})$ as the conformal image curves $(\RandCurve{1}, \ldots, \RandCurve{N}) = (\confmap^{-1} (\RandCurveD{1} ), \ldots, \confmap^{-1} ( \RandCurveD{N} ) )$. 
\item[3)] (Initial segments; Figure~\ref{fig: init segments}.) Denote by $\InitSegmDelta{\delta}$ the initial segments of the random curve $\RandCurveD{1}$ in $\UnitD$ starting from the boundary point $\Unitp_1$, until the continuous modification of the hitting time of the $\delta$-neighbourhood of the boundary arc $(\Unitp_2 \Unitp_{2N})$. For all $\delta > 0$, $\InitSegmDelta{\delta}$ is described by the local multiple SLE growth process~\eqref{eq: SDE definition of N-SLE} (with the partition function $\PartF_N$). As $\delta \shrinkto 0$, $\InitSegmDelta{\delta}$ almost surely tend to a closed curve $\InitSegmDelta{0}$ from $\Unitp_1$ to $(\Unitp_2 \Unitp_{2N})$.
\item[4a)] (Conditional laws for $\kappa \in (0, 4]$; Figure~\ref{fig: kappa < 4 sampling}.) The initial segment $\InitSegmDelta{0}$ will almost surely hit the arc $(\Unitp_2 \Unitp_{2N})$ at some even-index marked boundary point, and forms one full random curve, $\InitSegmDelta{0} = \RandCurveD{1}$ (Figure~\ref{fig: kappa < 4 sampling}(left)). The regular conditional distribution of the remaining curves  $\RandCurveD{2}, \ldots, \RandCurveD{N}$ are two independent local-to-global multiple $\SLE(\kappa)$:s in the relevant connected components  of $\UnitD \setminus \RandCurveD{1}$ and with the relevant marked boundary points~(the brown and green domains, curves, and boundary points in Figure~\ref{fig: kappa < 4 sampling}(right)).
\item[4b)] (Conditional laws for $ \kappa \in (4, 8)$; Figure~\ref{fig: kappa > 4 sampling}.) The initial segment $\InitSegmDelta{0}$ will almost surely not hit the arc $(\Unitp_2 \Unitp_{2N})$ at any of the marked boundary points $\Unitp_2, \ldots, \Unitp_{2N}$, and thus $\UnitD \setminus \InitSegmDelta{0} $ has two connected components adjacent to the remaining boundary points $\Unitp_2, \ldots, \Unitp_{2N}$: one with an even and one with an odd number of them; see Figure~\ref{fig: kappa > 4 sampling}(left).
Declare the tip of the initial segment $\InitSegmDelta{0}$ as a new marked boundary point in the ``odd'' component~(in brown in Figure~\ref{fig: kappa > 4 sampling}(right)), so that both components now have an even number of boundary points. The regular conditional distribution of the remainder of the curves $\RandCurveD{1}, \ldots, \RandCurveD{N}$ are two independent local-to-global multiple $\SLE(\kappa)$:s in the relevant connected components of $\UnitD \setminus \RandCurveD{1}$ and with the relevant marked boundary points~(the brown and green domains, curves, and boundary points in Figure~\ref{fig: kappa < 4 sampling}(right)).
\end{itemize}

\begin{figure}
\includegraphics[width=0.4\textwidth]{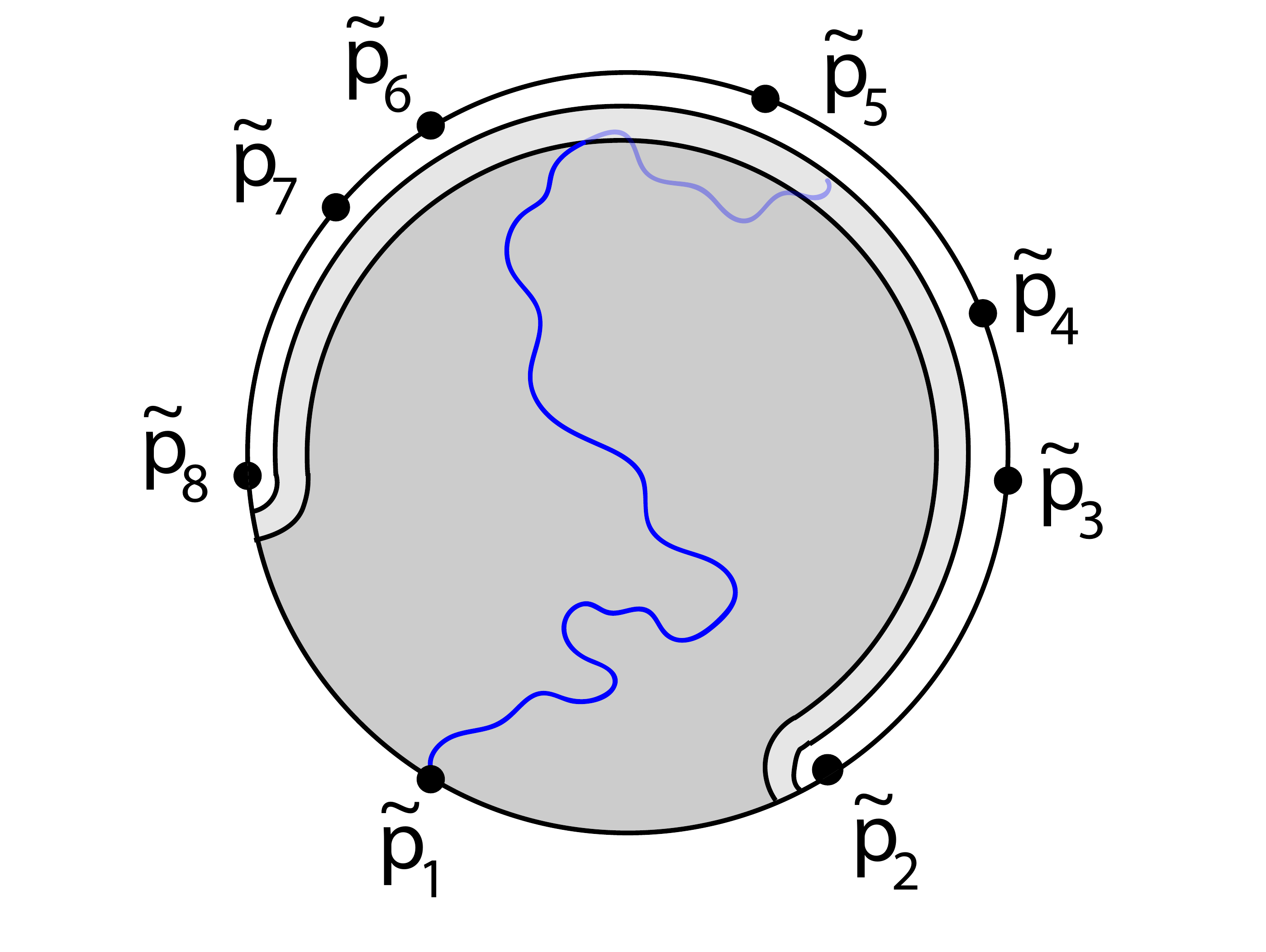}
\caption{\label{fig: init segments}
A schematic illustration of step (3) in the definition of the {local-to-global multiple SLE.}
}
\end{figure}

\begin{figure}
\includegraphics[width=0.4\textwidth]{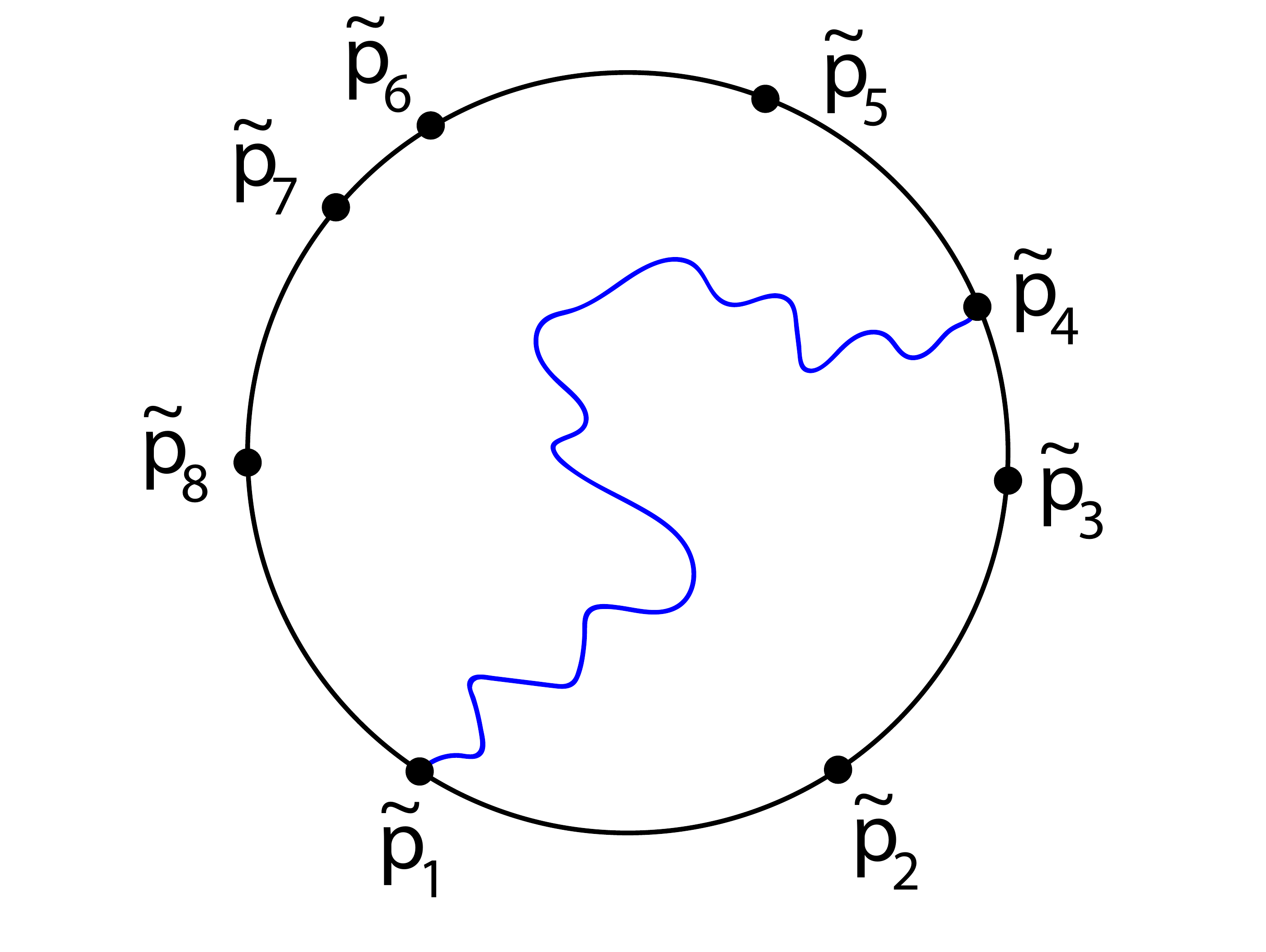} \quad
\includegraphics[width=0.4\textwidth]{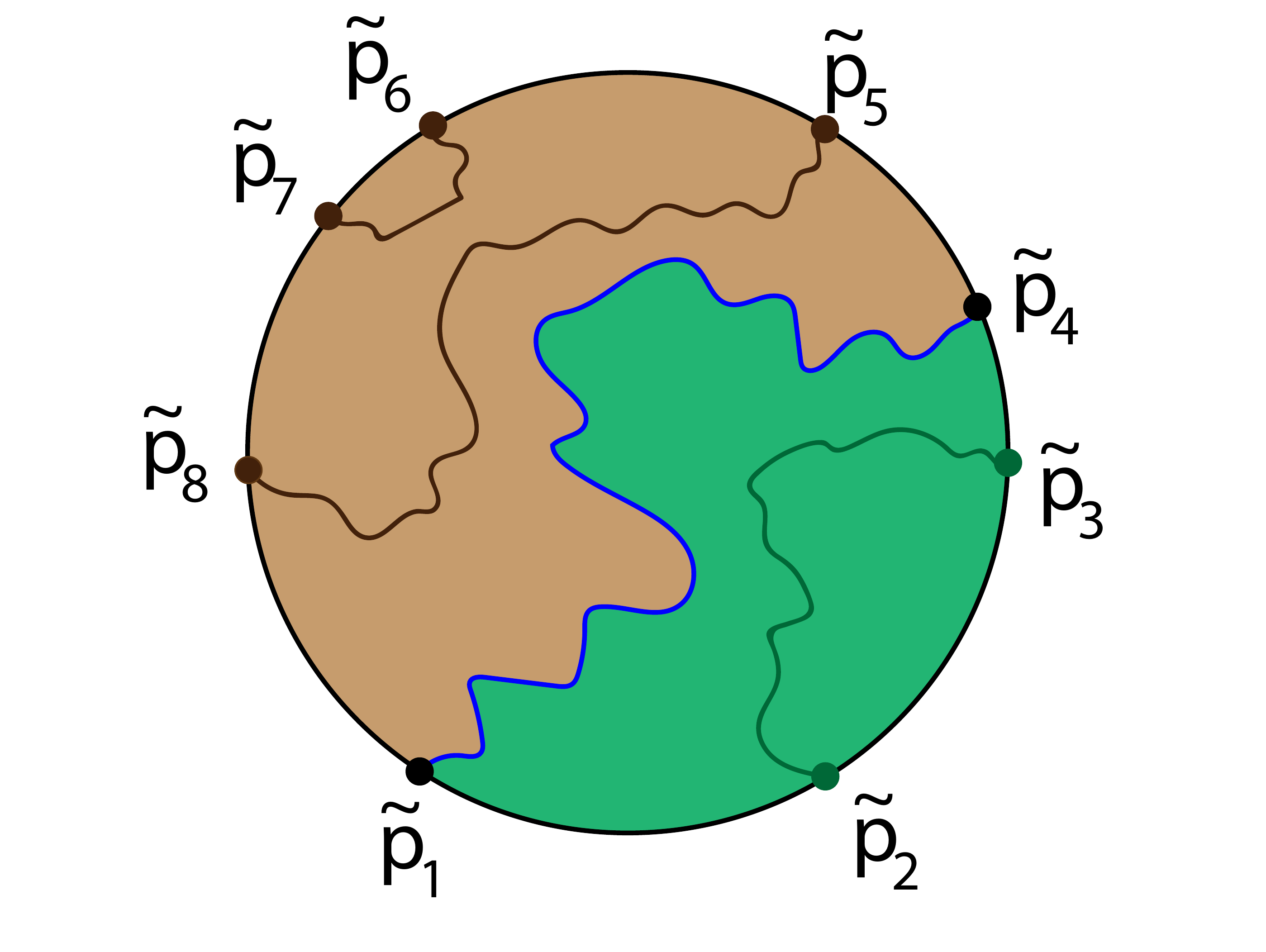} 
\caption{\label{fig: kappa < 4 sampling}
A schematic illustration of step (4a) in the definition of the local-to-global multiple SLE for $\kappa \in (0, 4]$.
}
\end{figure}

\begin{figure}
\includegraphics[width=0.4\textwidth]{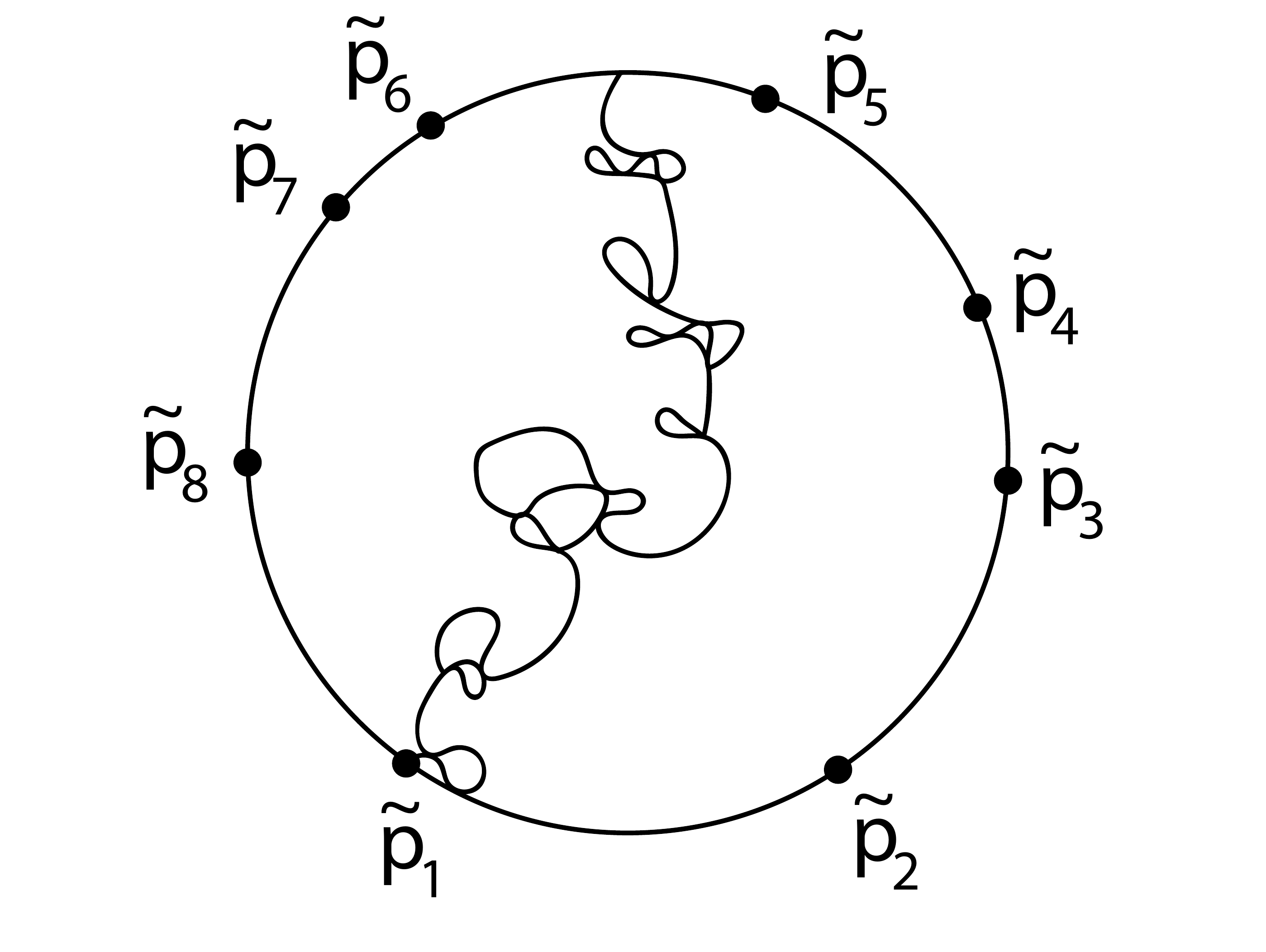} \quad
\includegraphics[width=0.4\textwidth]{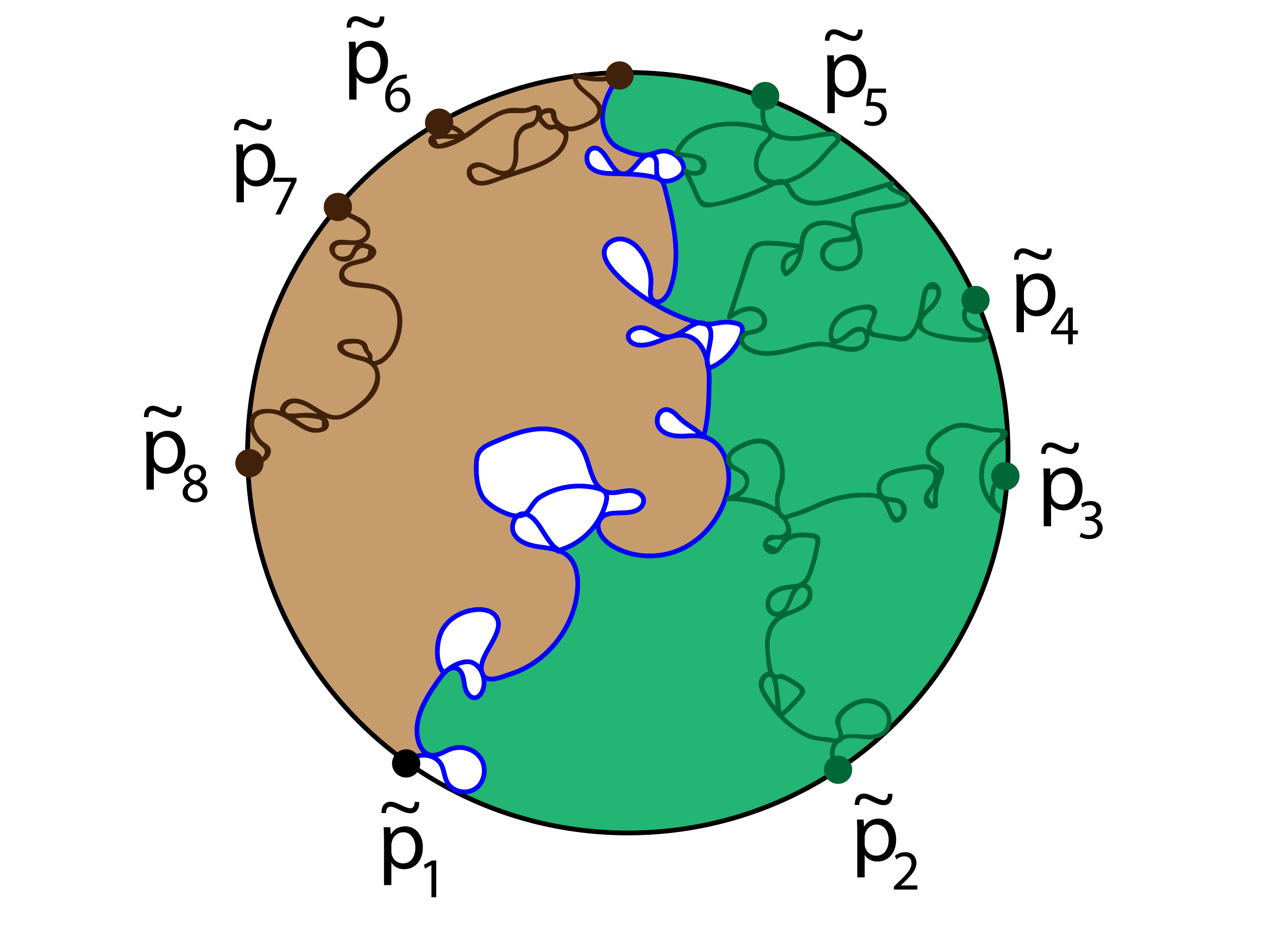}
\caption{\label{fig: kappa > 4 sampling}
A schematic illustration of step (4b) in the definition of the local-to-global multiple SLE for $\kappa \in (4, 8)$.
}
\end{figure}

For $0 < \kappa \le 4$, we will also consider conditional lattice models. The conditional local-to-global multiple SLE is defined almost identically, except that the collection of partition functions $\PartF_\alpha$ is now indexed by $N$ and link patterns $\alpha \in \LP_N$. Only step (4) is slightly modified:

\begin{itemize}
\item[4a')] The initial segment $\InitSegmDelta{0}$ will almost surely hit the arc $(\Unitp_2 \Unitp_{2N})$ at the even-index marked boundary point linked to $1$ in the link pattern $\alpha$, and forms one full random curve, $\InitSegmDelta{0} = \RandCurveD{1}$. The regular conditional distribution of the remaining curves  $\RandCurveD{2}, \ldots, \RandCurveD{N}$ are two independent conditional $\SLE(\kappa)$:s in the relevant connected components of $\UnitD \setminus \RandCurveD{1}$, with the relevant marked boundary points, and conditional on the relevant link patterns.
\end{itemize}

\subsubsection*{\textbf{Important remarks}}


%

The existence of curves given by the definition above is not immediate, and we will rely on realizing them as scaling limits. Some non-trivial obstacles that we will take care of are:
\begin{itemize}
\item The operation of $\confmap^{-1}$ maps curves to curves and the obtained collection of curves $(\RandCurveD{1}, \ldots, \RandCurveD{2N})$ is a measurable random variable (cf.~\cite{mie}).
\item The initial segment $\InitSegmDelta{0}$ exists as a closed curve up to and \emph{including} its end point (cf.~\cite{PW18}).
\item Being regular conditional laws requires some regularity properties from the local-to-global multiple SLEs.
\end{itemize}

The above regular conditional laws can be interpreted as a sampling procedure, and above we always sampled the initial segment of $\RandCurveD{1}$ first. This is only for definiteness; one can start by sampling the initial or final segment adjacent to a boundary point of choice. This will be proven basing on the undelying discrete models.

In step (4) we did not specify how the boundary points in the two connected components are re-labelled (although parities must be preserved). This need not be done due to the properties that we will require of the underlying lattice models.

Only very specific multiple SLE partition functions should yield the initial segments of the above type. These are called symmetric and pure partition functions in~\cite{KP-pure_partition_functions_of_multiple_SLEs, PW, Wu17}, corresponding to the unconditional and conditional cases, respectively, and finding such functions is not always easy. Nevertheless, in the strategy of this paper, the partition functions will be found as a step in the limit identification.


\bigskip{}

\section{Preliminaries}
\label{sec: preli}

This section introduce some notations, definitions, and concepts that are used throughout the paper.

\subsection{Lattice models}

\subsubsection{\textbf{Discrete random curve models}}

We start with the setup and notation that we refer to as \emph{discrete random curve models}. Various examples will be given in Section~\ref{sec: application examples}.
\begin{itemize}
\item $\InfiniteGr = (\Vert(\InfiniteGr) , \Edg(\InfiniteGr) )$ is a (possibly infinite) connected planar graph with fixed planar embedding, such as $\Z^2$, or an isoradial graph as in~\cite{CS-discrete_complex_analysis_on_isoradial}, or a more general graph as in~\cite{Chelkak-Robust_DCA_toolbox}. We call $\InfiniteGr$ a \emph{lattice}.
\item $\domain_\Gr$ is a bounded simply-connected planar domain, whose boundary consists of edges and vertices in $\InfiniteGr$.
\item $\Gr=(\Vert, \Edg)$ is the following graph: its vertices $\Vert$ consist of interior vertices $\Vert^\circ = \Vert(\InfiniteGr) \cap \domain$ and boundary vertices $\Vert^\bdry = \Vert(\InfiniteGr) \cap \bdry \domain .$ Its edges $\Edg$ consist of $\Edg = \Edg(\InfiniteGr) \cap \overline{\domain_\Gr}$. We call boundary edges $\bdry \Edg \subset \Edg$ the edges that connect $\Vert^\circ$ to $\Vert^\bdry$, and edges running between interior vertices are called interior edges $\Edg^\circ$. We call $\Gr$ the \textit{simply-connected subgraph} of $\InfiniteGr$ corresponding to $\domain_\Gr$.
\item Let $\Gr$ be as above $e_1, \ldots, e_{2N} \in \bdry \Edg$ be $2N$ distinct boundary edges, indexed counterclockwise along $\bdry \domain_\Gr$. A \emph{measure with random curves} on $(\Gr; e_1,  \ldots, e_{2N})$ is a pair $(\PR^{(\Gr; e_1,  \ldots, e_{2N})}, (\gamma_{\Gr; 1}, \ldots, \gamma_{\Gr; N}))$ of a probability measure and a measurable random variable $(\gamma_{\Gr; 1}, \ldots, \gamma_{\Gr; N})$,
supported on $N$-tuples of \nolinebreak{(vertex-)disjoint} simple paths on the graph $\Gr$, pairing the boundary edges $e_1, \ldots, e_{2N}$, and otherwise running in the interior edges $\Edg^\circ$. We choose here the convention that every path travels from odd to even boundary boundary edge, $\RandCurveD{1}$ starting from $e_1$, $\RandCurveD{2}$ from $e_3$, etc.
\item A \emph{discrete (random) curve model} on $\InfiniteGr$ is a collection of measures with random curves \newline
$(\PR^{(\Gr; e_1,  \ldots, e_{2N})}, (\gamma_{\Gr; 1}, \ldots, \gamma_{\Gr; N}))$. This collection is indexed by some positive integers $N$ and some simply-connected subgraphs $(\Gr; e_1,  \ldots, e_{2N})$ of the lattice $\InfiniteGr$ with $2N$ distinct marked boundary edges. 
\item Given a discrete random curve model, we define a \emph{discrete conditional (random) curve model} as a collection measures
\begin{align*}
\PR_{\alpha}^{(\Gr; e_1,  \ldots, e_{2N})} [\; \cdot \; ] = \PR^{(\Gr; e_1,  \ldots, e_{2N})} [ \; \cdot \; \vert \text{curves form the link pattern } \alpha ] = \PR^{(\Gr; e_1,  \ldots, e_{2N})} [ \; \cdot \; \vert  \alpha ]
\end{align*}
with the associated random curves.
This collection is indexed by $N$, graphs $(\Gr; e_1,  \ldots, e_{2N})$, and link patterns $\alpha \in \LP_N$ such that the measures $\PR^{(\Gr; e_1,  \ldots, e_{2N})}$ and the conditioning on $\alpha$ make sense.
\end{itemize}


\subsubsection{\textbf{Discrete domain Markov property}}

A key property of the discrete random curve models will be the discrete domain Markov property (DDMP). To be able to define the DDMP with a reasonably light notation, we will make one more assumption on the discrete random curve models. Namely, 
we will assume that if $(\PR^{(\Gr; e_1,  \ldots, e_{2N})}, (\gamma_1, \ldots, \gamma_N))$ is one measure with random curves in our discrete random curve model, and if we denote $( \hat{e}_1,  \ldots, \hat{e}_{2N} ) = ( e_3,  e_4, \ldots,  e_{2N}, e_1,  e_2)$, then also $ ( \PR^{(\Gr; \hat{e}_1,  \ldots, \hat{e}_{2N} )}, (\hat{\gamma}_1, \ldots, \hat{\gamma}_N ))$ is a pair in that model, and have the equality in distribution
\begin{align*}
(\hat{\gamma}_1, \ldots, \hat{\gamma}_N ) \eqd (\gamma_2, \ldots, \gamma_N , \gamma_1 ).
\end{align*}
If this is satisfied, we say that a symmetric random curve model has \emph{alternating boundary conditions}. In other words, we may re-label the edges from $e_1,  \ldots, e_{2N}$ to $\hat{e}_1,  \ldots, \hat{e}_{2N}$, using any cyclic permutation that labels the edges counterclockwise and \emph{preserves the parities}, and our random curve model yields the same random curves on $(\Gr; e_1,  \ldots, e_{2N})$ and $(\Gr; \hat{e}_1,  \ldots, \hat{e}_{2N})$; there are hence two kinds of boundary segments between the marked edges.

 
Informally speaking, a discrete random curve model satisfies the DDMP if conditioning the model on $\Gr$ on an initial segment or a full curve is equivalent to reducing the graph $\Gr$ by that initial segment or curve. Formally, let the collection of pairs $(\PR^{(\Gr; e_1,  \ldots, e_{2N})}, (\gamma_1, \ldots, \gamma_N ))$ be a discrete random curve model  on $\InfiniteGr$ with alternating boundary conditions. We say that it satisfies the DDMP the following hold:
\begin{itemize}
\item[i)] Consider the curves given by the random curve model on $(\Gr; e_1, \ldots, e_{2N})$,
 and for some $ 1 \le j \le 2N$, condition them 
 on a sequence of vertices $\lambda(0), \lambda(1), \ldots, \lambda(t)$, appearing in this order as the first vertices along the random curve adjacent to $e_j = \edgeof{\lambda(0)}{\lambda(1)}$, started from that edge. Denote by $\Gr_{t}$ the graph corresponding to the simply-connected domain $\domain_{t}$, whose boundary consists of $\bdry \domain$ and the graph path $\lambda(0), \lambda(1), \ldots, \lambda(t-1)$.
Now, on this condition, the remaining parts of the curves given by the random curve model on $(\Gr; e_1, e_2, \ldots, e_{2N})$ are in distribution equal to the random curve model on $(\Gr_{t}; e_1, \ldots,  e_{j-1}, \edgeof{\lambda(t-1)}{\lambda(t)}, e_{j+1}, \ldots, e_{2N})$.
\item[ii)] Condition the curves given by the random curve model on $(\Gr; e_1, e_2, \ldots, e_{2N})$, where $N \ge 2$, on the full random curve $\gamma_{\Gr; *}$ that reaches the boundary via $e_j$. The simple curve $\gamma_{\Gr; *}$ divides the domain $\domain_\Gr$ into two simply-connected domains $D_L$ and $D_R$. Let $\Gr_L$ and $\Gr_R$, be the corresponding simply-connected subgraphs of $\InfiniteGr$.
 Then, given the curve $\gamma_{\Gr; *}$, the conditional law  the remaining random curves is the following: the curves in $\Gr_L$ are independent of those in $\Gr_L$, the marginal law of the curves on each of these, say for definiteness $\Gr_L$, is (up to relabelling the curves) given by the random curve model on $(\Gr_L; \hat{e}_{1}, \ldots, \hat{e}_{2M})$, where $\hat{e}_{1}, \ldots, \hat{e}_{2M}$ are those of the marked boundary edges $e_1, e_2, \ldots, e_{2N}$ left in $\Gr_L$, relabelled countarclockwise in such a way that the parities are preserved.
\end{itemize}

\subsubsection*{\textbf{Remarks}}

As a consequence, analogues of properties~(i) and~(ii) hold for the conditional measures $\PR_{\alpha}^{(\Gr; e_1, e_2, \ldots, e_{2N})}$. The analogue of~(i) is obvious. In the analogue of~(ii), the conditional laws of the curves on $\Gr_L$ and $\Gr_R$ are independent of each other and given by the \emph{conditional} random curve model on $\Gr_L$ and $\Gr_R$, where the link pattern formed by the curves between the boundary edges in $\Gr_L$ and $\Gr_R$ are those inherited from $\alpha$.

Property~(i) can be equivalently stated in terms of stopping times. Then, in stead of some fixed first vertices $\lambda(0), \lambda(1), \ldots, \lambda(t)$, one conditions on the first vertices $\lambda(0), \lambda(1), \ldots, \lambda(\tau)$ up to a stopping time $\tau$ in the filtration $\mathcal{F}_1, \mathcal{F}_2, \ldots$, where $\mathcal{F}_t$ is the sigma-algebra generated by the first $t$ vertices $\lambda(0), \ldots, \lambda(t)$. 

From measures with $N=1$ curves, only property~(i) is required. 

Property~(ii) can be used inductively to deduce the distribution of the remaining curves, 
given any collection of full random curves. In particular, if we condition on all but one curves, so that the remaining one has to stay on the simply-connected subgraph $\tilde{\Gr}$ with marked boundary edges $e_{k_1}, e_{k_2}$, with $k_1$ odd, then the remaining one is described by the random curve model on $(\tilde{\Gr}; e_{k_1}, e_{k_2})$.

\subsection{Approximations of planar domains}

In this subsection, we introduce the concepts related to approximations of planar domains. All notations and convetions are identical to the previous paper of the author~\cite{mie}.

\subsubsection{\textbf{Prime ends}}
\label{subsubsec: prime ends}


We are dealing with simply-connected domains $\domain$ with possibly a very rough boundary. The notion of boundary points must thus be replaced with that of prime ends. Define first the following.
\begin{itemize}
\item A cross cut $S$ is an open Jordan arc in $\domain$ such that $\overline{S} = S \cup \{ a, b \}$, with $a, b \in \bdry \domain$.
\item A sequence $(S_n)_{n \in \N}$ of cross cuts is a null chain if $S_n \cap S_{n+1} = \emptyset$ for all $n$, $S_n$ separates $S_{n+1}$ from $S_0$ in $\domain$ for all $n \ge 1$, and $\diam (S_n) \to 0$ as $n \to \infty$.
\item Null chains $(S_n)_{n \in \N}$ and $(\tilde{S})_{n \in \N}$ are equivalent if for any large enough $m$ there exists $n$ such that $S_m$ separates $\tilde{S}_n$ from $\tilde{S}_0$ and $\tilde{S}_m$ separates ${S}_n$ from ${S}_0$.
\item A \emph{prime end} is an equivalence class of null chains under this equivalence relation.
\end{itemize}
A conformal map $\confmap: \domain \to \UnitD$ induces a bijection $\hat{\confmap}$ between the prime ends of $\domain$ and $\bdry \UnitD$, such that if a null chain $S_n$ determines a prime end $p$ of $\domain$, then $\confmap(S_n)$ is a null chain in $\UnitD$ and determines the prime end $\hat{\confmap} (p) \in \bdry \UnitD$~\cite[Theorem~2.15]{Pommerenke-boundary_behaviour_of_conformal_maps}. In this sense, prime ends are ``the conformal notion of boundary points''.

 \subsubsection{\textbf{Radial limits of conformal maps}} 
 \label{subsubec: rad cont ext of conf maps}


We will extend the conformal map $\confmap^{-1}: \UnitD \to \domain$ to $\bdry \UnitD$ by radial limits whenever they exist:
denote by $P_\eps : \overline{\UnitD} \to {\UnitD}$ the radial projection  on $\overline{ \UnitD }$, 
\begin{align*}
P_\eps (z) = \frac{z}{ \vert z \vert } \min \{ 1- \eps, \vert z \vert \},  
\end{align*}
where $0 < \eps < 1$, and for $z \in \bdry \UnitD$ denote by $\confmap^{-1} (z)$ the limit
\begin{align*}
\lim_{\eps \shrinkto 0} \confmap^{-1} \circ P_\eps (z)
\end{align*}
whenever it exists (by Fatou's theorem, it exists for Lebesgue-almost every $z \in \bdry \UnitD$ if $\domain$ is bounded).

It holds true that the existence and value of such a radial limit of a conformal map $\confmap^{-1}$ at some $z = \hat{\confmap} (p) \in \bdry \UnitD$ only depends on the corresponding prime end $p$ of $\domain$, but not the choice of the conformal map $\confmap: \domain \to \UnitD$~\cite[Corollary~2.17]{Pommerenke-boundary_behaviour_of_conformal_maps}. We thus say that \emph{radial limits exist at $p$} or do not exist at $p$. In particular, radial limits exist at degenerate prime ends $p$. We will often restrict our consideration to prime ends $p$ with radial limits, and we will then with a slight abuse of notation also denote the radial limit point in $\C$ by $p$. 


\subsubsection{\textbf{Carath\'{e}odory convergence of domains}}
\label{subsubsec: Cara convergence}

The notion of domain approximations that we use will be \nolinebreak{Carath\'{e}odory} convergence.
Let $( \domain_n )_{n \in \N}$ and $\domain$ by simply-connected open sets $\domain, \domain_n \subsetneq \C$, all containing a common point $u$. We say that $\domain_n \to \domain$ in the sense of kernel convergence with respect to $u$ if 
\begin{itemize}
\item[i)] every $z \in \domain$ has some neighbourhood $V_z$ such that $V_z \subset \domain_n$ for all large enough $n$; and
\item[ii)] for every point $p \in \bdry \domain$, there exists a sequence $p_n \in \bdry \domain_n$ such that $p_n \to p$.
\end{itemize}
Let $\confmap_n$ be the \emph{Riemann uniformization maps} from $\domain_n$ to $\UnitD$ normalized at $u$, i.e., $\confmap_n(u)=0$ and $\confmap_n'(u) > 0$. Let $\confmap$ be the Riemann uniformization map from $\domain$ to $\UnitD$. The kernel convergence
$\domain_n \to \domain$  with respect to $u$
holds if and only if the inverses $\confmap_n^{-1}$ converge uniformly on compact subsets of $\UnitD$ to $\confmap^{-1}$~\cite[Theorem~1.8]{Pommerenke-boundary_behaviour_of_conformal_maps}. Then, also $\confmap_n \to \confmap$ uniformly on compact subsets of $\domain$.

It is easy to see that if $\domain_n \to \domain$ in the sense of kernel convergence with respect to $u$, then the same convergence holds with respect to any $\tilde{u} \in \domain$, taking the tail of the sequence $\domain_n$ if needed. We then say that $\domain_n \to \domain$ \emph{in the Carath\'{e}odory sense} as $n \to \infty$, or that $\domain_n$ are  \emph{Carath\'{e}odory approximations} of $\domain$. Working from the point of view of uniformization maps, this relates to the following elemntary lemma, whose proof we leave to the reader.

\begin{lem}
\label{lem: Cara iff conv of confmaps}
$\domain_n \to \domain$ in the Carath\'{e}odory sense if and only if there exist some conformal maps $\confmap_n: \domain_n \to \UnitD$ and $\confmap: \domain \to \UnitD$ such that $\confmap_n^{-1}$ converge uniformly on compact subsets of $\UnitD$ to $\confmap^{-1}$
\end{lem}


For domains with marked prime ends, we say $(\domain_n; p_1^{(n)}, \ldots, p_m^{(n)} ) \to (\domain; p_1,  \ldots, p_m )$ in the Carath\'{e}odory sense as $n \to \infty$, if $\domain_n \to \domain$ in the Carath\'{e}odory sense and $(\hat{\confmap}_n (p_1^{(n)}), \ldots, \hat{\confmap}_n (p_m^{(n)}) ) \to (  \hat{\confmap} (p_1), \ldots, \hat{\confmap} (p_m) )$ as $n \to \infty$, where $\hat{\confmap}_n$ and $\hat{\confmap}$ are the induced maps of prime ends.

\subsubsection{\textbf{Close approximations of prime ends with radial limits}}
\label{subsubsec: close approximations}

A Carath\'{e}odory approximation of domains $(\domain_n; p_1^{(n)}, \ldots, p_m^{(n)} ) \to (\domain; p_1,  \ldots, p_m )$ allows wild behaviour of the boundaries at the marked prime ends. We wish to consider compact curves ending at these prime ends. For such compact curves to exist at all, the prime ends $p_1^{(n)}, \ldots, p_m^{(n)}$ and $ p_1,  \ldots, p_m $ must possess radial limits. Furthermore, to avoid bad boundary approximations, we need to restrict to close Carath\'{e}odory approximations. Informally, being a close approximation means that a chordal curve in $\domain_n$ starting from $p_1^{(n)}$ is not forced to wiggle a macroscopic distance to enter into $\domain_n \cap \domain$. This concept was introduced by the author in~\cite{mie}, and we repeat the definition below.

Assume that $(\domain_n; p_1^{(n)}, \ldots, p_m^{(n)} ) \to (\domain; p_1,  \ldots, p_m )$ in the Carath\'{e}odory sense, and that the radial limits exist at the prime ends $p_1^{(n)}, \ldots, p_m^{(n)}$  and $p_1,  \ldots, p_m$ of the respective domains. We say that
$p_1^{(n)}$ are \textit{close} approximations of a prime end $p_1$, if $p_1^{(n)} \to p_1$ as $n \to \infty$ (as points in $\C$), and in addition the following holds:
for any $r > 0$, $r < d (p_1, u)$ (where $u$ denotes the reference point of the approximation $\domain_n \to \domain$), denote by $S_r$ be the connected component of $ \bdry B(p_1, r)$ disconnecting $p_1$ from $u$ in $\domain$ that lies \emph{innermost}, i.e., closest to $p_1$ in $\domain$. Such a component exists by the existence of radial limits at the prime end $p_1$. Let $w_r \in S_r$ be any fixed reference point; the precise choice makes no difference. Now, $p_1^{(n)}$ are close approximations of $p_1$ if for any fixed $0 < r < d (p_1, u)$, taking a large enough $n$, $p_1^{(n)}$ is connected to $w_r$ inside $\domain_n \cap B(p_1, r)$.

\subsection{The different metric spaces}

In this subsection, we introduce the different metric spaces. All notations and convetions are identical to the previous paper of the author~\cite{mie}.

\subsubsection{\textbf{Space of plane curves modulo reparametrization}}
\label{subsubsec: space of curves}

A planar \emph{curve} is a continuous function $\gamma: [0,1] \to \C$. 
Define an equivalence relation $\sim$ on curves: $\gamma \sim \tilde{\gamma}$ if
\begin{align*}
\inf_{\psi} \left\{  \sup_{t \in [0,1]} \vert \gamma(t) - \tilde{\gamma} \circ \psi (t) \vert  \right\} = 0,
\end{align*}
where the infimum is taken over all reparametrizations (continuous increasing bijections) $ \psi: [0,1] \to [0,1]$. The space of curves modulo this equivalence relation is denoted by $X(\C)$.

We equip $X(\C)$ with the following metric. For two curves $\gamma$, $\tilde{\gamma}$ the distance between their equivalence classes $[\gamma]$ and $[\tilde{\gamma} ]$ in this metric is
\begin{align}
\label{eq: metric of curves}
d([\gamma], [\tilde{\gamma}]) = \inf_{\psi} \left\{  \sup_{t \in [0,1]} \vert \gamma(t) - \tilde{\gamma} \circ \psi (t) \vert  \right\},
\end{align}
where the infimum is taken over all reparametrizations $\psi$. The closed subset of  $X(\C)$ consisting of curves that stay in $\clos{\UnitD}$ is denoted by  $X(\clos{\UnitD})$. The  spaces $X (\C )$ and $X (\clos{\UnitD})$ are both complete and separable. 
We will in this paper only study curves $\gamma$ via the space $X(\C)$. As there is thus no danger of confusion, we will denote the equivalence class $[\gamma]$ by $\gamma$ for short.


The space $X(\C)^N$ of collections of $N$ curves modulo reparametrization is equipped with the metric
\begin{align*}
d((\gamma_1, \ldots, \gamma_N), (\tilde{\gamma}_1, \ldots, \tilde{\gamma}_N) ) = \max_{1 \le i \le N} d ({\gamma}_i, \tilde{\gamma}_i ),
\end{align*}
where the distance on the right-hand side is given by~\eqref{eq: metric of curves}. This space is complete and separable, too.

\subsubsection{\textbf{Space of continuous functions}}
\label{subsubsec: space of fcns}

We equip the space $\ctsfcns$ of continuous functions $W_\cdot : \R_{\ge 0} \to \R$ with the metric of uniform convergence over compact subsets
\begin{align}
\label{eq: metric of uniform convergence over compact subsets for functions of reals}
d(W, \tilde{W}) = \sum_{n \in \N} 2^{-n} \min \{ 1, \sup_{t \in [0,n]} \vert \tilde{W}_t - W_t \vert \}.
\end{align}
The space $\ctsfcns$ is then complete and separable.
The space $\ctsfcns^N$ of collections of continuous functions will be equipped with the metric
\begin{align*}
d((W_1, \ldots, W_N),(\tilde{W}_1, \ldots, \tilde{W}_N)) = \max_{1 \le i \le N} d(W_i, \tilde{W}_i),
\end{align*}
where the distance on the right-hand side is given by~\eqref{eq: metric of uniform convergence over compact subsets for functions of reals}. This space is complete and separable, too.

\section{Precompactness theorems}
\label{sec: precompactness}

\subsection{The main precompactness theorem}

\subsubsection{\textbf{Setup and notation}}
\label{subsec: setup and notation}

The setup and notation for the main theorem~\ref{thm: precompactness thm multiple curves} of this subsection is the following.

Let $(\InfiniteGr_n)_{n \in \N}$ be a sequence of lattices, and $(\Gr_n; e_1^{(n)}, \ldots, e_{2N}^{(n)})$, for each $n$, simply-connected subgraphs of $\InfiniteGr_n$, with $N$ giving the number of boundary edges fixed. Assume that for each $n$, we have a discrete random curve model on $\InfiniteGr_n$, defined on the subgraph $(\Gr_n; e_1^{(n)}, \ldots, e_{2N}^{(n)})$. Denote the measures with random curves on $(\Gr_n; e_1^{(n)}, \ldots, e_{2N}^{(n)})$ by $(\PR^{(\Gr_n; e_1^{(n)}, \ldots, e_{2N}^{(n)}) }, ( \gamma_{\Gr_n; 1}, \ldots, \gamma_{\Gr_n; 2N}) ) = (\PR^{(n)}, (\gamma_1^{(n)}, \ldots, \gamma_{2N}^{(n)}) )$.

Let $\domain_n = \domain_{\Gr_n}$ be the simply-connected domains corresponding to $\Gr_n$, and let $p_1^{(n)}, \ldots, p_{2N}^{(n)}$ be the prime ends of $\domain_n$ where the edges $e_1^{(n)}, \ldots, e_{2N}^{(n)}$, respectively, land. Assume that $(\domain_n; p_1^{(n)}, \ldots, p_{2N}^{(n)})$ are close Carath\'{e}odory approximations of a domain $(\domain; p_1, \ldots, p_{2N})$ with marked prime ends where radial limits exist. (The limiting prime ends need not be distinct for the statement and proof of Theorem~\ref{thm: precompactness thm multiple curves}, but they will be in all the applications in this paper.) Let $\confmap_n: \domain_n \to \UnitD$ and $\confmap: \domain \to \UnitD$ be any conformal maps such that $\confmap_n^{-1} \to \confmap^{-1}$ uniformly over compact subsets of $\UnitD$. Denote $(\gamma_{\UnitD; 1}^{(n)}, \ldots, \gamma_{\UnitD; 2N}^{(n)}) = (\confmap_n(\gamma_1^{(n)}), \ldots, \confmap_n(\gamma_{2N}^{(n)})) \in X(\overline{\UnitD})^N$. We also assume that $\domain_n$ are uniformly bounded.

Fix a point $\Unitp_\infty \in \bdry \UnitD$ on the counterclockwise arc of $\bdry \UnitD$ from $\Unitp_{2N} = \confmap(p_{2N})$ to $\Unitp_1 = \confmap(p_1)$, and a conformal map $\confmapDH$ taking $(\UnitD; \Unitp_\infty)$ to $(\bH, \infty)$. Let $U_j$ be a localizations neighbourhoods  of $\Unitp_j$ in $\UnitD$, for each $j$, i.e., $U_j $ only contain the marked boundary point $\Unitp_j$ and $\confmapDH (U_j)$ are compact $\bH$-hulls. (The neighbourhoods need not be disjoint.) For each $1 \le j \le 2N$, denote by $\InitSegmLatt{n}{j}$ the initial segment of the one of the random curves $\gamma_{\UnitD; 1}^{(n)}, \ldots, \gamma_{\UnitD; 2N}^{(n)}$ adjacent to $p_j^{(n)}$, started from that boundary point, up to the continuous modification $\tau^{(n)}_j$ of the hitting time $T^{(n)}_j$ of $\overline{(\UnitD \setminus U_j )}$ by the curve. Let $\DrFcnLattNotime{n}{j}$  be the driving function of the curve $\confmapDH( \InitSegmLatt{n}{j} )$ in $\overline{\bH}$, stopped at the half-plane capacity corresponding to $\tau_j^{(n)}$.

Conditional discrete random curve models are studied in an identical notation, with the only difference that also the link pattern $\alpha \in \LP_N$ is fixed, and $\PR^{(n)}$ then denotes the corresponding conditional measures, $\PR^{(n)} [\cdot] = \PR^{(\Gr_n, e_1^{(n)}, \ldots, e_{2N}^{(n)})}[\cdot \vert \alpha]$.

\subsubsection{\textbf{Statement of the theorem}}

We now state the main theorem of this section, giving the precompactness results needed for convergence proofs to local-to-global multiple SLEs. Analogues of this result for models with only one curve have been given in~\cite{KS, mie}, see also~\cite{AB-regularity_and_dim_bounds_for_random_curves, Wu17}.

\begin{thm}
\label{thm: precompactness thm multiple curves}
Consider the setup and notation of Section~\ref{subsec: setup and notation}, for discrete random curve models (resp. conditional discrete random curve models). Assume that the discrete curve models on $\InfiniteGr_n$ (resp. on which we impose the conditioning) have alternating boundary conditions and satisfy the DDMP. Assume in addition that the collection of one-curve measures $\PR^{(\Gr, e_1, e_2)}$ in these random curve models, indexed by $n$ and simply-connected subgraphs $\Gr$ of $\InfiniteGr_n$, satisfy the equivalent conditions (C) and (G), as defined below. Then, the following hold: \\
\textbf{A)} The measures $\PR^{(n)}$ are precompact in the following senses:
\begin{itemize}
\item[i)] as laws of the collections of curves $(\gamma_1^{(n)}, \ldots, \gamma_{N}^{(n)})$ on the space $X(\C)^N$;
\item[ii)] as laws of the collections of curves $(\gamma_{\UnitD; 1}^{(n)}, \ldots, \gamma_{\UnitD; N}^{(n)})$ on the space $X(\overline{\UnitD})^N \subset X(\C)^N$;
\item[iii)] as laws of the curves $\InitSegmLatt{n}{j}$ on the space $X(\overline{\UnitD}) \subset X(\C)$, for any $j$; and
\item[iv)] as laws of the driving functions $\DrFcnLattNotime{n}{j}$ on the space $\ctsfcns$, for any $j$.
\end{itemize}
In other words, there exist subsequences $(n_k)_{k \in \N}$ such that the random objects above converge weakly. \\
\textbf{B)} For a subsequence $(n_k)_{k \in \N}$, a weak convergence takes place in topology~(i), $(\gamma_1^{(n_k)}, \ldots, \gamma_{N}^{(n_k)}) \to (\gamma_1, \ldots, \gamma_{N})$, if and only if it takes place in topology~(ii),  $(\gamma_{\UnitD; 1}^{(n_k)}, \ldots, \gamma_{\UnitD; N}^{(n_k)}) \to (\gamma_{\UnitD; 1}, \ldots, \gamma_{\UnitD; N})$. Furthermore, we then have the equality
\begin{align*}
(\gamma_1, \ldots, \gamma_{N}) \eqd  (\confmap^{-1} ( \gamma_{\UnitD; 1} ), \ldots, \confmap^{-1} ( \gamma_{\UnitD; N}) )
\end{align*}
in distribution\footnote{
More precisely, the random variable $(\confmap^{-1} ( \gamma_{\UnitD; 1} ), \ldots, \confmap^{-1} ( \gamma_{\UnitD; N}) )$ in $X (\C)^N$ denotes the following: the map $\confmap^{-1}$, as extended by radial limits to $\bdry \UnitD$ whenever possible, is almost surely defined on all points of the curves $\gamma_{\UnitD; i} $, $1 \le i \le N$. Picking a parametrization of the curves $\gamma_{\UnitD; i} : [0,1] \to \C$, the functions $t \mapsto \confmap^{-1} (\gamma_{\UnitD; i} (t))$ are almost surely curves. The  collection of curves $(\confmap^{-1} ( \gamma_{\UnitD; 1} ), \ldots, \confmap^{-1} ( \gamma_{\UnitD; N}) )$, as an element of $X(\C)^N$, is almost surely equal to an $X(\C)^N$-valued random variable measurable with respect to the sigma algebra of $(\gamma_{\UnitD; 1}, \ldots, \gamma_{\UnitD; N}) \in X(\C)^N$.  This $X(\C)^N$-valued random variable is denoted, slightly abusively, by $(\confmap^{-1} ( \gamma_{\UnitD; 1} ), \ldots, \confmap^{-1} ( \gamma_{\UnitD; N}) )$ in the statement.}, and $(\gamma_{\UnitD; 1}, \ldots, \gamma_{\UnitD; N})$ has the unique distribution on $X(\overline{\UnitD})^N$ satisfying this equality. \\
\textbf{C)} For a subsequence $(n_k)_{k \in \N}$, a weak convergence takes place in topology~(iii), $\InitSegmLatt{n}{j} \to \InitSegm{j}$, if and only if it takes place in topology~(iv), $\DrFcnLattNotime{n}{j} \to W_j$. Furthermore, $\InitSegm{j}$ and $W_j$ are then Loewner transforms of each other\footnote{
More precisely, the curve $\InitSegm{j} \in X(\overline{\UnitD})$ almost surely has a Loewner transform, and the Loewner driving function obtained from this transform is almost surely equal to a $\ctsfcns$-valued random variable measurable with respect to the sigma algebra of $\InitSegm{j} \in X(\C)$. This random variable is in distribution equal to $W_j$. 

Conversely, the driving function $W_j$ almost surely has a Loewner transform curve, and the curve in $X(\overline{\UnitD})$ obtained from this transform is almost surely equal to an $X(\overline{\UnitD})$-valued random variable measurable with respect to the sigma algebra of $W_j \in \ctsfcns$. This random variable is in distribution equal to $\InitSegm{j}$.
}. \\
\textbf{D)} If the weak convergences of part~(B) above takes place, then so do the weak convergences of part~(C).
\end{thm}

Informally speaking, Theorem~\ref{thm: precompactness thm multiple curves} above proves two commutative diagrams:
\begin{equation}
\label{dia: N-curve commutation with conf maps}
\begin{gathered}
\xymatrix@C+1.5pc{
(\gamma_{\UnitD; 1}^{(n)}, \ldots, \gamma_{\UnitD; N}^{(n)})  \ar[d]^{n \to \infty} \ar[r]^{\text{ conformal }}  & (\gamma_{\ 1}^{(n)}, \ldots, \gamma_{ N}^{(n)})  \ar[l]  \ar[d]^{n \to \infty } \\
(\gamma_{\UnitD; 1}, \ldots, \gamma_{\UnitD; N}) \ar[r]^{\text{ conformal }}  & (\gamma_{ 1}, \ldots, \gamma_{ N}) \ar[l] 
}
\end{gathered}
\end{equation}
and
\begin{equation}
\label{dia: N-curve commutation with Loewner transform}
\begin{gathered}
\xymatrix@C+0.75pc{
\InitSegmLatt{n}{j}  \ar[d]^{n \to \infty} \ar[r]^{\text{ Loewner }}  & \DrFcnLattNotime{n}{j} \ar[l]  \ar[d]^{n \to \infty } \\
\InitSegm{j} \ar[r]^{\text{ Loewner }}  & W_j. \ar[l] 
}
\end{gathered}
\end{equation}

\begin{rem}
\label{rem: precompactness for irregular boundary}
The assumptions that the limiting prime ends $p_1, \ldots, p_{2N}$ possess radial limits and that $p_1^{(n)}, \ldots, p_{2N}^{(n)}$ are their close approximations are only needed in order to study the curves $(\gamma_1^{(n)}, \ldots, \gamma_{N}^{(n)})$ in the natural planar topology~(i). Removing these assumptions, statements~(A)(ii)--(iv) and~(C) still hold. Also~(D) holds, with the modification that ``weak convergences of part (B)'' should be replaced with ``weak convergence in topology (ii)''.
\end{rem}

\subsubsection{\textbf{Hypotheses of the theorem}}
\label{subsubsec: hypotheses of the main precompactness thm}

The hypotheses of Theorem~\ref{thm: precompactness thm multiple curves}, i.e., the equivalent conditions (C) and (G), are the well-established crossing conditions of Kemppainen and Smirnov~\cite{KS}. The same hypotheses will later used in Theorem~\ref{thm: one curve precompactness}, where we recall some prior results from~\cite{KS} and~\cite{mie}. The latter, and hence conditions (C) and (G) below, are given in a more general setup with measures $\PR^{(n)}$ with random curves $\gamma^{(n)}$ on some simply-connected planar graphs $(\Gr_n; e^{(n)}_1,  e^{(n)}_2)$. These measures need not originate in a random curve model on a lattice $\InfiniteGr_n$ (and in particular not a DDMP model).

Both conditions (C) and (G) require the following filtrations: consider $\gamma^{(n)}$ as a path on the graph $\Gr^{(n)}$. Let $\mathcal{F}^{(n)}_m$ be the sigma algebras generated by the $m$ first vertices of $\gamma^{(n)}$. We call $(\mathcal{F}^{(n)}_1, \mathcal{F}^{(n)}_2, \ldots )$ the filtration of the path $\gamma^{(n)}$. We denote by $T_2^{(n)}$ the ending time of the path $\gamma^{(n)}$, i.e., the time when $\gamma^{(n)}$ uses the edge $e_2^{(n)} $.

Let us start with condition (G). Let $0 < r < R$.
Denote open annuli by,
\begin{align*}
A(z, r, R) = B(z, R) \setminus \overline{B(z, r)}.
\end{align*}
Let $\domain_\Gr$ be a simply-connected domain. (In our case, there is always an underlying planar graph $\Gr$, hence the notation.)
We say that an annulus $A(z, r, R)$ is \textit{on the boundary} of a simply-connected domain $\domain_\Gr$ if $ B(z, r) \cap \bdry \domain_\Gr \ne \emptyset$. Let $p_1^{(\Gr)}$ and $p_2^{(\Gr)}$ be prime ends of $\domain_\Gr$. A chordal curve $\gamma_\Gr$ from $p_1^{(\Gr)}$ to $p_2^{(\Gr)}$ in $\domain_\Gr$ makes an \textit{unforced crossing} of $A(z, r, R)$ if for some connected component $C$ of $A(z, r, R) \cap \domain_\Gr$ which does not disconnect $p_1^{(\Gr)}$ from $p_2^{(\Gr)}$ in $\domain_\Gr$, there exists a subinterval $[t_0, t_1] \subset [0 ,1]$ such that $\gamma_\Gr ([t_0, t_1])$ intersects both connected components of $\C \setminus A(z, r, R)$, but for $t \in (t_0, t_1)$ we have $\gamma_\Gr (t) \in C$.

Condition (G): We say that the measures $\PR^{(n)}$ with random curves $\gamma^{(n)}$ \emph{satisfy condition (G)} if for all $\eps > 0$  there exists $M > 0$, independent of $n$, such that the following holds for all stopping times $u^{(n)}<T^{(n)}_2$ with respect to the filtrations of the paths $\gamma^{(n)}$: for any annulus $A(z, r, R)$ with $R/r \ge M$ on the boundary of $\domain_n \setminus \gamma^{(n)}([0, u^{(n)} ])$, we have
\begin{align*}
\PR^{(n)} [\gamma^{(n)}([u^{(n)} , T_2^{(n)}]) \text{ makes a crossing of }  A(z, r, R)
\text{ unforced in } \domain_n \setminus \gamma^{(n)}([0,u^{(n)} ])
\; \vert \; \mathcal{F}^{(n)}_{ u^{(n)} } ] \le \eps.
\end{align*}

Let us now work towards condition (C). A \textit{topological quadrilateral} $(Q; S_0, S_1, S_2, S_3)$ consists of a planar domain $Q$ homeomorphic to a square, and arcs $S_0, S_1, S_2, S_3$ of its boundary, indexed counterclockwise, that correspond to the closed edges of the square under the homeomorphism. There is a one-parameter family of classes of conformally equivalent topological quadrilaterals with labelled sides, and the equivalence class of $(Q; S_0, S_1, S_2, S_3)$ is captured by the \textit{modulus} $m(Q)$. It is the unique $L > 0$ such that there exists a biholomorphism between $Q$ and the rectangle $(0, L) \times (0,1)$, so that the sides $S_0, S_1, S_2, S_3$ of $Q$ correspond to the edges of the rectangle, and $S_0$ to $\{ 0 \} \times [0,1]$. (There is an alternative terminology and notation: $m(Q)$ is the extremal distance $d_Q (S_0, S_2)$ of $S_0$ and $S_2$ in $Q$, see, e.g.,~\cite[Chapter~4]{Ahlfors}.)

Let $\domain_\Gr$ be a simply-connected planar domain. We say that a topological quadrilateral $(Q; S_0, S_1, S_2, S_3) $ is \textit{on the boundary} of $\domain_\Gr$, if $Q \subset \domain_\Gr$ and $S_1, S_3 \subset \bdry \domain_\Gr$, while $S_0$ and $S_2$ lie inside $\domain_\Gr$, except for their end points. Let $p_1^{(\Gr)}$ and $p_2^{(\Gr)}$ be prime ends of $\domain_\Gr$. A chordal curve $\gamma_\Gr$ from $p_1^{(\Gr)}$ to $p_2^{(\Gr)}$ in $\domain$ is said to make an \textit{crossing} of $Q$ if there is a subinterval $[t_0, t_1] \subset [0 ,1]$ such that $\gamma_\Gr ([t_0, t_1])$ intersects both $S_0$ and $S_2$, but for $t \in (t_0, t_1)$ we have $\gamma_\Gr (t) \in Q$. The crossing is \emph{unforced} if $Q$ does not disconnect $p_1^{(\Gr)}$ from $p_2^{(\Gr)}$ in $\domain_\Gr$.

Condition (C): We say that the measures $\PR^{(n)}$ with random curves $\gamma^{(n)}$ \emph{satisfy condition (C)} if for all $\eps > 0$  there exists $M > 0$, independent of $n$, such that the following holds for all stopping times $u^{(n)}<T^{(n)}_2$ with respect to the filtrations of the paths $\gamma^{(n)}$: for any topological quadrilateral $Q$ with $m(Q) \ge M$ on the boundary of $\domain_n \setminus \gamma^{(n)}([0, u^{(n)}])$, we have
\begin{align*}
\PR^{(n)} [\gamma^{(n)}([u^{(n)}, T_2^{(n)}]) \text{ makes a crossing of } Q \text{ unforced in } \domain_n \setminus \gamma^{(n)}([0, u^{(n)}])
\; \vert \; \mathcal{F}^{(n)}_{u^{(n)}} ] \le \eps.
\end{align*}

\begin{rem}
\label{rem: crossing conditions and DDMP}
If the measures $\PR^{(n)}$ with random curves $\gamma^{(n)}$, for each $n$, originate in random curve models on $\InfiniteGr_n$ with the DDMP, we know that conditioning on an initial segment $\gamma^{(n)}([0, u^{(n)}])$ is equivalent to reducing the graph $\Gr^{(n)}$ by that segment. Thus, we may assume that $u^{(n)}=0$ in the conditions above, with the cost that in stead of merely the graphs $(\Gr_n; e_1^{(n)}, e_2^{(n)})$, we will have to consider all simply-connected subgraphs $(\Gr; e_1, e_2)$ of the lattices $\InfiniteGr_n$ that may appear as such reduced graphs. 
\end{rem}

\subsection{Proof of Theorem~\ref{thm: precompactness thm multiple curves}}

\subsubsection{\textbf{An analogous theorem for $N=1$ curve}}

The proof of Theorem~\ref{thm: precompactness thm multiple curves} relies heavily on the analogue of that theorem for $N=1$ curves, given in~\cite{KS} and~\cite{mie}.
To state this analogue, consider the setup described in Section~\ref{subsec: setup and notation} with $N=1$, and omitting the assumption that the measures with random curve $ (\PR^{(n)}, \gamma )$ on $(\Gr_n; e_1^n, e_2^n)$ originate in some random curve model. Note that the choice of conformal maps $\confmap_n$ is free in Section~\ref{subsec: setup and notation}, as long as they converge. In the special case of the conformal maps $\confmap_n$ chosen so that in addition $(\domain_n; p^{(n)}_1, p^{(n)}_2) $ maps to $(\UnitD; -1, 1)$, 
denote by $\tilde{\gamma}_\UnitD^{(n)} = \confmap_n(\gamma^{(n)})$ the curves from $-1$ to $1$ in $\UnitD$. Denote by $\Dr^{(n)}$ the Loewner driving functions of the curves $\tilde{\gamma}_\UnitD^{(n)}$ (where a conformal map $ (\UnitD; -1, 1) \to (\bH; 0, \infty)$ is fixed independent of $n$).

\begin{thm}
\label{thm: one curve precompactness} \emph{(\cite[Theorems~1.5 and~1.7]{KS} and~\cite[Theorem~4.4 and Proposition~4.7]{mie})}
In the setup and notation given above, suppose that the measures with random curves $(\PR^{(n)}, \gamma^{(n)})$ satisfy the equivalent conditions (C) and (G).
Then the following hold:\\ 
 \textbf{A)} The measures $\PR^{(n)}$ 
 are precompact in the following senses:
 \begin{itemize}
 \item[i)] as laws of the curves $\gamma^{(n)}$ on the space $X(\C)$;
  \item[ii)] as laws of the curves 
  $\tilde{\gamma}^{(n)}_\UnitD$ (or ${\gamma}^{(n)}_\UnitD$, obtained with any converging conformal maps) on the space $X(\overline{\UnitD})$; and
 \item[iii)] as laws of the driving functions $\Dr^{(n)}$ on the space $\ctsfcns$.
 \end{itemize}
\textbf{B)} If for some subsequence $(n_k)_{k \in \N}$ weak convergence takes place in one of the topologies above, it also takes place in the two other ones. Furthermore, denoting the respective weak limits by $\gamma$, $\tilde{\gamma}_\UnitD$, ${\gamma}_\UnitD$, and $\Dr$, it holds that $\tilde{\gamma}_\UnitD$ and $\Dr$ are Loewner transforms of each other, while $\gamma$ and ${\gamma}_\UnitD$
 satisfy
\begin{align*}
\gamma \eqd \confmap^{-1} (\gamma_\UnitD),
\end{align*}
and $\gamma_\UnitD$ has the unique distribution on $X(\overline{\UnitD})$ satisfying this.
\end{thm}

The statements that $\tilde{\gamma}_\UnitD$ and $\Dr$ are Loewner transforms of each other and $\gamma \eqd \confmap^{-1} (\gamma_\UnitD)$ are formally interpreted as in Theorem~\ref{thm: precompactness thm multiple curves}. Taking the conformal maps so that $\gamma_\UnitD = \tilde{\gamma}_\UnitD$, this theorem can be summarized in the commutative diagram
\begin{equation}
\label{dia: 1-curve commutative diagram}
\begin{gathered}
\xymatrix@C+0.75pc{
\gamma^{(n)} \ar[d]^{n \to \infty} \ar[r]^{\text{ conformal }} & \tilde{\gamma}_{\UnitD }^{(n)} \ar[l] \ar[d]^{ n \to \infty } \ar[r]^{\text{ Loewner }}  & \Dr^{(n)} \ar[l] \ar[d]^{n \to \infty}  \\
\gamma \ar[r]^{\text{ conformal }} & \tilde{\gamma}_{\UnitD} \ar[l] \ar[r]^{\text{ Loewner }} & \Dr. \ar[l]
}
\end{gathered}
\end{equation}

\subsubsection{\textbf{One-curve marginals for general $N$}}

To prove Theorem~\ref{thm: precompactness thm multiple curves}, our strategy is based on establishing Theorem~\ref{thm: one curve precompactness} for the one-curve marginal laws. We thus start with an analogue of condition (G) for multiple curves. Note that the assumptions of the below lemma hold in the setup of Theorem~\ref{thm: precompactness thm multiple curves}.


\begin{lem}
\label{lem: multi-G}
Let $\InfiniteGr_n$ be a sequence of lattices, and assume that we have, for each $n \in \N$, a discrete random curve model $( \PR^{(\Gr; e_1, \ldots, e_{2N})}, (\gamma_{\Gr;1}, \ldots, \gamma_{\Gr;N} ) )$ on some simply-connected subgraphs $\Gr$ of $\InfiniteGr_n$. Assume that these models have alternating boundary conditions and satisfy the DDMP. Fix $N \in \N$ and $\alpha \in \LP_N$, and from these symmetric curve models, extract the collection of conditioned measures $\PR_{\alpha}^{(\Gr; e_1, \ldots, e_{2N})} [\cdot]= \PR^{(\Gr; e_1, \ldots, e_{2N})} [\cdot \vert \alpha]$ with random curves $(\gamma_{\Gr;1}, \ldots, \gamma_{\Gr;2N} )$, indexed by $n$ and the simply-connected subgraphs $\Gr$ of $\InfiniteGr_n$ on which these measures make sense. This collection of measures with random curves 
 satisfies the condition (multi-G) below.
\end{lem}

Condition (multi-G): We say that a collection of measures with random curves $(\PR_{\alpha}^{(\Gr; e_1, e_2, \ldots, e_{2N})}, (\gamma_{\Gr;1}, \ldots, \gamma_{\Gr;2N} )) $, where the link pattern formed by the curves $\gamma_{\Gr;1}, \ldots, \gamma_{\Gr;2N}$ is always $\alpha$, \textit{satisfies condition (multi-G)} if for all $\eps > 0$ there exists $M > 0$ such that the following holds: 
for any annulus $A(z, r, R)$ with $R/r \ge M$ on the boundary of the simply-connected domain corresponding to $\Gr$,
\begin{align}
\label{eq: multi-G}
\PR^{(\Gr; e_1, e_2, \ldots, e_{2N})} [\text{for some $j$, }\gamma_{\Gr; j} \text{ makes an unforced crossing of } A(z, r, R) ] \le \eps.
\end{align}

Condition (multi-G) is a direct analogue of condition (G) in the case $\tau = 0$ for multiple curves. Note that we need to fix the link pattern $\alpha$ in order to be able to talk about forced and unforced crossings of $\gamma_{\Gr; j}$.

\begin{proof}[Proof of Lemma~\ref{lem: multi-G}]
We prove the proposition by induction on $N$. In the base case $N=1$ condition (multi-G) becomes simply condition (G) with $\tau=0$, and holds by assumption. Assume now that the claim holds for each number of curves $\ell = 1, \ldots, N$ and any link patterns with that number of curves, with some $M= M(\ell, \eps)$. (We may assume that $M(\ell, \eps)$ does not depend on the link pattern since, for any $\ell$, there are finitely many link patterns.) Let us study the model with $N+1$ curves, and fix an index $j \in \{1, \ldots, N+1 \}$ of the considered curve. To satisfy~\eqref{eq: multi-G} in Condition~(multi-G) it clearly suffices to show that for any $j$ 
\begin{align}
\label{eq: new multicurve crossing cond}
p =: \PR_\alpha^{(\Gr; e_1, \ldots, e_{2N + 2})} [\gamma_{\Gr; j} \text{ makes an unforced crossing of } A(z, r, R) ] \le \eps/(N+1).
\end{align} 
We claim that~\eqref{eq: new multicurve crossing cond} holds when we choose $M(N+1, \eps) = M(N, \eps/(4N + 4)) M(1, \eps/(4N + 4))$. For the rest of this proof, let us fix the index $n$ of our random curve model and the subgraph $(\Gr; e_1, \ldots, e_{2N + 2})$ of $\InfiniteGr_n$, and show that $p \le \eps/(N+1)$ irrespective of these choices. We will also drop all subscripts $\Gr$ for short.

%

Notice first that the curve $\gamma_j$ lies in the connected component $\domain'$ of $ \domain \setminus \{ \gamma_1, \ldots, \gamma_{2j-2}, \gamma_{2j}, \ldots, \gamma_{2N} \}$ containing the $j$:th odd-index edge $e_{2j-1}$. Let us denote the corresponding simply-connected subgraph of $\Gr$ by $\Gr'$, and the two marked boundary edges left in that graph by $e_1', e_2'$, where $e_1' = e_{2j-1}$ is the odd one. By the DDMP, when we condition on the remaining curves $(\gamma_1, \ldots, \gamma_{2j-2}, \gamma_{2j}, \ldots, \gamma_{2N}$), $\gamma_j$ is in distribution equal to $\gamma_{\Gr'}$ from the pair $( \PR^{(\Gr', e_1', e_2')}, \gamma_{\Gr'})$ in our random curve model on $\InfiniteGr_n$.

Consider now the event in~\eqref{eq: new multicurve crossing cond} that $\gamma_{ j}$ makes an unforced crossing of $A(z, r, R)$. Let us first study the case that that $\gamma_j$ makes this unforced crossing from the outside of the annulus $A (z, r, R)$. (Formally, $\bdry B(z, R)$ separates the crossed component $C$ of $A (z, r, R)$ in $\domain$ from the end points of $\gamma_j$.) Denote the probability of such an unforced crossing from outside by $p'$. Divide $A(z, r, R)$ into boundary annuli $A(z, r, r')$ and $A(z, r', R)$, where $r'/r =  M(1, \eps/(4N + 4))$ and $R/r' =  M(N, \eps/(4N + 4))$. 
A crossing of $C$ includes a crossing of the inner subannulus $A(z, r, r')$. At least one of the following two thus has to occur: 
\begin{itemize}
\item[i)] some connected component of $A(z, r, r')$ not disconnecting in $\domain'$ is crossed by $\gamma_j$; or
\item[ii)] there is a component $C'$ of  $A(z, r, r')$ in $\domain'$, with $C' \subset C$ for some non-disconnecting component $C$ of $A(z, r, R)$ in $\domain$, such that $C'$ is disconnecting in $\domain'$.
\end{itemize} 
Case~(i) occurs with probability $\le \eps / (4N + 4)$, by the conditional law of $\gamma_j$ deduced above. 
In case~(ii), recall that $C$ is separated from the end points of $\gamma_j$ in $\domain$, by $\bdry B(z, R)$. It follows that $C'$ is separated from the end points of $\gamma_j$ in $\domain$, and hence also in $\domain'$, by $\bdry B(z, r')$.
On the other hand, $C'$ is disconnecting in $\domain'$ if and only if both the clockwise and counterclockwise boundary arcs of $\bdry \domain'$ from $e_1'$ to $e_2'$ touch $C'$. Likewise, since $C$ does not disconnect $e_1'$ from $e_2'$ in $\domain$, we know that one of the arcs of $\bdry \domain$ from $e_1'$ to $e_2'$, say for definiteness the clockwise one, does not touch $C$ and is thus separated from $C$ by $\bdry B(z, R)$. In other words, for $C'$ to be disconnecting in $\domain'$, one of the remaining curves $\gamma_i$, $i \ne j$, starting and ending on the clockwise arc of $\bdry \domain$, has to cross $\bdry B(z, R)$, then enter $C$ and enter it deep enough to touch $\bdry B(z, r')$, and finally touch $C'$. In particular, this curve $\gamma_i$ crosses the annulus $A(z, r', R)$ inside $C$. Now, study the component $\domain''$ of $\domain \setminus \gamma_j$ containing $\gamma_i$, and let $C'' \subset C$ be the component of $A(z, r', R)$ in $\domain''$ containing a crossing of $\gamma_i$ (if there are several, pick one). We claim that the crossing of $C''$ by $\gamma_i$ is unforced in $\domain''$: indeed, the clockwise boundary of $\domain''$ between the end points of $\gamma_i$ is contained in that of $\domain$ between the end points of $\gamma_j$, and we already know that the latter does not touch $C$. Thus, the small clockwise boundary arc of $\domain''$ between the end points of $\gamma_i$ does not touch the smaller set $C''$, so the crossing of $C''$ is unforced in $\domain''$. This holds for any $C$, and $C''$ is always a connected component of the same annulus. Thus, by the DDMP, such an unforced crossing by $\gamma_i$, for some $i$, occurs with probability $\le M(N, \eps/(4N + 4))$.

Finally, summing up the contributions of cases (i) and (ii) above, we notice that $p' \le 2\eps/(4N + 4)$. Crossings of $A(z, r, R)$ from the inside are treated similarly. We thus obtain
\begin{align*}
p \le  2\eps/(4N + 4) + 2\eps/(4N + 4) = \eps/(N+1),
\end{align*}
as required.
\end{proof}

The lemma above allows us to apply Theorem~\ref{thm: one curve precompactness} for the one-curve marginals of the random curve collections in Theorem~\ref{thm: precompactness thm multiple curves}

\begin{cor}
\label{cor: marginal precompactness}
Consider the setup of Theorem~\ref{thm: precompactness thm multiple curves}, in the version where the measures $\PR^{(n)}$ are those conditional on a fixed link pattern $\alpha \in \LP_N$, and let the assumptions in that theorem hold. Then, the measures with random curves $(\PR^{(n)}, \gamma^{(n)}_j)$, for any fixed $j$, satisfy condition (G), and hence all the consequences of Theorem~\ref{thm: one curve precompactness}
\end{cor}

\begin{proof}
Condition (G) for the curve $\gamma^{(n)}_j$ follows immediately by combining Remark~\ref{rem: crossing conditions and DDMP} and condition (multi-G) obtained in Lemma~\ref{lem: multi-G}.
\end{proof}

\subsubsection{\textbf{Proofs of the statements about curves}}

We can now rather straightforwardly prove the statements of Theorem~\ref{thm: precompactness thm multiple curves} that only employ random variables in the spaces of curves $X(\C)$ and $X(\overline{\UnitD})$.

\begin{proof}[Proof of Theorem~\ref{thm: precompactness thm multiple curves}(A)(i)--(iii)]
Let us first prove part~(i). Consider first Theorem~\ref{thm: precompactness thm multiple curves} with the measures $\PR^{(n)}$ being those conditional on a fixed link pattern $\alpha \in \LP_N$.
Recall that by Prohorov's theorem, tightness and precompactness are equivalent for measures on Polish spaces (i.e., complete separable metric spaces). Thus, by Corollary~\ref{cor: marginal precompactness} on the measures with random curves $(\PR^{(n)}, \gamma^{(n)}_j)$ and Theorem~\ref{thm: one curve precompactness}(A)(i), we know that $\PR^{(n)}$ are tight as laws of the random curves $\gamma^{(n)}_j$. In other words, for all $\eps > 0$, there exists a (sequentially) compact set $K^{(j)}_\eps \subset X(\C)$ such that $\PR^{(n)} (\gamma^{(n)}_j \in K^{(j)}_\eps ) \ge 1 - \eps/N$ for all $n$. This holds for all $1 \le j \le N$.

Next, take for all $1 \le j \le N$ the sets $K^{(j)}_\eps \subset X(\C)$ as above. Their product set $K^{(1)}_\eps \times \ldots \times K^{(N)}_\eps \subset X(\C)^N$ is a (sequentially) compact set in the space $X(\C)^N$. Furthermore, it clearly holds that 
\begin{align}
\label{eq: tightness of RVs implies tightness of tensor RV}
\PR^{ (n)} [ (\gamma^{(n)}_1, \ldots, \gamma^{(n)}_{N} ) \in K^{(1)}_\eps \times \ldots \times K^{(N)}_\eps ) \ge 1 - \eps.
\end{align}
Thus, the measures $\PR^{(n)}$ are tight in topology (i). Using Prohorov's theorem to the converse direction, we deduce that they are precompact. This proves Theorem~\ref{thm: precompactness thm multiple curves}(A)(i) for measures $\PR^{(n)}$ conditional on a fixed link pattern $\alpha \in \LP_N$

For the symmetric measures $\PR^{(n)}$, notice that
\begin{align}
\label{eq: curve measure as a convex combi of LP conditional measures}
\PR^{(n)} [\; \cdot \; ] = \sum_{\alpha \in \LP_N} \PR^{(n)} [\alpha] \PR^{(n)} [\; \cdot \;  \vert \alpha ].
\end{align}
Since there are finitely many link patterns $\alpha \in \LP_N$, we can extract a subsequence so that the numbers $\PR^{(n)} [\alpha]$ converge for all $\alpha \in \LP_N$. Then, we can use the precompactness of the conditional measures $\PR^{(n)} [\cdot \vert \alpha ]$, deduced in the previous paragraph, to extract a further subsequence where the conditional measures $\PR^{(n)} [\cdot \vert \alpha ]$ converge weakly as laws of $(\gamma^{(n)}_1, \ldots, \gamma^{(n)}_{N} ) \in X(\C)^N$, for all $\alpha \in \LP_N$. Then, also the symmetric measures $\PR^{(n)}$ converge weakly as laws of $(\gamma^{(n)}_1, \ldots, \gamma^{(n)}_{N} )$ along this subsequence.

The proof of part~(ii) is identical to the proof of part~(i) above.

For part~(iii), the initial segments $\InitSegmLatt{n}{j}$ obtained from the continuous stopping time is a continuous function of the full curves $(\gamma^{(n)}_{\UnitD; 1}, \ldots, \gamma^{(n)}_{\UnitD; N} )$, see Appendix~\ref{app: continuous stopping times}. Thus, the weak convergence of the former follows from that of the latter.
\end{proof}
%

\begin{proof}[Proof of Theorem~\ref{thm: precompactness thm multiple curves}(B)]
The proof for $N=1$ curve is given by the author in~\cite[Theorem~4.4 and Proposition~4.7]{mie}. For $N \ge 2$, the proof essentially identical, and we thus only outline the proof here. Let us first give the proof in the case when the measures $\PR^{(n)}$ are those conditional on a fixed link pattern $\alpha \in \LP_N$. By Corollary~\ref{cor: marginal precompactness}, the curves $\gamma^{(n)}_j$, for all $j$, satisfy condition (G), which is the hypothesis in~\cite[Theorem~4.4 and Proposition~4.7]{mie}. The proofs of~\cite{mie} can now be straightforwardly repeated for the collections of curves $(\gamma^{(n)}_1, \ldots, \gamma^{(n)}_{2N})$.

To handle the case when the measures $\PR^{(n)}$ are not conditional on some link pattern, we have to be careful with Condition~(G), which does not make sense any more now that a curve $\gamma^{(n)}_j$ does not have a single target point. Condition (G) only appears in the proof of~\cite[Theorem~4.4 and Proposition~4.7]{mie} inside the proof of the Key lemma~4.6. The conclusion of that lemma (for which we refer to~\cite{mie}) thus has to be reached differently: Condition (G), and thus~\cite[Key lemma~4.6]{mie} holds for the curves $\gamma^{(n)}_j$ under the conditional measures $\PR^{(n)} [\cdot \vert \alpha ]$, for all $\alpha \in \LP_N$. Since there are finitely many link patterns $\alpha \in \LP_N$, we can make the conclusion of that lemma hold for all of them simultaneously. Since $\PR^{(n)}$ is a convex combination of such conditional measures, we then obtain the conclusion of that lemma also for curves $\gamma^{(n)}_j$ under $\PR^{(n)}$, irrespective of the convex weights $\PR^{(n)}[\alpha]$. The rest of the arguments in~\cite[Theorem~4.4 and Proposition~4.7]{mie} can be repeated straightforwardly.
\end{proof}

\subsubsection{\textbf{Proofs of the statements about driving functions}}

Let us now prove the statements about driving functions in Theorem~\ref{thm: precompactness thm multiple curves}. The strategy will be once again to first prove the theorem in the case when the random curves are conditional on some particular link pattern $\alpha \in \LP_N$. Let us introduce some notation in that case. Note first that under the assumptions of Theorem~\ref{thm: precompactness thm multiple curves}, Corollary~\ref{cor: marginal precompactness} and the conclusions of Theorem~\ref{thm: one curve precompactness} hold for the curves $\gamma^{(n)}_j$ starting from the odd boundary points. The same deduction also applies for their reversals, starting from the even boundary points. In particular, if we choose the conformal maps $\confmap_n$ so that the $i$:th boundary point, for a fixed $1 \le i \le 2N$, maps to $-1$ and the other endpoint of that curve to $+1$, the driving functions $\Dr^{(n)}_i$ of the corresponding curves are precompact, as detailed in Theorem~\ref{thm: one curve precompactness}. Unfortunately, the driving functions $\DrFcnLattNotime{n}{i}$ in Theorem~\ref{thm: precompactness thm multiple curves} are driving functions of some conformal images of the curves described by $\Dr^{(n)}_i$. Let us relate $\DrFcnLattNotime{n}{i}$ and $\Dr^{(n)}_i$, and let us for a moment fix $i$ and omit the subscripts $i$.

Now, more precisely, $\Dr^{(n)}$, parametrized by time $t$,  are the Loewner driving functions of some random growing hulls $K^{(n)}_t$ from zero to infinity in $\bH$. Likewise $\DrFcnNoInd^{(n)}$, parametrized by time $s$, are the driving functions of $\confmapSH_n (K^{(n)}_t)$; here $\confmapSH_n$ are suitable conformal (M\"{o}bius) maps $\overline{\bH} \to \overline{\bH}$ that converge uniformly over compacts, $\confmapSH_n \to \confmapSH$, and the time parametrization $s$ is different than $t$, $s = s^{(n)} (t)$. Recall that $\DrFcnNoInd^{(n)}$ is stopped when the corresponding hulls $\confmapSH_n (K^{(n)}_t)$ reach the continuous exit time of the localization neighbourhood $\confmapDH(U_i)$ of the $i$:th boundary point in $\bH$; we denote that that value of $s$ by $\sigma^{(n)}$.


\begin{lem}
\label{lem: driving functions of conformally perturbed curves - 1 curve}
Under the notation above and assumptions of Theorem~\ref{thm: one curve precompactness}, the continuously stopped driving functions $\DrFcnNoInd^{(n)}_{s \wedge \sigma^{(n)}}$ are precompact in $\ctsfcns$. Furthermore, if the driving functions $\Dr^{(n)}_t$ converge weakly in $\ctsfcns$ to $\Dr_t$, describing some random growing hulls $K_t$, then the driving functions $\DrFcnNoInd^{(n)}_{s \wedge \sigma^{(n)}}$ converge weakly in $\ctsfcns$ to $\DrFcnNoInd_{s \wedge \sigma}$, describing the hulls $\confmapSH(K_t)$ up to their continuous exit time $\sigma$ of $\confmapDH(U_i)$.
\end{lem}

This lemma can be summarized in the commutative diagram
\begin{equation}
\label{dia: dr fcns and conf maps commutative diagram}
\begin{gathered}
\xymatrix@C+1.25pc{
\Dr^{(n)} \ar[d]^{n \to \infty} \ar[r]^{\text{ conformal }} & \DrFcnNoInd^{(n)}_{s \wedge \sigma^{(n)}} \ar[d]^{ n \to \infty }  \\
\Dr \ar[r]^{\text{ conformal }} & \DrFcnNoInd_{s \wedge \sigma}. 
}
\end{gathered}
\end{equation}

\begin{proof}[Proof of Lemma~\ref{lem: driving functions of conformally perturbed curves - 1 curve}]
The proof becomes more transparent if we operate with the conformal maps $\confmapSH_n$ and $\confmapSH$ on driving functions in stead of their hulls, with the undestanding that after this operation, the driving functions are stopped at their continuous exit times of $\confmapDH(U_i)$. For instance, we replace $\DrFcnNoInd^{(n)}_{s \wedge \sigma^{(n)}}$ by $\confmapSH_n (\Dr^{(n)}_t)$ and $\DrFcnNoInd_{s \wedge \sigma}$  by $\confmapSH (\Dr_t)$. By Corollary~\ref{cor: coordinate change of LE is cts fcn of dr fcns}, conformal maps operate continuously on driving functions (when interpreted this way), so $\confmapSH_n (\Dr^{(n)}_t)$ and $\confmapSH (\Dr_t)$ are measurable random variables.

Now, by Theorem~\ref{thm: one curve precompactness}, the functions $\Dr^{(n)}_t$ are precompact. It thus suffices to show that 
 if $\Dr^{(n)}_t \to \Dr_t$ weakly in $\ctsfcns$, then also $\confmapSH_n (\Dr^{(n)}_t)\to \confmapSH (\Dr_t)$ weakly in $\ctsfcns$. Take thus $f: \ctsfcns \to \R$ a bounded, Lipschitz continuous test function, and compute
\begin{align}
\nonumber
 & \vert \; \EX^{(n)} [f (\confmapSH_n (\Dr^{(n)}_t))] - \EX [ f (\confmapSH (\Dr_t))] \; \vert \\
 \label{eq: another triangle ineq}
 \le &  \vert \; \EX^{(n)} [f (\confmapSH_n (\Dr^{(n)}_t)) - f (\confmapSH (\Dr^{(n)}_t))] \; \vert 
 + 
 \vert  \; \EX^{(n)} [f (\confmapSH (\Dr^{(n)}_t))] - \EX [ f (\confmapSH (\Dr_t))] \; \vert.
\end{align}

Consider first the latter term on the right-hand side of~\eqref{eq: another triangle ineq}. By Corollary~\ref{cor: coordinate change of LE is cts fcn of dr fcns}, the mapping $\confmapSH$ operates continuously on driving functions, and hence we deduce that $\confmapSH (\Dr^{(n)}_t) \to \confmapSH (\Dr_t)$ weakly in $\ctsfcns$. Thus, the latter term tends to zero as $n \to \infty$.

Cosider now the former term on the right-hand side of~\eqref{eq: another triangle ineq}. Since the random variables $\Dr^{(n)} \in \ctsfcns$ are tight and $f$ is bounded, we may with an arbitrarily small error restrict our consideration to the case where $\Dr^{(n)}$ belongs to a suitable compact set of the space $\ctsfcns$. Now, by Corollary~\ref{cor: LE coordinate changes with converging maps converge uniformly over compacts}, we have
\begin{align*}
 d_{\ctsfcns} ( \; \confmapSH_n ( \cdot ), \; \confmapSH ( \cdot )\; ) \to 0 \qquad \text{as } n \to \infty,
\end{align*}
uniformly over any compact set of $\ctsfcns$. Since $f$ is Lipschitz, it follows that also
\begin{align*}
\vert f (\confmapSH_n ( \cdot )) - f (\confmapSH ( \cdot )) \vert \to 0
\qquad \text{as } n \to \infty,
\end{align*}
uniformly over any compact set of $\ctsfcns$. This shows that the former term on the right-hand side of~\eqref{eq: another triangle ineq} tends to zero as $n \to \infty$.

Having analyzed both terms of~\eqref{eq: another triangle ineq}, we now deduce that
\begin{align*}
\vert \; \EX^{(n)} [f (\confmapSH_n (\Dr^{(n)}_t))] - \EX [ f (\confmapSH (\Dr_t))]  \;\vert \to 0
\qquad \text{as } n \to \infty.
\end{align*}
This shows that $\confmapSH_n (\Dr^{(n)}_t)\to \confmapSH (\Dr_t)$ weakly in $\ctsfcns$, and hence completes the proof.
\end{proof}

\begin{rem}
\label{rem: weak conv of stopped dr fcns w different stopping times}
Lemma~\ref{lem: driving functions of conformally perturbed curves - 1 curve} and its proof also apply for the slightly different definition of the stopping times $\sigma^{(n)}$, as in Remarks~\ref{rem: exit times of neighbourhoods before or after conf map} and~\ref{rem: exit times of neighbourhoods before or after conf map 2}.
\end{rem}

With the above lemma at hand, we can finsh the proof of Theorem~\ref{thm: precompactness thm multiple curves}.

\begin{proof}[Proof of Theorem~\ref{thm: precompactness thm multiple curves}(A)(iv)]
In the case where the measures $\PR^{(n)}$ are those conditional on a link pattern $\alpha \in \LP_N$,
the precompactness of the stopped driving function $\DrFcnLattNotime{n}{i}$ was stated and proven in Lemma~\ref{lem: driving functions of conformally perturbed curves - 1 curve} above. In the general case, it follows from Equation~\eqref{eq: curve measure as a convex combi of LP conditional measures}.
\end{proof}

\begin{proof}[Proof of Theorem~\ref{thm: precompactness thm multiple curves}(C)]

We prove the equivalence as two implications.

Implication 1: Let us first assume that the initial segments $\InitSegmLatt{n}{i}$ converge weakly, $\InitSegmLatt{n}{i} \to \InitSegm{i}$. Start by showing that the corresponding stopped driving functions $\DrFcnLatt{n}{i}{s \wedge \sigma^{(n)}_i}$ then converge weakly to the stopped driving functions $\DrFcnNoInd_{i; s \wedge \sigma_i}$ of $\InitSegm{i}$. (We suppress here and in continuation the subsequence notation $n_k$.) By the precompactness of $\DrFcnLatt{n}{i}{s}$, it suffices to prove that the claimed convergence holds for some futher subsequence. Employing this strategy, we can pick a subsequence of $n$:s so that the probabilities $\PR^{(n)}[\alpha]$ of all link patterns $\alpha \in \LP_N$ converge. It thus suffices to prove the claim subsequentially for the conditional measures $\PR^{(n)}_\alpha$.

Now, for the conditional measure $\PR^{(n)}_\alpha$, apply Corollary~\ref{cor: marginal precompactness} and thus Theorem~\ref{thm: one curve precompactness} to deduce the commutative diagram~\eqref{dia: 1-curve commutative diagram}, where the curves $\gamma^{(n)}$ and $\tilde{\gamma}_\UnitD^{(n)}$ are those starting from the boundary point $i$ in the link pattern $\alpha$, and $\Dr^{(n)}$ their driving functions. We can pick a subsequence so that all the weak convergences in that diagram hold. By~\cite[Proposition~4.3]{mie}, also the curves ${\gamma}_\UnitD^{(n)}$ obtained by a different confromal map to $\UnitD$ then converge weakly to the suitable conformal images of $\tilde{\gamma}_\UnitD^{(n)}$. 
Lemma~\ref{lem: driving functions of conformally perturbed curves - 1 curve} now states that the stopped driving functions $\DrFcnLatt{n}{i}{s}$ converge weakly to that of $\gamma_\UnitD$ up to $\sigma_i$, in other words, the driving function of $\InitSegmNoInd_i$. This proves the claimed weak convergence.

We yet have to prove the random variable measurability claimed in Theorem~\ref{thm: precompactness thm multiple curves}(C). Here $\DrFcn{i}{s \wedge \sigma_i }$ is a measurable function of $\Dr_{t \wedge \tau_i}$ by Corollary~\ref{cor: coordinate change of LE is cts fcn of dr fcns} and $\Dr_{t \wedge \tau_i}$ is a measurable function of $\InitSegm{i}$ by the commutative diagram~\eqref{dia: 1-curve commutative diagram} and~\cite[Proposition~4.3]{mie} (see also Appendix~\ref{app: continuous stopping times} on the measurability of the the stopping times).

Implication 2: Let us now assume that the stopped driving functions $\DrFcnLatt{n}{i}{s}$ converge weakly, $\DrFcnLatt{n}{i}{s} \to \DrFcn{i}{s}$. Show first that the corresponding initial segments $\InitSegmLatt{n}{i}$ then also converge weakly, $\InitSegmLatt{n}{i} \to \InitSegm{i}$, where $\InitSegm{i}$ is the curve described by the driving function $\DrFcn{i}{s}$ up to time $\sigma_i$. As in the previous implication, it suffices to prove this claim subsequentially for the conditional measures $\PR^{(n)}_\alpha$.

The argument to prove the subsequential convergence is now identical to the previous implication. To argue the measurabilities of random variables, $\InitSegm{i}$ is a measurable with function of $\Dr_{t \wedge \tau_i}$ by the arguments above, and $\Dr_{t \wedge \tau_i}$ with respect to $\DrFcn{i}{s \wedge \sigma_i }$ by Corollary~\ref{cor: coordinate change is a bijection}. This finishes the proof of Theorem~\ref{thm: precompactness thm multiple curves}(C).
\end{proof}

We finish the proof of Theorem~\ref{thm: precompactness thm multiple curves} with a triviality.

\begin{proof}[Proof of Theorem~\ref{thm: precompactness thm multiple curves}(D)]
An initial segment $\InitSegm{i}$ up to a continuous exit time are a continuous function of the full curves $(\gamma_{1; \UnitD}, \ldots, \gamma_{N; \UnitD})$, see Appendix~\ref{app: continuous stopping times}.
\end{proof}

\subsection{A precompactness result for local multiple SLEs}
\label{subsec: precompactness for local NSLE}

We have now in complete detail proven Theorem~\ref{thm: precompactness thm multiple curves}, sufficient for the main results of this paper. In this subsection, we will discuss some slight improvements to that theorem, holding under the same assumptions. These improvements are required in order to prove convergence of the collection of initial segments to a local multiple SLEs, in stead of convergence of full curves to a local-to-global multiple SLEs. We will keep the discussion in this subsection on an informal level, trusting that the interested reader can formalize the arguments presented here.

Let us still consider the setup of Section~\ref{subsec: setup and notation}. The prime ends $p_1, \ldots, p_{2N}$ of $\domain$ now need to be distinct, and their localization neighbourhoods $U_1, \ldots, U_{2N}$ disjoint, but otherwise the assumptions of Theorem~\ref{thm: precompactness thm multiple curves} can for the discussion of this subsection be relaxed as in Remark~\ref{rem: precompactness for irregular boundary}.

The local multiple SLE describes the collection of initial segments $\InitSegmLatt{n}{1}, \ldots, \InitSegmLatt{n}{2N} \in X(\overline{\UnitD})$ in terms of their \emph{iterated driving functions} $\IterDrFcnLattNotime{n}{1}, \ldots, \IterDrFcnLattNotime{n}{2N} \in \ctsfcns$. These are described as follows: the first function $\IterDrFcnLattNotime{n}{1}$ is the driving function of $ \confmapDH (\InitSegmLatt{n}{1})$ stopped at the continuous exit time $\sigma^{(n)}_1$ of $ \confmapDH ( U_1 )$, i.e., $\IterDrFcnLattNotime{n}{1} = \DrFcnLattNotime{n}{1}$. Given the first one, denote by $g_{\InitSegmLatt{n}{1}}$ the mapping-out function of the initial segment $\confmapDH ( \InitSegmLatt{n}{1} )$, obtained from the Loewner equation with the first driving function $\IterDrFcnLattNotime{n}{1}$. Then, $\IterDrFcnLattNotime{n}{2} $ is the driving function of the mapped-out second initial segment $g_{\InitSegmLatt{n}{1}} ( \confmapDH ( \InitSegmLatt{n}{2}) )$. Similarly, the further iterated driving functions are defined as the driving functions of the conformally mapped initial segment, when one maps out the previous initial segments. Note that the local multiple SLE is defined via its iterated driving functions, see Section~\ref{subsubsec: loc N-SLE in D}.

Now, our precompactness result for the local multiple SLEs is the precompactness of the collections of driving functions $( \IterDrFcnLattNotime{n}{1}, \ldots, \IterDrFcnLattNotime{n}{2N} ) \in \ctsfcns^{2N}$ and the commutative diagram
\begin{equation}
\label{dia: local multiple SLE commutation}
\begin{gathered}
\xymatrix@C+2.5pc{
( \InitSegmLatt{n}{1}, \ldots, \InitSegmLatt{n}{2N} )  \ar[d]^{n \to \infty} \ar[r]^{\text{ iter. Loewner }}  & ( \IterDrFcnLattNotime{n}{1}, \ldots, \IterDrFcnLattNotime{n}{2N} )  \ar[l]  \ar[d]^{n \to \infty } \\
(\InitSegm{ 1}, \ldots, \InitSegm{2N} ) \ar[r]^{\text{ iter. Loewner }}  & (\IterDrFcnNoInd_{ 1}, \ldots, \IterDrFcnNoInd_{2N}), \ar[l] 
}
\end{gathered}
\end{equation}
formally interpreted similarly to all previous commutative diagrams.

Let us sketch the proof of this result. First, one utilizes the earlier commutation relations and to establishe the diagram
\begin{equation}
\label{dia: local multiple SLE first commutation}
\begin{gathered}
\xymatrix@C+1.5pc{
( \InitSegmLatt{n}{1}, \ldots, \InitSegmLatt{n}{2N} )  \ar[d]^{n \to \infty} \ar[r]^{\text{ Loewner }}  & ( \DrFcnLattNotime{n}{1}, \ldots, \DrFcnLattNotime{n}{2N} )  \ar[l]  \ar[d]^{n \to \infty } \\
(\InitSegm{ 1}, \ldots, \InitSegm{2N} ) \ar[r]^{\text{ Loewner }}  & (\DrFcnNoInd_{ 1}, \ldots, \DrFcnNoInd_{2N}), \ar[l] 
}
\end{gathered}
\end{equation}
with the usual (not iterated) driving functions.
This diagram is a direct multicurve analogue of diagram~\eqref{dia: N-curve commutation with Loewner transform}, i.e., Theorem~\ref{thm: precompactness thm multiple curves}(C). The tightness of the collections of curves (resp. driving functions) follows from that of the individual curves (resp. driving functions) similarly to the proof of Theorem~\ref{thm: precompactness thm multiple curves}(A)(i)~(see especially Equation~\eqref{eq: tightness of RVs implies tightness of tensor RV}). The commutation in diagram~\eqref{dia: local multiple SLE first commutation} follows by straightforwardly repeating for multiple curves the arguments that yielded the analogous commutation relation for a single initial segment, needed to prove Theorem~\ref{thm: precompactness thm multiple curves}(C), given in~\cite{KS} and the previous subsection.

Having proven~\eqref{dia: local multiple SLE first commutation}, one next establishes the commutative diagram
\begin{equation}
\label{dia: local multiple SLE second commutation}
\begin{gathered}
\xymatrix@C+1.5pc{
( \DrFcnLattNotime{n}{1}, \ldots, \DrFcnLattNotime{n}{2N} )  \ar[d]^{n \to \infty} \ar[r]^{\text{ iterated }}  & ( \IterDrFcnLattNotime{n}{1}, \ldots, \IterDrFcnLattNotime{n}{2N} )  \ar[l]  \ar[d]^{n \to \infty } \\
(\DrFcnNoInd_{ 1}, \ldots, \DrFcnNoInd_{2N}) \ar[r]^{\text{ iterated }}  & (\IterDrFcnNoInd_{ 1}, \ldots, \IterDrFcnNoInd_{2N}), \ar[l] 
}
\end{gathered}
\end{equation}
This diagram follows from the observation that the mapping $f$ from the original driving functions $ ( \DrFcnLattNotime{n}{1}, \ldots, \DrFcnLattNotime{n}{2N} )$ to the iterated ones $ ( \IterDrFcnLattNotime{n}{1}, \ldots, \IterDrFcnLattNotime{n}{2N} )$ is a continuous bijection, whose inverse is also continuous.
Let us now illustrate how this is proven for $N=1$ with two initial segments --- larger $N$ are treated identically.

Let us first argue that $f$ is continuous. The first functions are identical, $\IterDrFcnLattNotime{n}{1} = \DrFcnLattNotime{n}{1}$. For the second functions, $\DrFcnLattNotime{n}{2} $ is the driving function of the second initial segment $\confmapDH (\InitSegmLatt{n}{2} )$, and $\IterDrFcnLattNotime{n}{2} $ that of the mapped-out second initial segment $g_{\InitSegmLatt{n}{1}} ( \confmapDH ( \InitSegmLatt{n}{2} ) )$. Now, the mapping-out function $g_{\InitSegmLatt{n}{1}}$ depends continuously on the first driving function $\DrFcnLattNotime{n}{1}$ (equipping analytic functions with the topology of uniform convergence over compacts), see~\cite[Lemma~5.1]{Kemppainen-SLE_book}, and the driving function of the conformal image $g_{\InitSegmLatt{n}{1}} ( \confmapDH ( \InitSegmLatt{n}{2} ) )$ thus depends continuously on the conformal map $g_{\InitSegmLatt{n}{1}}$ and the original driving function $\DrFcnLattNotime{n}{2}$, see Corollary~\ref{cor: coordinate change with converging maps} and Remark~\ref{rem: exit times of neighbourhoods before or after conf map}. This shows that $f$ is continuous.

To argue that $f$ is bijective and its inverse is continuous, we can find $f^{-1}$ and prove its continuity identically to the previous paragraph. The only difference is that here we have to use the inverse mapping-out function $g_{\InitSegmLatt{n}{1}}^{-1}$, whose continuity follows from \cite[Lemma~5.8]{Kemppainen-SLE_book}. This finishes our sketch of proof of the commutative diagram~\eqref{dia: local multiple SLE second commutation}

Finally, combining the commutative diagrams~\eqref{dia: local multiple SLE first commutation} and~\eqref{dia: local multiple SLE second commutation} yields~\eqref{dia: local multiple SLE commutation}.

\bigskip{}

\section{Local-to-global properties}
\label{sec: loc-2-glob}
%

In this section, we will assume that the study of a discrete random curve model is in a phase where the precompactness conditions of Theorem~\ref{thm: precompactness thm multiple curves} have been verified, and the scaling limit $W_j$ of the driving function of an initial segment has been identified as that of a local multiple SLE.
The objective is to show that the discrete domain Markov property is (under some reasonable assumptions on the discrete curve model) inherited to a domain Markov type property of the scaling limit, so that this identification of the scaling limit of one initial segment actually identifies the scaling limit of the full collection of full curves.

Our results of this type are stated in two theorems, Theorem~\ref{thm: local multiple SLE convergence} addressing the convergence of collections of initial segments to local multiple SLEs, and Theorem~\ref{thm: loc-2-glob multiple SLE convergence, kappa le 4} addressing convergence of collections of full curves to local-to-global multiple SLEs. The latter has two alternative assumptions and proofs, a simpler one with a strong \emph{a priori} estimate only possible for SLE scaling limits with $\kappa \le 4$, and a longer one with an assumption applicable for general $\kappa < 8$. Finally, in Theorems~\ref{thm: relation to global multiple SLE 1} and~\ref{thm: relation to global multiple SLE 2}, we explicate the relation of the obtained scaling limits to the global multiple SLEs of~\cite{PW, BPW}, introduced in Section~\ref{sec: intro}.

\subsection{Statements of the main theorems}

This subsection contains the statements of the main theorems of this 
section. The rest of this section will constitute their proofs.

\subsubsection{\textbf{Notations and domain discretizations}}
\label{subsubsec: relaxed assumptions}

We will continue in the notation and setup introduced in Section~\ref{subsec: setup and notation}. The only difference is that in some statements, we should only assume that the domains $(\domain_n; p_1^{(n)}, \ldots, p_{2N}^{(n)})$ with marked boundary points are uniformly bounded and converge in the Carath\'{e}odory sense to $(\domain; p_1, \ldots, p_{2N})$, without additional requirements on existence of radial limits at $p_1, \ldots, p_{2N}$ or closeness of approximations (cf. Remark~\ref{rem: precompactness for irregular boundary}). If the assumptions in Section~\ref{subsec: setup and notation} are relaxed in this way, we say that we have \emph{relaxed regularity at marked boundary points}.

\subsubsection{\textbf{Assumptions on the discrete curve models}}
\label{subsubsec: setup and assumptions}

Throughout this section, we consider discrete random curve models under the following setup and assumptions.

We have a sequence of lattices $\InfiniteGr_n$, and a discrete random curve model on each lattice. These lattices describe a scaling limit, in the sense that the maximal length of a lattice edge in any bounded domain tends to zero as $n \to \infty$. (In many applications $\InfiniteGr_n$ are just scalings of $\InfiniteGr_1$, and we may thus think of having a single random curve model on $\InfiniteGr_1$.) The assumptions imposed on the random curve models in Theorem~\ref{thm: precompactness thm multiple curves} are satisfied, i.e., they have alternating boundary conditions,  satisfy the DDMP, and the collection of one-curve measures $\PR^{(\Gr, e_1, e_2)}$ in these random curve models, indexed by $n$ and simply-connected subgraphs $\Gr$ of $\InfiniteGr_n$, satisfy the equivalent conditions (C) and (G).

The above assumptions are well established properties for many discrete curve models. To state our nontrivial standing assumption about convergence of driving functions, consider the setup of Section~\ref{subsec: setup and notation} with relaxed regularity at marked boundary points.
 By Theorem~\ref{thm: precompactness thm multiple curves} and Remark~\ref{rem: precompactness for irregular boundary}, the stopped driving functions $\DrFcnLattNotime{n}{j}$ of the curve initial segment starting from the $j$:th boundary point, for any $j$ and any localization neighbourhood $U_j$ in $\UnitD$, are precompact.
Throughout this section, we will assume that the weak limit of any such initial segment has been in that case identified as a local multiple SLE with partition function $\PartF_N$:
\begin{ass}
\label{ass: dr fcns converge to loc mult SLE}
For any $(\Gr_n; e_1^{(n)}, \ldots, e_{2N}^{(n)})$ converging in the Carath\'{e}odory sense, and any localization neighbourhood $U_j$,
the stopped driving functions $\DrFcnLattNotime{n}{j}$ of the random curves converge weakly in $\ctsfcns$,
\begin{align*}
\DrFcnLattNotime{n}{j} \to \DrFcnNoInd_j \qquad \text{as } n \to \infty,
\end{align*}
where $\DrFcnNoInd_j$ is the driving function of the local multiple SLE, stopped at the continuous exit time of $U_j$, and with partition function $\PartF_N$.
%
%
\end{ass}

The partition functions $\PartF_N$ above, indexed by $N$, should be partition functions with the same parameter $\kappa \in (0, 8)$. We will mostly suppress the notation $\kappa$ in our discussions. Let us make two remarks about this assumption.

First, it is largely for simplicity that we restrict our consideration to unconditional measures, or partition function $\PartF_N$ here. The proofs and statements of Theorem~\ref{thm: local multiple SLE convergence} and Theorem~\ref{thm: loc-2-glob multiple SLE convergence, kappa le 4} under Assumption~\ref{ass: quantitative no boundary visits assumption} (i.e., the \emph{a priori} estimate only possible for $\kappa \le 4$) can be straightforwardly generalized to convergence of conditional discrete curve models. Assumption~\ref{ass: dr fcns converge to loc mult SLE} above should then be modified so that the convergence holds for the conditional discrete curve models with any link pattern $\alpha \in \LP_N$, to local multiple SLEs with the corresponding partition functions $\PartF_\alpha$. However, the alternative Assumption~\ref{ass: cond C'} in Theorem~\ref{thm: loc-2-glob multiple SLE convergence, kappa le 4}, suitable for general $0 < \kappa < 8$, intrinsically requires considering non-conditional measures.

Second, note that whether Assumption~\ref{ass: dr fcns converge to loc mult SLE} holds or not is independent of the choice of conformal maps from $\domain_n$ to the upper half-plane $\bH$, as long as the conformal maps converge to that of $\domain$. Namely, the local multiple SLE is conformally invariant, and on the other hand, so is the weak limit of $\DrFcnLattNotime{n}{j}$, by applying the commutative diagram~\eqref{dia: N-curve commutation with Loewner transform} with different conformal maps.

\subsubsection{\textbf{Convergence to local multiple SLEs}}
\label{subsec: conv to loc N-SLE: statement}



Our first theorem states, roughly, that condition (C), discrete the domain Markov property, and the identification of one initial segment in Assumption~\ref{ass: dr fcns converge to loc mult SLE} together identify the collection of initial segments as a local multiple SLE.


\begin{thm}
\label{thm: local multiple SLE convergence}
Consider the setup of Section~\ref{subsec: setup and notation} with relaxed regularity at marked boundary points, and let the discrete curve models satisfy the assumptions of Section~\ref{subsubsec: setup and assumptions}. Then, the iterated driving functions $( \IterDrFcnLattNotime{n}{1}, \ldots, \IterDrFcnLattNotime{n}{2N} ) \in \ctsfcns^{2N}$ (resp. the initial segments $( \InitSegmLatt{n}{1}, \ldots, \InitSegmLatt{n}{2N} )\in X(\overline{\UnitD})^{2N}$) converge weakly to the iterated driving functions (resp. to the initial segments) of the local multiple SLE in $(\UnitD; \Unitp_1, \ldots, \Unitp_{2N})$, stopped at the continuous exit times of $U_1, \ldots, U_{2N}$, and with partition function $\PartF_N$.
\end{thm}


\begin{cor}
\label{cor: local multiple SLE convergence -  strong topology}
If the regularity at marked boundary points is not relaxed above, then the weak convergence to local multiple SLE also takes place in the sense that initial segments of any subsequential weak limit in topology~(i) of Theorem~\ref{thm: precompactness thm multiple curves} are in distribution equal to the conformal images $( \confmap^{-1}(\InitSegm{1}), \ldots, \confmap^{-1}(\InitSegm{2N)} )) $ of local multiple SLE curves $( \InitSegm{1}, \ldots, \InitSegm{2N} ) $ in $\UnitD$.
%
\end{cor}

\subsubsection{\textbf{Additional assumptions on the discrete curve models}}
\label{subsubsec: approximability}

For the other main result of this section, we need two more assumptions on our discrete curve models, in addition to those in Section~\ref{subsubsec: setup and assumptions}.

First, informally, the discrete curve models must allow increasingly dense-mesh discretizations of any desired limiting domain. This assumption is necessary in the proofs since contrary to the case of local multiple SLEs, we cannot directly rely the  Carath\'{e}odory stability of the scaling limit, but it has to be deduced via the discrete curve models\footnote{
As regards the Carath\'{e}odory stability of local multiple SLE driving functions, we rather consider it as an input from basic SDE theory than SLE theory.
}.

\begin{ass}
\label{ass: approximability}
For any bounded simply-connected domain $(\domain; {p}_1, \ldots, {p}_{2N})$ with marked prime ends that possess radial limits, there exist close lattice approximations $( \Gr^{(n)}; e^{(n)}_1, \ldots, e^{(n)}_{2N} )$ on the graphs $\InfiniteGr^{(n)}$, such that the measures $\PR^{  ( \Gr^{(n)}; e^{(n)}_1, \ldots, e^{(n)}_{2N} )  }$ are defined and $\Gr^{(n)}$ lies inside of $\domain$ for all large enough $n$. 
\end{ass}

Second, in order to handle full curves in stead of initial segments, an initial segment in a very lage neighbourhood must yield some information about the tagret of the curve we are following. There are two alternative assumptions ensuring this, the first one possible only for discrete models corresponding to (multiple) SLEs with parameter $\kappa \le 4$, and the second one applicable for all $0 < \kappa < 8$.

\begin{ass}
\label{ass: quantitative no boundary visits assumption}
For any fixed sequence of Carath\'{e}odory converging graphs $(\Gr_n; e_1^{(n)}, \ldots, e_{2N}^{(n)})$, we have the following.
For any $\delta'> 0$ and any $\eps>0$, taking $\delta$ small enough, the following holds. Denote by $({\gamma}^{(n)}_{\UnitD; 1}, \ldots, {\gamma}^{(n)}_{\UnitD; N}) \in E( \delta, \delta')$ if some of the curves $({\gamma}^{(n)}_{\UnitD; 1}, \ldots, {\gamma}^{(n)}_{\UnitD; N})$ visits the $\delta$-neighbourhood of $\bdry \UnitD$ outside of the $\delta'$-neighbourhoods of its end points (see Figure~\ref{fig: a priori no bdry visits}). Then,
\begin{align*}
\PR^{( \Gr^{(n)}; e^{(n)}_1, \ldots, e^{(n)}_{2N} ) } [ ({\gamma}^{(n)}_{\UnitD; 1}, \ldots, {\gamma}^{(n)}_{\UnitD; N}) \in  E( \delta, \delta') ] < \eps
\end{align*}
for all large enough $n$.
\end{ass}

\begin{figure}
\includegraphics[width=0.5\textwidth]{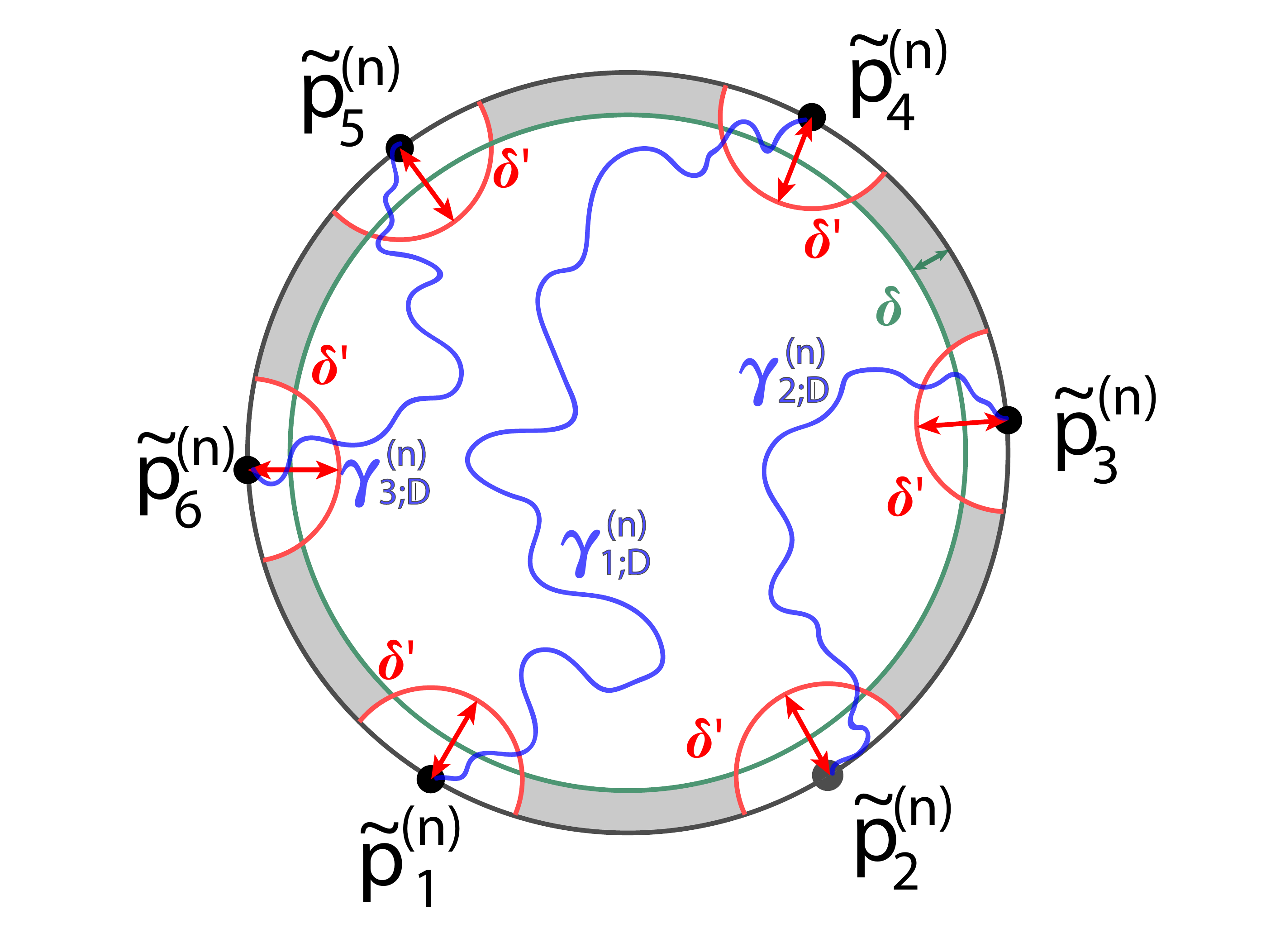}
\caption{
\label{fig: a priori no bdry visits}
Assumption~\ref{ass: quantitative no boundary visits assumption} states that boundary visits of the curves can be excluded \emph{a priori}. Using the notation there, we have  $({\gamma}^{(n)}_{\UnitD; 1}, {\gamma}^{(n)}_{\UnitD; 2}, {\gamma}^{(n)}_{\UnitD; 3}) \not \in E( \delta, \delta')$ in this figure.
}
\end{figure}

The assumption alternative to Assumption~\ref{ass: quantitative no boundary visits assumption} is a slightly improved version of condition (C), which we call condition (C').  

\begin{ass}
\label{ass: cond C'}
The collection of measures with random curves $( \PR^{(\Gr; e_1, \ldots, e_{2N})}, (\gamma_{\Gr;1}, \ldots, \gamma_{\Gr;N} ) )$, indexed by $n$ and all simply-connected subgraphs $(\Gr; e_1, \ldots, e_{2N})$ of $\InfiniteGr_n$ on which the random curve model of $\InfiniteGr_n$ is defined, satisfies condition (C') below.
\end{ass}

Let us formulate condition (C'). For this purpose, we need some new terminology. Consider a measure with random curves $( \PR^{(\Gr; e_1,  \ldots, e_{2N})}, (\gamma_{\Gr;1}, \ldots, \gamma_{\Gr;N} ) )$ on a simply-connected planar graph $(\Gr; e_1, \ldots, e_{2N})$ with $2N$  distinct boundary edges. These edges divide the boundary $\bdry \domain_\Gr$ of the corrresponding planar domain into disjoint \emph{marked boundary arcs} which we call either \emph{odd} or \emph{even}, running counterclockwise between neighbouring boundary edges $e_{2k-1}, e_{2k}$, or $e_{2k}, e_{2k+1}$, respectively (here $e_{2N} = e_0$). Let now $Q$ be a topological quadrilateral on the boundary of $\domain_\Gr$, as defined in Section~\ref{subsubsec: hypotheses of the main precompactness thm}. We say that a crossing of $Q$ by some of the curves $\gamma_{\Gr;1}, \ldots, \gamma_{\Gr;N} $ is \emph{unforced for 
the collection of random curves} $(\gamma_{\Gr;1}, \ldots, \gamma_{\Gr;N} )$ if the sides $S_1, S_3$ of $Q$ that lie on $\bdry \domain_\Gr$ are entirely inside marked boundary arcs of the same parity. Equivalently, the mark boundary edges $e_1, \ldots, e_{2N}$ all land outside of $Q$ on $\bdry \domain_\Gr$, and both connected components of $\domain_\Gr \setminus Q$ contain an even number of them.

\textbf{Condition (C')} is now stated as follows: a collection of measures with random curves $( \PR^{(\Gr; e_1, \ldots, e_{2N})}, (\gamma_{\Gr;1}, \ldots, \gamma_{\Gr;N} ) )$ satisfies condition (C') if for all $\eps > 0$  there exists $M = M(N, \eps) > 0$, such that for any topological quadrilateral $Q$ with $m(Q) \ge M$ on the boundary of $\domain_\Gr$, we have
\begin{align*}
\PR^{(\Gr; e_1, \ldots, e_{2N})} [ \text{crossing of } Q \text{ unforced for 
the collection of random curves } (\gamma_{\Gr;1}, \ldots, \gamma_{\Gr;N} )  ] \le \eps.
\end{align*}
for all the measures with random curves.

Let us make some remarks about condition (C'). First, taking $N=1$, it becomes condition (C) with the stopping time $0$, i.e., for DDMP models. Second, condition (C') is \emph{not} compatible with conditional discrete curve models; it is easy to come up with examples of graphs $(\Gr; e_1, \ldots, e_{2N})$ and $Q$, where conditioning on a link pattern $\alpha$ forces a crossing of $Q$ which is unforced for the collection of random curves.

As a final and most important remark, we give a sufficient alternative condition in terms of annuli. We say that a crossing of a boundary annulus $A(z, r, R)$ by some of the curves $\gamma_{\Gr;1}, \ldots, \gamma_{\Gr;N} $ is unforced for that collection of random curves, if it occurs in a connected component $\mathcal{C}$ of $A(z, r, R) \cap \domain_\Gr$ such that $\bdry \mathcal{C} \cap \bdry \domain_\Gr$ lies entirely inside marked boundary arcs of the same parity.

\textbf{Condition (G')}: a collection of measures with random curves $( \PR^{(\Gr; e_1, \ldots, e_{2N})}, (\gamma_{\Gr;1}, \ldots, \gamma_{\Gr;N} ) )$ satisfies condition (G') if for all $\eps > 0$  there exists $M = M(N, \eps) > 0$, such that for an annulus $A(z, r, R)$ with $R/r \ge M$ on the boundary of $\domain_\Gr$, we have
\begin{align*}
\PR^{(\Gr; e_1, \ldots, e_{2N})} [ \text{crossing of } A(z, r, R) \text{ unforced for 
the collection of random curves } (\gamma_{\Gr;1}, \ldots, \gamma_{\Gr;N} )  ] \le \eps.
\end{align*}
for all the measures with random curves.

\begin{lem} Condition (G') for $R/r \ge M$ implies Condition (C') for $m(Q) \ge 4 (M+ 1)^2$.
\end{lem}

\begin{proof}
The proof is identical to showing that condition (G) implies condition (C) for the one-curve models, see~\cite[Proof of Proposition~2.6]{KS}.
\end{proof}

\subsubsection{\textbf{Convergence to local-to-global multiple SLEs}}
\label{subsec: conv to loc-2-glob N-SLE: statement}

%
%

We now state the main theorem of this section.

\begin{thm}
\label{thm: loc-2-glob multiple SLE convergence, kappa le 4}
Consider the setup of Section~\ref{subsec: setup and notation} with relaxed regularity at marked boundary points, and let discrete curve models satisfy the assumptions of Section~\ref{subsubsec: setup and assumptions}, as well as Assumptions~\ref{ass: approximability}, and either~\ref{ass: quantitative no boundary visits assumption} or~\ref{ass: cond C'}.
Then, the curves $({\gamma}^{(n)}_{\UnitD; 1}, \ldots, {\gamma}^{(n)}_{\UnitD; N})$ converge weakly in $\ctsfcns^N$ to the local-to-global multiple SLE with partition function $\PartF_N$ on the domain $(\UnitD; \Unitp_1, \ldots, \Unitp_{2N})$.
If the regularity at marked boundary points is not relaxed from Section~\ref{subsec: setup and notation}, then also the curves $({\gamma}^{(n)}_{1}, \ldots, {\gamma}^{(n)}_{N})$ converge weakly in $\ctsfcns^N$ to the local-to-global multiple SLE with partition function $\PartF_N$ on the domain $(\domain; p_1, \ldots, p_{2N})$.
\end{thm}

%

\subsubsection{\textbf{Relation to global multiple SLE}}

The final theorems of this section connect the limits of Theorem~\ref{thm: loc-2-glob multiple SLE convergence, kappa le 4} to the global multiple SLEs. We however emphasize once again that \emph{the SLE convergence proof in Theorem~\ref{thm: loc-2-glob multiple SLE convergence, kappa le 4} by no means relies on global multiple SLEs}.

First, since Theorem~\ref{thm: loc-2-glob multiple SLE convergence, kappa le 4} addresses scaling limits unconditional curve models, while the global multiple SLEs address the conditional ones, the two models cannot be the same. Proposition~\ref{prop: relation to global multiple SLE} below however guarantees that the scaling limits from Theorem~\ref{thm: loc-2-glob multiple SLE convergence, kappa le 4} are convex combinations of measures satisfying the Markov stationarity that defines global multiple SLEs for $\kappa \in (0,4]$ and conjecturally also for $\kappa \in (4,8)$.

\begin{prop}
\label{prop: relation to global multiple SLE}
The scaling limits $({\gamma}_{\UnitD; 1}, \ldots, {\gamma}_{\UnitD; N})$ from Theorem~\ref{thm: loc-2-glob multiple SLE convergence, kappa le 4} satisfy the following property: for any $1 \le j \le N$, the regular conditional law of $\gamma_{\UnitD; j}$ given all the other curves is the chordal SLE in between the remaining marked boundary points in the remaining domain; the boundary points almost surely lie adjacent to the same simply-connected component of the complement of the remaining curves in $\UnitD$, so this chordal SLE makes sense. 
\end{prop}

Next, an interesting question is if all the link patterns $\alpha \in \LP_N$ appear with positive probability in the scaling limits of Theorem~\ref{thm: loc-2-glob multiple SLE convergence, kappa le 4}. The answer is positive at least if Assumption~\ref{ass: cond C'} holds. In the sense formalized below, this means that Theorem~\ref{thm: loc-2-glob multiple SLE convergence, kappa le 4} also guarantees convergence of the conditional discrete models to global multiple SLEs.

\begin{thm}
\label{thm: relation to global multiple SLE 1}
Suppose that the assumptions of Theorem~\ref{thm: loc-2-glob multiple SLE convergence, kappa le 4} (with relaxed regularity at marked boundary points), including Assumption~\ref{ass: cond C'}, are satisfied. Then, all link patterns $\alpha \in \LP_N$ appear with positive probability in the scaling limit $({\gamma}_{\UnitD; 1}, \ldots, {\gamma}_{\UnitD; N})$. In particular, the conditional discrete models converge weakly to a measure satisfying the Markov stationarity, which defines the local multiple SLE if $\kappa \le 4$ and cojecturally also for $\kappa \in (4, 8)$.
\end{thm}

Conversely, suppose now that the conditional discrete models are known to converge to global multiple SLEs with $\kappa \in (0,4]$. Theorem~\ref{thm: relation to global multiple SLE 2} below guarantees that the conditional models then also converge in the sense of Theorem~\ref{thm: loc-2-glob multiple SLE convergence, kappa le 4}, to conditional local-to-global multiple SLEs.

To be precise, we say that \emph{the conditional discrete curves converge to global multiple $\SLE(\kappa)$}, for short, if the following holds.
The discrete random curves $({\gamma}^{(n)}_{\UnitD; 1}, \ldots, {\gamma}^{(n)}_{\UnitD; N})$ obtained from corresponding conditional discrete models $\PR^{(n)}_\alpha [\cdot] = \PR^{(n)} [\cdot \vert \alpha]$ converge weakly to the global multiple $\SLE(\kappa)$ on $(\UnitD; \Unitp_1, \ldots, \Unitp_{2N})$ with link pattern $\alpha$, and this holds for any $\alpha \in \LP_N$ and any Carath\'{e}odory converging domain approximations $(\Gr^{(n)}; e_1^{(n)}, \ldots, e_{2N}^{(n)})$. 

\begin{thm}
\label{thm: relation to global multiple SLE 2}
 Suppose that the assumptions imposed on the discrete curve models in the precompactness theorem~\ref{thm: precompactness thm multiple curves} and in (the usually trivial) Assumption~\ref{ass: approximability} are satisfied. Suppose also that the conditional discrete curves of the discrete curve models converge to global multiple $\SLE(\kappa)$, for some $\kappa \le 4$. Then, also the remaining Assumptions~\ref{ass: dr fcns converge to loc mult SLE} and~\ref{ass: quantitative no boundary visits assumption} of Theorem~\ref{thm: loc-2-glob multiple SLE convergence, kappa le 4} are satisfied by the conditional measures $\PR^{(n)}_\alpha $, the former in its conditional form and with the partition functions $\PartF_\alpha$ as given in~\cite[Equation~(3.7)]{PW}. Thus, Theorem~\ref{thm: loc-2-glob multiple SLE convergence, kappa le 4} holds in the sform giving convergence to conditional local-to-global multiple SLEs.
\end{thm}

\subsection{Proof of Theorem~\ref{thm: local multiple SLE convergence}}

\begin{lem}
\label{lem: loc multiple SLE cond exp property}
In the setup of Theorem~\ref{thm: local multiple SLE convergence}, let $(\IterDrFcnNoInd_{ 1}, \ldots, \IterDrFcnNoInd_{2N}) \in \ctsfcns^{2N}$ be any subsequential scaling limit of the iterated driving functions $( \IterDrFcnLattNotime{n}{1}, \ldots, \IterDrFcnLattNotime{n}{2N}  )$. 
Let $f: \ctsfcns^{m} \to \R$, where $1 \le m \le 2N - 1$, and $g: \ctsfcns \to \R$ be bounded continuous test functions. Then, we have
\begin{align*}
\EX [f(\IterDrFcnNoInd_{ 1}, \ldots, \IterDrFcnNoInd_{m}) g ( \IterDrFcnNoInd_{m+ 1} ) ] = \EX [f(\IterDrFcnNoInd_{ 1}, \ldots, \IterDrFcnNoInd_{m}) \EXNSLE_{( \IterDrFcnNoInd_{ 1}, \ldots, \IterDrFcnNoInd_{m})} [ g ( \DrFcnNoInd_{m+ 1} ) ] ],
\end{align*}
where by the random variable $ \DrFcnNoInd_{m+ 1} \in \ctsfcns$ under the measure $\PRNSLE_{( \IterDrFcnNoInd_{ 1}, \ldots, \IterDrFcnNoInd_{m})}$ we mean the driving function of the the local multiple SLE with partition function $\PartF_N$, when growing the initial segment starting from the $(m+1)$:st boundary point, and with the initial configuration of the marked boundary points being where the $m$ first iterated growth processes $(\IterDrFcnNoInd_{ 1}, \ldots, \IterDrFcnNoInd_{m})$ end, up to the stopping time corresponding to the continuous exit time of $U_j$.
\end{lem}

\begin{proof}
Let us assume that a converging subsequence has been extracted, so that 
$( \IterDrFcnLattNotime{n}{1}, \ldots, \IterDrFcnLattNotime{n}{2N}  ) \to (\IterDrFcnNoInd_{ 1}, \ldots, \IterDrFcnNoInd_{2N}) $ weakly in $\ctsfcns^{2N}$.
Start with the triangle inequality,
\begin{align}
\nonumber
\vert 
\EX & [f(\IterDrFcnNoInd_{ 1}, \ldots, \IterDrFcnNoInd_{m}) g ( \IterDrFcnNoInd_{m+ 1} ) ] 
- 
\EX [f(\IterDrFcnNoInd_{ 1}, \ldots, \IterDrFcnNoInd_{m}) \EXNSLE_{ ( \IterDrFcnNoInd_{ 1}, \ldots, \IterDrFcnNoInd_{m})} [ g ( \DrFcnNoInd_{m+ 1} ) ] ]
 \vert  \\
 \label{eq: three-eps proof}
\le & \vert 
\EX [f(\IterDrFcnNoInd_{ 1}, \ldots, \IterDrFcnNoInd_{m}) g ( \IterDrFcnNoInd_{m+ 1} ) ]  
-
 \EX^{(n)} [f(\IterDrFcnLattNotime{n}{1}, \ldots, \IterDrFcnLattNotime{n}{m} ) g ( \IterDrFcnLattNotime{n}{m + 1} ) ] \vert \\
\nonumber
+ & 
 \EX^{(n)} [f(\IterDrFcnLattNotime{n}{1}, \ldots, \IterDrFcnLattNotime{n}{m} ) g ( \IterDrFcnLattNotime{n}{m + 1} ) ]
-
 \EX^{(n)} [f(\IterDrFcnLattNotime{n}{1}, \ldots, \IterDrFcnLattNotime{n}{m} ) 
 \EXNSLE_{(\IterDrFcnLattNotime{n}{1}, \ldots, \IterDrFcnLattNotime{n}{m} )} [ g ( \DrFcnNoInd_{m+ 1}) ] ] \vert   \\
\nonumber
 + &
\vert
 \EX^{(n)} [f(\IterDrFcnLattNotime{n}{1}, \ldots, \IterDrFcnLattNotime{n}{m} ) 
 \EXNSLE_{(\IterDrFcnLattNotime{n}{1}, \ldots, \IterDrFcnLattNotime{n}{m} )} [ g ( \DrFcnNoInd_{m+ 1} ) ] ]
 -
\EX [f(\IterDrFcnNoInd_{ 1}, \ldots, \IterDrFcnNoInd_{m})  \EXNSLE_{(\IterDrFcnNoInd_{ 1}, \ldots, \IterDrFcnNoInd_{m}) } [ g (\DrFcnNoInd_{m+ 1} ) ] ]
\vert.
\end{align}
We claim that, taking a large enough $n$, the right-hand side of~\eqref{eq: three-eps proof} can be made arbitrarily small. The first term becomes arbitrarily small by the weak convergence  $( \IterDrFcnLattNotime{n}{1}, \ldots, \IterDrFcnLattNotime{n}{2N}  ) \to (\IterDrFcnNoInd_{ 1}, \ldots, \IterDrFcnNoInd_{2N}) $, likewise the third one. (This uses the fact that the stopped local multiple SLE driving function $\hat{W}_{m+1}$ is continuous with respect to the initial configuration of the marked boundary points.)

Let us examine the second term of~\eqref{eq: three-eps proof}. First, by the DDMP,
\begin{align}
 \EX^{(n)} [f(\IterDrFcnLattNotime{n}{1}, \ldots, \IterDrFcnLattNotime{n}{m} ) g ( \IterDrFcnLattNotime{n}{m + 1} ) ] 
\label{eq: DDMP cond exp}
& =  \EX^{(n)} [f(\IterDrFcnLattNotime{n}{1}, \ldots, \IterDrFcnLattNotime{n}{m} )
 \EX^{(n)}_{(\IterDrFcnLattNotime{n}{1}, \ldots, \IterDrFcnLattNotime{n}{m} )} [ g ( \IterDrFcnLattNotime{n}{m+1} ) ] ],
\end{align}
where we denoted by $ \EX^{(n)}_{(\IterDrFcnLattNotime{n}{1}, \ldots, \IterDrFcnLattNotime{n}{m} )}$ the measure from the discrete random curve model on $\InfiniteGr_n$, on the graph obtained by reducing the original graph $ \Gr^{(n)}$ by the initial segments described by the driving functions $(\IterDrFcnLattNotime{n}{1}, \ldots, \IterDrFcnLattNotime{n}{m} )$. Under this measure, $\IterDrFcnLattNotime{n}{m+1}$ is the driving function of conformal image of the $(m+1)$:st curve initial segment $\InitSegmLatt{n}{m+1}$, after mapping-out of the previous initial segments $\InitSegmLatt{n}{1}, \ldots, \InitSegmLatt{n}{m}$.

Let us state the next step of the proof as a separate lemma.

\begin{lem}
\label{lem: uniform convergence}
For any fixed compact set $K \subset \ctsfcns^{m}$, we have the convergence
\begin{align*}
 \EX^{(n)}_{(\IterDrFcnLattNotime{n}{1}, \ldots, \IterDrFcnLattNotime{n}{m} )} [ g ( \IterDrFcnLattNotime{n}{m+1} ) ] 
\stackrel{n \to \infty }{\longrightarrow }
 \EXNSLE_{(\IterDrFcnLattNotime{n}{1}, \ldots, \IterDrFcnLattNotime{n}{m} )} [ g ( \DrFcnNoInd_{m+ 1}) ] ]
\qquad
\text{as } n \to \infty,
\end{align*}
uniformly over $(\IterDrFcnLattNotime{n}{1}, \ldots, \IterDrFcnLattNotime{n}{m} )$ describing possible initial segments and belonging to $K$.
\end{lem}

\begin{proof}
Assume for a contradiction than such uniform convergence does not occur, i.e., for infinitely many $n$, there exist deterministic iterated driving functions $( \DetIterDrFcn^{(n)}_{1 } , \ldots, \DetIterDrFcn^{(n)}_{m} ) \in K$ that can each appear as iterated driving functions $(\IterDrFcnLattNotime{n}{1}, \ldots, \IterDrFcnLattNotime{n}{m} )$ of the initial segments in our lattice models (i.e., they describe lattice curves), and
\begin{align}
\label{eq: counter-assumption}
\vert \EX^{(n)}_{ ( \DetIterDrFcn^{(n)}_{1 } , \ldots, \DetIterDrFcn^{(n)}_{m} ) } [ g ( \IterDrFcnLattNotime{n}{m+1} ) ] 
- \EXNSLE_{ ( \DetIterDrFcn^{(n)}_{1 } , \ldots, \DetIterDrFcn^{(n)}_{m} ) } [ g ( \DrFcnNoInd_{m+ 1}) ] \vert > \delta
\end{align} 
for some $\delta > 0$.

Now, by compactness, we may extract a convergent subsequence (which we suppress in notation), $ ( \DetIterDrFcn^{(n)}_{1 } , \ldots, \DetIterDrFcn^{(n)}_{m} ) \to ( \DetIterDrFcn_{1 } , \ldots, \DetIterDrFcn_{m} ) $. In Assumption~\ref{ass: dr fcns converge to loc mult SLE}, we assumed that the convergence of a single driving function to local multiple SLE is verified, so this implies\footnote{
The stopping times used here are slightly different that in Assumption~\ref{ass: dr fcns converge to loc mult SLE}. However, this does not change the weak convergence by Remark~\ref{rem: weak conv of stopped dr fcns w different stopping times}.
}
\begin{align*}
\EX^{(n)}_{ ( \DetIterDrFcn^{(n)}_{1 } , \ldots, \DetIterDrFcn^{(n)}_{m} ) } [ g ( \IterDrFcnLattNotime{n}{m+1} ) ] 
\stackrel{n \to \infty }{\longrightarrow }
\EXNSLE_{ ( \DetIterDrFcn_{1 } , \ldots, \DetIterDrFcn_{m} ) } [ g ( \DrFcnNoInd_{m+ 1}) ].
\end{align*}
On the other hand, the continuity of the local multiple SLE driving function with respect to the initial configuration implies
\begin{align*}
\EXNSLE_{ ( \DetIterDrFcn^{(n)}_{1 } , \ldots, \DetIterDrFcn^{(n)}_{m} ) } [ g ( \DrFcnNoInd_{m+ 1}) ] 
\stackrel{n \to \infty }{\longrightarrow }
\EXNSLE_{ ( \DetIterDrFcn_{1 } , \ldots, \DetIterDrFcn_{m} ) } [ g ( \DrFcnNoInd_{m+ 1}) ].
\end{align*}
These two convergences contradict~\eqref{eq: counter-assumption}, proving the lemma.
\end{proof}

Let us now finish the proof of Lemma~\ref{lem: loc multiple SLE cond exp property}, by bounding the second term of~\eqref{eq: three-eps proof}. First, by Prohorov's theorem, weak convergence implies tightness. Since the functions $f$ and $g$ are bounded, we can thus with arbitrarily small error assume that $(\IterDrFcnLattNotime{n}{1}, \ldots, \IterDrFcnLattNotime{n}{m} ) \in K$ for a suitable compact set $K \subset \ctsfcns^{m}$. Applying then~\eqref{eq: DDMP cond exp} and Lemma~\ref{lem: uniform convergence}, we observe that the the second term of~\eqref{eq: three-eps proof} tends to zero as $n \to \infty$. This concludes the proof.
\end{proof}

\begin{proof}[Proof of Theorem~\ref{thm: local multiple SLE convergence}] By the commutative diagram~\eqref{dia: local multiple SLE commutation}, it suffices to prove the weak convergence of the iterated driving function.
Consider a subsequential weak limit $(\IterDrFcnNoInd_{ 1}, \ldots, \IterDrFcnNoInd_{2N})$.
Lemma~\ref{lem: loc multiple SLE cond exp property}, with fixed $g: \ctsfcns \to \R$, holds for all  continuous and bounded $f: \ctsfcns^m \to \R$. Thus, the weak limit $(\IterDrFcnNoInd_{ 1}, \ldots, \IterDrFcnNoInd_{2N}) $ satisfies
\begin{align}
\label{eq: conditional expectation property of loc mult SLE}
\EX [g (\IterDrFcnNoInd_{ m+ 1} ) \; \vert \; \sigma (\IterDrFcnNoInd_{ 1}, \ldots, \IterDrFcnNoInd_{m}) ] = \EXNSLE_{(\IterDrFcnNoInd_{ 1}, \ldots, \IterDrFcnNoInd_{m})} [ g (  \DrFcnNoInd_{m+ 1} ) ].
\end{align}
The right-hand side is a continuous function of $(\IterDrFcnNoInd_{ 1}, \ldots, \IterDrFcnNoInd_{m})$ by the stability of local multiple SLE with respect to the initial configuration.

By Proposition~\ref{prop: conditional expectation determine conditional law} from the appendices, the fact that~\eqref{eq: conditional expectation property of loc mult SLE} holds for all continuous functions $g: \ctsfcns \to \R$ means that the regular conditional law of the $(m+1)$:st iterated driving function $\IterDrFcnNoInd_{ m+ 1}$ given the previous ones $ (\IterDrFcnNoInd_{ 1}, \ldots, \IterDrFcnNoInd_{m}) $ is the local multiple SLE growth driving function, launched from the boundary point configuration where the previous ones $ (\IterDrFcnNoInd_{ 1}, \ldots, \IterDrFcnNoInd_{m}) $ end. Inductively on $m$, this shows that $ (\IterDrFcnNoInd_{ 1}, \ldots, \IterDrFcnNoInd_{m}) $ are local multiple SLE iterated driving functions.
\end{proof}

\begin{proof}[Proof of Corollary~\ref{cor: local multiple SLE convergence -  strong topology}]
%
By the commutative diagram~\eqref{dia: N-curve commutation with conf maps}, depicting Theorem~\ref{thm: precompactness thm multiple curves}(B), the initial segments of any subsequential weak limit $(\gamma_1, \ldots, \gamma_{N})$ are the conformal images of the initial segments $( \InitSegm{1}, \ldots, \InitSegm{2N} ) $ on $\UnitD$.
\end{proof}

\subsection{Proof of Theorem~\ref{thm: loc-2-glob multiple SLE convergence, kappa le 4} under Assumption~\ref{ass: quantitative no boundary visits assumption}}

In this subsection, we present the proof of Theorem~\ref{thm: loc-2-glob multiple SLE convergence, kappa le 4} under Assumption~\ref{ass: quantitative no boundary visits assumption}. This proof is easier and notationally lighter than the one under the alternative assumption~\ref{ass: cond C'}. By ``assumptions of Theorem~\ref{thm: loc-2-glob multiple SLE convergence, kappa le 4}'' we refer to the set of assumptions with Assumption~\ref{ass: quantitative no boundary visits assumption}.

\subsubsection{\textbf{Identifying the scaling limit of one-curve marginals}}
\label{subsubsec: identify one curve marginals}

Note that by of Theorem~\ref{thm: precompactness thm multiple curves} and  Remark~\ref{rem: precompactness for irregular boundary}, the curves $({\gamma}_{\UnitD; 1}^{(n)}, \ldots, {\gamma}_{\UnitD; N}^{(n)})$ are precompact.
Fix a subsequential scaling limit $({\gamma}_{\UnitD; 1}, \ldots, {\gamma}_{\UnitD; N})$, and consider the marginal law of one curve. For notational simplicity, we choose this special curve to be ${\gamma}_{\UnitD; 1}$ in this and following computations, but the straightforward analogues hold for all curves ${\gamma}_{\UnitD; j}$, as well as their reversals. Let us denote by $\InitSegmDelta{\delta}$ the intial segment of ${\gamma}_{\UnitD; 1}$, up to the continuous exit time of a very large localization neighbourhood $U(\delta)$ of the first boundary point, consisting of all of $\UnitD$ except a $\delta$-neighbourhood of the arc $( \Unitp_2 \Unitp_{2N} )$ of the other 
boundary points. Note that by Assumption~\ref{ass: dr fcns converge to loc mult SLE}, $\InitSegmDelta{\delta}$ is described by the local multiple SLE growth process, for any subsequential scaling limit. 

\begin{lem}
\label{lem: from localization to one curve}
Under the setup and assumptions of Theorem~\ref{thm: loc-2-glob multiple SLE convergence, kappa le 4}, the curve ${\gamma}_{\UnitD; 1}$ almost surely visits the boundary $\bdry \UnitD$ only at its end points, and only at times $0$ and $1$, and $\InitSegmDelta{\delta} \to \gamma_{\UnitD; 1}$ almost surely as $\delta \shrinkto 0$.
\end{lem}

\begin{proof}
The property that the curve ${\gamma}_{\UnitD; 1}$ almost surely visits the boundary $\bdry \UnitD$ only at its end points follows straightforwardly from Assumption~\ref{ass: quantitative no boundary visits assumption} and Portmanteau's theorem on weak convergence. The fact that these boundary visits occur at times $0$ and $1$, i.e., the end points are almost surely not double points of the curve ${\gamma}_{\UnitD; 1}$, follows from the annulus crossing estimates~\cite[Theorem~1.5]{KS}.

The convergence $\InitSegmDelta{\delta} \to {\gamma}_{\UnitD; 1}$ almost surely as $\delta \shrinkto 0$ is proven by the following argument: the curve $\InitSegmDelta{\delta}$ is the initial segment of ${\gamma}_{\UnitD; 1}$, as grown up to the continuous exit time $\tau(\delta)$ of $U(\delta)$. For a fixed realization of ${\gamma}_{\UnitD; 1}$, it is easy to show that as $\delta \shrinkto 0$, the stopping times $\tau(\delta)$ tend to the hitting time of the arc $( \confmap(p_2) \confmap(p_{2N}) )$. Since the boundary visits of ${\gamma}_{\UnitD; 1}$ a.s. only occur at times $0$ and $1$, we deduce that, a.s., $\tau(\delta) \to 1$ as $\delta \shrinkto 0$, and thus also $\InitSegmDelta{\delta} \to {\gamma}_{\UnitD; 1}$.
\end{proof}

\begin{cor}
\label{cor: scaling limits of marginal laws identified}
In the setup of Theorem~\ref{thm: loc-2-glob multiple SLE convergence, kappa le 4}, the marginal law of the curve ${\gamma}_{\UnitD; 1}$ is the same for all subsequential scaling limits $({\gamma}_{\UnitD; 1}, \ldots, {\gamma}_{\UnitD; N})$. It is the local multiple SLE which, when defined in the increasing neighbourhoods $U(\delta)$ as $\delta \shrinkto 0$, almost surely yields a continuous closed curve between two marked boundary points.
\end{cor}

Note that the above properties of the local multiple SLE initial segment from $\Unitp_1$ to $(\Unitp_2 \Unitp_{2N})$ would be difficult to prove directly by SLE theory, but are now easy by the underlying full curve ${\gamma}_{\UnitD; 1}$ obtained from the lattice model.

\begin{proof}[Proof of Corollary~\ref{cor: scaling limits of marginal laws identified}] Let us first identify the marginal scaling limit ${\gamma}_{\UnitD; 1}$.
It suffices to show that for any bounded continuous function $g: X(\overline{\UnitD}) \to \R$, the expectation $\EX[g ({\gamma}_{\UnitD; 1} ) ]$ is the same for all subsequential scaling limits $({\gamma}_{\UnitD; 1}, \ldots, {\gamma}_{\UnitD; N})$. Almost sure convergence implies weak convergence, so by Lemma~\ref{lem: from localization to one curve}, we have
\begin{align*}
\EX[g ( \gamma_{\UnitD; 1} ) ] = \lim_{\delta \shrinkto 0} \EX[ g( \InitSegmDelta{\delta} )],
\end{align*}
for any subsequential limit ${\gamma}_{\UnitD; 1}$.
By Assumption~\ref{ass: dr fcns converge to loc mult SLE}, $\InitSegmDelta{\delta}$ is the local multiple SLE initial segment, and in particular, the right-hand side above the same for any subsequential scaling limit. The fact that the local multiple SLE determines a full curve is an immediate consequence of Lemma~\ref{lem: from localization to one curve}.
\end{proof}

\subsubsection{\textbf{Identifying the full scaling limit}}

We now prove the following statements inductively on the number of curves $N$. Statement~(ii) in the proposition below is Theorem~\ref{thm: loc-2-glob multiple SLE convergence, kappa le 4}.

\begin{prop}
\label{thm: k le 4 local-to-global NSLE}
Under the setup and assumptions of Theorem~\ref{thm: loc-2-glob multiple SLE convergence, kappa le 4}, the following hold.
\begin{itemize}
\item[i) ] The local-to-global multiple SLE with partition function $\PartF_N$, on any domain $(\domain; {p}_1, \ldots, {p}_{2N})$ with $2N$ distinct marked prime ends with radial limits, exists as a random variable in $X(\C)^N$. Furthermore, interpreting its conditional-law definition as a sampling procedure of curve initial segments from the local multiple SLE in a given order, sampling the initial or final segments as local multiple SLEs in any order yields the same distribution of full curves.
\item[ii) ] The weak limits $ ( {\gamma}_{\UnitD; 1}, \ldots, {\gamma}_{\UnitD; N} )$ and $ ( {\gamma}_{1}, \ldots, {\gamma}_{N} )$ (the latter only when considering non-relaxed regularity at marked boundary points) of the curves $ ( \gamma^{(n)}_{\UnitD; 1}, \ldots, \gamma^{(n)}_{\UnitD; N} )$ and $ ( {\gamma}_{1}^{(n)}, \ldots, {\gamma}^{(n)}_{N} )$, are local-to-global multiple SLE with partition function $\PartF_N$, on domains $(\UnitD; \Unitp_1, \ldots, \Unitp_{2N})$ and $(\domain; {p}_1, \ldots, {p}_{2N})$, respectively. 
\item[iii)] The multiple SLE of part~(i) above is Carath\'{e}odory stable in the following precise sense: if  $(\domain_m; p_1^{(m)}, \ldots, p^{(m)}_{2N})$ and $(\domain; p_1, \ldots, p_{2N})$ are uniformly bounded simply-connected planar domains with $2N$ distinct marked prime ends with radial limits, the former being close Carath\'{e}odory approximations of the latter as $m \to \infty$, then 
 the local-to-global multiple SLEs on $(\domain_m; {p}_1^{(m)}, \ldots, {p}_{2N}^{(m)})$ converge weakly in $X(\C)^N$ to local-to-global multiple SLE on  $(\domain; {p}_1, \ldots, {p}_{2N})$. 
\end{itemize}
\end{prop}

\begin{proof}[Proof of the base case $N=1$]
For the base case $N=1$, we will prove the claim using the weak convergence $\gamma_{\UnitD; 1}^{(n)} \to \gamma_{\UnitD; 1}$ of the underlying lattice model\footnote{
The statement in the special case $N=1$ can be seen as more or less standard properties of chordal SLEs. Instead, we use here arguments that are less standard for chordal SLEs, but generalize to $N \ge 2$.
}. For the existence in part (i),
in the preceding subsection, we defined the marginal law of one curve in the local-to-global multiple SLE as the weak limit of the curves ${\gamma}_{\UnitD; 1}^{(n)}$. Thus, if there is only one curve, the local-to-global multiple SLE on the unit disc exists as this weak limit. The existence in general domains follows by Assumption~\ref{ass: approximability} and Theorem~\ref{thm: precompactness thm multiple curves}(B). For the order of sampling in part (i), there are two possible starting points from which we can grow the curve $ \gamma_{\UnitD; 1}$ when sampling with the Loewner growth processes. Lemma~\ref{lem: from localization to one curve} and the discussion on one-curve marginals holds for both starting points, so sampling the local multiple SLE growth from either starting point, we get the Loewner description of the limiting curve $ \gamma_{\UnitD; 1}$. This finishes part (i) in the base case $N=1$.

As $N=1$, part (ii) follows directly from the identification of one-curve marginals in Corollary~\ref{cor: scaling limits of marginal laws identified}.

Part~(iii) can be proven by the following argument, relying on Assumption~\ref{ass: approximability}: For each $m$, let $\gamma^{(n)}_m$ be the discrete curves on the lattice approximations $(\domain_{n; m}; {p}_1^{(n, m)}, {p}_{2}^{(n, m)})$ of $(\domain_m; {p}_1^{(m)}, {p}_{2}^{(m)})$ given by Assumption~\ref{ass: approximability}, and let $\gamma_m$ denote the local-to-global multiple SLE on  $(\domain_m; {p}_1^{(m)}, {p}_{2}^{(m)})$, i.e., the weak limit of $\gamma^{(n)}_m$ as $n \to \infty$. Recall that weak convergence is metrizable. It is easy to see that one can define inductively an increasing sequence $n(m)$ such that $(\domain_{n(m); m}; {p}_1^{(n(m), m)}, {p}_{2}^{(n(m), m)})$ are close Carath\'{e}odory approximations of $(\domain; {p}_1, {p}_{2})$, and the distance of $\gamma^{(n(m))}_m$ and $\gamma_m$ in the metric of weak convergence tends to zero as $m \to \infty$. By part~(ii), $\gamma^{(n(m))}_m$ tends weakly to $\gamma$ as $m \to \infty$, so also the distance of $\gamma$ and $\gamma_m$ in the metric of weak convergence tends to zero as $m \to \infty$.
\end{proof}

\begin{proof}[Proof of the induction step]
Let us now assume that the three properties in the statement of Proposition~\ref{thm: k le 4 local-to-global NSLE} hold for any number of curves $1, 2, \ldots, (N-1)$, and show that they then also hold for $N$ curves. We first prove property~(ii). Let us start with an analogy of Lemma~\ref{lem: loc multiple SLE cond exp property}. 

\begin{lem}
\label{lem: domain Markov of scaling limit curves}
Any subsequential limit $ ( {\gamma}_{\UnitD; 1}, \ldots, {\gamma}_{\UnitD; N} )$ of the curves $( {\gamma}^{(n)}_{\UnitD; 1}, \ldots, {\gamma}^{(n)}_{\UnitD; N} ) $ satisfies, for any bounded continuous functions $g: X(\overline{\UnitD})^{N-1} \to \R$, and $f: (\overline{\UnitD}) \to \R$,
\begin{align*}
\EX [f ( {\gamma}_{\UnitD; 1} )  g( {\gamma}_{\UnitD; 2}, \ldots, {\gamma}_{\UnitD; N}   ) ] = \EX [ f ( {\gamma}_{\UnitD; 1} )  \EXSLEcurves{(N-1)}_{\UnitD \setminus {\gamma}_{\UnitD; 1}  } [ g( \SLEcurve_{1},  \ldots, \SLEcurve_{N-1}   ) ] ],
\end{align*}
where $\EXSLEcurves{(N-1)}_{\UnitD \setminus {\gamma}_{\UnitD; 1}  }$ on the right-hand side denotes the following: $( \SLEcurve_{1},  \ldots, \SLEcurve_{N-1})$ are the local-to-global multiple SLE of $(N-1)$ curves on $\UnitD \setminus \gamma_{\UnitD; 1} $ with the remaining marked boundary points. If $\UnitD \setminus \gamma_{\UnitD; 1} $ is not simply connected, it should be interpreted as two independent local-to-global multiple SLEs on the two connected components of $\UnitD \setminus \gamma_{\UnitD; 1} $ that are adjacent to the remaining marked boundary points.
\end{lem}

Note that the one or two connected components of $\UnitD \setminus {\gamma}_{\UnitD; 1} $ adjacent to the remaining marked boundary points are (almost surely) simply-connected by Lemma~\ref{lem: from localization to one curve}, and the local-to-global multiple SLEs on them exist by the inductive assumption~(i).

\begin{proof}[Proof of Lemma~\ref{lem: domain Markov of scaling limit curves}]
Assume for notational simplicity that a weakly converging subsequence has been extracted, so that $ ( {\gamma}^{(n)}_{\UnitD; 1}, \ldots, {\gamma}^{(n)}_{\UnitD; N} ) \to  ( {\gamma}_{\UnitD; 1}, \ldots, {\gamma}_{\UnitD; N} )$,
Start with the triangle inequality:
\begin{align}
\nonumber
\vert   
\EX &[f ( {\gamma}_{\UnitD; 1} )  g( {\gamma}_{\UnitD; 2}, \ldots, {\gamma}_{\UnitD; N}   ) ] 
-
 \EX [ f ( {\gamma}_{\UnitD; 1} )  \EXSLEcurves{(N-1)}_{\UnitD \setminus {\gamma}_{\UnitD; 1}  } [ g( \SLEcurve_{1},  \ldots, \SLEcurve_{N-1}   ) ] ]
 \vert  \\
  \label{eq: another three eps proof}
 \le & \vert   \EX [f ( {\gamma}_{\UnitD; 1} )  g( {\gamma}_{\UnitD; 2}, \ldots, {\gamma}_{\UnitD; N}   ) ] 
  - 
  \EX^{(n)} [ f ( \gamma^{(n)}_{\UnitD; 1}  )   g( \gamma^{(n)}_{\UnitD; 2}, \ldots, \gamma^{(n)}_{\UnitD; N} )  ] \vert \\
 \nonumber
 &
 + \vert   \EX^{(n)} [ f ( \gamma^{(n)}_{\UnitD; 1}  )   g( \gamma^{(n)}_{\UnitD; 2}, \ldots, \gamma^{(n)}_{\UnitD; N} )  ] 
 -
  \EX^{(n)} [ f ( \gamma_{\UnitD; 1}^{(n)}  )  \EXSLEcurves{(N-1)}_{\UnitD \setminus \gamma_{\UnitD; 1}^{(n)}  } [ g(\SLEcurve_1, \ldots, \SLEcurve_{N-1}   ) ] ] \vert \\
\nonumber
 &
+ \vert   \EX^{(n)} [ f ( \gamma_{\UnitD; 1}^{(n)}  )  \EXSLEcurves{(N-1)}_{\UnitD \setminus \gamma_{1; \UnitD}^{(n)}  } [ g(\SLEcurve_1, \ldots, \SLEcurve_{N-1} ) ] ] 
 - \EX [ f ( \gamma_{\UnitD; 1}  )  \EXSLEcurves{(N-1)}_{\UnitD \setminus \gamma_{1; \UnitD  } } [ g(\SLEcurve_1, \ldots, \SLEcurve_{N-1}   ) ] ] \vert,
\end{align}
where the SLE curves on both $\UnitD \setminus \gamma_{\UnitD; 1}$ and $\UnitD \setminus \gamma_{\UnitD; 1}^{(n)}$ run between the limitng marked points $\Unitp_1, \ldots, \Unitp_{2N} \in \bdry \UnitD$.

We claim that all terms in~\eqref{eq: another three eps proof} can be made arbitrarily small by choosing $n$ large enough. For the first term, this holds by the weak convergence. Likewise, the third term follows by weak convergence: namely,
\begin{align*}
\gamma_{\UnitD; 1} \mapsto \EXSLEcurves{(N-1)}_{\UnitD \setminus \gamma_{\UnitD; 1}  } [ g (\SLEcurve_1, \ldots, \SLEcurve_{N-1}   ) ] ]
\end{align*}
is a continuous function of the curve $\gamma_{\UnitD; 1}$ not visiting $\bdry \UnitD$ except at its end points, by the inductive assumption~(iii).

For the second term, notice that by the DDMP,
\begin{align*}
\EX^{(n)} [ f ( \gamma^{(n)}_{\UnitD; 1}  )   g( \gamma^{(n)}_{\UnitD; 2}, \ldots, \gamma^{(n)}_{\UnitD; N} )  ]
 = \EX^{(n)} [ f ( \gamma^{(n)}_{\UnitD; 1}  )  \EX^{(n)}_{\UnitD \setminus \gamma^{(n)}_{\UnitD; 1} } [ g( \gamma^{(n)}_{\UnitD; 2}, \ldots, \gamma^{(n)}_{\UnitD; N} )  ] ]
\end{align*}
The rest is similar to the proof of Lemma~\ref{lem: loc multiple SLE cond exp property}: by tightness of the sequence $\gamma^{(n)}_1  $, we take a compact set $K_\eps \subset X(\overline{\UnitD})$, containing $1- \eps$ of probability mass  of $\gamma^{(n)}_{\UnitD; 1}$ for all $n$. It then suffices to show that the convergence 
\begin{align}
\label{eq: unif conv}
\EX^{(n)}_{\UnitD \setminus \gamma_{\UnitD; 1}^{(n)}  } [ g( \gamma^{(n)}_{\UnitD; 2}, \ldots, \gamma^{(n)}_{\UnitD; N} )  ] \to 
\EXSLEcurves{(N-1)}_{\UnitD \setminus {\gamma}_{\UnitD; 1}^{(n)}  } [ g( \SLEcurve_{1},  \ldots, \SLEcurve_{N-1}   ) ]
\end{align}
is uniform over $ \gamma_{\UnitD; 1}^{(n)} \in K_\eps$. The intersection of a compact set with a closed set is compact. Thus, by Assumption~\ref{ass: quantitative no boundary visits assumption}, we may assume that $K_\eps$ is such that  $ \gamma_{\UnitD; 1}^{(n)}  $ never visits at distance $< \delta$-neighbourhood of $\bdry \UnitD$, except at disctance $\le \delta'$ from its end points. 

 Assume for a contradiction that the convergence~\eqref{eq: unif conv} is not uniform over $ K_\eps$. I.e., there exist deterministic curves $\DetCurve^{(n)} \in K_\eps$, each possible to be observed as the curves $\gamma_{\UnitD; 1}^{(n)} $, such that for some $\ell > 0$ and infinitely many $n$ we have
\begin{align}
\label{eq: yet another counterassumption}
\vert \EX^{(n)}_{\UnitD \setminus \DetCurve^{(n)}  } [  g( \gamma^{(n)}_{\UnitD; 2}, \ldots, \gamma^{(n)}_{\UnitD; N} )] - \EXSLEcurves{(N-1)}_{\UnitD \setminus \DetCurve^{(n)}  } [ g( \SLEcurve_{1},  \ldots, \SLEcurve_{N-1}   )  ] \vert > \ell.
\end{align}
By compactness, extract a convergent subsequence (which we suppress in notation), $\DetCurve^{(n)} \to \DetCurve$ for which~\eqref{eq: yet another counterassumption} holds. Now, by the inductive assumption~(iii), we have 
\begin{align}
\label{eq: conv for contradiction 1}
\EXSLEcurves{(N-1)}_{\UnitD \setminus \DetCurve^{(n)}  } [ g( \SLEcurve_{1},  \ldots, \SLEcurve_{N-1}   )  \to \EXSLEcurves{(N-1)}_{\UnitD \setminus \DetCurve  } [ g( \SLEcurve_{1},  \ldots, \SLEcurve_{N-1}   )  ],
\end{align}
where the multiple SLE in $\UnitD \setminus \DetCurve$ makes sense as we restricted the boundary visits of $\DetCurve$ by our choice of $K_\eps$.

Now, recall that the curves $\gamma^{(n)}_2, \ldots, \gamma^{(n)}_N  $ originate in a DDMP lattice model where the one-curve model satisfies the conformally invariant condition (C). DDMP and condition (C) clearly also hold for the conformal images $\gamma^{(n)}_{\UnitD; 2}, \ldots, \gamma^{(n)}_{\UnitD; N} $. Also the assumptions imposed on the discrete domains in Theorem~\ref{thm: precompactness thm multiple curves} hold for the curves $\gamma^{(n)}_{\UnitD; 2}, \ldots, \gamma^{(n)}_{\UnitD; N} $, and thus we can deduce precompactness by that theorem. Now, by inductive assumption~(ii), the conformal images of the curves $\gamma^{(n)}_2, \ldots, \gamma^{(n)}_N  $, when the two connected components of $\confmap_n^{-1 } ( \UnitD \setminus \DetCurve^{(n)} )$ are both mapped to $\UnitD$, tend weakly to a local-to-global multiple SLE. But these are also the conformal images of $\gamma^{(n)}_{\UnitD; 2}, \ldots, \gamma^{(n)}_{\UnitD; N} $ on $\UnitD$, so by Theorem~\ref{thm: precompactness thm multiple curves}(B), also 
the curves $(\gamma^{(n)}_{\UnitD; 2}, \ldots, \gamma^{(n)}_{\UnitD; N} )$ converge weakly to the local-to-global multiple SLE in the two connected components of $\UnitD \setminus \DetCurve$. Thus,
\begin{align}
\label{eq: conv for contradiction 2}
\EX^{(n)}_{\UnitD \setminus \DetCurve^{(n)}  } [  g( \gamma^{(n)}_{\UnitD; 2}, \ldots, \gamma^{(n)}_{\UnitD; N} )]  \to \EXSLEcurves{(N-1)}_{\UnitD \setminus \DetCurve  } [ g( \SLEcurve_{1},  \ldots, \SLEcurve_{N-1}   )  ].
\end{align}

The two convergences~\eqref{eq: conv for contradiction 1} and~\eqref{eq: conv for contradiction 2} together contradict~\eqref{eq: yet another counterassumption}, finishing the proof.
\end{proof}

Let us now finish the induction step in the proof of Proposition~\ref{thm: k le 4 local-to-global NSLE}. Notice that Lemma~\ref{lem: domain Markov of scaling limit curves} implies that for any bounded continuous function $g: X^{N} \to \R$, we have
\begin{align*}
\EX[ g( {\gamma}_{\UnitD; 2}, \ldots, {\gamma}_{\UnitD; N}   ) \; \vert \; \sigma ( {\gamma}_{\UnitD; 1} ) ]
= \EXSLEcurves{(N-1)}_{\UnitD \setminus {\gamma}_{\UnitD; 1}  } [ g( \SLEcurve_{1},  \ldots, \SLEcurve_{N-1}  ) ].
\end{align*}
(By inductive assumption~(iii), the right-hand side above is a continuous function of ${\gamma}_{\UnitD; 1}$.)
 By Proposition~\ref{prop: conditional expectation determine conditional law}, this shows that the regular conditional law of $({\gamma}_{\UnitD; 2}, \ldots, {\gamma}_{\UnitD; N}   )$ given ${\gamma}_{\UnitD; 1}$ is the local-to-global multiple SLE of $(N - 1)$ curves. Since the marginal law of ${\gamma}_{\UnitD; 1}$ is by Corollary~\ref{cor: scaling limits of marginal laws identified} the one-curve marginal of the local-to-global multiple SLE of $N$ curves, this identifies the law of the curves $( {\gamma}_{\UnitD; 1}, \ldots, {\gamma}_{\UnitD; N}   )$ as the local-to-global multiple SLE. For domains with non-relaxed regularity at marked boundary points, the weak convergence of $(\gamma_1^{(n)}, \ldots, \gamma_N^{(n)})$ is guaranteed by Theorem~\ref{thm: precompactness thm multiple curves}(B). This proves the induction step for property~(ii).
 
 For property~(i), the local-to-global multiple SLE exists as the scaling limit in the proof of property~(ii). The independence on sampling order follows from that in the discrete case. The proof of property~(iii) is identical to the base case $N=1$.
\end{proof}

\subsection{Proof of Theorem~\ref{thm: loc-2-glob multiple SLE convergence, kappa le 4} under Assumption~\ref{ass: cond C'}}

We now prove Theorem~\ref{thm: loc-2-glob multiple SLE convergence, kappa le 4} under Assumption~\ref{ass: cond C'}.

\subsubsection{\textbf{Identifying the scaling limit of up-to-swallowing initial segments}}
\label{subsubsec: up-to-swalliwing segments}

Let us start an analogue of Section~\ref{subsubsec: identify one curve marginals}. Continue in the notation introduced there. Denote by $\InitSegmDelta{0}$ the initial segment of ${\gamma}_{\UnitD; 1}$ up to hitting the closed boundary arc $( \Unitp_2 \Unitp_{2N} ) \subset \bdry \UnitD$. We call $\InitSegmDelta{0}$ the \emph{up-to-swallowing initial segment} of ${\gamma}_{\UnitD; 1}$.  Denote by $\FinalSegmDelta{\delta}$ and $\FinalSegmDelta{0}$ the remainder of the curve ${\gamma}_{\UnitD; 1}$ after the initial segments $\InitSegmDelta{\delta}$ and $\InitSegmDelta{0}$, respectively. Repeating the arguments of Lemma~\ref{lem: from localization to one curve} and Corollary~\ref{cor: scaling limits of marginal laws identified}, one readily obtains:

\begin{lem}
\label{lem: SLE up-to-swallowing segments exist}
For any subsequential scaling limit $({\gamma}_{\UnitD; 1}, \ldots, {\gamma}_{\UnitD; N})$, we have $\InitSegmDelta{\delta} \to \InitSegmDelta{0}$ and  $\FinalSegmDelta{\delta} \to \FinalSegmDelta{0}$ almost surely as $\delta \shrinkto 0$. In particular, the marginal law of $\InitSegmDelta{0}$ is the local multiple SLE for all subsequential limits.
\end{lem}

By the almost sure convergence above, the local multiple SLE initial segment from $\Unitp_1$ to $(\Unitp_2 \Unitp_{2N})$ exists as a closed curve and yields $\InitSegmDelta{0}$, analogously to Corollary~\ref{cor: scaling limits of marginal laws identified}. 


One would expect that $\InitSegmDelta{0}$ terminates at an even-index marked boundary point if and only if the SLE has parameter $\kappa \in (0, 4]$. This indeed holds true, by an \textit{a posteriori} proof, relying on the properties of the chordal $\SLE (\kappa)$. In order not to mix \textit{a posteriori} and \textit{a priori} properties of the scaling limits, we will first finish the proof of Theorem~\ref{thm: loc-2-glob multiple SLE convergence, kappa le 4} without knowledge of where $\InitSegmDelta{0}$ terminates, and return to this discussion in Proposition~\ref{prop: initial segment end points} in Section~\ref{subsec: termination points of init segments}.

\subsubsection{\textbf{Distance from the initial segment to marked boundary arcs}}

With the marginal law of $\InitSegmDelta{0}$ determined, note that by definition, $\InitSegmDelta{0}$ only visits the boundary arc $( \Unitp_2 \Unitp_{2N} ) \subset \bdry \UnitD$ at its end point. We now explicate two simple but important consequences of this trivial observation.

First, let $B_1$ and $B_2$ be closed marked boundary arcs of $ \UnitD$, between some neighbouring marked boundary points $\Unitp_1, \ldots, \Unitp_{2N}$. Assume that $B_1$ and $B_2$ are not adjacent to the boundary point $\Unitp_1$ and not adjacent to each other. As $\delta' \shrinkto 0$, we have the following approximation of events:
\begin{align}
\label{eq: bdary visit event 1}
 & \{ \InitSegmDelta{0} \text{ visits } \delta' \text{-close to } B_1 \} \cap \{ \InitSegmDelta{0} \text{ visits } \delta' \text{-close to } B_2 \} \\
 \nonumber
& \shrinkto \{ \InitSegmDelta{0} \text{ visits } B_1 \} \cap \{ \InitSegmDelta{0} \text{ visits } B_2 \} = \emptyset.
\end{align}
In particular, taking $\delta'$ small enough, we can make the probability of the event~\eqref{eq: bdary visit event 1} arbitrarily small. 

Second, let $B_1$ and $B_2$ now be closed marked boundary arcs of $\bdry \UnitD$, not adjacent to the boundary point $\Unitp_1$ but possibly adjacent to each other. Fix a (small) $\delta' > 0$. This time, as $\delta'' \shrinkto 0$, we have
\begin{align}
\label{eq: bdary visit event 2}
 & \{ \InitSegmDelta{0} \text{ visits } \delta'' \text{-close to the $\delta'$-interior of } B_1 \} \cap \{ \InitSegmDelta{0} \text{ visits } \delta'' \text{-close to } B_2 \} \\
 \nonumber
& \shrinkto \{ \InitSegmDelta{0} \text{ visits the $\delta'$-interior of } B_1 \} \cap \{ \InitSegmDelta{0} \text{ visits } B_2 \} = \emptyset.
\end{align}

Useful interpretations of these computations play a role analogous to Assumption~\ref{ass: quantitative no boundary visits assumption} in the proof of Theorem~\ref{thm: loc-2-glob multiple SLE convergence, kappa le 4}, see illustration in Figure~\ref{fig: bdary visit events}.

\begin{figure}
\includegraphics[width=0.5\textwidth]{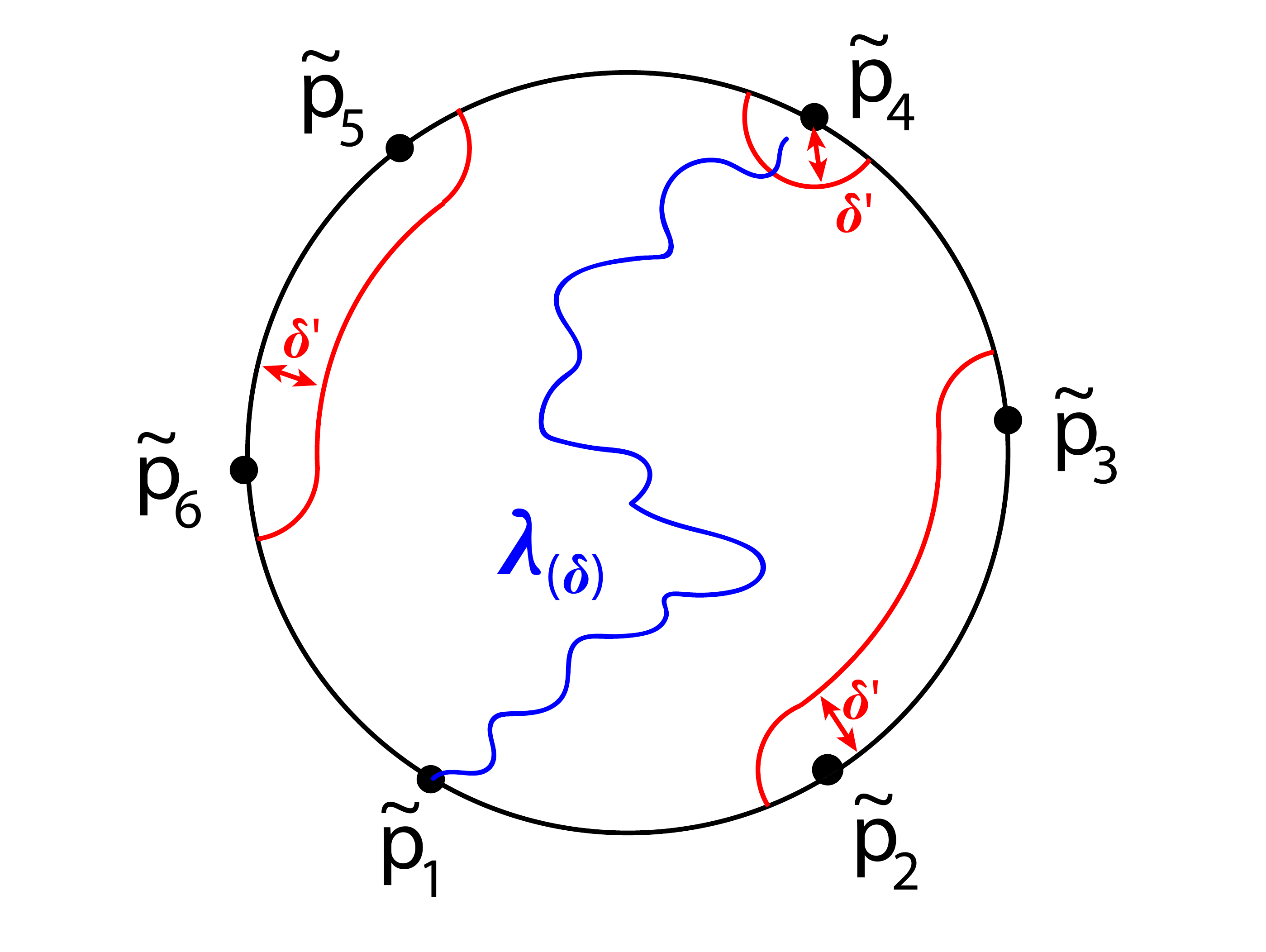}%
\includegraphics[width=0.5\textwidth]{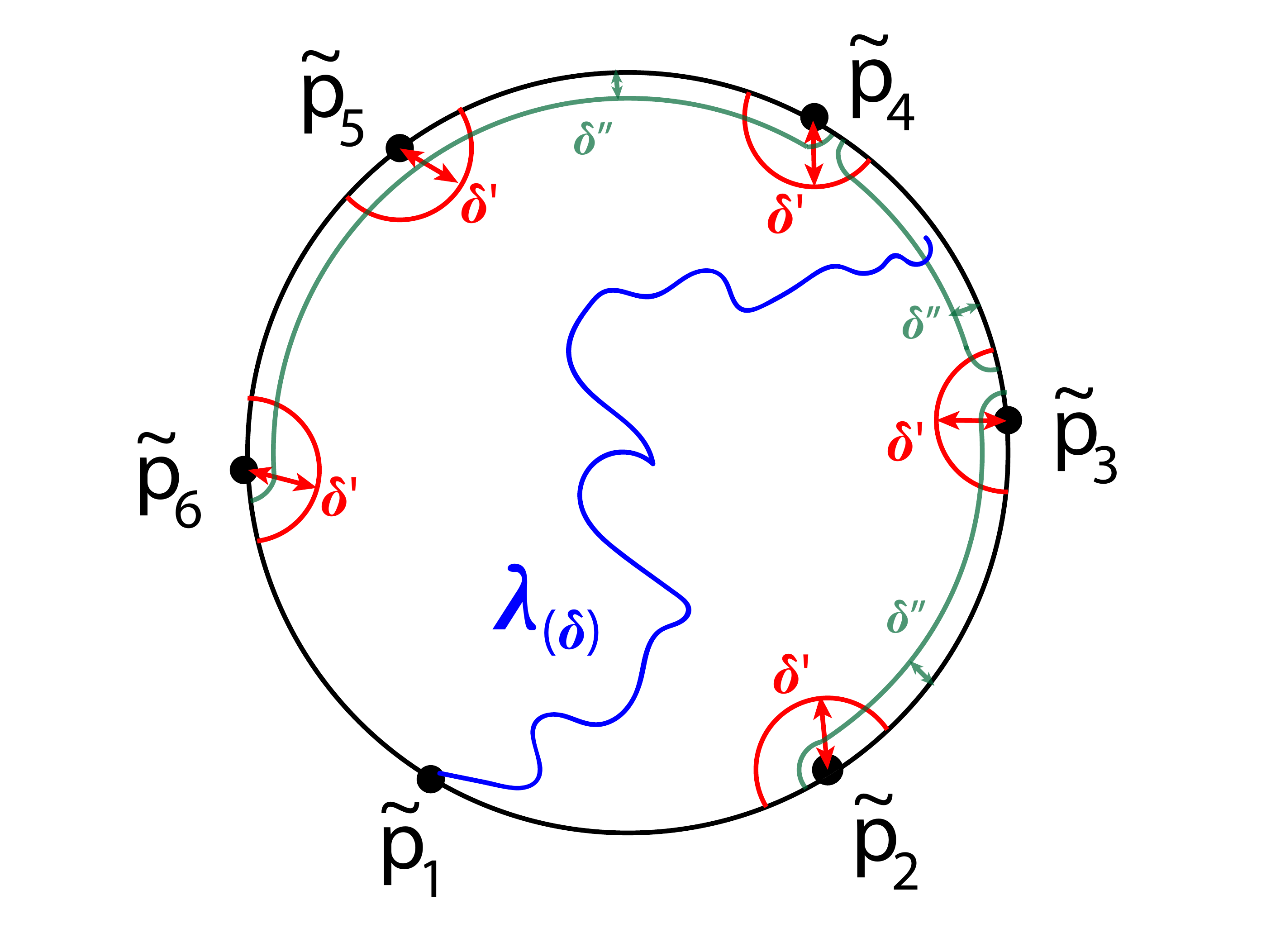}
\caption{ \label{fig: bdary visit events} Interpretation of the event approximations~\eqref{eq: bdary visit event 1} and~\eqref{eq: bdary visit event 2}.
\textbf{Left:} Fix any $\eps > 0$. Take $\delta'$ small enough so that the event~\eqref{eq: bdary visit event 1} has probability $\le \eps$, and take any $\delta < \delta'$. If the tip of the initial segment $\InitSegmDelta{\delta}$ is $\delta'$-close to a marked boundary point, the initial segment $\InitSegmDelta{\delta}$ will leave free a $\delta'$-neighbourhood of the boundary arcs not close to its end points, with probability $\ge 1 - \eps$. 
\textbf{Right:} Given the $\delta'$ above, take $\delta''$ small enough so that the event~\eqref{eq: bdary visit event 2} has probability $\le \eps$, and assume that $\delta$ is small enough so that also $\delta < \delta''$. If the tip of the initial segment $\InitSegmDelta{\delta}$ is not $\delta'$-close to any marked boundary point, it has to be $\delta''$-close to the $\delta'$-interior of some marked boundary arc. Then, the initial segment will leave free a $\delta''$-neighbourhood of the boundary arcs not close to its end points, with probability $\ge 1 - \eps$. }
\end{figure}

\subsubsection{\textbf{Identifying the full scaling limit}}

To complete the proof of Theorem~\ref{thm: loc-2-glob multiple SLE convergence, kappa le 4} under Assumption~\ref{ass: cond C'}, one proves Proposition~\ref{thm: k le 4 local-to-global NSLE} with that assumption in stead of Assumption~\ref{ass: quantitative no boundary visits assumption}. The proof of the Proposition remains identical, except for Lemma~\ref{lem: domain Markov of scaling limit curves}, which now has to be replaced by the following.

\begin{lem}
\label{lem: domain Markov of scaling limit curves 2}
Under Assumption~\ref{ass: cond C'} in the setup of Theorem~\ref{thm: loc-2-glob multiple SLE convergence, kappa le 4}, any subsequential limit $ ( {\gamma}_{\UnitD; 1}, \ldots, {\gamma}_{\UnitD; N} )$ of the curves $( {\gamma}^{(n)}_{\UnitD; 1}, \ldots, {\gamma}^{(n)}_{\UnitD; N} ) $ satisfies the following: the tip $\InitSegmDelta{0}(1)$ is almost surely not an odd-index marked boundary point, and for any bounded, non-negative, Lipschitz continuous test functions $g: X(\overline{\UnitD})^{N} \to \R$, and $f: X(\overline{\UnitD}) \to \R$,
\begin{align*}
\EX &[f ( \InitSegmDelta{0} )  g( \FinalSegmDelta{0}, {\gamma}_{\UnitD; 2}, \ldots, {\gamma}_{\UnitD; N}   ) ]  \\
= & 
\EX [ \mathbb{I} \{ \InitSegmDelta{0}(1) \text{ is an even-index marked boundary point } w \}  f ( \InitSegmDelta{0} )  \EXSLEcurves{(N-1)}_{\UnitD \setminus \InitSegmDelta{0} } [ g( w, \SLEcurve_{1},  \ldots, \SLEcurve_{N-1}   ) ] \\
& +
\EX [  \mathbb{I}\{ \InitSegmDelta{0}(1) \text{ is not a marked boundary point}\}   f ( \InitSegmDelta{0} ) \EXSLEcurves{\mathbf{e}}_{\UnitD \setminus \InitSegmDelta{0} }   [  \EXSLEcurves{\mathbf{o}}_{\UnitD \setminus \InitSegmDelta{0} } [  g( \mathbf{o}, \mathbf{e} ) ],
\end{align*}
where the notation $\EXSLEcurves{(N-1)}_{\UnitD \setminus \InitSegmDelta{0} }$ in the first expectation, i.e., when $\InitSegmDelta{0}$ traverses between two marked boundary points, is interpreted as the similar notation in Lemma~\ref{lem: domain Markov of scaling limit curves}; the notations $\EXSLEcurves{\mathbf{e}}_{\UnitD \setminus \InitSegmDelta{0} } $ and $\EXSLEcurves{\mathbf{o}}_{\UnitD \setminus \InitSegmDelta{0} } $ in turn are multiple SLE expectations in the components of $( \UnitD \setminus \InitSegmDelta{0} )$ with an even and odd number of marked boundary points, respectively, between the remaining marked boundary points and, in the odd component, the tip of the initial segment $\InitSegmDelta{0} $ (the notation $ g( \mathbf{o}, \mathbf{e} )$ is a slightly abusive shorthand since the original argument $\FinalSegmDelta{0}, {\gamma}_{\UnitD; 2}, \ldots, {\gamma}_{\UnitD; N} $ of $g$ should actually be replaced by the obvious re-labelling of the curves $\mathbf{o}$ and $\mathbf{e} $).
\end{lem}

\begin{proof}
First, the annulus crossing condition (G) guarantees that $ {\gamma}_{\UnitD; 1}$ will almost surely not hit an odd boundary point. Next, note that it then suffices to prove the claim for nonnegative functions $f$ such that for some $\tilde{\delta} > 0$, $f(\InitSegmDelta{0})$ takes the value $0$ if $\InitSegmDelta{0}$ visits the $\tilde{\delta}$-neighbourhood of some odd-index boundary point other than the first one. Indeed, the case of a general $f$ then follows by taking increasing approximations of $f$ and using Monotone convergence. We will assume this property of $f$, keeping $\tilde{\delta} > 0$ fixed throughout the proof.

Next, take a small auxiliary radius $\delta' < \tilde{\delta}$ and denote by $\InitSegmDelta{0} (1)$ the tip of the curve $\InitSegmDelta{0} $.
By our assumption on $f$, either $B(\InitSegmDelta{0} (1), \delta')$ contains an even-index marked boundary point, or it contains no marked boundary points, or $f( \InitSegmDelta{0} )=0$. We will treat the two first nontrivial cases separately, and in the end combine the results again in the limit $\delta' \shrinkto 0$. To formalize this, write $f= 1\cdot f$ and decompose $1$ into a sum continuous cutoff functions of $\lambda_0$; to be explicit, for instance
\begin{align*}
c_1^{(\delta')} (  \InitSegmDelta{0}   ) = \min \left\{ \; \frac{d( \; \InitSegmDelta{0}  (1), B(\Unitp_2, \delta' ) \cup \ldots \cup B(\Unitp_{2N} , \delta' ) \;)}{\delta' } \; ; 1 \; \right\}, \qquad
c_2^{(\delta')} ( \InitSegmDelta{0}  ) = 1 - c_1^{(\delta')} ( \InitSegmDelta{0}  ).
\end{align*}
We observe that both $c_1^{(\delta')} f$ and $c_2^{(\delta')} f$ are bounded, Lipschitz continuous, nonnegative functions, for each $\delta' > 0$. The function $c_1^{(\delta')} f$ takes nonzero values only when the tip $\InitSegmDelta{0}  (1)$ is at a distance $\ge \delta'$ from all the marked boundary points $\Unitp_2, \ldots, \Unitp_{2N}$, whereas the function $c_2^{(\delta')} f$ takes nonzero values only when the tip $\InitSegmDelta{0} (1)$ is at a distance $\le 2 \delta'$ from some even-index boundary point. 

\textbf{The term $c_2^{(\delta')} f$:} Consider first the term $c_2^{(\delta')} f$ and start by computing
\begin{align}
\label{eq: weak conv 1}
\EX &[c_2^{(\delta')} ( \InitSegmDelta{0}  )  f ( \InitSegmDelta{0}  )  g( \FinalSegmDelta{0} , {\gamma}_{\UnitD; 2}, \ldots, {\gamma}_{\UnitD; N}   ) ] \\
\label{eq: init segm conv}
&= \EX [c_2^{(\delta')} ( \InitSegmDelta{\delta} )  f ( \InitSegmDelta{\delta}  )  g( \FinalSegmDelta{\delta} , {\gamma}_{\UnitD; 2}, \ldots, {\gamma}_{\UnitD; N}   ) ] + o^{(\delta')}_\delta (1) \\
\label{eq: weak conv}
&= \EX^{(n)} [c_2^{(\delta')} ( \InitSegmDelta{\delta}^{(n)} )  f ( \InitSegmDelta{\delta}^{(n)} )  g( \FinalSegmDelta{\delta}^{(n)}, \gamma_{\UnitD; 2}^{(n)}, \ldots, \gamma_{\UnitD; N}^{(n)}   ) ] + o^{(\delta, \delta')}_n (1) + o^{(\delta')}_\delta (1)
\end{align}
where the equality~\eqref{eq: init segm conv} holds by the almost sure convergences in Lemma~\ref{lem: SLE up-to-swallowing segments exist} and $o^{(\delta')}_\delta (1)$ stands for ``$o(1)$ as $\delta \shrinkto 0$ for any fixed $\delta'$'', while~\eqref{eq: weak conv} holds by the weak convergence of the discrete models and $o^{(\delta, \delta')}_n (1)$ stands for ``$o(1)$ as $n \to \infty$ for any fixed $\delta$ and $\delta'$''. 

Recall that the function in the expectation~\eqref{eq: weak conv} takes values $\ne 0$ only when the tip $\InitSegmDelta{\delta}^{(n)} (1)$ is at a distance $\le 2 \delta'$  from some even-index boundary point $\Unitp_2, \ldots, \Unitp_{2N}$; call it $w$. Assume that $\delta<\delta'$ so that this is possible and condition on such a $\InitSegmDelta{\delta}^{(n)} $. Now, when having grown the initial segment $\InitSegmDelta{\delta}^{(n)} $ only up to the first hitting of $B(w, 2\delta')$, the circle arcs $\bdry B(w, 2\delta')$ and $\bdry B(w, 2\sqrt{\delta'} )$ allow one to define two topological quadrilaterals of modulus $1/o_{\delta'}(1)$ that separate the tip of that smaller segment and the boundary point $w$ from all other marked boundary points; see Figure~\ref{fig: thick quads}. By Condition (C'), it thus occurs with probability $\ge 1 - o_{\delta'}(1)$ that $\gamma_{\UnitD; 1}$ actually connects to the boundary point $w$, where the term $o_{\delta'}(1)$ is independent of $n$, $\delta$, or the shape of our initial segment of $\InitSegmDelta{\delta}^{(n)} $. Recalling from Corollary~\ref{cor: marginal precompactness} that the curves $\gamma_{\UnitD; 1}^{(n)}$ satisfies the condition (multi-G), we deduce that also with probability $\ge 1 - o_{\delta'}(1)$, the remainder of the curve $\gamma_{\UnitD; 1}^{(n)}$ after the first hitting of $B(w, 2\delta')$ never exits $B(w,  2 \sqrt{\delta'})$. In particular, this holds for the smaller final segment $\FinalSegmDelta{\delta}^{(n)} $. 

\begin{figure}
\includegraphics[width=0.45\textwidth]{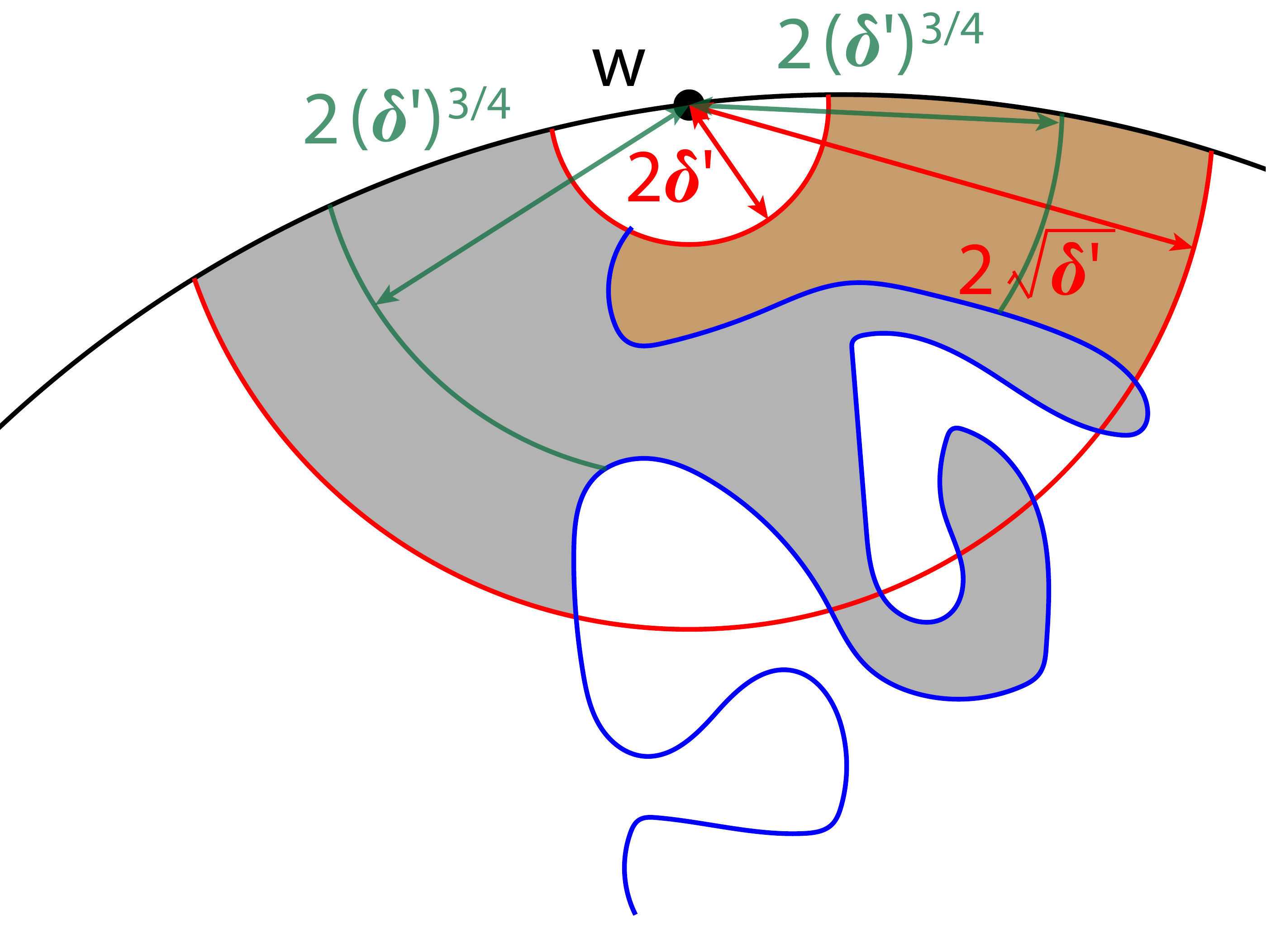}%
\includegraphics[width=0.45\textwidth]{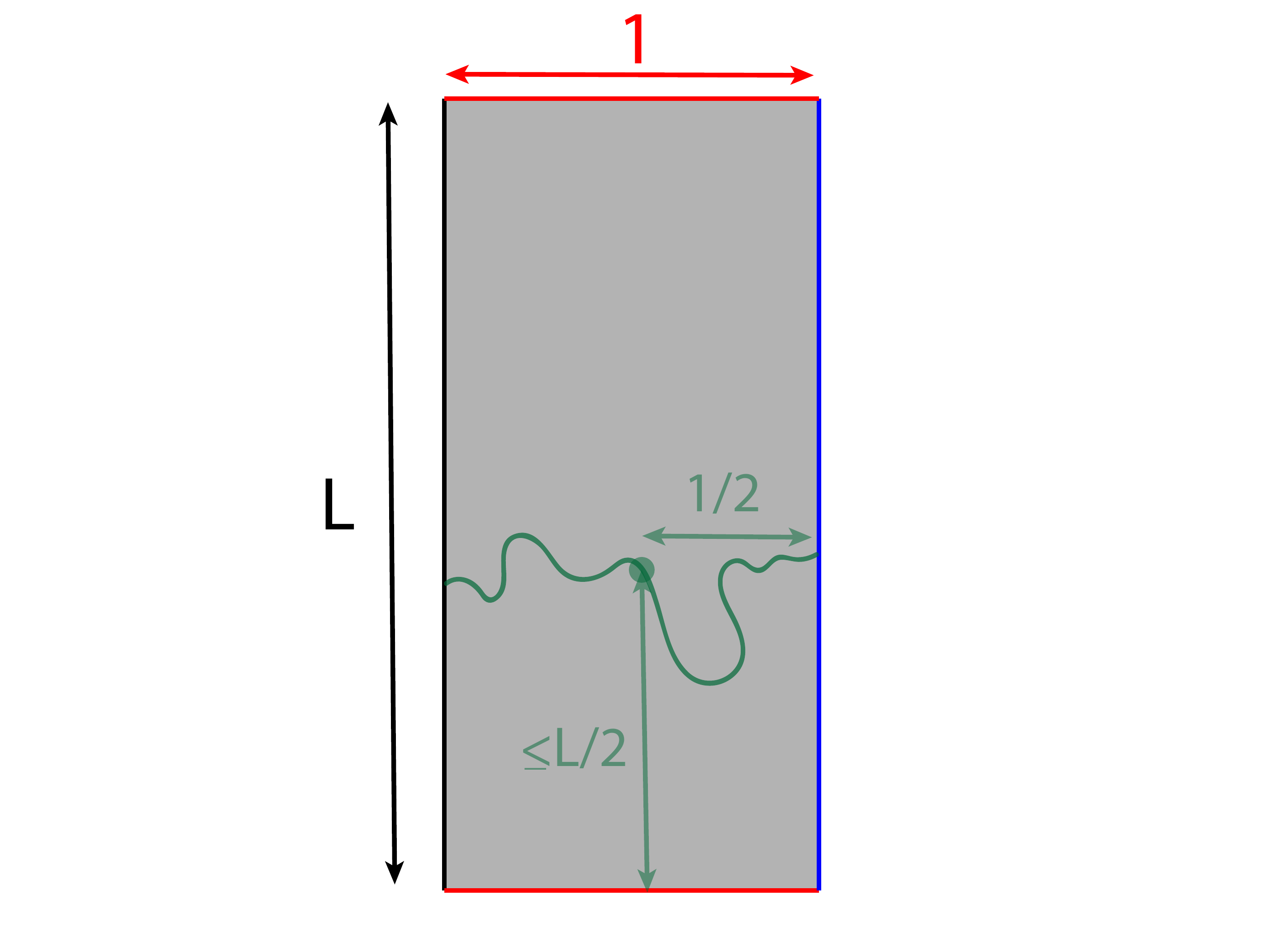}
\caption{\label{fig: thick quads}
Left: In blue the initial segment of $\InitSegmDelta{\delta}^{(n)} $ up to the first hitting of $B(w, 2\delta')$. In gray and brown the two topological quadrilaterals with large conformal moduli (both with two sides drawn in red, one in blue, and one in black). Left and Right: a sketch of proof for a lower bound $L= 1/o_{\delta'} (1)$ for the conformal modulus $L$ of the gray quadrilateral, independent of the shape of the blue initial segment: an arc of the circle $B(w, 2(\delta')^{3/4})$ (in green) crosses it from black side to blue side. Beurling's estimate shows that a positive power of $\delta'$ upper bounds the harmonic measure of the red sides as seen form \emph{any} point on the green arc. Mapping everything conformally to the rectangle $(0, 1)\times (0, L)$, the image of the green arc contains a point at distance $1/2$ from the black and blue sides and $\le L/2$ from one red side. Since the harmonic measure of the red side is upper bounded by a power of $\delta'$, we obtain a lower bound $L= 1/o_{\delta'} (1)$ for $L$.
}
\end{figure}

Let $c^{( \delta')} ( \FinalSegmDelta{\delta}^{(n)}  )$ now be a continuous cutoff function taking value $1$ if $\FinalSegmDelta{\delta}^{(n)} $ never exits the ball $B(w, 2 \sqrt{\delta'})$ and $0$ if it exits the ball $A(w, 3 \sqrt{\delta'})$. By the above paragraph, making an error $o_{\delta'}(1)$ uniform over $n$ and $\delta$ with $\delta < \delta'$, we can add a factor $c^{( \delta')} ( \FinalSegmDelta{\delta}^{(n)}  )$ to~\eqref{eq: weak conv}. Likewise, since the function $g$ is bounded and Lipschitz continuous, we can replace $\FinalSegmDelta{\delta}^{(n)}$ in its argument by the curve consisting of the single point $w$ in~\eqref{eq: init segm conv} within the same error:
\begin{align*}
\eqref{eq: weak conv 1} &= 
\EX^{(n)} [  c^{( \delta')} ( \FinalSegmDelta{\delta}^{(n)}  ) c_2^{(\delta')} ( \InitSegmDelta{\delta}^{(n)}  )  f ( \InitSegmDelta{\delta}^{(n)} )  g( w, \gamma_{\UnitD; 2}^{(n)}, \ldots, \gamma_{\UnitD; N}^{(n)}   ) ] ]
+ o_{\delta'}(1)  + o^{(\delta, \delta')}_n (1) + o^{(\delta')}_\delta (1).
\end{align*}
The limit of the above expectation as $n \to \infty$ can now be treated as in the proof of Lemma~\ref{lem: domain Markov of scaling limit curves}, by using the inductively assumed Proposition~\ref{thm: k le 4 local-to-global NSLE}(ii) for the curves $\gamma_{\UnitD; 2}^{(n)}, \ldots, \gamma_{\UnitD; N}^{(n)} $ (the estimate in Figure~\ref{fig: bdary visit events}(left) replaces Assumption~\ref{ass: quantitative no boundary visits assumption}).
This yields
\begin{align*}
\EX^{(n)} & [ c^{( \delta')} ( \FinalSegmDelta{\delta}^{(n)}  )  c_2^{(\delta')} ( \InitSegmDelta{\delta}^{(n)}  )  f ( \InitSegmDelta{\delta}^{(n)} )  g( w, \gamma_{\UnitD; 2}^{(n)}, \ldots, \gamma_{\UnitD; N}^{(n)}   ) ] ] \\
&=
\EX [ c^{( \delta')} ( \FinalSegmDelta{\delta}^{(n)}  )  c_2^{(\delta')} ( \InitSegmDelta{\delta}^{(n)}  )  f ( \InitSegmDelta{\delta}^{(n)} )  \EXSLEcurves{(N-1)}_{\UnitD \setminus \gamma_{\UnitD; 1}^{(n)}  } [ g( w, \SLEcurve_{1},  \ldots, \SLEcurve_{N-1}   ) ] + o^{(\delta, \delta')}_n (1) \\
&=
\EX [ c^{( \delta')} ( \FinalSegmDelta{\delta}  ) c_2^{(\delta')} ( \InitSegmDelta{\delta} )  f ( \InitSegmDelta{\delta} )  \EXSLEcurves{(N-1)}_{\UnitD \setminus \gamma_{\UnitD; 1}  } [ g( w, \SLEcurve_{1},  \ldots, \SLEcurve_{N-1}  ) ] + o^{(\delta, \delta')}_n (1),
\end{align*}
where the last step used the weak convergence $\gamma_{\UnitD; 1}^{(n)} \to \gamma_{\UnitD; 1}$. (This is possible since when $ c^{( \delta')} ( \FinalSegmDelta{\delta}^{(n)}  )  c_2^{(\delta')} ( \InitSegmDelta{\delta}^{(n)}  ) \ne 0$, the expectation $\EXSLEcurves{(N-1)}_{\UnitD \setminus \gamma_{\UnitD; 1}  } $ makes sense by the estimate in Figure~\ref{fig: bdary visit events}(left) and is a continuous function of $\gamma_{\UnitD; 1}$ by inductive assumption~(iii).) Finally, substituting this back, we have
\begin{align*}
\eqref{eq: weak conv 1} &= \EX [ c^{( \delta')} ( \FinalSegmDelta{\delta} )  c_2^{(\delta')} ( \InitSegmDelta{\delta} )  f ( \InitSegmDelta{\delta} )  \EXSLEcurves{(N-1)}_{\UnitD \setminus \gamma_{\UnitD; 1}  } [ g( w, \SLEcurve_{1},  \ldots, \SLEcurve_{N-1}  ) ] + o_{\delta'}(1)  + o^{(\delta, \delta')}_n (1) + o^{(\delta')}_\delta (1).
\end{align*}

Note that in the above equation, $n$ only appears in the Landau o-terms. We would like to achieve this for $\delta$ and $\delta'$, too.
Thus, study the expectation above first in the limit $\delta \shrinkto 0$ using the almost sure convergence of Lemma~\ref{lem: SLE up-to-swallowing segments exist}, and then in the limit $\delta' \shrinkto 0$. Using Bounded convergence theorem in these limits yields
\begin{align*}
\eqref{eq: weak conv 1} =& \EX [ c^{( \delta')} ( \FinalSegmDelta{0} )  c_2^{(\delta')} ( \InitSegmDelta{0} )  f ( \InitSegmDelta{0} )  \EXSLEcurves{(N-1)}_{\UnitD \setminus \gamma_{\UnitD; 1}  } [ g( w, \SLEcurve_{1},  \ldots, \SLEcurve_{N-1}    ) ] + o_{\delta'}(1)  + o^{(\delta, \delta')}_n (1) + o^{(\delta')}_\delta (1) \\
=& \EX [  \mathbb{I} \{ \FinalSegmDelta{0} \text{ is a point}  \}  \mathbb{I} \{ \InitSegmDelta{0} \text{ ends at even-index boundary point } w \}  f ( \InitSegmDelta{0} )  \EXSLEcurves{(N-1)}_{\UnitD \setminus \gamma_{\UnitD; 1}  } [ g( w, \SLEcurve_{1},  \ldots, \SLEcurve_{N-1}   ) ]\\ 
& + o_{\delta'}(1)  + o^{(\delta, \delta')}_n (1) + o^{(\delta')}_\delta (1)
\end{align*}

Now, the three Landau o-terms above can be all made simultaneously arbitrarily small by choosing first $\delta'$ small enough, then $\delta$ small enough, and then $n$ large enough. In addition, Condition (G) for the curves $\gamma_{\UnitD; 1}^{(n)}$ guarantees that $\gamma_{\UnitD; 1}$ almost sure only visits an even-index marked boundary point once, namely, its target point in the end of the curve. Thus, we obtain
\begin{align*}
\lim_{\delta' \shrinkto 0} \eqref{eq: weak conv 1} 
&= \EX [ \mathbb{I} \{ \InitSegmDelta{0}(1) \text{ is an even-index boundary point } w \}  f ( \InitSegmDelta{0} )  \EXSLEcurves{(N-1)}_{\UnitD \setminus \InitSegmDelta{0}   } [ g( w, \SLEcurve_{1},  \ldots, \SLEcurve_{N-1}  ) ].
\end{align*}
This finishes our treatment of the term $c_2^{(\delta')} f$.


\textbf{The term $c_1^{(\delta')} f$:} Consider next the term $c_1^{(\delta')} f$ and start by computing
\begin{align}
\label{eq: DMP 1}
\EX &[c_1^{(\delta')} ( \lambda_0 )  f ( \lambda_0 )  g( \eta_0, \gamma_{\UnitD; 2}, \ldots, \gamma_{\UnitD; N}   ) ] \\
\nonumber
&= \EX [c_1^{(\delta')} ( \InitSegmDelta{\delta} )  f ( \InitSegmDelta{\delta} )  g( \FinalSegmDelta{\delta}, \gamma_{\UnitD; 2}, \ldots, \gamma_{\UnitD; N}   ) ] + o^{(\delta')}_\delta (1) \\
\nonumber
&= \EX^{(n)} [c_1^{(\delta')} ( \InitSegmDelta{\delta}^{(n)} )  f ( \InitSegmDelta{\delta}^{(n)} )  g( \FinalSegmDelta{\delta}^{(n)}, \gamma_{\UnitD; 2}^{(n)},  \ldots, \gamma_{\UnitD; N}^{(n)}   ) ] + o^{(\delta, \delta')}_n (1) + o^{(\delta')}_\delta (1) \\
\label{eq: DMP}
&= \EX^{(n)} [c_1^{(\delta')} (\InitSegmDelta{\delta}^{(n)} ) f ( \InitSegmDelta{\delta}^{(n)} )  \EX^{(n)}_{\UnitD \setminus \InitSegmDelta{\delta}^{(n)} } [ g( \FinalSegmDelta{\delta}^{(n)}, \gamma_{\UnitD; 2}^{(n)}, \ldots, \gamma_{\UnitD; N}^{(n)}   ) ] ] + o^{(\delta, \delta')}_n (1) + o^{(\delta')}_\delta (1),
\end{align}
where the first two steps are similar to~\eqref{eq: init segm conv} and~\eqref{eq: weak conv}, and step~\eqref{eq: DMP} is an application of the discrete domain Markov property.

Recall that $c_1^{(\delta')} (\InitSegmDelta{\delta}^{(n)} )  f ( \InitSegmDelta{\delta}^{(n)} )  $ takes nonzero values only when the tip $\InitSegmDelta{\delta}^{(n)} (1)$ is at a distance $\ge \delta'$ from all the marked boundary points. Condition on such a $\InitSegmDelta{\delta}^{(n)} $, and assume that $\sqrt{\delta} < \delta'$ so that the tip of $\InitSegmDelta{\delta}^{(n)} $ is very close to $\bdry \UnitD$.
In that case, removing the curve $\InitSegmDelta{\delta}^{(n)} $ and the ball $B(\InitSegmDelta{\delta}^{(n)} (1), \sqrt{\delta} )$ at its tip, the remaining boundary $\bdry \UnitD \setminus \InitSegmDelta{\delta}^{(n)} \setminus B(\InitSegmDelta{\delta}^{(n)} (1), \sqrt{\delta} ) $ will consist of two connected components, containing altogether $N-1$ marked boundary points. Both of these boundary arcs contain a nonzero amount of boundary points, and it is natural to call \emph{even} the arc with an even number of them, and the other one \emph{odd}.
We call the connected components of $ \UnitD \setminus \InitSegmDelta{\delta}^{(n)} \setminus B(\InitSegmDelta{\delta}^{(n)} (1), \sqrt{\delta} ) $ adjacent to these arcs even and odd, respectively.
Note also that given $\InitSegmDelta{\delta}^{(n)} $, we know the indices of the curves $\gamma_{\UnitD; 1}^{(n)}, \gamma_{\UnitD; 2}^{(n)}, \ldots$ that start from the even and odd arc, as each curve starts from an odd-index boundary point. We denote by $\mathbf{e}^{(n)}$ the collection of curves starting from the even arc and by $\mathbf{o}^{(n)}$ the rest, i.e., the remainder curve $\FinalSegmDelta{\delta}^{(n)}$ and the curves starting in the odd components. We will also denote, with a slight abuse of notation in the curve indexing
\begin{align*}
 g( \FinalSegmDelta{\delta}^{(n)}, \gamma_{\UnitD; 2}^{(n)}, \ldots, \gamma_{\UnitD; N}^{(n)}   ) = g( \mathbf{o}^{(n)},  \mathbf{e}^{(n)} ) .
\end{align*}

First, identically to the case of the term $c_2^{(\delta')} f $, observe that given $\InitSegmDelta{\delta}^{(n)} $ there is a conditional probability $1-o_{\delta}(1)$ that the even curves $\mathbf{e}^{(n)}$ only intersect one connected component of $\UnitD \setminus \InitSegmDelta{\delta}^{(n)} \setminus B( \InitSegmDelta{\delta}^{(n)} (1),  \sqrt{\delta })$. Denote this event by $P(\delta)$. Note that on this event, the curves $\mathbf{e}^{(n)}$ are only adjacent to the even-component boundary points.
Note also that $P(\delta)$ only depends on $\InitSegmDelta{\delta}^{(n)} $ and the even curves $\mathbf{e}^{(n)}$. With a slight abuse of notation, we will denote the corresponding Borel set on the space of collections of curves $X(\overline{\UnitD})^N$ also by $P(\delta)$.


Next, take $\delta'' > 0$ with $\delta < \delta'' < \delta'$. We denote $\InitSegmDelta{\delta} \in E( \delta', \delta'')$ if the tip $\InitSegmDelta{\delta}^{(n)} (1)$ is at a distance $\ge \delta'$ from all the marked boundary points (i.e., $c_1^{(\delta')} (\InitSegmDelta{\delta}^{(n)} )  f (\InitSegmDelta{\delta}^{(n)} ) \ne 0$), and the curve $\InitSegmDelta{\delta}^{(n)}$ visits at distance $\le \delta''$ from some marked boundary arc not closest to the tip of $\InitSegmDelta{\delta}^{(n)} $ or adjacent to $\Unitp_1$. It follows from the deduction in Figure~\ref{fig: bdary visit events} that for the measure $\PR$ of the weak limit $\InitSegmDelta{\delta}$
\begin{align*}
\PR[\InitSegmDelta{\delta} \in E(\delta', \delta'')] = o_{\delta''}^{(\delta')}(1),
\end{align*}
where the $o_{\delta''}^{(\delta')}(1)$-term is independent of $\delta$ apart from the requirement $\delta < \delta'$. Note that  $E(\delta', \delta'')$ is  a closed set in $X(\overline{\UnitD})$.
 By Portmanteau's theorem on weak convergence, and the weak convergence $\InitSegmDelta{\delta}^{(n)} \to \InitSegmDelta{\delta}$, it follows that
\begin{align*}
\PR^{(n)} [\InitSegmDelta{\delta}^{(n)} \in E( \delta', \delta'')] = o_{\delta''}^{(\delta')}(1),
\end{align*}
where $o_{\delta''}^{(\delta')}(1)$ is small uniformly over all $n$ large enough, $n > n_0 (\delta, \delta', \delta'') $. 

Using the boundedness of the involved functions, we can add the indicator functions $\mathbb{I}_{ P(\delta) } (\InitSegmDelta{\delta}^{(n)}, \mathbf{e}^{(n)})$ and $\mathbb{I}_{E( \delta', \delta'')^C} ( \InitSegmDelta{\delta}^{(n)} )$ to~\eqref{eq: DMP} within errors of $o_{\delta}(1)$ and $o_{\delta''}^{(\delta')}(1)$, respectively:
\begin{align*}
\eqref{eq: DMP 1} =& \EX^{(n)} [  \mathbb{I}_{E( \delta', \delta'')^C} ( \InitSegmDelta{\delta}^{(n)} )  c_1^{(\delta')} (\InitSegmDelta{\delta}^{(n)} )   f ( \InitSegmDelta{\delta}^{(n)} ) \EX^{(n)}_{\UnitD \setminus \InitSegmDelta{\delta}^{(n)} } [ \mathbb{I}_{ P(\delta) } (\InitSegmDelta{\delta}^{(n)}, \mathbf{e}^{(n)}) g( \mathbf{o}^{(n)},  \mathbf{e}^{(n)}  ) ] ]+ o_{\delta''}^{(\delta')}(1) + o^{(\delta, \delta')}_n (1) + o^{(\delta')}_\delta (1) \\
=& \EX^{(n)} [ \mathbb{I}_{ E( \delta', \delta'')^C } (\InitSegmDelta{\delta}^{(n)}) c_1^{(\delta')} (\InitSegmDelta{\delta}^{(n)} )   f ( \InitSegmDelta{\delta}^{(n)} ) \EX^{(n)}_{\UnitD \setminus \InitSegmDelta{\delta}^{(n)} } [ \mathbb{I}_{ P(\delta) } (\InitSegmDelta{\delta}^{(n)}, \mathbf{e}^{(n)}) \EX^{(n)}_{\UnitD \setminus \InitSegmDelta{\delta}^{(n)} \setminus \mathbf{e}^{(n)} } [ g( \mathbf{o}^{(n)},  \mathbf{e}^{(n)}  ) ] ] ]\\
& + o_{\delta''}^{(\delta')}(1) + o^{(\delta, \delta')}_n (1) + o^{(\delta')}_\delta (1),
\end{align*}
where the latter step used the DDMP

Next, we replace the odd curves by SLEs: By tightness we may assume that $\mathbf{e}^{(n)}$ and $\InitSegmDelta{\delta}^{(n)} $ belong to a compact set $K_\eps$ carrying a large probability mass. The intersection of $K_\eps$ with the closed sets $\overline{E(\delta, \delta', \delta'')^C}$ and $ \overline{P(\delta) }$ is compact. \\
\textbf{Claim:} We can choose the former compact sets $K_\eps$ suitably, so that uniformly over all $\mathbf{e}^{(n)}$ and $\InitSegmDelta{\delta}^{(n)} $ in the latter compact sets, we have the convergence as $n \to \infty$
 \begin{align}
 \label{eq: another uniform convergence}
\vert  \EX^{\mathbf{o}^{(n)} }_{\UnitD \setminus \InitSegmDelta{\delta}^{(n)} \setminus \mathbf{e}^{(n)} }  [  g( \mathbf{o}^{(n)},  \cdot  ) ] 
- \EXSLEcurves{\mathbf{o}}_{\UnitD \setminus \InitSegmDelta{\delta}^{(n)} \setminus \mathbf{e}^{(n)} }   [  g( \mathbf{o}, \cdot ) ] \vert  = o^{(\delta, \delta', \delta'')}_n (1),
 \end{align}
 also uniformly over any arguments\footnote{We will later wish to repeat almost identical computations for the even curves, and for that purpose, it is more transparent to leave this argument unspecified.} $\cdot$ of $g$; here the $m$-SLE curves $\mathbf{o}$ now run in $\UnitD \setminus \InitSegmDelta{\delta}^{(n)} \setminus \mathbf{e}^{(n)} $ between the tip of $\InitSegmDelta{\delta}^{(n)}$ and the limiting boundary points $\Unitp_1, \ldots, \Unitp_{2N}$ on $\bdry \UnitD$. The proof of this claim is based on the inductively assumed SLE convergence in Proposition~\ref{thm: k le 4 local-to-global NSLE}(ii). We have chosen to leave this proof to the reader, since we present a very similar but perhaps more difficult proof in the next paragraph.
 Substituting~\eqref{eq: another uniform convergence} into~\eqref{eq: DMP 1} yields
 \begin{align*}
\eqref{eq: DMP 1}
=& \EX^{(n)} [ \mathbb{I}_{ E( \delta', \delta'')^C } (\InitSegmDelta{\delta}^{(n)}) c_1^{(\delta')} (\InitSegmDelta{\delta}^{(n)} )   f ( \InitSegmDelta{\delta}^{(n)} ) \EX^{(n)}_{\UnitD \setminus \InitSegmDelta{\delta}^{(n)} } [ \mathbb{I}_{ P(\delta) } (\InitSegmDelta{\delta}^{(n)}, \mathbf{e}^{(n)}) \EXSLEcurves{\mathbf{o}}_{\UnitD \setminus \InitSegmDelta{\delta}^{(n)} \setminus \mathbf{e}^{(n)} }   [  g( \mathbf{o}, \mathbf{e}^{(n)} ) ] ] ]\\
& + o_{\delta''}^{(\delta')}(1) + o^{(\delta, \delta', \delta'')}_n (1) + o^{(\delta')}_\delta (1),
\end{align*}

Now, we remove the even curves from the domain of the SLEs: Equip both $\UnitD \setminus \InitSegmDelta{\delta}^{(n)} \setminus \mathbf{e}^{(n)}$ and $\UnitD \setminus \InitSegmDelta{\delta}^{(n)}$ with the odd amount of marked boundary points $\Unitp_1, \ldots, \Unitp_{2N}$ on the odd side and one at the tip of the curve $\InitSegmDelta{\delta}^{(n)}$; we keep these marked points implicit in the notation. Now, by tightness, $\gamma^{(n)}_{\UnitD; 1}$ lie on a compact set $K_\eps$ with probability $1-\eps$. \\
\textbf{Claim:} the sets $K_\eps$ may be chosen so that uniformly over $\gamma^{(n)}_{\UnitD; 1} \in K_\eps$ with $\InitSegmDelta{\delta}^{(n)} \in  \overline{E( \delta', \delta'')^C}$ and
$\mathbf{e}^{(n)}$ such that $ (\InitSegmDelta{\delta}^{(n)}, \mathbf{e}^{(n)}) \in \overline{P(\delta)}$, we have the convergence
 \begin{align}
 \label{eq: yet another uniform convergence}
\vert  \EXSLEcurves{\mathbf{o}}_{\UnitD \setminus \InitSegmDelta{\delta}^{(n)} \setminus \mathbf{e}^{(n)} }   [  g( \mathbf{o}, \cdot ) ] 
- \EXSLEcurves{\mathbf{o}}_{\UnitD \setminus \InitSegmDelta{\delta}^{(n)} }   [  g( \mathbf{o}, \cdot ) ] 
\vert  = o^{( \delta', \delta'')}_\delta (1)
 \end{align}
 for any arguments $\cdot$ of $g$ that are same in both expectations; here the $o^{( \delta', \delta'')}_\delta (1)$ term is independent of $n$.
 
\begin{proof}[Proof of claim]
Fix $\delta'$ and $\delta''$, and take a sequence of $\delta$:s converging to zero. We suppress the sequence notation, as well as all subsequence notations to come.
Assume for a contradiction that the claim does not hold, i.e., we can choose a (sub)sequence of $n = n(\delta)$:s (not necessarily growing to infinity!) and deterministic curves $\DetCurve^{(\delta)} \in K_\eps$, with initial segments $\DetCurve^{(\delta)}_\delta$ and arguments $a_\delta$ of $g$, satisfying
\begin{align}
\label{eq: SLE stability}
\vert  \EXSLEcurves{\mathbf{o}}_{\UnitD \setminus \DetCurve^{(\delta)}_\delta \setminus \mathbf{e}^{(n)} }   [  g( \mathbf{o}, a_\delta ) ] 
- \EXSLEcurves{\mathbf{o}}_{\UnitD \setminus \DetCurve^{(\delta)}_\delta }   [  g( \mathbf{o}, a_\delta ) ] 
\vert  > \ell
\end{align}
for some $\ell > 0$.

We would now like to use the Carath\'{e}odory stability of multiple SLEs, i.e., inductive assumption~(iii). First, recall that domains that are bounded from inside and outside are sequentially compact with respect to Carath\'{e}odory convergence (with respect to a reference point in the domain bounding from inside). Here, the domains $\UnitD \setminus \DetCurve^{(\delta)}_\delta \setminus \mathbf{e}^{(n)}$ and $\UnitD \setminus \DetCurve^{(\delta)}_\delta $ are bounded from inside due to the event of $E( \delta', \delta'')$, and we may thus assume that they converge in the Carath\'{e}odory sense (with a suitable reference point at a distance $< \delta''$ from the odd boundary arc of $\bdry \UnitD$). It is easily deduced that both sequences of domains converge to the same limit. By Schwarz reflection of conformal maps over $\bdry \UnitD$, this Carath\'{e}odory convergence can be extended to domains with the marked boundary points on $\bdry \UnitD$.  Also closeness of these boundary approximations is trivial. To apply inductive assumption~(iii), we have to reach these conclusions for the marked boundary points at the tip of $\DetCurve^{(\delta)}_\delta$.


First, by the compactness of $K_\eps$, we may  assume that the curves $\DetCurve^{(\delta)}$ converge, $\DetCurve^{(\delta)} \to \DetCurve$ in $X(\overline{\UnitD})$. Assume that the curves $\DetCurve^{(\delta)}$ and $ \DetCurve$ come as parametrized representatives such that this uniform convergence takes place as functions, too. Next, recall that the curves $\gamma^{(n)}_{\UnitD; 1}$ satisfy condition (G) by Corollary~\ref{cor: marginal precompactness}. In the proof of its consequence \cite[Theorems~1.5]{KS} (stated as Theorem~\ref{thm: one curve precompactness} in this paper), the compact sets $K_\eps$ of $X(\overline{\UnitD})$ are chosen so that the curves in them can be described by a Loewner equation. Choosing our $K_\eps$ in this manner, we thus know that $\DetCurve$ has a Loewner description (when mapped to $\bH$ so that its end point is at infinity).

Now, by compactness of the interval $[0, 1]$, we may assume that the times $T_\delta$ at which $\DetCurve^{(\delta)}$ is stopped to obtain $\DetCurve^{(\delta)}_\delta$ converge, $T_\delta \to T$. It follows that $\DetCurve^{(\delta)}_\delta = \DetCurve^{(\delta)} ([0, T_\delta]) \to \DetCurve ([0, T])$ in $X(\overline{\UnitD})$. It also follows that $\DetCurve ([0, T])$ hits $\bdry \UnitD$ at time $T$. Also, since $\DetCurve$ has a Loewner description and connects to the odd component, the tip $\DetCurve (T)$ of its initial segment $\DetCurve ([0, T])$ is on the boundary of the odd component of $\UnitD \setminus \DetCurve ([0, T])$. Now, the Carath\'{e}odory convergence $(\UnitD \setminus \DetCurve^{(\delta)}_\delta;  \DetCurve^{(\delta)} ( T_\delta) ) \to (\UnitD \setminus \DetCurve ([0, T]);  \DetCurve ( T) )$ can be deduced, e.g., by showing that the harmonic measures in $\UnitD \setminus \DetCurve^{(\delta)}_\delta$ of the boundary segment from $\Unitp_2$ clockwise to the tips of the curves $\DetCurve^{(\delta)}_\delta$ converge to those with curves $\DetCurve$. (The same holds for the Carath\'{e}odory convergence $(\UnitD \setminus \DetCurve^{(\delta)}_\delta \setminus \mathbf{e}^{(n)};  \DetCurve^{(\delta)} ( T_\delta) ) \to (\UnitD \setminus \DetCurve ([0, T]);  \DetCurve ( T) )$  with any $\mathbf{e}^{(n)}$ such that $\overline{P(\delta)}$ occurs.)

We yet need to show that the existence of radial limits and closeness in these Carath\'{e}odory approximations. For the first property, it is easy to that the boundary of $\UnitD \setminus \DetCurve ([0, T])$ is locally connected, using the continuity of $\DetCurve$. This implies (among stronger consequences) that radial limits exist at the prime end $\DetCurve ( T)$ of $\UnitD \setminus \DetCurve ([0, T])$, see~\cite[Theorem~2.1]{Pommerenke-boundary_behaviour_of_conformal_maps}.
Closeness of the approximations follows by comparing how $\DetCurve^{(\delta)} ([T_\delta, 1])$ and $\DetCurve ([T, 1])$ exit small neighbourhoods of $\DetCurve ( T)$.

Finishing the proof is now easy. Due to the inductive assumption~(iii) and the multiple SLEs $\mathbf{o}$ in the domains $\UnitD \setminus \DetCurve^{(\delta)}_\delta$ and $\UnitD \setminus \DetCurve^{(\delta)}_\delta \setminus \mathbf{e}^{(n)}$ converge weakly to the same limit $\mathbf{O}$, the multiple SLEs on the odd component of $\UnitD \setminus \DetCurve ([0, T])$. In particular, both are hence tight. On the other hand, by the Arzel\`{a}--Ascoli theorem and $g$ being bounded and Lipschitz, we may pick a further subsequence so that $g(\mathbf{o}, a_\delta)$ converges as functions of $\mathbf{o}$, uniformly over compacts, to some function $h(\mathbf{o})$. Combining this with the tightness, it follows that both expectations in~\eqref{eq: SLE stability} tend to 
$\EX [  h( \mathbf{O}) ]$,
a contradiction.
\end{proof} 
 
Let us continue the proof we were working on. Substituting~\eqref{eq: yet another uniform convergence} into~\eqref{eq: DMP 1}, we get
 \begin{align*}
\eqref{eq: DMP 1}
=& \EX^{(n)} [ \mathbb{I}_{ E( \delta', \delta'')^C } (\InitSegmDelta{\delta}^{(n)}) c_1^{(\delta')} (\InitSegmDelta{\delta}^{(n)} )   f ( \InitSegmDelta{\delta}^{(n)} ) \EX^{(n)}_{\UnitD \setminus \InitSegmDelta{\delta}^{(n)} } [ \mathbb{I}_{ P(\delta) } (\InitSegmDelta{\delta}^{(n)}, \mathbf{e}^{(n)}) \EXSLEcurves{\mathbf{o}}_{\UnitD \setminus \InitSegmDelta{\delta}^{(n)}  }   [  g( \mathbf{o}, \mathbf{e}^{(n)} ) ]  ] ]\\
& + o_{\delta''}^{(\delta')}(1) + o^{(\delta, \delta', \delta'')}_n (1) + o^{(\delta', \delta'')}_\delta (1).
\end{align*}
Removing the indicator functions by identical arguments as they were introduced with, and using Fubini's theorem, we get
\begin{align*}
\eqref{eq: DMP 1} = & \EX^{(n)} [  c_1^{(\delta')} (\InitSegmDelta{\delta}^{(n)} )   f ( \InitSegmDelta{\delta}^{(n)} ) \EXSLEcurves{\mathbf{o}}_{\UnitD \setminus \InitSegmDelta{\delta}^{(n)} }   [  \EX^{\mathbf{e}^{(n)} }_{\UnitD \setminus \InitSegmDelta{\delta}^{(n)} } [  g(\mathbf{o}, \mathbf{e}^{(n)}  ) ] 
  ]  ] + o_{\delta''}^{(\delta')}(1) + o^{(\delta, \delta', \delta'')}_n (1) + o^{(\delta', \delta'')}_\delta (1).
\end{align*}

Now, the next steps are to introduce similar indicator functions as before, but switching the roles of odd and even curves in the definition of $P(\delta)$. Then, repeating the uniform convergence arguments~\eqref{eq: SLE stability} and~\eqref{eq: yet another uniform convergence} for the even curves, 
%
%
%
%
one obtains 
\begin{align*}
\eqref{eq: DMP 1} = & \EX^{(n)} [  c_1^{(\delta')} (\InitSegmDelta{\delta}^{(n)} )   f ( \InitSegmDelta{\delta}^{(n)} ) \EXSLEcurves{\mathbf{e}}_{\UnitD \setminus \InitSegmDelta{\delta}^{(n)} }   [  \EXSLEcurves{\mathbf{o}}_{\UnitD \setminus \InitSegmDelta{\delta}^{(n)} } [  g( \mathbf{o}, \mathbf{e} ) ] 
  ]  ] + o_{\delta''}^{(\delta')}(1) + o^{(\delta, \delta', \delta'')}_n (1) + o^{(\delta', \delta'')}_\delta (1).
\end{align*}


Note that the only discrete curve left above is $\InitSegmDelta{\delta}^{(n)}$.
Next, we would like to use the weak convergence $\InitSegmDelta{\delta}^{(n)} \to \InitSegmDelta{\delta}$.
For this purpose, we will again have to restrict our consideration on the compact sets $K_\eps$ that guarantee the Loewner regularity of $\gamma_{\UnitD; 1}$.
As in the proof of the claim above, we can choose the compact sets $K_\eps$ so that for any converging sequence of curves $\DetCurve^{(n)} \to \DetCurve$ in $K_\eps$, on the additional events of $\overline{P(\delta)}$ and $\overline{E(\delta', \delta'')^C}$, it holds that the tip of the initial segment $\DetCurve_\delta$ is a prime end with radial limits in $\UnitD \setminus \DetCurve_\delta$, and $( \UnitD \setminus \DetCurve_\delta^{(n)}; \DetCurve_\delta^{(n)} (1) )$ are close Carath\'{e}odory approximations of $( \UnitD \setminus \DetCurve_\delta; \DetCurve_\delta (1) )$. Thus, by the inductive assumption~(iii), the expectation
\begin{align}
\label{eq: cont fcn}
\EXSLEcurves{\mathbf{e}}_{\UnitD \setminus \InitSegmDelta{\delta}^{(n)} }   [  \EXSLEcurves{\mathbf{o}}_{\UnitD \setminus \InitSegmDelta{\delta}^{(n)} } [  g( \mathbf{o}, \mathbf{e} ) ]
\end{align}
with respect to two independent multiple SLEs in $\UnitD \setminus \InitSegmDelta{\delta}^{(n)}$, is a  continuous function of $\gamma_{\UnitD; 1}^{(n)}$ on the intersection of $\gamma_{\UnitD; 1}^{(n)} \in K_\eps$ with $\overline{P(\delta)}$ and $\overline{E(\delta', \delta'')^C}$. Now, by tightness, take a compact set in $X(\overline{\UnitD})^N$ carrying a large probability mass and intersect it with $\gamma_{\UnitD; 1}^{(n)} \in K_\eps$ and $\overline{P(\delta)}$ and $\overline{E(\delta', \delta'')^C}$. The latter set is also compact with a large probability mass. Now, by Tietze's extension theorem, the continuous function~\eqref{eq: cont fcn} on the latter compact set can be continued to yield a bounded continuous function the whole space $X(\overline{\UnitD})^N$.
Using now the weak convergence $(\gamma_{\UnitD; 1}^{(n)}, \ldots, \gamma_{\UnitD; N}^{(n)}) \to ( \gamma_{\UnitD; 1}, \ldots, \gamma_{\UnitD; N})$, 
we obtain
\begin{align*}
\eqref{eq: DMP 1} = & \EX [  c_1^{(\delta')} (\InitSegmDelta{\delta} )   f ( \InitSegmDelta{\delta} ) \EXSLEcurves{\mathbf{e}}_{\UnitD \setminus \InitSegmDelta{\delta} }   [  \EXSLEcurves{\mathbf{o}}_{\UnitD \setminus \InitSegmDelta{\delta} } [  g( \mathbf{o}, \mathbf{e} ) ] 
  ]  ] + o_{\delta''}^{(\delta')}(1) + o^{(\delta, \delta', \delta'')}_n (1) + o^{(\delta', \delta'')}_\delta (1).
\end{align*}

Finally, fix a realization of $ \gamma_{\UnitD; 1}$ that can be described by the Loewner equation; it is then easy to argue that $\UnitD \setminus \InitSegmDelta{\delta} $, with the marked boundary points on $\bdry \UnitD$ and at the tip of $\InitSegmDelta{\delta}$, are close Carath\'{e}odory approximations of $\UnitD \setminus \InitSegmDelta{0} $. Thus, the almost sure convergence $\InitSegmDelta{\delta} \to  \InitSegmDelta{0}$ and the Bounded convergence theorem, we have
\begin{align*}
\eqref{eq: DMP 1} = & \EX [  c_1^{(\delta')} (\InitSegmDelta{0} )   f ( \InitSegmDelta{0} ) \EXSLEcurves{\mathbf{e}}_{\UnitD \setminus \InitSegmDelta{0} }   [  \EXSLEcurves{\mathbf{o}}_{\UnitD \setminus \InitSegmDelta{0} } [  g( \mathbf{o}, \mathbf{e} ) ] 
  ]  ] + o_{\delta''}^{(\delta')}(1) + o^{(\delta, \delta', \delta'')}_n (1) + o^{(\delta', \delta'')}_\delta (1).
\end{align*}
We would like to find the limit of~\eqref{eq: DMP} as $\delta' \shrinkto 0$. The expectation above can be treated using the Bounded convergence theorem, yielding
\begin{align*}
 \EX & [  c_1^{(\delta')} (\InitSegmDelta{0} )   f ( \InitSegmDelta{0} ) \EXSLEcurves{\mathbf{e}}_{\UnitD \setminus \InitSegmDelta{0} }   [  \EXSLEcurves{\mathbf{o}}_{\UnitD \setminus \InitSegmDelta{0} } [  g( \mathbf{o}, \mathbf{e} ) ] 
  ]  ] \\
  &= \EX [  \mathbb{I}\{ \InitSegmDelta{0}(1) \text{ is not a marked boundary point}\}   f ( \InitSegmDelta{0} ) \EXSLEcurves{\mathbf{e}}_{\UnitD \setminus \InitSegmDelta{0} }   [  \EXSLEcurves{\mathbf{o}}_{\UnitD \setminus \InitSegmDelta{0} } [  g( \mathbf{o}, \mathbf{e} ) ] + o_{\delta'} (1).
\end{align*}
 Finally, by taking first $\delta'$ small enough, and then $\delta''$, and then $\delta$, and then $n$ large enough, all the Landau o-terms above can all be made arbitrarily small. Thus, as $\delta' \shrinkto 0$, we have
\begin{align*}
\lim_{\delta' \shrinkto 0} \eqref{eq: DMP 1} = & \EX [  \mathbb{I}\{ \InitSegmDelta{0}(1) \text{ is not a marked boundary point}\}   f ( \InitSegmDelta{0} ) \EXSLEcurves{\mathbf{e}}_{\UnitD \setminus \InitSegmDelta{0} }   [  \EXSLEcurves{\mathbf{o}}_{\UnitD \setminus \InitSegmDelta{0} } [  g( \mathbf{o}, \mathbf{e} ) ].
\end{align*}
This finishes our discussion on the second term.

%

\textbf{Conclusion:}
Finally, combining the analyses of the two terms above, we observe that 
\begin{align*}
\EX &[f ( \InitSegmDelta{0} )  g( \FinalSegmDelta{0}, {\gamma}_{\UnitD; 2}, \ldots, {\gamma}_{\UnitD; N}   ) ]  \\
\text{(any $\delta' \in (0, \tilde{\delta})$) } = & \EX [ c_2^{(\delta')} (\lambda_0 )  f ( \InitSegmDelta{0} )  g( \FinalSegmDelta{0}, {\gamma}_{\UnitD; 2}, \ldots, {\gamma}_{\UnitD; N}   ) ] + \EX [c_1^{(\delta')} (\lambda_0 )  f ( \InitSegmDelta{0} )  g( \FinalSegmDelta{0}, {\gamma}_{\UnitD; 2}, \ldots, {\gamma}_{\UnitD; N}   ) ] \\
\text{(limit $\delta' \shrinkto 0$) } = & 
\EX [ \mathbb{I} \{ \InitSegmDelta{0}(1) \text{ is an even-index marked boundary point } w \}  f ( \InitSegmDelta{0} )  \EXSLEcurves{(N-1)}_{\UnitD \setminus \InitSegmDelta{0} } [ g( w, \SLEcurve_{1},  \ldots, \SLEcurve_{N-1}   ) ] \\
& +
\EX [  \mathbb{I}\{ \InitSegmDelta{0}(1) \text{ is not a marked boundary point}\}   f ( \InitSegmDelta{0} ) \EXSLEcurves{\mathbf{e}}_{\UnitD \setminus \InitSegmDelta{0} }   [  \EXSLEcurves{\mathbf{o}}_{\UnitD \setminus \InitSegmDelta{0} } [  g( \mathbf{o}, \mathbf{e} ) ],
\end{align*}
and thus the claim holds
\end{proof}

The proof of Proposition~\ref{thm: k le 4 local-to-global NSLE} with Assumption~\ref{ass: cond C'} can now be finished identically to the case with Assumption~\ref{ass: quantitative no boundary visits assumption}.

\subsubsection{\textbf{Termination points of initial segment}s}
\label{subsec: termination points of init segments}

Let us return to the question where the initial segments $\InitSegmDelta{0}$ terminate, left open in Section~\ref{subsubsec: up-to-swalliwing segments}.

\begin{prop}
\label{prop: initial segment end points}
For scaling limits with SLE parameter $\kappa \in (0, 4]$, the initial segment $\InitSegmDelta{0}$ almost surely terminates at an even-index marked boundary point.
For scaling limits with $\kappa \in (4, 8)$, $\InitSegmDelta{0}$ almost surely does not terminate at an even-index marked boundary point.
\end{prop}

\begin{proof}[Proof of Proposition~\ref{prop: initial segment end points} for $4 < \kappa < 8$]
Consider first the case $\kappa \in (4, 8)$. Take any subsequential limit ${\gamma}_{\UnitD; 1}$. By condition (G) for ${\gamma}_{\UnitD; 1}^{(n)}$, the end point of ${\gamma}_{\UnitD; 1}$ is almost surely not a double point of that curve, so we can study the final segment of ${\gamma}_{\UnitD; 1}$ (the initial segment of the reversed curve ${\gamma}_{\UnitD; 1}$) to answer whether ${\gamma}_{\UnitD; 1}$ hits $\bdry \UnitD$ somewhere else before hitting an even-index marked boundary point. By Theorem~\ref{thm: local multiple SLE convergence}, the initial segments converge to a local multiple SLE initial segment. Now, a chordal $\SLE(\kappa)$ initial segment with $\kappa \in (4, 8)$ almost surely hits the boundary outside of its starting point in any small neighbourhood of the end points. By absolute continuity, so does the local multiple SLE initial segment. Thus, (irrespective of which local multiple SLE initial segment turns out to be the final segment of  ${\gamma}_{\UnitD; 1}$) we can conclude that $\InitSegmDelta{0}$ almost surely does not terminate at an even-index marked boundary point. 
\end{proof}

In order to prove Proposition~\ref{prop: initial segment end points} for $\kappa \in (0,4]$, we will first need to prove Proposition~\ref{prop: relation to global multiple SLE}.

\begin{proof}[Proof of Proposition~\ref{prop: relation to global multiple SLE}]
By Proposition~\ref{thm: k le 4 local-to-global NSLE}(i), we can freely choose the order in which we inductively sample the different up-to-swallowing initial segments to obtain the collection of curves $({\gamma}_{\UnitD; 1}, \ldots,{\gamma}_{\UnitD; N})$. Sampling in an order that leaves ${\gamma}_{\UnitD; 1}$ last, it follows that ${\gamma}_{\UnitD; 1}$ is a chordal SLE in the domain left for it.
\end{proof}

\begin{proof}[Proof of Proposition~\ref{prop: initial segment end points} for $0 < \kappa \le 4$]
 By Proposition\ref{prop: relation to global multiple SLE}, ${\gamma}_{\UnitD; 1}$ is a chordal $\SLE(\kappa)$ in the domain left for it. It follows that ${\gamma}_{\UnitD; 1}$ almost surely only visits $\bdry \UnitD$ at its end points.
\end{proof}

\subsection{Proofs of Theorems~\ref{thm: relation to global multiple SLE 1} and~\ref{thm: relation to global multiple SLE 2}}

\begin{proof}[Proof of Theorem~\ref{thm: relation to global multiple SLE 1}]
%

We will show by induction over $N$ that all link patterns $\alpha \in \LP_N$ have a probability $\ge p$ to occur in the scaling limit $(\gamma_{\UnitD; 1}, \ldots, \gamma_{\UnitD; N})$, given that the distances between the marked boundary points $\Unitp_1, \ldots, \Unitp_{2N}$ are bounded from below by some number ($p$ of course depends on this number).  The base case $N=1$ is obvious, since there is only one link pattern.

Let us sketch the induction step with $N \ge 1$. Fix $\alpha \in \LP_N$, and a small tubular neighbourhood of the straight line segment connecting $\Unitp_1$ in $\UnitD$ to its pair boundary point given by $\alpha$. There is a positive probability that the usual chordal $\SLE(\kappa)$ from $\Unitp_1$ to $\Unitp_\infty$ has an initial segment in $U(\delta)$ (fixed but small $\delta$) that stays inside this tubular neighbourhood. (This follows from an analogous property of the Brownian motion: there is a positive probability that the driving function of the chordal SLE stays close to that of the straight line.) By absolute continuity (see, e.g.,~\cite{KP-pure_partition_functions_of_multiple_SLEs} for the expicit Radon-Nikodym derivatives) the initial segment $\InitSegmDelta{\delta}$ of the local multiple SLE also has a positive probability to stay in this tube. By weak convergence, this also holds for the discrete initial segments $\InitSegmDelta{\delta}^{(n)}$. Now, Assumption~\ref{ass: cond C'} (see especially Figure~\ref{fig: thick quads}) guarantees that the curve $\gamma_{\UnitD; 1}^{(n)}$ is then likely to pair the boundary point $\Unitp_1^{(n)}$ to its pair given by $\alpha$, and so that its remainder $\FinalSegmDelta{\delta}^{(n)}$ after $\InitSegmDelta{\delta}^{(n)}$ stays close to the tip of $\InitSegmDelta{\delta}^{(n)}$. The same conclusion holds for the weak limit $\gamma_{\UnitD; 1}$. Now, we have obtained a positive probability that $\gamma_{\UnitD; 1}$ connects $\Unitp_1$ to its pair in $\alpha$ and stays close to the corresponsing straight line. By the conditional law definition of the local-to-global multiple SLE and the inductive assumption, the remaning curves also have a positive probability to pair the marked boundary as given by $\alpha$.
\end{proof}

\begin{proof}[Proof of Theorem~\ref{thm: relation to global multiple SLE 2}]
The driving function of a global multiple $\SLE(\kappa)$ one-curve marginal has been identified in~\cite{PW}. Together with the precompactness theorem~\ref{thm: precompactness thm multiple curves}, this guarantees that Assumption~\ref{ass: dr fcns converge to loc mult SLE} holds in its conditional form. Assumption~\ref{ass: quantitative no boundary visits assumption} holds in the $\kappa \le 4$ case \emph{a posteriori}, relying on chordal $\SLE(\kappa)$:s having no boundary visits, and the weak convergence to global multiple SLEs.
\end{proof}

\bigskip{}

\section{Application examples}
\label{sec: application examples}

In this section, we show how our main results can be applied to deduce the convergence of multiple simultaneous random curves in various random models. Also relation to prior literature is discussed.

We will in this section deduce multiple SLE convergence for various discrete curve models using Theorem~\ref{thm: loc-2-glob multiple SLE convergence, kappa le 4}. For simplicity, we have chosen to state the convergence results in the topology of curves, $X(\C)$. Analogous convergences to local or local-to-global multiple SLEs naturally also hold in the other topologies of Theorems~\ref{thm: local multiple SLE convergence} and~\ref{thm: loc-2-glob multiple SLE convergence, kappa le 4} and Corollary~\ref{cor: local multiple SLE convergence -  strong topology}, and in these topologies also under the relaxed boundary regularity assumptions. Also the connection to global multiple SLEs, given in Theorem~\ref{thm: relation to global multiple SLE 1} holds.

\subsection{Three priorly known examples: Ising, FK-Ising and Percolation}

We start by discussing three models for which convergence results for multiple curves have appeared in prior literature. The purpose of this discussion is to demonstrate the applicability and practical application of our results.

\subsubsection{\textbf{The Ising model}}

Consider first the Ising model on the faces of the square lattice $\Z^2$ at critical temperature. We consider this model in simply-connected subgraphs $(\Gr; e_1, \ldots, e_{2N})$ of $ \Z^2$. We set boundary conditions that fix the spins on faces (edge-)adjacent to the boundary $\bdry \domain_\Gr$ of the corresponding planar domain, and the spin signs of these boundary conditions alter precisely at the edges $e_1, \ldots, e_{2N}$. (This of course puts some limitations on the subgraph $(\Gr; e_1, \ldots, e_{2N})$.)
The random curves $(\gamma_{\Gr; 1}, \ldots, \gamma_{\Gr; N})$ in $(\Gr; e_1, \ldots, e_{2N})$ are the magnetization cluster interfaces that surround the clusters adjacent to the boundary, with the convention of turning left when there are multiple ways to choose the interface. See, e.g.,~\cite{BPW} for a precise definition of the model, boundary conditions and the random curves.

Consider now the lattices $\delta_n \Z^{2} = \InfiniteGr_n$, where $\delta_n \shrinkto 0$ as $n \to \infty$, and their simply-connected subgraphs $(\Gr^{(n)}; e_1^{(n)}, \ldots, e^{(n)}_{2N})$ converging to some domain $(\domain; p_1, \ldots, p_{2N})$ in the Carath\'{e}odory sense. Study these discretizations under the assumptions and notation of Section~\ref{subsec: setup and notation}.

\begin{prop}
\label{prop: Ising convergence}
In the setup described above, the Ising interfaces $(\gamma^{(n)}_1, \ldots, \gamma^{(n)}_N)$ converge weakly in $X(\C)$ to the local-to-global multiple $\SLE(3)$ in $(\domain; p_1, \ldots, p_{2N})$ with the partition functions
\begin{align}
\label{eq: Ising part fcns}
\PartF_N (x_1, \ldots, x_{2N}) = \Pf \bigg( \big( \frac{1}{x_i- x_j} \big)_{i, j = 1}^{2N} \bigg),
\end{align}
where $\Pf(\cdot)$ denotes the Pfaffian of a matrix.
\end{prop}

Note that it is not immediate, but verified in~\cite[Proposition~4.6]{KP-pure_partition_functions_of_multiple_SLEs}, that~\eqref{eq: Ising part fcns} actually is a local multiple SLE partition function, as defined in Section~\ref{subsubsec: local multiple SLE partition functions}. 

\begin{proof}[Proof of Proposition~\ref{prop: Ising convergence}]
We wish to apply Theorem~\ref{thm: precompactness thm multiple curves} to deduce precompactness and Theorem~\ref{thm: loc-2-glob multiple SLE convergence, kappa le 4} to identify the scaling limit. In order to apply these results, we have to check that the discrete curve model satisfies their assumptions.

Assumptions of  Theorem~\ref{thm: precompactness thm multiple curves}:
\begin{itemize}
\item Alternating boundary conditions and DDMP are trivially satisfied.
\item Condition~(C) for the one-curve model is non-trivial, but has been verified in~\cite[Corollary~1.7]{CDCH}.
\end{itemize}
In addition to the above, applying Theorem~\ref{thm: loc-2-glob multiple SLE convergence, kappa le 4} requires the following assumptions to be satisfied:
\begin{itemize}
\item Assumption~\ref{ass: dr fcns converge to loc mult SLE}, i.e., convergence of driving functions to local multiple SLE with the partition function~\eqref{eq: Ising part fcns}, holds
by~\cite[Theorem~1.1]{Izyurov-critical_Ising_interfaces_in_multiply_connected_domains}\footnote{
\cite[Theorem~1.1]{Izyurov-critical_Ising_interfaces_in_multiply_connected_domains} is stated under some boundary regularity assumptions that need to be removed in Assumption~\ref{ass: dr fcns converge to loc mult SLE}. This assumption is made there in order to shorten the discussion on the convergence of the martingale observable needed in the scaling limit identification. However, this boundary regularity assumption can be relaxed, as discussed in~\cite[Section~1.1]{Izyurov-critical_Ising_interfaces_in_multiply_connected_domains}.
}. In the case of $N=1$ curve, Assumption~\ref{ass: dr fcns converge to loc mult SLE}, i.e., convergence to usual chordal SLE, was verified in~\cite[Theorem~1]{CDHKS-convergence_of_Ising_interfaces_to_SLE}.
\item Assumption~\ref{ass: approximability} holds trivially, and~\ref{ass: cond C'}, i.e., Condition (C'), is also a direct consequence of~\cite[Corollary~1.7]{CDCH}.
\end{itemize}
We have now verified all the assumptions of Theorems~\ref{thm: precompactness thm multiple curves} and Theorem~\ref{thm: loc-2-glob multiple SLE convergence, kappa le 4}, and the conclusions of the latter thus hold.
\end{proof}

\subsubsection*{\textbf{Prior results on the Ising model}}

The convergence of multiple Ising interfaces is by now understood rather completely, and the above proposition hence only provides a new proof for a known result, and a slightly different characterization of the weak limit. Convergence of one initial segment to that of a local multiple SLE was established in~\cite[Theorem~1.1]{Izyurov-critical_Ising_interfaces_in_multiply_connected_domains} via martingale observables. The weak convergence of full curves under the conditional measures $\PR^{(n)}[\; \cdot \; \vert \; \alpha]$ to global multiple SLEs, for any link pattern $\alpha$, was established in~\cite[Proposition~1.3]{BPW}. Later on,~\cite[Theorem~1.1]{PW18} established the convergence of the connection probabilities $\PR^{(n)}[ \alpha]$. Combining this with the convergence of the conditional measures $\PR^{(n)}[\; \cdot \; \vert \; \alpha]$, the weak convergence of the full curves under the unconditional measures
\begin{align*}
\PR^{(n)} [\cdot] = \sum_{\alpha \in \LP_N} \PR^{(n)}[ \alpha] \PR^{(n)}[\cdot \; \vert \; \alpha]
\end{align*}
follows. Interestingly, the results of~\cite{PW18} rely on the local convergence of~\cite{Izyurov-critical_Ising_interfaces_in_multiply_connected_domains}. This manifests the principle from Section~\ref{sec: intro} proving the convergence of $\PR^{(n)}[ \alpha]$ is roughly equivalent to finding a converging matringale observale. The two-interface case is discussed in~\cite{Wu17}.

\subsubsection{\textbf{Percolation}}

\begin{figure}
\includegraphics[width=0.5\textwidth]{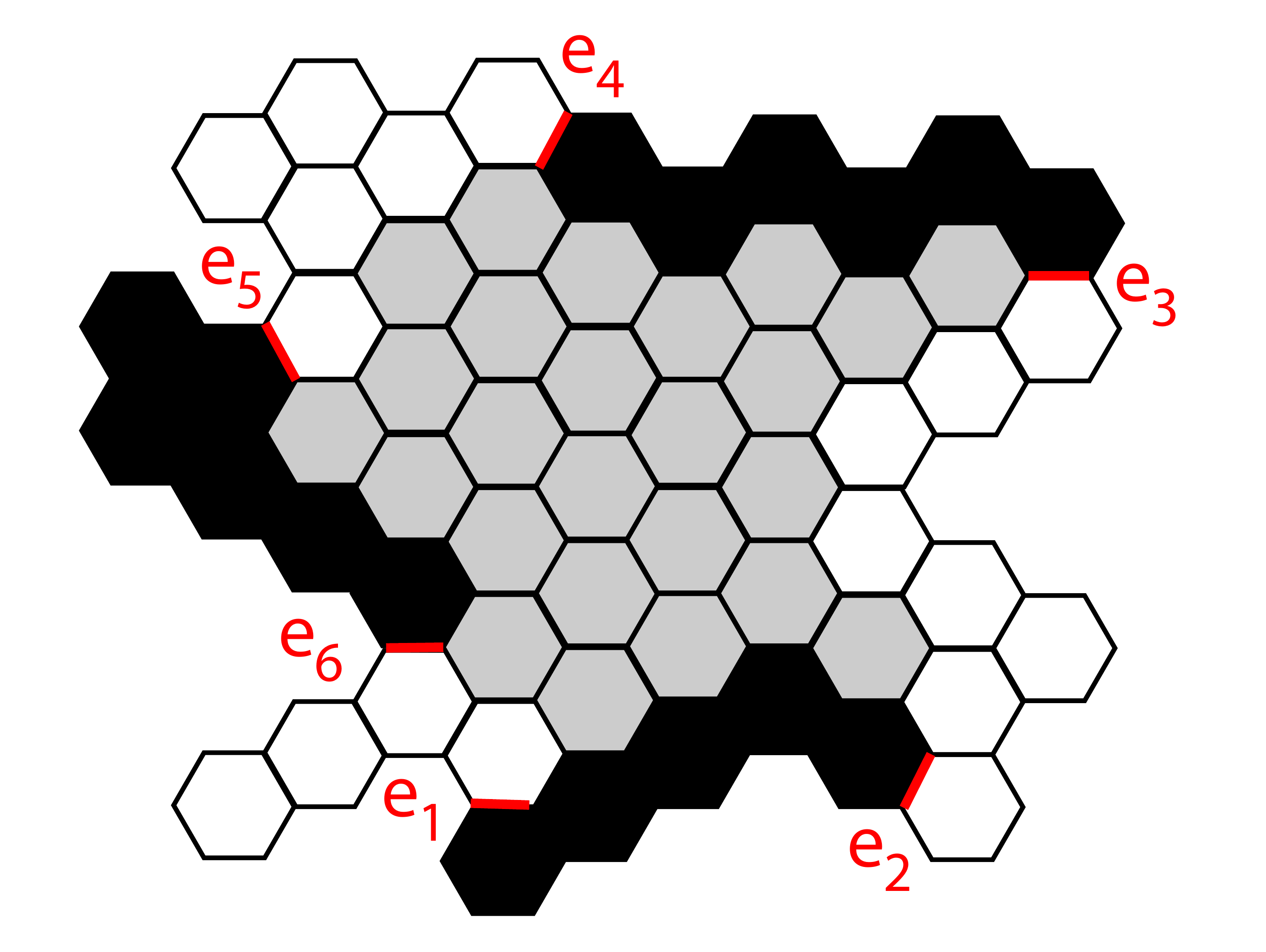}%
\includegraphics[width=0.5\textwidth]{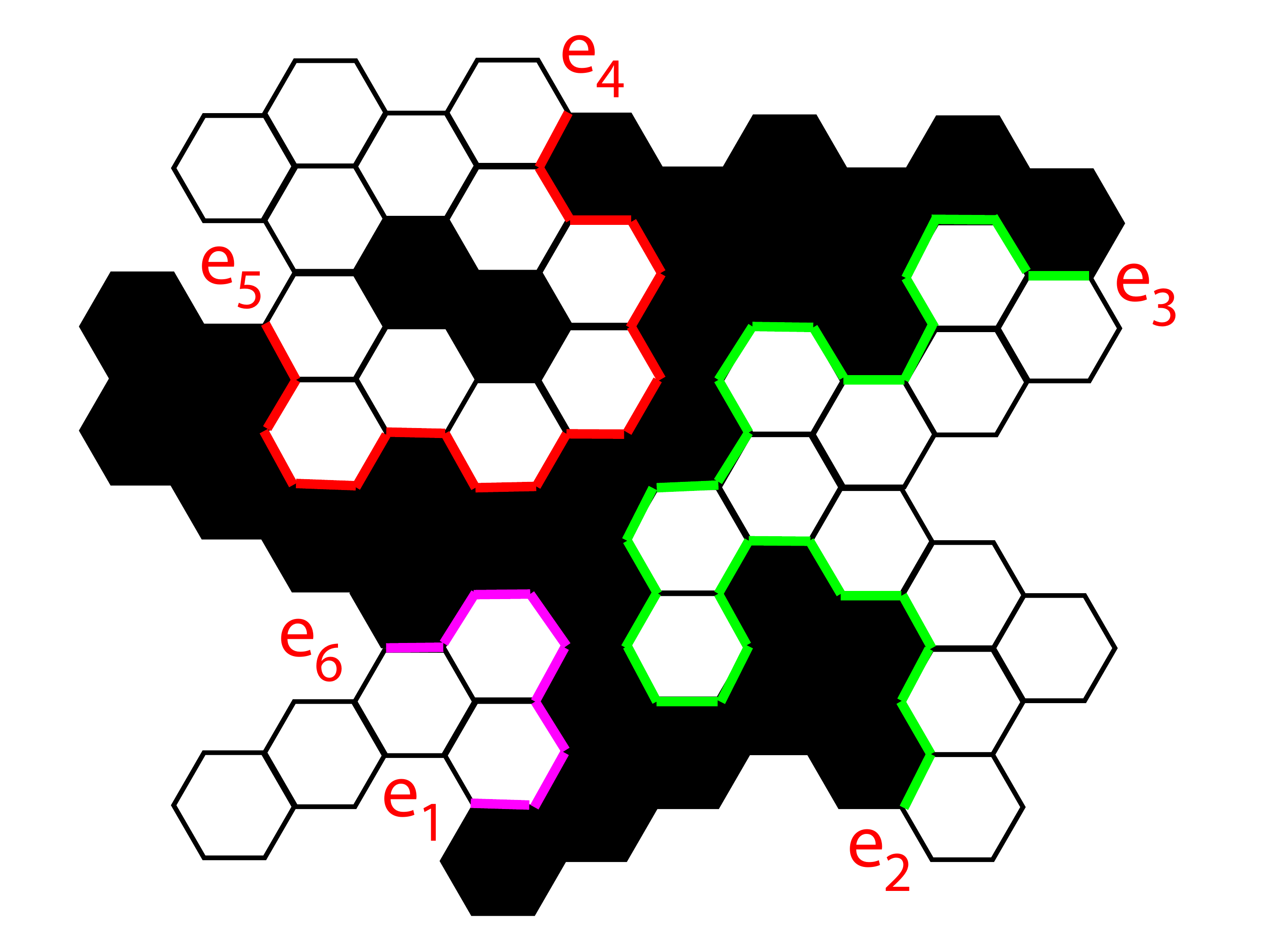}
\caption{
\label{fig: perco}
Left: A simply-connected subgraph $\Gr$ of $H$, with boundary faces altering colour between black and white over the marked boundary edges $e_1, \ldots, e_{6}$. The remaining faces are coloured gray. Right: Percolation colouring of the remaining faces, and the obtained random curves on $H$ bounding the black and white clusters adjacent to the boundary of $\Gr$.
}
\end{figure}

Consider now the critical percolation on the faces of the honeycomb lattice $H$, i.e., colouring each face independently either black or white, both with probability $1/2$. We consider this model in simply-connected subgraphs $(\Gr; e_1, \ldots, e_{2N})$ of $H$, fixing the colours of the faces adjacent to a boundary vertex, so that these boundary conditions alter colour precisely at the edges $e_1, \ldots, e_{2N}$.
The random curves  $(\gamma_{\Gr; 1}, \ldots, \gamma_{\Gr; N})$  in $(\Gr; e_1, \ldots, e_{2N})$ are the outer boundaries of the black or white clusters adjacent to the boundary, see Figure~\ref{fig: perco}.

Consider now the lattices $\delta_n H = \InfiniteGr_n$, where $\delta_n \shrinkto 0$ as $n \to \infty$, and their simply-connected subgraphs $(\Gr^{(n)}; e_1^{(n)}, \ldots, e^{(n)}_{2N})$ converging to some domain $(\domain; p_1, \ldots, p_{2N})$ in the Carath\'{e}odory sense. Study these discretizations under the assumptions and notation of Section~\ref{subsec: setup and notation}.

\begin{prop}
\label{prop: percolation convergence}
In the setup described above, the percolation interfaces $(\gamma^{(n)}_1, \ldots, \gamma^{(n)}_N)$ converge weakly in $X(\C)$ to the local-to-global multiple $\SLE(6)$ in $(\domain; p_1, \ldots, p_{2N})$ with the partition functions
\begin{align}
\label{eq: percolation part fcns}
\PartF_N (x_1, \ldots, x_{2N}) = 1.
\end{align}
\end{prop}

\begin{proof}
It is trivial to check that~\eqref{eq: percolation part fcns} are local multiple SLE partition functions with $\kappa = 6$ (this was observed, e.g., in~\cite[Proposition~4.9]{KP-pure_partition_functions_of_multiple_SLEs}). Observe also that the local multiple SLE initial segment from $p_1$ is then equal in distribution to the initial segment of a chordal $\SLE(6)$ from $p_1$ targeting at, say, $p_2$ (the precise choice of target is irrelevant due to the locality of the chordal $\SLE(6)$).

We now check the assumptions of Theorem~\ref{thm: precompactness thm multiple curves}:
\begin{itemize}
\item Alternating boundary conditions and DDMP are trivially satisfied.
\item Condition~(G) for the one-curve model follows from the Russo--Seymour--Welsh estimates.
\end{itemize}
The additional assumptions for Theorem~\ref{thm: loc-2-glob multiple SLE convergence, kappa le 4}:
\begin{itemize}
\item Assumption~\ref{ass: dr fcns converge to loc mult SLE} holds since the initial segment both in the percolation model and in the local multiple SLE~\eqref{eq: percolation part fcns} are independent of the number and locations of the other marked boundary points. Thus, the proof of convergence to chordal SLE for $N=1$ interface suffices. The latter has been addressed by various authors, see, e.g.~\cite{Smirnov-critical_percolation, CN07, Beffara-easy}.
\item Assumption~\ref{ass: approximability} holds trivially, and~\ref{ass: cond C'}, i.e., Condition (C'), is verified via condition (G'), which in turn is also a direct consequence of the Russo--Seymour--Welsh estimates.
\end{itemize}
We can now apply Theorem~\ref{thm: loc-2-glob multiple SLE convergence, kappa le 4} to complete the proof.
\end{proof}

\subsubsection*{\textbf{Prior results on percolation}}

Percolation interfaces are very well understood. (Indeed, the main reason for our discussion on it is the warning example of Section~\ref{subsec: warning example}.) Convergence results to multiple SLE type curves have been addressed in~\cite[Section~3]{KS18} and~\cite[Remark~1.5]{BPW}. Also the scaling limit of the full collection of percolation interfaces has been identified~\cite{CN-perco_full_limit}.

\subsubsection{\textbf{The FK cluster and FK-Ising models}}

\subsubsection*{\textbf{Definition of the models}}

Let us discuss the FK cluster model on the square lattice --- the FK-Ising model is later addressed as an important special case. We follow the conventions of the literature, and refer the reader to, e.g.,~\cite{DC-conf_inv_in_latt_models} for a good introduction. 

\begin{figure}
\includegraphics[width=0.5\textwidth]{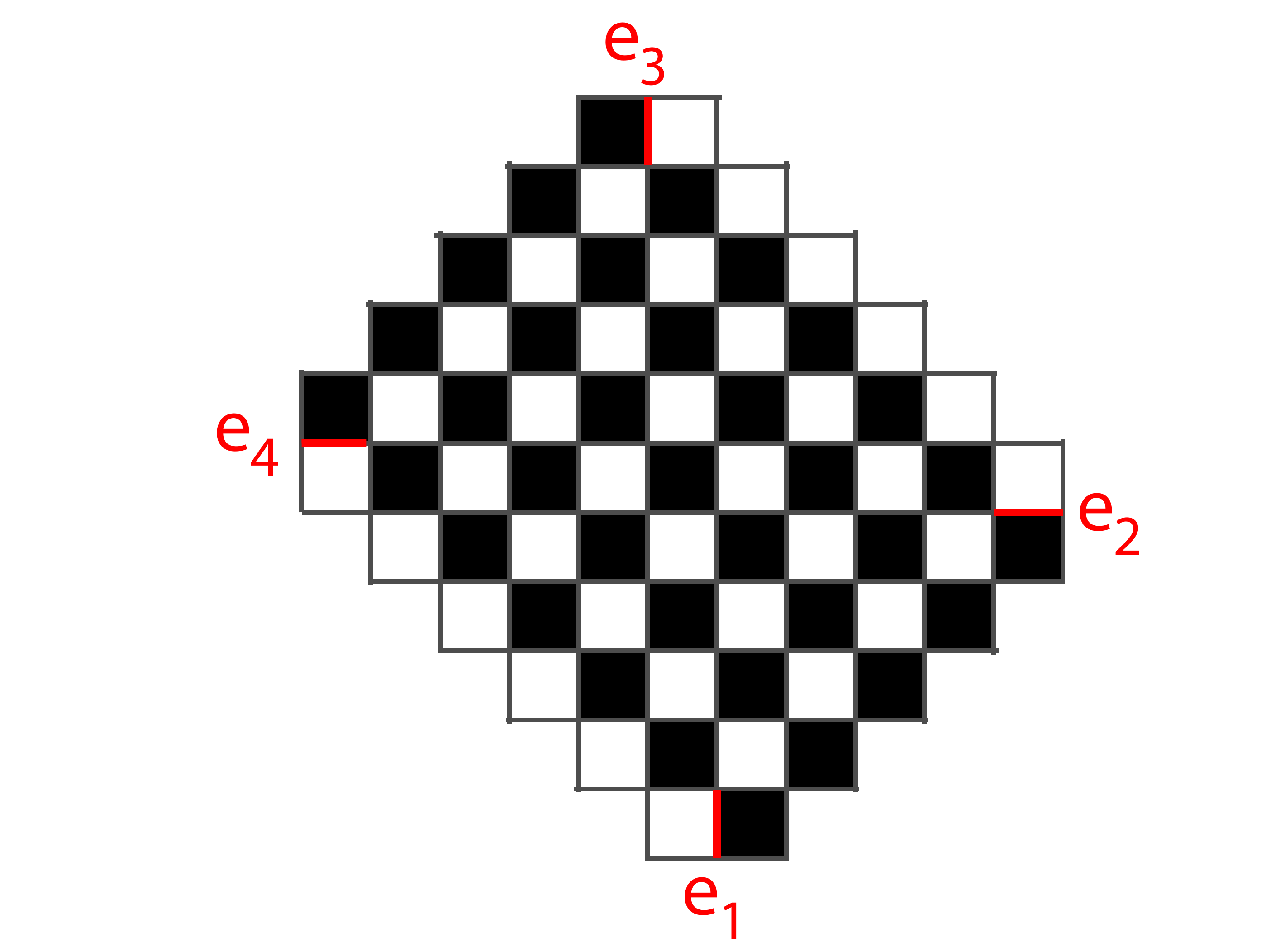}%
\includegraphics[width=0.5\textwidth]{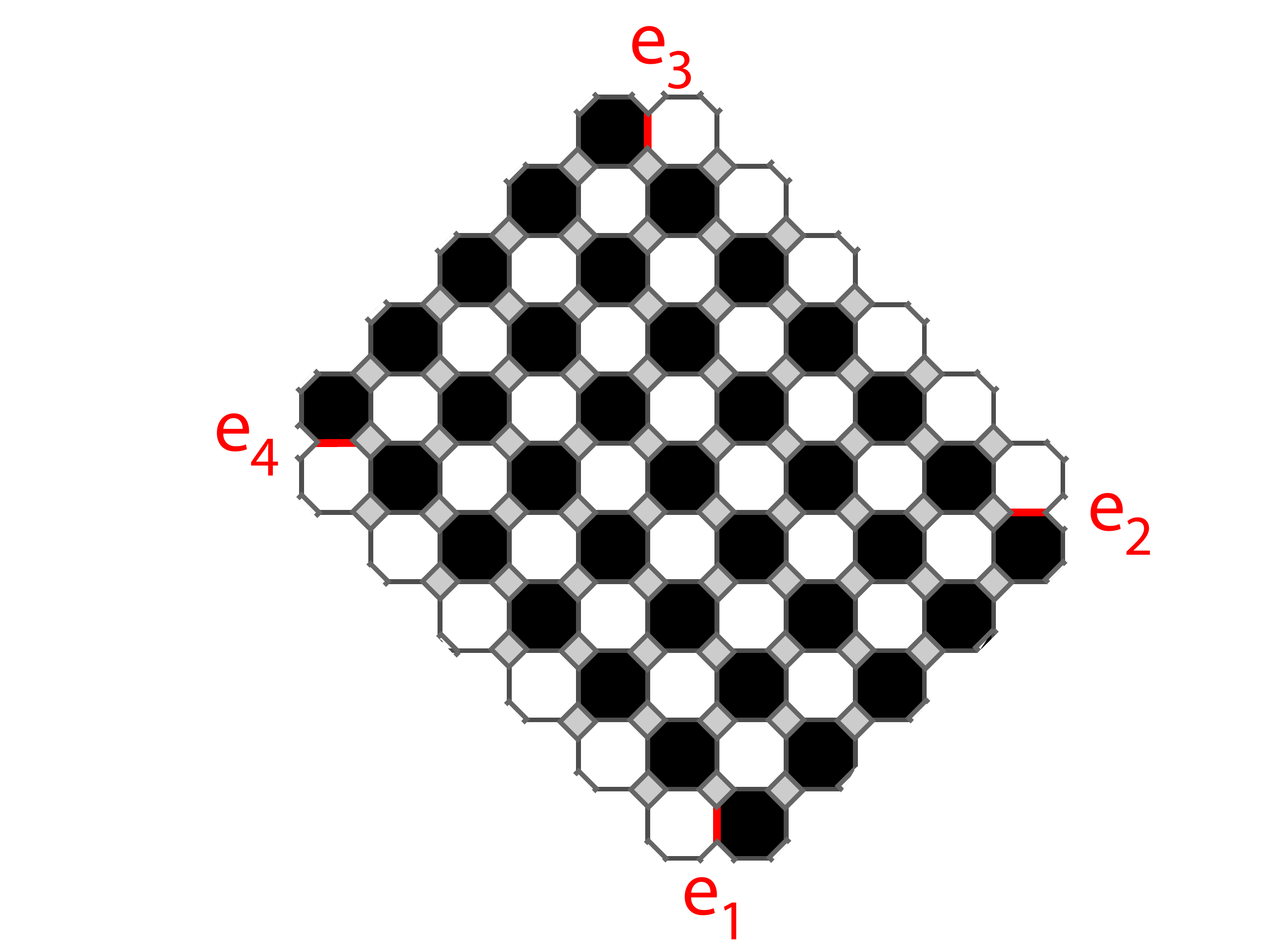} \\
\includegraphics[width=0.5\textwidth]{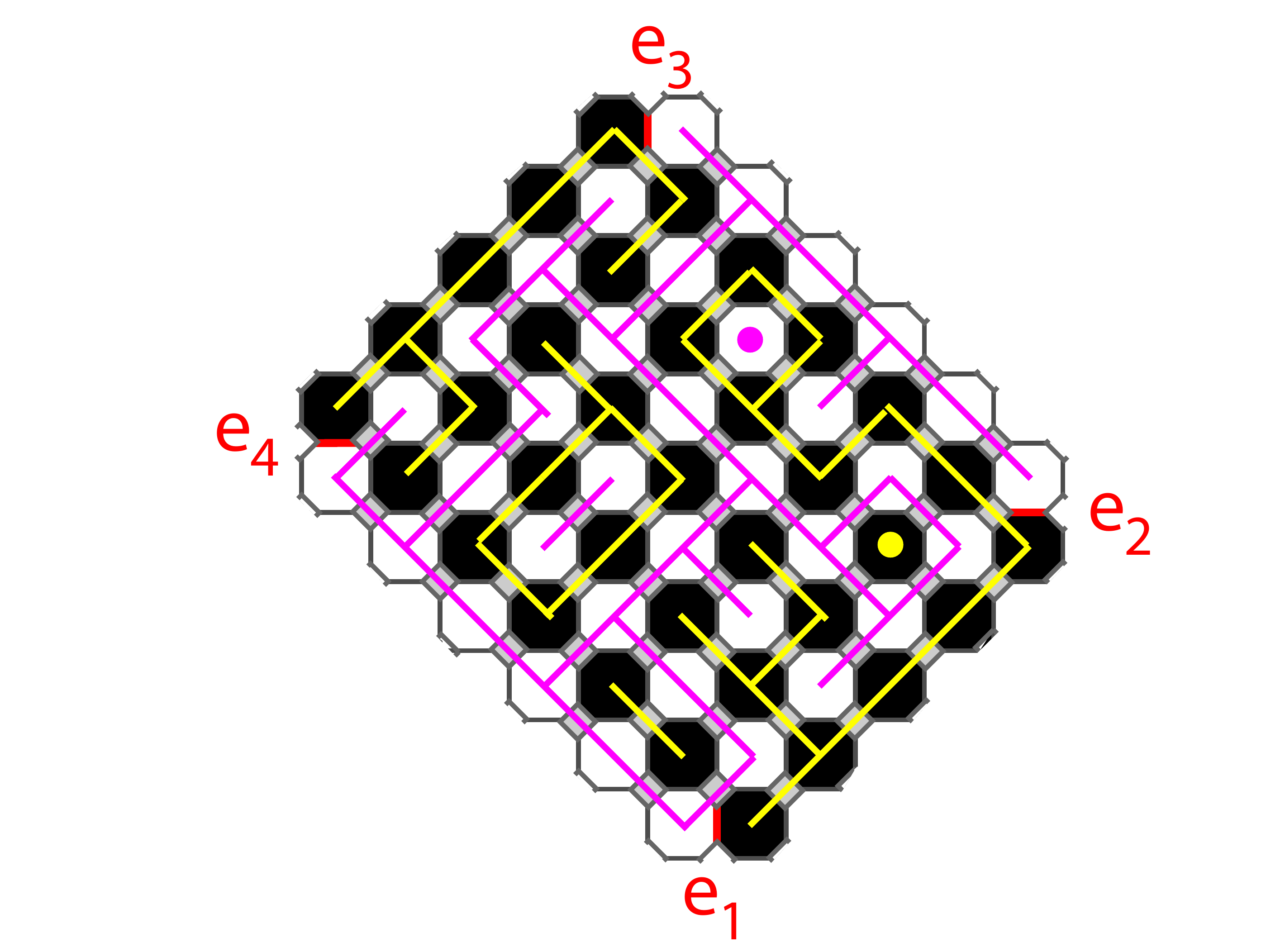}%
\includegraphics[width=0.5\textwidth]{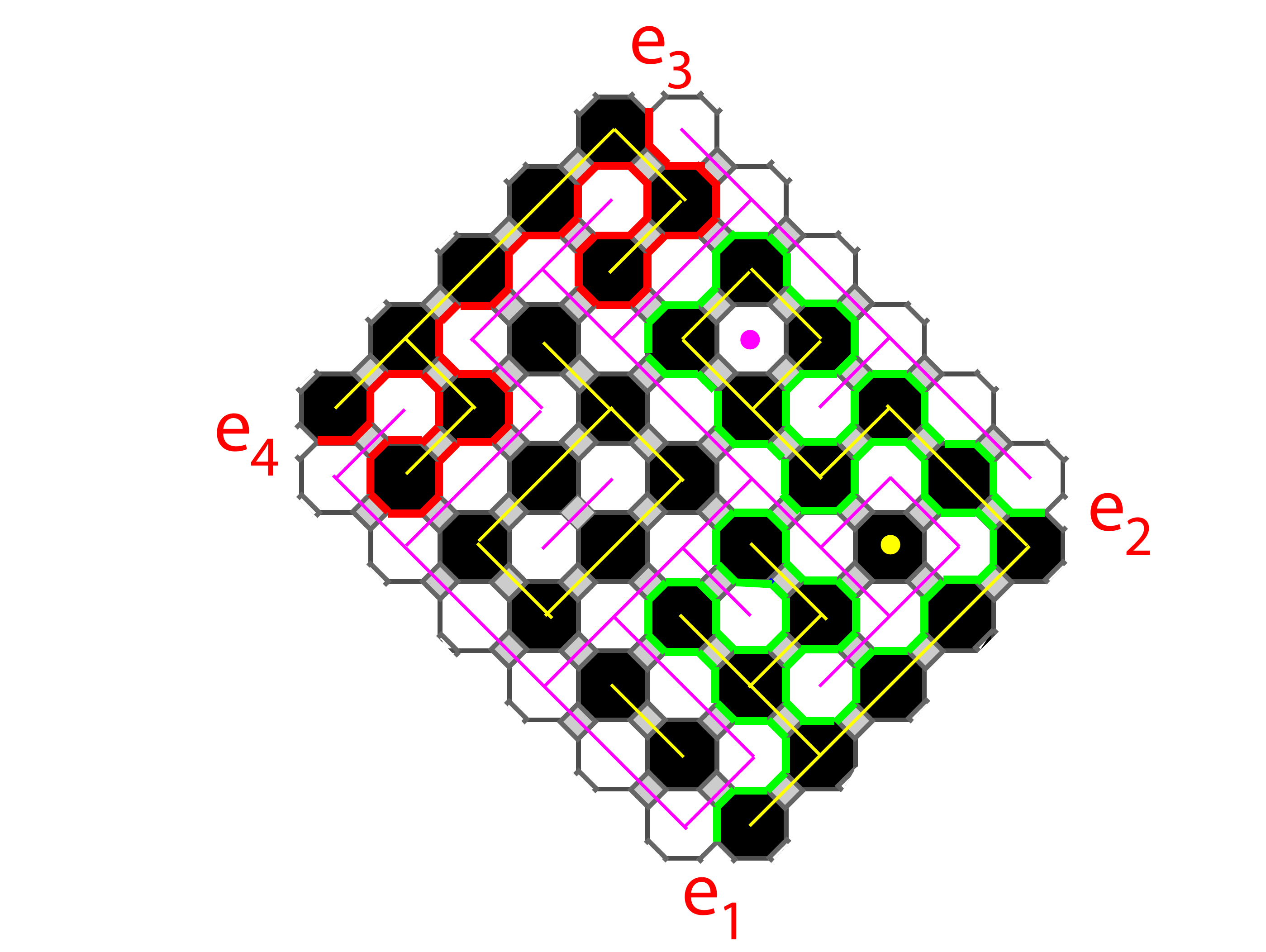}
\caption{
\label{fig: FK} The FK cluster model.
Top left: A simply-connected subgraph of $\Z^2$, with the faces chessboard coloured, and the colour of the faces edge-adjacent to the boundary changing at the marked boundary edges. Top right: the corresponding simply-connected subgraph of the square-octagon lattice $L$. Bottom left: the cluster-model subgraphs $\FKsub$ of $B_\Gr$ (in yellow) and and its dual subgraph $\FKsub^*$ of $W_\Gr$ (in magenta). Bottom right: the corresponding chordal random curves on the simply-connected subgraph of the square-octagon lattice, in red and green.
}
\end{figure}

First, colour the squares of $\Z^2$ black and white in a chessboard manner. The black (resp. white) squares form a scaled and rotated $\Z^2$ lattice, which we call the black (resp. white) lattice. These lattices are mutual duals. In the original $\Z^2$ lattice, take a simply-connected subgraph $\Gr$ whose boundary consists of $N$ black and $N$ white segments; by a black (resp. white) segment mean here that the $\Z^2$ squares inside $\Gr$ edge-adjacent to that boundary segment are all black (resp. white). The $2N$ marked boundary edges $e_1, \ldots, e_{2N}$ of $\Gr$ separate black and white boundary-neighbouring squares (top left in Figure~\ref{fig: FK}).

Next, on  subgraph of the black lattice inside $\Gr$, we impose wired boundary conditions, i.e., the black squares adjacent to the each of black boundary segment are identified, producing $N$ black boundary segment vertices. Call this graph $B_\Gr$. On the white squares inside $\Gr$, we impose slightly different boundary conditions: the white squares adjacent to the white boundary segments are \emph{all} identified, producing a \emph{single one} white boundary segment vertex. Denote this graph by $W_\Gr$. The graphs $B_\Gr$ and $W_\Gr$ are mutual planar duals.

Then, we run the \emph{FK cluster model} on $B_\Gr$ with parameters $p$ and $q$: choose a random subgraph $\FKsub$ of $B_\Gr$, whose vertices are all vertices in $B_\Gr$ but whose edges are a subset of the edges of $B_\Gr$, so that the probability of each different such subgraph $\FKsub$ is proportional to
\begin{align*}
p^{\# \{ \text{edges of } B_\Gr \text{ present in } \FKsub \} } (1 - p)^{ \# \{ \text{edges of } B_\Gr \text{ not present in } \FKsub \} } q^{ \# \{ \text{connected components of } \FKsub \} }.
\end{align*}
We will only consider the self-dual parameters satisfying $p=\sqrt{q}/(1 + \sqrt{q})$; this means that the dual subgraph $\FKsub^*$ of $W_\Gr$, consisting of all the vertices of $W_\Gr$ and the edges of $W_\Gr$ not crossed by $\FKsub$, is in distribution equal to the FK clusted model $W_\Gr$ with the same parameters $q$ and $p=\sqrt{q}/(1 + \sqrt{q})$. We will identify $\FKsub$ (resp. $\FKsub^*$) with the subgraph of the black (resp. white) lattice obtained from the edges of $\FKsub$ (resp. $\FKsub^*$) and the edges connecting black (resp. white) vertices of same black (resp. white) boundary segment (bottom left in Figure~\ref{fig: FK}).

Finally, the related random curve model is obtained from the loop representation of the FK clusters, which we describe next. First, we modify $\Gr$ slighty: every corner of the $\Z^2$ lattice is rounded by putting there a \emph{small square}, making the lattice into a square-octagon lattice, which we denote by $L$. Round the corners of the graph $\Gr$ to obtain a simply-connected subgraph of $L$, i.e., include the small squares at concave corners of $\Gr$ and exclude the ones at the convex or $180^\circ$ corners (top right in Figure~\ref{fig: FK}). Slightly abusively, let us in continuation refer by $\Gr$ to this subgraph of $L$. Now, with our convention of regarding $\FKsub$ (resp. $\FKsub^*$) as a subgraph of the black (resp. white) lattice, two opposite sides of each small square of $\Gr$ are crossed by exactly one edge from either $\FKsub$ or $\FKsub^*$; this is visible in Figure~\ref{fig: FK}(bottom left). In particular, $\FKsub$ can thus be bijectively encoded into the pairs of opposite non-crossed small-square edges of $\Gr$. Let us add to this collection of edges of $\Gr$ all the black-boundary edges of $L \cap \bdry \domain_\Gr$ and all the edges of $\Gr$ originating from $\Z^2$ and not on $\bdry \domain_\Gr$. The bijection with $\FKsub$ of course pertains. However, in the new collection of edges of $\Gr$, each vertex of $\Gr$ has either $0$ or $2$ edges adjacent to it: the edges form a collection of disjoint simple loops on $\Gr$.
This is the \emph{loop representation} of the FK cluster model. Each loop is adjacent to black (resp. white) squares of $\Z^2$ from exactly one connected component of $\FKsub$ (resp. $\FKsub^*$). We can sample $\FKsub$ via sampling its loop representation, in which case the probability of a loop configuration in proportional to
\begin{align}
\label{FK loop ptt fcn}
\sqrt{q}^{\# \{ \text{loops} \}}.
\end{align}

Consider now those loops that contain the black boundary segments of $\bdry \domain_\Gr$. In addition to the boundary segments, this collection of loops contains $N$ chordal paths inside $\Gr$, pairing the marked boundary edges $e_1, \ldots, e_{2N}$. The measures with random curves $(\PR^{( \Gr; e_1, \ldots, e_{2N})}, (\gamma_{\Gr; 1}, \ldots, \gamma_{\Gr; N}))$ are the FK cluster loop representations and these chordal paths on $\Gr$ (bottom right in Figure~\ref{fig: FK}).

The FK-Ising model is the FK cluster model with parameters $q=2$ and $p=\sqrt{q}/(1 + \sqrt{q}) =\sqrt{2}/(1 + \sqrt{2})$.

\subsubsection*{\textbf{Precompactness of the FK cluster models}}

It has been conjectured that the $N$ random curves in the self-dual FK cluster model introduced above, with parameter $q \in [0, 4)$, converges to SLE type scaling limits, with the SLE parameter $\kappa$ depending on the cluster model parameter via
\begin{align*}
\kappa = \frac{4 \pi}{\arccos (- \sqrt{q}/2)}.
\end{align*}
For such predictions, see, e.g.,~\cite{Schramm-ICM, Smirnov-ICM} for $N=1$ curve and chordal SLEs,~\cite{BPW} for general $N$ and global multiple SLEs. Regarding such convergence proofs, the precompactness part has been established~\cite{DC-conf_inv_in_latt_models,DST-q-Potts_phase_tr,  BPW}, but the limit identification step is missing, except in the FK-Ising case $q=2$. We now check that also when following the convergence proof strategy and of this paper, only the limit identification step, i.e., Assumption~\ref{ass: dr fcns converge to loc mult SLE} is missing.

\begin{prop}
\label{prop: FK cluster precompactness}
The discrete curve models obtained from the loop representation of the FK cluster model with $q \ge 1$ satisfy the assumptions of Theorem~\ref{thm: precompactness thm multiple curves}. Also the assumptions of Theorem~\ref{thm: loc-2-glob multiple SLE convergence, kappa le 4} except for possibly Assumption~\ref{ass: dr fcns converge to loc mult SLE} hold.
\end{prop}

\begin{proof}
For the assumptions of Theorem~\ref{thm: precompactness thm multiple curves}, the discrete models clearly have alternating boundary conditions and satisfy the DDMP. Condition (G) for the one-curve model has been verified in~\cite[Theorem~6]{DST-q-Potts_phase_tr}.
As regards the assumptions of Theorem~\ref{thm: loc-2-glob multiple SLE convergence, kappa le 4}, Assumption~\ref{ass: approximability} holds trivially. Assumption~\ref{ass: cond C'}, i.e., Condition (C'), is verified via condition (G'), which in turn is proven identically to condition (G) in~\cite{DST-q-Potts_phase_tr}.
\end{proof}

\subsubsection*{\textbf{Convergence of two FK-Ising interfaces}}

Let us now discuss the weak convergence in the FK-Ising model, i.e., the FK cluster model with $q = 2$ with two curves. We keep the discussion here largely informal, referring to the more complete account in~\cite{KS-bdary_loops_FK, KS18} for those parts. Multiple interfaces in FK cluster and FK Ising models have been studied priorly in~\cite{KS-bdary_loops_FK, KS18, BPW}, and the scaling limits in the setups considered below could be identified (with slightly different characterizations) by combining results from those papers. Following~\cite{KS18}, we consider a slightly modified FK model, so that in the loop representation probabilities~\ref{FK loop ptt fcn}, boundary-touching loops are not counted.

Note first that Proposition~\ref{prop: FK cluster precompactness} applies for the FK-Ising model. (Conditions (C) and (C') can also be verified directly then~\cite{CDCH}.) Thus, in order to apply the main theorem~\ref{thm: loc-2-glob multiple SLE convergence, kappa le 4} of this paper, it remains to verify Assumption~\ref{ass: dr fcns converge to loc mult SLE}. To that end, first, the scaling limit of $N=1$ curve has been identified in~\cite[Theorem~2]{CDHKS-convergence_of_Ising_interfaces_to_SLE} as a chordal $\SLE(16/3)$. For $N=2$ curves, the driving process of the initial segment of one curve has been identified in~\cite[Equation~(94)]{KS-bdary_loops_FK}.
Recalling that the proof of Theorem~\ref{thm: loc-2-glob multiple SLE convergence, kappa le 4} is based on an induction over $N$, we can thus apply it for the FK-Ising model with $N=2$ curves. We conclude the following.

\begin{prop}
\label{prop: FK Ising convergence}
The curves $(\gamma_{\UnitD; 1}^{(n)}; \gamma_{\UnitD; 2}^{(n)})$ under the FK-Ising model with $N=2$ curves converges weakly to the following limit: the up-to-swallowing initial segment $\InitSegmDelta{0}$ is described by the Loewner growth in~\cite[Equation~(94)]{KS-bdary_loops_FK}. Given $\InitSegmDelta{0}$, the regular conditional laws of the remainder of the curves are two independent chordal $\SLE(16/3)$ curves in the respective domains of $\UnitD \setminus \InitSegmDelta{0}$, with the three remaining marked boundary points and one at the tip of $\InitSegmDelta{0}$.
\end{prop}

The curves $(\gamma_{\UnitD; 1}^{(n)}; \gamma_{\UnitD; 2}^{(n)})$ under the FK-Ising model conditional on a link pattern were studied in~\cite[Theorem~1.1]{KS18}. The initial segment $\InitSegmDelta{0}$ is then described by the hypergeometric $\SLE(16/3)$. The following convergence of a pair of curves was stated there without explicit proof.

\begin{prop}
\label{prop: conditional FK Ising convergence}
Proposition~\ref{prop: FK Ising convergence} holds for the curves $(\gamma_{\UnitD; 1}^{(n)}; \gamma_{\UnitD; 2}^{(n)})$ under the FK-Ising model conditional on a link pattern, with $\InitSegmDelta{0}$ changed to the hypergeometric SLE of~\cite[Equation~(2)]{KS18}.
\end{prop}

\begin{proof}
Consider first the unconditional scaling limit of Proposition~\ref{prop: FK Ising convergence}. The boundary point $\Unitp_1$ connects to $\Unitp_2$ (resp. $\Unitp_4$) if and only if the tip of $\InitSegmDelta{0}$ is on the arc $(\Unitp_2, \Unitp_3)$ (resp. $(\Unitp_3, \Unitp_4)$) of $\bdry \UnitD$, by Proposition~\ref{prop: initial segment end points}. Both of these occur with positive probability, given explicitly in~\cite[Equation~(4)]{KS18}. Condition now the unconditional scaling limit of Proposition~\ref{prop: FK Ising convergence} on the tip of $\InitSegmDelta{0}$ lying one of these arcs, say $(\Unitp_2, \Unitp_3)$.  On the one hand, Proposition~\ref{prop: FK Ising convergence} gives the regular conditional distributions of the remaining curves as chordal SLEs. On the other hand, this conditioning reveals the link pattern and thus~\cite[Theorem~1.1]{KS18} tells that the law of $\InitSegmDelta{0}$ under this condition is the hypergeometric SLE.
\end{proof}

\subsection{A warning example: not all local multiple SLEs are global}
\label{subsec: warning example}

The identification of the scaling limit of percolation in the previous subsection was particularly interesting due to the following consequence.

\begin{prop}
Let $N \ge 3$ and $\Unitp_1, \ldots, \Unitp_{2N} \in \bdry \UnitD$ be any $2N$ distinct points. There is a collection of $N$ chordal random curves in $X(\overline{\UnitD})$, pairing the boundary points $\Unitp_1, \ldots, \Unitp_{2N}$, such that the initial and final segments of these curves in any localization neighbourhoods are those of the local multiple $\SLE(6)$ in Proposition~\ref{prop: percolation convergence}, but the full curves are not the local-to-global multiple $\SLE(6)$ in Proposition~\ref{prop: percolation convergence},
\end{prop}

\begin{figure}
\includegraphics[width=0.4\textwidth]{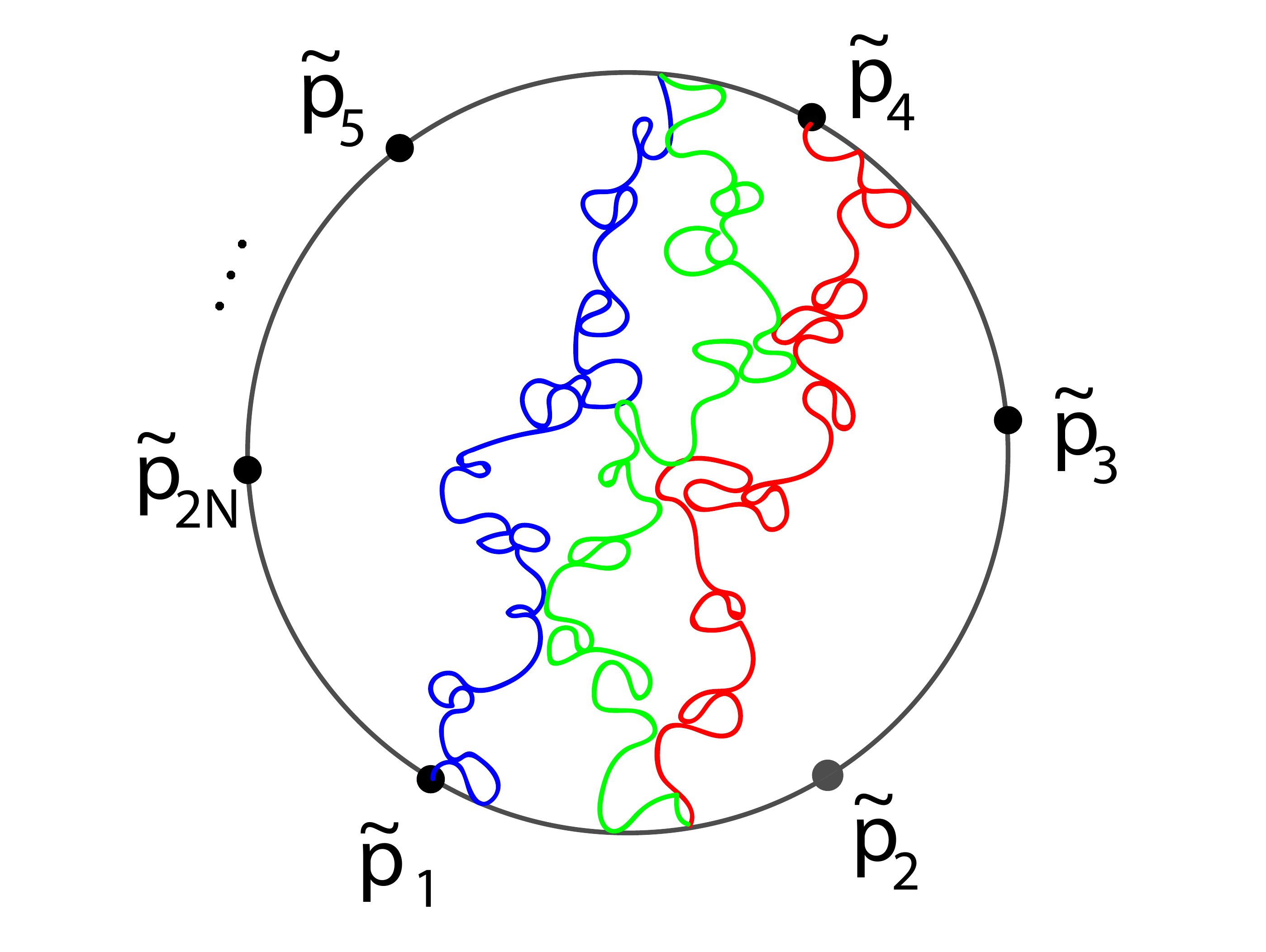}%
\includegraphics[width=0.4\textwidth]{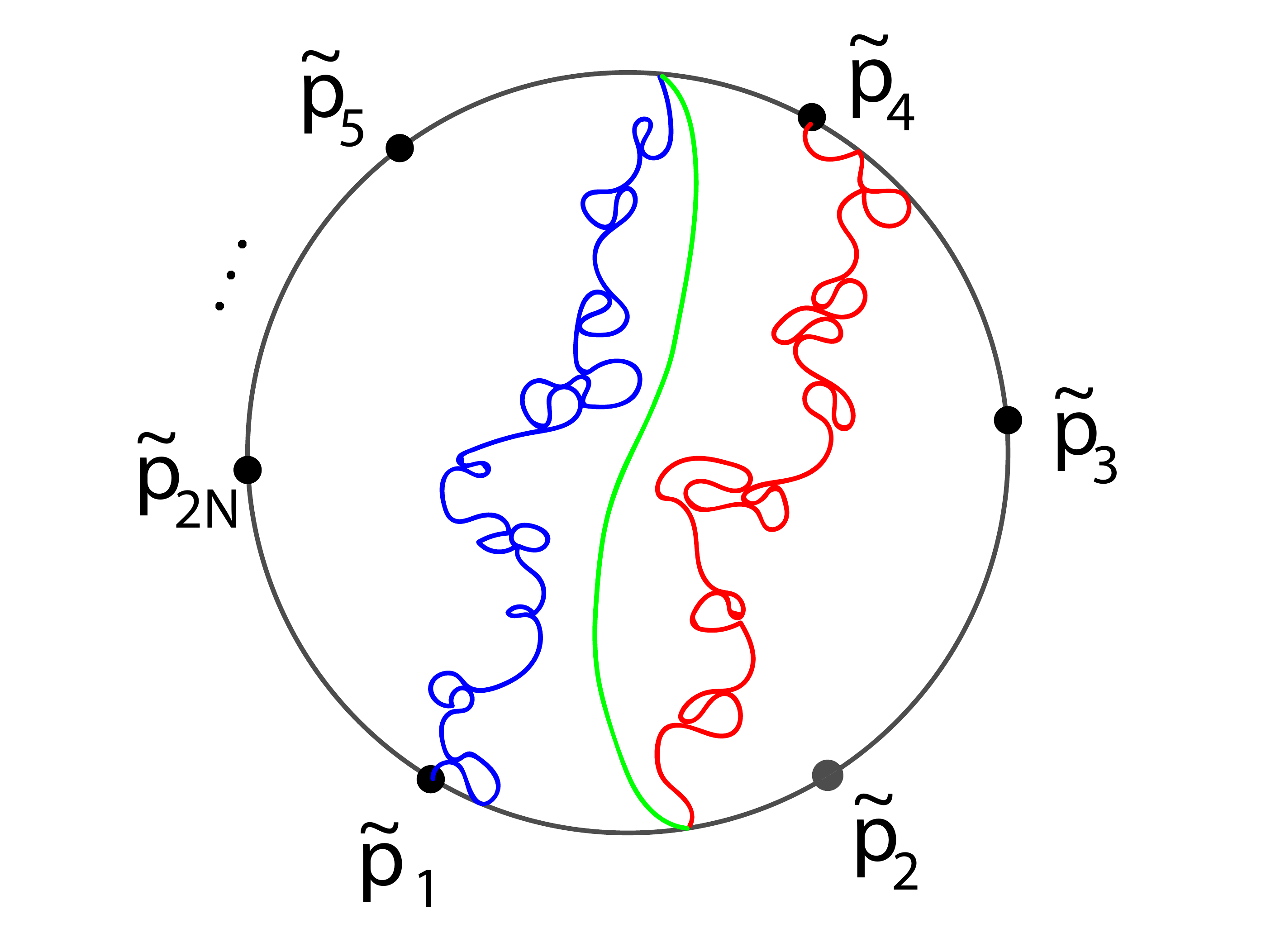}
\caption{
\label{fig: warning ex}
Schematic illustrations. Left: one curve of the local-to-global multiple $\SLE(6)$ in Proposition~\ref{prop: percolation convergence}, divided into three segments. Right: the corresponding three curve segments in a different collection of random curves, whose localizations in any localization neighbourhoods are the same as those of the curves on the left.
}
\end{figure}

\begin{proof}
Start from the local-to-global multiple $\SLE(6)$ in Proposition~\ref{prop: percolation convergence}. By Proposition~\ref{thm: k le 4 local-to-global NSLE}, it has the following property: if the up-to-swallowing initial segment from $\Unitp_1$~(in blue in Figure~\ref{fig: warning ex}(left)) hits the arc $(\Unitp_2, \Unitp_{2N})$ in $(\Unitp_4, \Unitp_5)$ and is disjoint from the up-to-swallowing initial segment from $\Unitp_4$~(in red in Figure~\ref{fig: warning ex}(left)), which hits the arc $(\Unitp_5, \Unitp_3)$ in $(\Unitp_1, \Unitp_2)$, then $\Unitp_1$ and $\Unitp_4$ are connected by a random curve. Denote this event by $E$. By the Russo--Seymour--Welsh estimates, $E$ has a positive probability. On the event $E$, the random curve from $\Unitp_1$ to $\Unitp_4$ is a concatenation of three curves (in order 1--2--3, in blue, green, and red in Figure~\ref{fig: warning ex}(left), respectively): 1) the up-to-swallowing initial segment from $\Unitp_1$ to $(\Unitp_4, \Unitp_5)$; 3) the reversal of the up-to-swallowing initial segment from $\Unitp_4$ to $(\Unitp_5, \Unitp_3)$; and 2) a chordal $\SLE(6)$ between the tips of these the two up-to-swallowing initial segments, in the domain restricted by them.

Now, on the event $E$, let us replace the curve (2) above by a hyperbolic geodesic, i.e., the chordal $\SLE(0)$, in the same domain; see Figure~\ref{fig: warning ex}(right). It is elementary to verify that after this replacement, we obtain a different family of random curves, whose all localizations are nevertheless the same as before this replacement operation. This proves the claim.
\end{proof}

\begin{rem}
The counterexample in the proof above is conformally invariant, and may be defined via conformal maps in any domain $(\domain; p_1, \ldots, p_{2N})$ with marked prime end that possess radial limits.
\end{rem}

\subsection{Outline of a new example: UST branches}
\label{subsec: LERW outline}

Let us return to direct applications of our main theorems~\ref{thm: precompactness thm multiple curves} and~\ref{thm: loc-2-glob multiple SLE convergence, kappa le 4}. The next discrete model that we will study is the uniform spanning tree (UST). Verifying the assumptions of these theorems in that model would take up some space, and thus we only outline the proofs in this subsection. Also, we will for simplicity restrict our consideration in this paper to the lattice $\Z^2$, even if all the results could be derived on any isoradial lattice, as defined in~\cite{CS-discrete_complex_analysis_on_isoradial}.

Let $(\Gr; e_1, \ldots, e_{2N})$ be a simply-connected subgraph of $\Z^2$ with marked boundary edges. 
Consider the uniform spanning tree on the graph $\Gr/\bdry$ obtained by identifying all the boundary vertices of $\Gr$. Each interior vertex $v \in \Vert^\circ$ thus connects to the boundary vertices $\bdry \Vert$ by a unique path on such a tree. Condition the UST on $\Gr/\bdry$ on the event that that such boundary paths from the interior vertices of the odd edges $e_1, e_3, \ldots, e_{2N - 1}$ reach $\bdry \Vert$ via the even edges $e_2, e_4, \ldots, e_{2N}$, each using a different even edge; see Figure~\ref{fig: UST} for illustration. (This conditioning making sense puts some very mild limitations on the subgraph $(\Gr; e_1, \ldots, e_{2N})$.) The probability measures $\PR^{(\Gr; e_1, \ldots, e_{2N})}$ that we are interested in are these conditional USTs, and the random chordal curves $\gamma_{\Gr; 1}, \ldots, \gamma_{\Gr; N}$ are the chordal graph paths consisting of the odd edges $e_1, e_3, \ldots, e_{2N - 1}$ and the boundary paths from their interior vertices. We call $\gamma_{\Gr; 1}, \ldots, \gamma_{\Gr; N}$ \emph{UST boundary branches}. This model is also sometimes called multiple loop-erased random walks (LERWs), due to the connection of the discrete models~\cite{Wilson-generating_random_spanning_trees}.

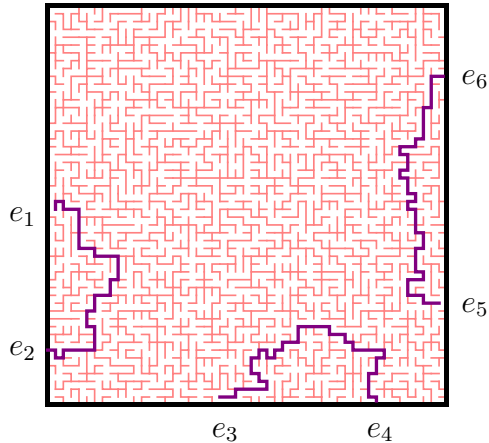
\begin{figure}
\centering
\begin{overpic}[width=0.35\textwidth]{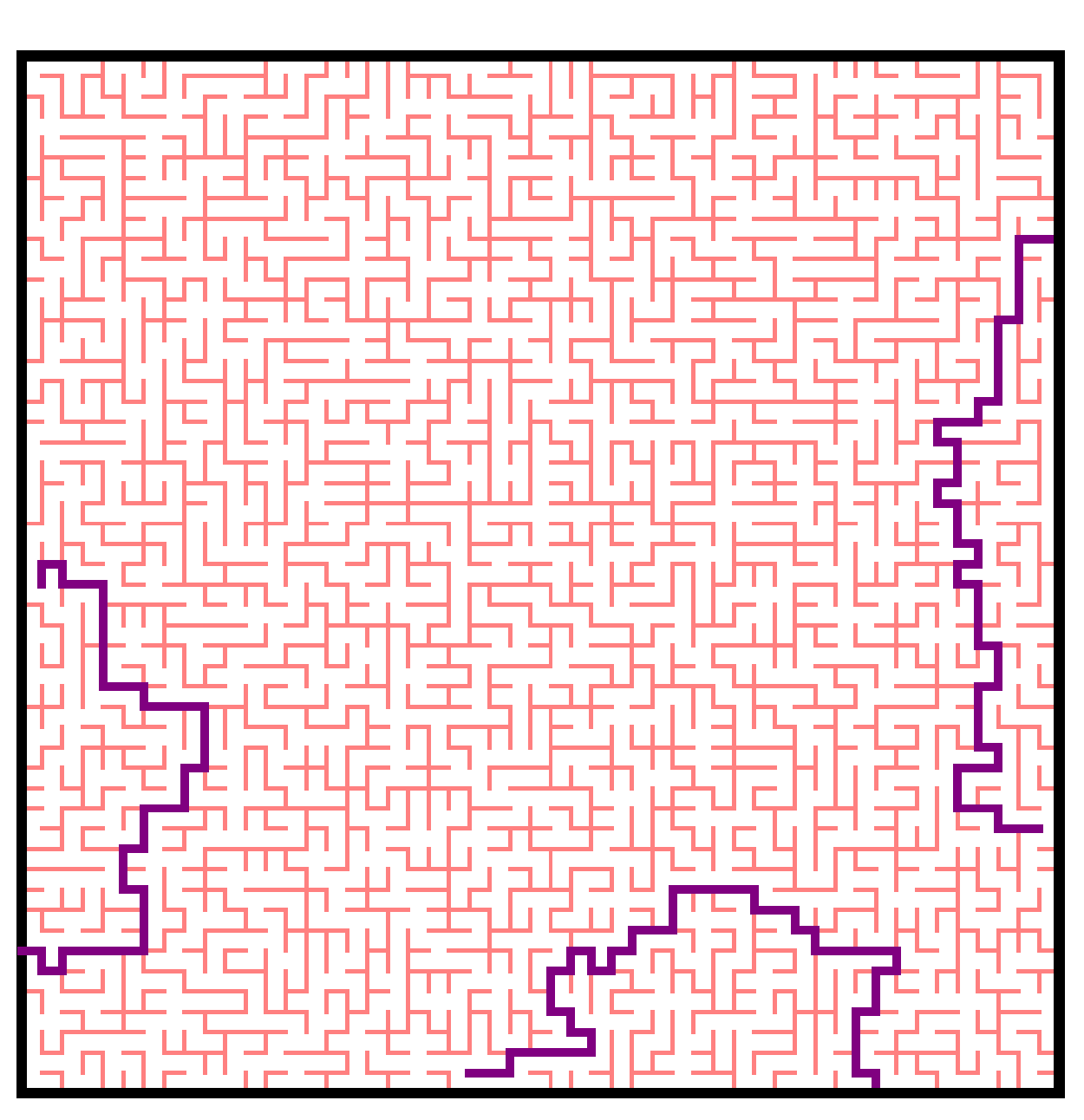}
 \put (98,77) {\Large $e_6$}
 \put (-7,45) {\Large $e_1$}
 \put (-7,14) {\Large $e_2$}
 \put (40,-5) {\Large $e_3$}
  \put (76,-5) {\Large $e_4$}
 \put (98,24) {\Large $e_5$}
\end{overpic}
\bigskip
\caption{\label{fig: UST} A uniform spanning tree, with the paths to the boundary from the interior vertices of the odd edges $e_1$, $e_3$, and $e_5$ reaching the boundary each via a different even edge $e_2$, $e_4$, or $e_6$.}
\end{figure}

Consider now the lattices $\delta_n \Z^{2} = \InfiniteGr_n$, where $\delta_n \shrinkto 0$ as $n \to \infty$, and their simply-connected subgraphs $(\Gr^{(n)}; e_1^{(n)}, \ldots, e^{(n)}_{2N})$ as above, converging to some domain $(\domain; p_1, \ldots, p_{2N})$ in the Carath\'{e}odory sense. Study these discretizations under the assumptions and notation of Section~\ref{subsec: setup and notation}.

\begin{thm}
\label{thm: UST convergence}
In the setup described above, the UST boundary branches $(\gamma^{(n)}_1, \ldots, \gamma^{(n)}_N)$, both unconditional and conditional on a link pattern $\alpha \in \LP_N$, converge weakly in $X(\C)$ to the local-to-global multiple $\SLE(2)$ in $(\domain; p_1, \ldots, p_{2N})$. The scaling limit in the conditional case is described by local multiple SLEs with the partition functions
\begin{align}
\label{eq: UST conditional part fcns}
\PartF_\alpha (x_1, \ldots, x_{2N})
\end{align}
given in~\cite[Equation~(3.14)]{KKP}, and in the unconditional case by
\begin{align}
\label{eq: UST unconditional part fcns}
\PartF_N (x_1, \ldots, x_{2N}) =  \sum_{\alpha \in \LP_N} \PartF_\alpha (x_1, \ldots, x_{2N}).
\end{align}
An alternative expression for $\PartF_N$ is given in~\cite[Lemma~4.12]{PW}.
\end{thm}

\begin{proof}[Proof precompactness]
Let us first verify the assumptions of Theorem~\ref{thm: precompactness thm multiple curves} on the UST discrete curve model. First, the model clearly has alternating boundary conditions. Actually, by the bijection argument in~\cite[Lemma~3.1]{KKP}, we can re-label the edges $e_1, \ldots, e_{2N}$ to $\hat{e}_1, \ldots, \hat{e}_{2N}$ counterclockwise starting from any edge, and the UST models on $(\Gr; e_1, \ldots, e_{2N})$ and $( \Gr; \hat{e}_1, \ldots, \hat{e}_{2N} )$ yield the same distribution of random curves.

The DDMP follows from the fact that for the UST on any graph $\Gr$, the UST conditional on a subtree is in distribution a UST on the graph obtained by identifying the vertices of that subtree. Now, property~(ii) in the definition of the DDMP follows from this property of the UST, likewise property~(i) when conditioning on a branch initial segment from an even boundary edge $e_2, e_4, \ldots, e_{2N}$. For property~(i) with a branch initial segment from an odd-index boundary edge, use the re-labelling argument of the previous paragraph.

Finally, Condition~(C) for the one-curve model has been verified in~\cite[Theorem~4.18]{KS}. The assumptions of Theorem~\ref{thm: precompactness thm multiple curves} for discrete curve model have thus been verified.
\end{proof}

\begin{proof}[Outline of identification]
The verification of the additional assumptions imposed in Theorem~\ref{thm: loc-2-glob multiple SLE convergence, kappa le 4} is postponed to a follow-up paper. The outline is the following:

1) Assumption~\ref{ass: dr fcns converge to loc mult SLE}, i.e., convergence of driving functions to local multiple SLE with the partition function~\eqref{eq: UST conditional part fcns} or~\eqref{eq: UST unconditional part fcns}: the precompactness part guarantees the existence of subsequential limits of the driving functions $\DrFcnNoInd_j$. Any subsequential limit is identified as a local multiple $\SLE(2)$ initial segment via a martingale observable, as in classical SLE convergence proofs. There are several alternative martingale observables, one (in the $\alpha$-conditional model) being the ratio of partition functions $Z_\beta / Z_\alpha$, with any $\beta \in \LP_N$, in the notation of~\cite[Theorem~3.12]{KKP}. With the expression given there for this observable and some discrete harmonic analysis, one can prove the convergence of the observable.

2) Assumption~\ref{ass: approximability} holds trivially.

3) We verify Assumption~\ref{ass: quantitative no boundary visits assumption} for the conditional models --- if it holds for the conditional models with any link pattern $\alpha$, it clearly holds for the unconditional model. In the conditional case, construct the uniform spanning tree by Wilson's algorithm \cite{Wilson-generating_random_spanning_trees}; the paths $\gamma_{\Gr; 1}, \ldots, \gamma_{\Gr; N}$ are then loop-erasures of suitable random walk excursions (see, e.g.,~\cite[Corollary~3.5(c)]{KKP}), conditional on the loop-erasure paths not crossing. Assumption~\ref{ass: quantitative no boundary visits assumption} is now first verified for the (traces of) the underlying random walk excursions; it is thus also satisfied by their loop-erasures. Finally, one shows that these loop-erasures are vertex-disjoint with a uniformly positive probability, and hence Assumption~\ref{ass: quantitative no boundary visits assumption} also holds for the loop-erasures conditional on this vertex-disjointness.

This finishes the outline of the identification step in Theorem~\ref{thm: UST convergence}. We conclude by remarking that step~(3), which took a lengthy outline above, relies on tools required for the martingale argument in step~(1); in other words, here as in the case of all other models, it is the verification of Assumption~\ref{ass: dr fcns converge to loc mult SLE} that is the core of the convergence proof.
\end{proof}

Convergence results for a single UST branch have been given at least in~\cite{LSW-LERW_and_UST, Zhan-scaling_limits_of_planar_LERW, YY-Loop-erased_random_walk_and_Poisson_kernel_on_planar_graphs, LV-natural_parametrization_for_SLE, CW-mLERW}. Convergence results for multiple branches have been predicted in various sources, e.g.,~\cite{KW-boundary_partitions_in_trees_and_dimers, KKP, Wu17}.

\subsection{A complete new example: the harmonic explorer}

In this subsection we give a new and complete example of multiple SLE convergence given by our main theorems. The model we consider is the natural multiple-curve generalization of the harmonic explorer on the honeycomb lattice. The harmonic explorer was introduced as a toy model for the study of the level lines of the discrete Gaussian free field. One very useful simplification in moving to the harmonic explorer was the arrival of the DDMP. This is also the reason why we consider the harmonic explorer but not the discrete Gaussian free field.

In the one-curve case, the convergence of the harmonic explorer to the chordal $\SLE(4)$ was proven in~\cite{SS05}, and the proof here employs similar ideas. The convergence of discrete Gaussian free field level lines to chordal $\SLE(4)$ was later proven in~\cite{SS09}. Multiple $\SLE(4)$:s have been studied via the (continuous) Gaussian free field in~\cite{PW}, but we are not aware of a prior lattice model convergence result addressing multiple $\SLE(4)$.

\subsubsection{\textbf{A first definition}}

The multiple harmonic explorer has to our knowledge not appeared anywhere previously, but it is a straightforward generalization of the harmonic explorer. We give a first definition here, and will soon find a useful equivalent definition.

Consider the honeycomb lattice $H$, and its simply-connected subgraph $(\Gr; e_1, \ldots, e_{2N})$ with distinct marked boundary edges. Colour of the faces adjacent to a boundary vertex black or white, so that colour of the boundary faces changes precisely at the edges $e_1, \ldots, e_{2N}$, say for definiteness so that boundary arcs counterclockwise from odd to even are black and even to odd white.

\begin{figure}
\includegraphics[width=0.6\textwidth]{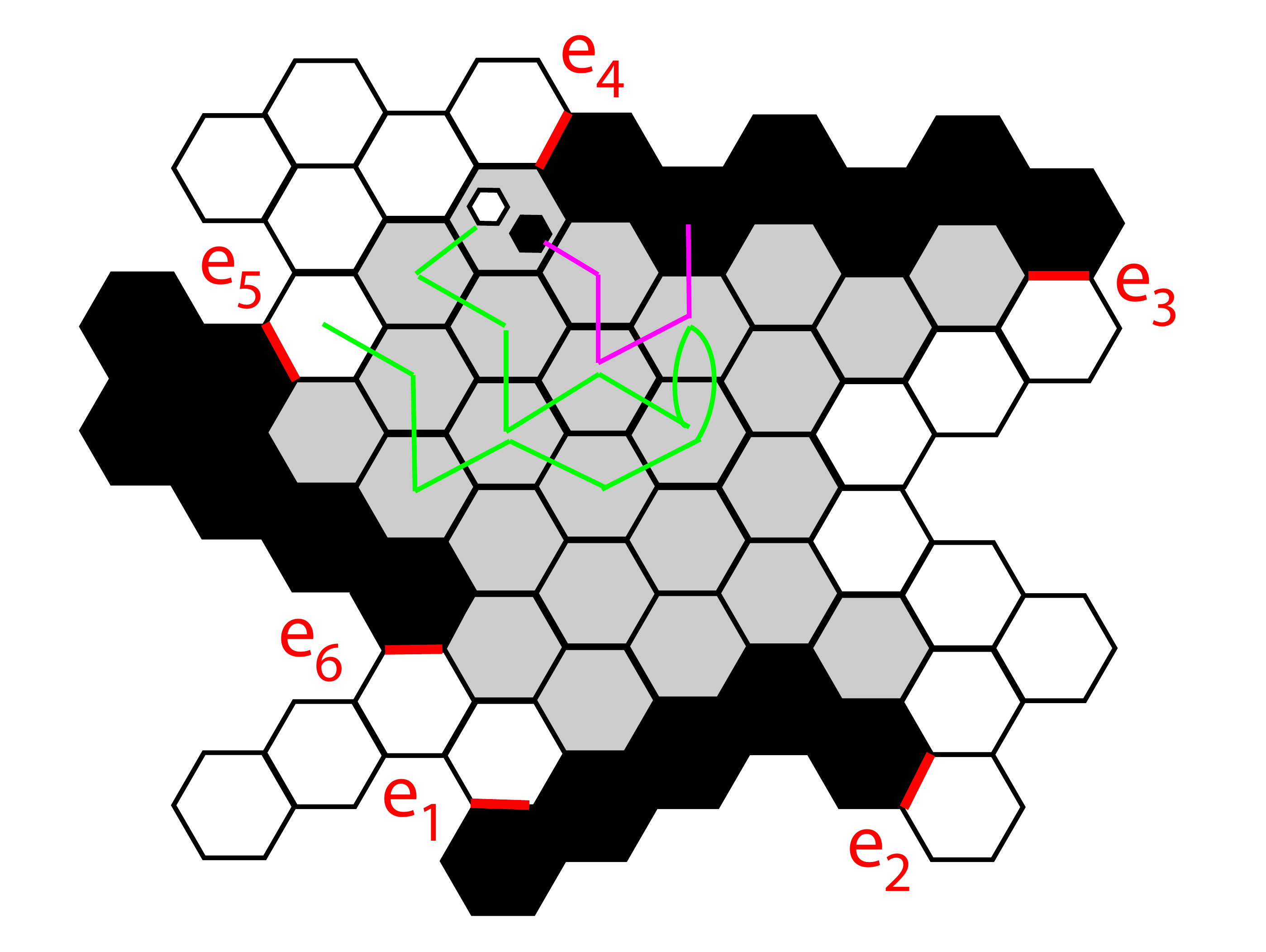}
\caption{
\label{fig: HE grow edge}
A simply-connected subgraph $\Gr$ of $H$, and its boundary faces, coloured black and white. In order to grow an edge at index $4$, one launches a random walk on the faces of $H$ from the face $F$ right in front of the edge $e_4$. If the random walk hits the boundary faces of $\Gr$ on a white face (green trajectory), one colours $F$ white. If the random walk first hits a black face (purple trajectory), one colours $F$ black.
}
\end{figure}

Let $\Gr_t$ be a graph with marked boundary edges as above (we suppress the edges in the notation; note also that the marked boundary edges determine boundary colouring and vice versa). Given such $\Gr_t$, we now define a procedure that yields $\Gr_{t+1}$, also of the above type. We call this procedure \emph{growing an edge at $i$}, where $1 \le i \le 2N$. Call the interior vertex of $e_i$ its tip vertex. There are three cases.
\begin{itemize}
\item[1)] If the tip vertex also adjacent to some marked boundary edge of $\Gr_t$ other than $e_i$, set $\Gr_{t+1} = \Gr_t$.
\end{itemize}
Otherwise, observe that the tip vertex of $e_i$ is adjacent to three faces of $\Gr_t$. Two of these faces are the boundary faces of $\Gr_t$ on either side of $e_i$. Call the third one $F$. We then determine a colour to $F$ as follows.
\begin{itemize}
\item[2)]  If $F$ is a boundary face of $\Gr_t$ it is already coloured in $\Gr_t$.
\item[3)] If $F$ is not a boundary face of $\Gr_t$, lauch a simple random walk on the faces of $\Gr_t$ from $F$. If it first hits the boundary faces of $\Gr_t$ at a black face, colour $F$ black, and otherwise colour $F$ white.
\end{itemize}
Note that cases (2) and (3) can be summarized as
\begin{align*}
\PR [F \text{ is black}] = \HarmMeas_{\Gr_t^*} ( F; \text{black boundary of }\Gr_t),
\end{align*}
where $\HarmMeas_{\Gr_t^*} ( F; \text{black})$ denotes the harmonic measure on the faces of $\Gr_t$ of the black boundary, as seen from $F$.

In cases (2) and (3) above, $\Gr_{t+1}$ is obtained from $\Gr_t$ by declaring the tip vertex of $e_i$ a boundary vertex. The $i$:th marked boundary edge of $\Gr_{t+1}$ then starts from this vertex, and goes  either clockwise or counterclockwise along the boundary of $F$, the direction chosen so that the boundary colourings of $F$ in $\Gr_{t+1}$ is as determined above. All other marked boundary edges of $\Gr_{t+1}$ are the same as in $\Gr_t$.

Finally, we define the discrete random curves $(\gamma_{\Gr; 1}, \ldots, \gamma_{\Gr; N})$ given by the \emph{multiple harmonic explorer} on $(\Gr; e_1, \ldots, e_{2N})$. These are the curves obtained by growing edges in the following order: start from $\Gr_1 = (\Gr; e_1, \ldots, e_{2N})$ as above. Then, inductively, $\Gr_{t+1}$ is obtained from $\Gr_t$ by growing the edge at $i$, where $i \equiv t$ modulo $2N$. I.e., we grow edges at $1, 2, \ldots, 2N, 1, 2, \ldots, 2N, 1, 2, \ldots$. Each growing of an edge is independent from the previous ones. We continue this until the graphs $\Gr_t$ stabilize, i.e., growing any edge leads to case (1) above.

\subsubsection{\textbf{An equivalent definition}}

Let $\Gr_1 = (\Gr; e_1, \ldots, e_{2N}), \Gr_2, \ldots$ be as above. Let $F_1$ be the face at the tip of the boundary edge $e_1$ of $\Gr$, and $F_2$ the at the tip of $e_2$. Suppose that we grow an edge of $\Gr_1$ at $2$, not at $1$ as in the first definition of the harmonic explorer, or equivalently colour $F_2$. We have
\begin{align*}
\PR [F_2 \text{ coloured first is black}] = \HarmMeas_{\Gr_1^*} ( F_2; \text{black boundary of }\Gr_1).
\end{align*}
Suppose now that we first grow an edge at $1$ and then at $2$. Then, we observe that
\begin{align}
\label{eq: swapping growth order in HE}
\PR [F_2 \text{ coloured second is black}] &= \HarmMeas_{\Gr_2^*} ( F_2; \text{black boundary of }\Gr_1)  + \HarmMeas_{\Gr_2^*} ( F_2; F_1 )  \PR [F_1 \text{ is black}] \\
\nonumber
&= \HarmMeas_{\Gr_2^*} ( F_2; \text{black boundary of }\Gr_1)  + \HarmMeas_{\Gr_2^*} ( F_2; F_1 ) \HarmMeas_{\Gr_1^*} ( F_1; \text{black boundary of }\Gr_1) \\
\nonumber
&= \HarmMeas_{\Gr_1^*} ( F_2; \text{black boundary of }\Gr_1) ,
\end{align}
where the last step follows by the strong Markov property of the simple random walk.
From the two computations above, we observe that the probability of colouring $F_2$ black does not depend on whether it was coloured first or second. Switching the roles of $F_1$ and $F_2$ in the computation above, we observe that \emph{growing and edge (i) first at $1$ and then at $2$, or (ii) first at $2$ and then at $1$, the pairs of edges grown in cases (i) and (ii) are equal in distribution}.

Using the the argument above one can deduce that the edges of the multiple harmonic explorer can be grown in an order chosen freely.


\begin{prop} \emph{(Equivalent definition of the multiple harmonic explorer.)}
\label{prop: equivalent def of HE}
Let $\Gr_1 = (\Gr; e_1, \ldots, e_{2N})$ be as above. Inductively, let $\Gr_{t+1}$ be obtained from $\Gr_t$ by growing the edge at $i= i(t)$, where the index $i(t)$ only depends on the graphs $\Gr_1, \ldots, \Gr_t$ up to time $t$. Assume that the indices $i(t)$ are chosen so that $\Gr_{t+1} \ne \Gr_t$ if it is possible to choose such an $i(t)$. The curves obtained by growing edges in this manner are in distribution equal to the multiple harmonic explorer on $\Gr_1$. Furthermore, conditional on any sequence of first the graphs $\Gr_1, \ldots, \Gr_t$, the remainder of the curves is distributed as a multiple harmonic explorer on $\Gr_t$.
\end{prop}

A rule to determine $i(t)$ given $\Gr_1, \ldots, \Gr_t$ such that $\Gr_{t+1} \ne \Gr_t$ if possible is called a \emph{valid growth rule}. The obtained process is called \emph{harmonic explorer under a valid growth rule}. Note also that by the above proposition, the one-curve harmonic explorer studied, e.g., in~\cite{SS05} is the special case $N=1$, namely under the growth rule of always growing an edge at $1$, which is valid if $N=1$.

\begin{proof}[Proof of Proposition~\ref{prop: equivalent def of HE}]
Start from the first definition of the multiple harmonic explorer, i.e., growing edges at indices  $1, 2, 3, \ldots, 2N, 1, 2, 3, \ldots$ By the argument above, we can swap the growth order of egdes at $1$ and $2$, growing edges at indices $2, 1, 3, \ldots, 2N, 2, 1, 3, \ldots$, and obtain same distribution of curves. Similarly, we can swap the indices of any two subsequent growth steps. Any permutation $\sigma$ of $2N$ is a composition of such swaps, so we can grow the edges in order $\sigma(1), \sigma(2),  \ldots, \sigma(2N), \sigma(1), \sigma(2), \ldots$ and still obtain same distribution of curves. Thus, the first index $i(1)= \sigma(1)$ of the edge to be grown can be chosen freely, and conditional on the obtained graph $\Gr_2$, the remainder of the curves (obtained by growing edges at indices $\sigma(2),  \ldots, \sigma(2N), \sigma(1), \sigma(2), \ldots$) is a harmonic explorer in $\Gr_2$. Both claims now follow inductively.
\end{proof}

\subsubsection{\textbf{The convergence theorem}}

Consider now the scaled honeycomb lattices $\delta_n H = \InfiniteGr_n$, where $\delta_n \shrinkto 0$ as $n \to \infty$, and their simply-connected subgraphs $(\Gr^{(n)}; e_1^{(n)}, \ldots, e^{(n)}_{2N})$ as above, converging to some domain $(\domain; p_1, \ldots, p_{2N})$ in the Carath\'{e}odory sense. Study these discretizations under the assumptions and notation of Section~\ref{subsec: setup and notation}.

\begin{thm}
\label{thm: HE convergence}
In the setup described above, the multiple harmonic explorer interfaces $(\gamma^{(n)}_1, \ldots, \gamma^{(n)}_N)$ converge weakly in $X(\C)$ to the local-to-global multiple $\SLE(4)$ in $(\domain; p_1, \ldots, p_{2N})$ with the partition functions
\begin{align}
\label{eq: HE part fcns}
\PartF_N (x_1, \ldots, x_{2N}) = \prod_{1 \le k < \ell \le 2N} (x_\ell - x_k)^{1/2 (-1)^{(\ell - k)}}.
\end{align}
\end{thm}

The fact that~\eqref{eq: HE part fcns} actually is a local multiple SLE partition function, as defined in Section~\ref{subsubsec: local multiple SLE partition functions},
is verified in~\cite[Proposition~4.8]{KP-pure_partition_functions_of_multiple_SLEs}. Note that combining this theorem with Theorems~\ref{thm: relation to global multiple SLE 1}, we also obtain the convergence of the conditional discrete curve models to the global multiple SLEs of~\cite{PW, BPW}. Furthermore, combining this with Theorem~\ref{thm: relation to global multiple SLE 2} we obtain their convergence of the conditional local-to-global multiple SLEs.

\begin{proof}[Proof of Precompactness in Theorem~\ref{thm: HE convergence}]
Let us first prove the precompactness part of Theorem~\ref{thm: HE convergence} by verifying the assumptions of Theorem~\ref{thm: precompactness thm multiple curves} for the discrete curve model. Using  Proposition~\ref{prop: equivalent def of HE}, one can easily find valid growth rules that show that the harmonic explorer has alternating boundary conditions and satisfies the DDMP. Condition (G) for the one-curve model is verified in~\cite[Proposition~6.3]{SS05}.
\end{proof}

The identification step will be done in Section~\ref{subsubsec: finishing HE convergence}.

\subsubsection{\textbf{A discrete martingale}}

Let us return to the discrete model:
consider the harmonic explorer with a valid growth rule on a simply-connected subgraph $\Gr_1 = (\Gr; e_1, \ldots, e_{2N})$ of $H$. Let $z$ be a face of $\Gr_1$. Note that it is also a face of $\Gr_t$, for any $t$. Define
\begin{align*}
M_t(z) = \HarmMeas_{\Gr_t^*} ( z; \text{black boundary of }\Gr_t).
\end{align*}

\begin{prop} 
\label{prop: HE disc mg}
$M_t(z)$ is for any face $z$ an $\mathcal{F}_t$-martingale, where $\mathcal{F}_t$ is the sigma algebra of the graphs $\Gr_1, \ldots, \Gr_t$.
\end{prop}

\begin{proof}
$M_t(z)$ is clearly bounded and $\mathcal{F}_t$-adapted, so it remains to show the conditional expectation property of discrete matringales.
Let $F_t$ be the face coloured to obtain $\Gr_{t+1}$ from $\Gr_t$. Note that given $\Gr_t$, we know the face $F_t$, since the growth rule is valid. The computation is now identical to~\eqref{eq: swapping growth order in HE}:
\begin{align*}
\EX [ M_{t+1} (z) \; \vert \; \mathcal{F}_t]
&= \HarmMeas_{\Gr_{t+1}^*} ( z ; \text{black boundary of }\Gr_t)  + \HarmMeas_{\Gr_{t+1}^*} ( z; F_t )  \PR [F_t \text{ is black} \; \vert \; \mathcal{F}_t] \\
&= \HarmMeas_{\Gr_{t+1}^*} ( z ; \text{black boundary of }\Gr_t)  + \HarmMeas_{\Gr_{t+1}^*} ( z; F_t )  \HarmMeas_{\Gr_t^*} ( F_t; \text{black boundary of }\Gr_t) \\
&= \HarmMeas_{\Gr_t^*} ( z; \text{black boundary of }\Gr_t)  = M_t(z).
\end{align*}
This concludes the proof.
\end{proof}

\subsubsection{\textbf{Verifying Assumption~\ref{ass: dr fcns converge to loc mult SLE}}}
\label{subsubsec: HE local identification}

Let us now 
verify Assumption~\ref{ass: dr fcns converge to loc mult SLE} for the multiple harmonic explorer. This is the core of the limit identification step in the proof of Theorem~\ref{thm: HE convergence}.

Note that in Assumption~\ref{ass: dr fcns converge to loc mult SLE}, we worked with Carath\'{e}odory converging graphs $(\Gr^{(n)}; e_1^{(n)}, \ldots, e^{(n)}_{2N})$, but with relaxed regularity at marked boundary points, see Section~\ref{subsubsec: relaxed assumptions}. Nevertheless, by Remark~\ref{rem: precompactness for irregular boundary}, since the assumptions of Theorem~\ref{thm: precompactness thm multiple curves} for the discrete curve model were verified in the precompactness part of Theorem~\ref{thm: HE convergence}, we know that the stopped driving functions $\DrFcnLattNotime{n}{j}$ are precompact also under this relaxed boundary regularity. Let us fix $j$, assume that a convergent subsequence has been extracted, and suppress it in notation, so that
\begin{align*}
\DrFcnLattNotime{n}{j} \stackrel{n \to \infty}{\longrightarrow} \DrFcnNoInd_j \qquad \text{weakly in } \ctsfcns.
\end{align*}
We now claim that Assumption~\ref{ass: dr fcns converge to loc mult SLE} is now satisfied:

\begin{prop}
\label{prop: dr fcn convergence for HE}
Any weak limit $\DrFcnNoInd_j$ as above is the stopped driving function of the local multiple $\SLE(4)$ with partition function~\eqref{eq: HE part fcns}.
\end{prop}

For $i \ne j$, let $\DrFcn{i}{t}$ describe the locations of the other marked boundary points under the Loewner equation driven by $\DrFcn{j}{t}$:
\begin{align*}
d \DrFcn{i}{t} = \frac{2 dt}{\DrFcn{i}{t} - \DrFcn{j}{t}}.
\end{align*}
Let $g_t$ be the mapping-out functions of this Loewner equation, so $\DrFcn{i}{t} = g_t(\DrFcn{i}{0})$. For any point $\zeta \in \bH$ outside of the localization neighbourhood $U_j$ of the $j$:th boundary point, define for times $t \le \tau_j$
\begin{align}
\label{eq: defn of continuous mgale}
\Mart_t ( \zeta ) &= \frac{1}{\pi} \Re \left( \log (g_t( \zeta ) - \DrFcn{2N}{t}) - \log (g_t( \zeta ) - \DrFcn{2N -1}{t}) + \ldots - \log (g_t( \zeta ) - \DrFcn{1}{t})   \right) \\
\nonumber
&= \frac{1}{\pi} \Re \left( \sum_{i = 1}^{2N} (-1)^i \log (g_t( \zeta ) - \DrFcn{i}{t}) \right),
\end{align}
and for $t \ge \tau_j$ set $\Mart_t ( \zeta ) = \Mart_{\tau_j} ( \zeta )$. (Here $\log$ denotes the natural complex logarithm; we choose the branch cut on the negative imaginary axis.)
Note that $\Mart_t (\zeta )$ is the continuum harmonic measure in $\bH$ of the counterclockwise odd-to-even marked boundary arcs between the marked boundary points $\DrFcn{1}{t}, \ldots, \DrFcn{2N}{t}$ as seen from $g_t ( \zeta )$. In other words, $\Mart_t ( \zeta )$ is the direct continuum analogue of the discrete martingale $M_t ( z )$ in the previous subsection.

\begin{lem}
\label{lem: HE martingale observable}
$\Mart_t (\zeta )$ is for any $\zeta$ as above a continuous bounded martingale with respect to the filtration $\mathscr{F}_t$ of $\DrFcn{j}{t}$.
\end{lem}

\begin{proof}
$\Mart_t (\zeta )$ is clearly continuous, bounded, and $\mathscr{F}_t$-adapted. It remains to verify the conditional expectation property of martingales. For this, we will show that
\begin{align*}
\Mart_t ( \zeta ) = \EX [ \Mart_{\tau_j} ( \zeta ) \; \vert \; \mathscr{F}_t ],
\end{align*}
since any conditional expectation is a martingale. The above holds if and only if for all continuous bounded functions $f_t : \ctsfcns \to \R$ of $\DrFcnNoInd_j$, measurable with respect to $\mathscr{F}_t$, i.e., only depending on $\DrFcnNoInd_j$ up to time $t$, we have
\begin{align}
\label{eq: cond exp alt def}
\EX [ \Mart_{t} ( \zeta ) f_t (\DrFcnNoInd_j) ] = \EX [ \Mart_{\tau_j} ( \zeta ) f_t (\DrFcnNoInd_j) ].
\end{align}

Let us verify~\eqref{eq: cond exp alt def}. Consider the analogue of $\Mart_t ( \zeta )$ with the discrete driving function $\IterDrFcnLatt{n}{j}{t}$, i.e., define  for times $t \le \tau_j^{(n)}$
\begin{align*}
\Mart^{(n)}_t ( \zeta ) = \frac{1}{\pi} \Re \left( \log (g^{(n)}_t( \zeta ) - \DrFcnLatt{n}{2N}{t}) - \log (g^{(n)}_t( \zeta ) - \DrFcnLatt{n}{2N - 1}{t}) + \ldots - \log (g^{(n)}_t( \zeta ) - \DrFcnLatt{n}{1}{t})   \right),
\end{align*}
and for $t \ge \tau_j^{(n)}$ set $\Mart^{(n)}_t ( \zeta ) = \Mart^{(n)}_{\tau_j^{(n)}} ( \zeta )$. Note again that $\DrFcnLatt{n}{i}{t}$ describe the locations of the other discrete marked boundary points under the Loewner equation driven by $\DrFcnLatt{n}{j}{t}$. Now, if $\DrFcnLatt{n}{i}{t}$ were lauched from the same locations as $\DrFcn{i}{t}$, $\Mart^{(n)}_t ( \zeta )$ and $\Mart_t ( \zeta )$ would just be continuous functions of $\DrFcnLatt{n}{j}{t}$ and $\DrFcn{i}{t}$, respectively. It takes a standard harmonic measure argument to show that a small change in the launching location does not play a role, and hence by the weak convergence $\DrFcnLattNotime{n}{j} \to  \DrFcnNoInd_j$ (using also the continuity of the stopping times), we get
\begin{align} 
\label{eq: Mart_tau}
\EX [ \Mart_{\tau_j}( \zeta ) f_t (\DrFcnNoInd_j) ] = \lim_{n} \EX^{(n)} [ \Mart^{(n)}_{\tau_j^{(n)}} ( \zeta ) f_t ( \DrFcnLattNotime{n}{j} ) ].
\end{align}
SImilarly, for $\Mart_{t} = \Mart_{t \wedge \tau_j}$ and $\Mart^{(n)}_{t}$, we get
\begin{align}
\label{eq: Mart_t}
\EX [ \Mart_{t} ( \zeta ) f_t(\DrFcnNoInd_j) ] = \lim_{n} \EX^{(n)} [ \Mart^{(n)}_{t} ( \zeta ) f_t( \DrFcnLattNotime{n}{j} ) ].
\end{align}

Let us next relate $\Mart^{(n)}$ to the discrete martingales  $M^{(n)}$  under the harmonic explorer in $\Gr^{(n)}$, given by  Proposition~\ref{prop: HE disc mg}. Consider first $\Mart^{(n)}_{\tau_j^{(n)}}$. Due to the convergence of discrete harmonic measures to the continuous ones, which is uniform over the family of discrete domains bounded from inside and outside~\cite[Theorem~3.12]{CS-discrete_complex_analysis_on_isoradial}, we have
\begin{align}
\label{eq: unif conv of mg observable 1}
\EX^{(n)} [ \Mart^{(n)}_{\tau_j^{(n)}} ( \zeta ) f_t ( \DrFcnLattNotime{n}{j} ) ] = \EX^{(n)} [ M^{(n)}_{ \lceil \tau_j^{(n)} \rceil } ( z^{(n)} ) f_t ( \DrFcnLattNotime{n}{j} ) + o_n (1)] ;
\end{align}
here $z^{(n)}$ is the face of $\Gr^{(n)}$ whose conformal image in $\bH$ contains $\zeta$; $\lceil \tau_j^{(n)} \rceil$ is the first time after $ \tau_j^{(n)} $ when the lattice initial segment reaches a vertex; and $o_n (1)$ denotes $o(1)$ as $n \to \infty$, and is uniform over $t$ and the possible initial segments up to time $\lceil \tau_j^{(n)} \rceil$, or, equivalently, over the  driving functions $\DrFcnLattNotime{n}{j}$.
Arguing identically, we also have
\begin{align}
\label{eq: unif conv of mg observable 2}
\EX [ \Mart^{(n)}_{t} ( \zeta ) f_t ( \DrFcnLattNotime{n}{j} ) ] = \EX [ M^{(n)}_{ \lceil t \wedge \tau_j^{(n)} \rceil } ( z^{(n)} ) f_t ( \DrFcnLattNotime{n}{j} ) + o_n (1)],
\end{align}
where the notations are defined analogously to the above.

Now, using the discrete martingale property of $M^{(n)} $ and the uniformity of the $o(1)$ terms in~\eqref{eq: unif conv of mg observable 1} and \eqref{eq: unif conv of mg observable 2}, we deduce
\begin{align*}
\EX [ \Mart^{(n)}_{t} ( \zeta ) f_t ( \DrFcnLattNotime{n}{j} ) ] = \EX [ \Mart^{(n)}_{\tau_j^{(n)}} ( \zeta ) f_t ( \DrFcnLattNotime{n}{j} ) ] + o_n (1).
\end{align*}
Substituting this into~\eqref{eq: Mart_t} and~\eqref{eq: Mart_tau}, we observe that
\begin{align*}
\EX [ \Mart_{\tau_j}( \zeta ) f_t (\DrFcnNoInd_j) ]  - \EX [ \Mart_{t}( \zeta ) f_t (\DrFcnNoInd_j) ] = \lim_n (0 +  o_n (1)) = 0.
\end{align*}
This shows that~\eqref{eq: cond exp alt def} holds and finishes the proof of the lemma.
\end{proof}

\begin{proof}[Proof of Proposition~\ref{prop: dr fcn convergence for HE}]
Observe first that the derivative of the expression~\eqref{eq: defn of continuous mgale} defining $\Mart_{t} ( \zeta )$ with respect to $\DrFcn{j}{t}$ is never zero. Thus, by the Implicit function theorem,  $\DrFcn{j}{t}$ can be expressed as a smooth function of $\Mart_{t} ( \zeta )$ and $\DrFcn{i}{t}$, $i \ne j$. In particular, since $\Mart_{t} ( \zeta )$ is a continuous bounded martingale and $\DrFcn{i}{t}$ are continuously differentiable in time, it follows that $\DrFcn{j}{t}$ is a semimartingale. 

Let us now apply It\^{o} calculus to the (complex) process
\begin{align*}
A_t = \sum_{i = 1}^{2N} (-1)^i \log (g_t( \zeta ) - \DrFcn{i}{t}),
\end{align*}
whose real part is a martingale by Lemma~\ref{lem: HE martingale observable}. We obtain
\begin{align*}
\ud A_t = & \sum_{\substack{i = 1 \\ i \ne j }}^{2N} (-1)^i \frac{1}{g_t( \zeta ) - \DrFcn{i}{t } } \left( \frac{2 \ud t}{ g_t( \zeta ) - \DrFcn{j}{t } } - \frac{2 \ud t}{ \DrFcn{i}{t } - \DrFcn{j}{t } } \right) \\
& + (-1)^j \frac{1}{g_t( \zeta ) - \DrFcn{j}{t } } \left( \frac{2 \ud t}{ g_t( \zeta ) - \DrFcn{j}{t } }  
- \ud \DrFcn{j}{t } \right)
+ 1/2 (-1)^j \left( - \frac{1}{ (g_t( \zeta ) - \DrFcn{j}{t } )^2 } \right) \ud \langle \DrFcnNoInd_j , \DrFcnNoInd_j \rangle_t.
\end{align*}
Now, $A_t$ is a semimartingale, and consist thus of a local martingale part and a finite variation (f.v.) part.
For the real part of $A_t$ to be a martingale, the real part of the f.v. part of $A_t$ must vanish. After some simplifications, we express the f.v. part as
\begin{align*}
\ud [ \text{f.v. part of }  A_t ] = & \sum_{\substack{i = 1 \\ i \ne j }}^{2N} (-1)^i \frac{2 \ud t}{ g_t( \zeta ) - \DrFcn{j}{t } } \left(  - \frac{1}{ \DrFcn{i}{t } - \DrFcn{j}{t } } \right) 
- (-1)^j \frac{1}{g_t( \zeta ) - \DrFcn{j}{t } } \ud [ \text{f.v. part of }  \DrFcn{j}{t } ]\\
& + (-1)^j \frac{1}{ (g_t( \zeta ) - \DrFcn{j}{t } )^2 } \left( 2 \ud t - \ud \langle \DrFcnNoInd_j , \DrFcnNoInd_j \rangle_t / 2 \right).
\end{align*}
Furthermore, the real part of the above must vanish for a \emph{continuum} of $\zeta$:s. On the other hand, the above is a second degree polynomial of the complex variable $1/( g_t( \zeta ) - \DrFcn{j}{t } )$ with real coefficients. Now, the real part of such a polynomial vanishes on an open set of $\zeta$:s if and only if the coefficients of $1/( g_t( \zeta ) - \DrFcn{j}{t } )$ and $1/( g_t( \zeta ) - \DrFcn{j}{t } )^2 $ both vanish. The latter gives 
\begin{align}
\label{eq: indentifying mg part of semimg}
 2 \ud t - \ud \langle \DrFcnNoInd_j , \DrFcnNoInd_j \rangle_t / 2 = 0 \quad \Rightarrow \quad  \langle \DrFcnNoInd_j , \DrFcnNoInd_j \rangle_t  = 4t,
\end{align}
and the former gives
\begin{align}
\label{eq: identifying fv part of semimg}
\ud [ \text{f.v. part of } \DrFcn{j}{t } ] = \sum_{\substack{i = 1 \\ i \ne j }}^{2N} (-1)^{i-j}  \left(  - \frac{2 \ud t}{ \DrFcn{i}{t } - \DrFcn{j}{t } } \right)  = 4 ( \partial_j \log \PartF_N )(\DrFcn{1}{t }, \ldots, \DrFcn{2N}{t }) \ud t,
\end{align}
where $\PartF_N$ is given by~\eqref{eq: HE part fcns}. Equations~\eqref{eq: indentifying mg part of semimg} and~\eqref{eq: identifying fv part of semimg} and the initial value $\DrFcn{j}{0}$ together identify the local semimartingale $\DrFcn{j}{t }$, giving the stochastic integral representation
\begin{align*}
\ud \DrFcn{j}{t } = \sqrt{4} \ud B_t + 4 ( \partial_j \log \PartF_N )(\DrFcn{1}{t }, \ldots, \DrFcn{2N}{t }) \ud t,
\end{align*}
where $B_t$ is a standard Brownian motion. By definition, this means that $\DrFcn{j}{t }$ is a local multiple $\SLE(4)$ driving function with partition function~\eqref{eq: HE part fcns}. This concludes the proof.
\end{proof}

\begin{rem}
A sophisticated guess for the partition function $\PartF_N$ from~\cite[Proposition~4.8]{KP-pure_partition_functions_of_multiple_SLEs} streamlined the proof above. However, this is not an inevitable logical input: $\PartF_N$ is determined (up to a multiplicative constant) by requiring that~\eqref{eq: identifying fv part of semimg} holds for all $1 \le j \le 2N$. 
\end{rem}

\subsubsection{\textbf{Finishing the proof of Theorem~\ref{thm: HE convergence}}}
\label{subsubsec: finishing HE convergence}

\begin{proof}[Identification part of Theorem~\ref{thm: HE convergence}]
Let us now finish the proof of Theorem~\ref{thm: HE convergence} by verifying that the assumptions needed for applying Theorem~\ref{thm: loc-2-glob multiple SLE convergence, kappa le 4} are satisfied. Assumption~\ref{ass: dr fcns converge to loc mult SLE} was just verified in Proposition~\ref{prop: dr fcn convergence for HE}. Assumption~\ref{ass: approximability} holds trivially, and~\ref{ass: cond C'} is verified via condition (G'): condition (G) is verified in~\cite[Lemma~6.3]{SS05} based on the discrete martingale of Proposition~\ref{prop: HE disc mg} for $N=1$ curves. The identical computation with general $N$ proves condition (G'). We can now apply Theorem~\ref{thm: loc-2-glob multiple SLE convergence, kappa le 4} to conclude the proof of Theorem~\ref{thm: HE convergence}.
\end{proof}

%
%
%
%
%
%
%
%
%
%
%
%
%
%

\bigskip{}

\appendix

%

\section{On regular conditional laws}
\label{sec: abstract nonsense}

\addtocontents{toc}{\setcounter{tocdepth}{1}}

In this appendix, we present for completeness some basic facts about regular conditional laws and their relation to conditional expectations. The notations in this appendix are independent of the notations in the rest of the article.


\subsection{\textbf{Regular conditional law given a sigma algebras}}

Let $(\Omega, \sF, \PR)$ be a probability space and $Y:(\Omega, \sF) \to (G, \sG)$ a measurable random variable. The regular conditional law of $Y$ given a sub-sigma algebra $\sH$ of $\sF$ 
is a map $\mu: \Omega \times \sG \to \R$ such that
\begin{itemize}
\item[i)] $\mu_\omega$ is a probability measure on $(G, \sG)$ for all $\omega \in \Omega$;
\item[ii)] $\omega \mapsto \mu_\omega [B] $ is, for all $B \in \sG$,  measurable $(\Omega, \sH) \to (\R, \sB)$, where $\sB$ denotes the standard Borel sigma algebra of $\R$; and
\item[iii)] $\PR [\omega \in A, \;  Y  \in B ] = \EX [\mathbb{I}_A (\omega) \mu_\omega [B]  ] $ for all $A \in \sH$ and $B \in \sG$.
\end{itemize}

\subsection{\textbf{Regular conditional law given a random variable}}

Let $(\Omega, \sF, \PR)$ be a probability space. Let $X: (\Omega, \sF) \to (E, \sE)$ and $Y:(\Omega, \sF) \to (G, \sG)$ be random variables taking values in compete separable metric spaces $E$ and $G$, respectively, with the Borel sigma algebras $\sE$ and $\sG$. We say that a map $\lambda$ from $E \times \sG $ to $ \R$, denoted $(x, B) \mapsto \lambda_x [B]$, is the \textit{(regular) conditional law of $Y $ given $X$} if
\begin{itemize}
\item[i)] $\lambda_x$ is a probability measure on $(G, \sG)$ for all $x \in E$;
\item[ii)] $x \mapsto \lambda_x [B] $ is, for all $B \in \sG$, a measurable map from $(E, \sE) \to (\R, \sB)$; and
\item[iii)] for all $A \in \sE$ and $B \in \sG$,
\begin{align*}
\PR [ X \in A, \;  Y  \in B ] = \EX [\mathbb{I}_A (X) \lambda_X [B]  ] = \int_{x \in E} \mathbb{I}_A (x) \lambda_x [B] dP_X (x),
\end{align*}
where $P_X$ denotes the law of $X$ on $(E, \sE)$.
\end{itemize}


We observe that the conditional law $\lambda$ of $Y $ given $X$, and the conditional law $\mu$ of $Y $ given the sigma algebra $\sigma(X) \subset \sF$ generated by $X$, are related by $\mu_\omega [\cdot] = \lambda_{X(\omega)} [\cdot]$ in the following precise sense. First, given a conditional law $\lambda$ of $Y $ given $X$, and taking $\mu_\omega [B] = \lambda_{X(\omega)} [B]$ one readily observes that $\mu$ is the conditional law of $Y$ given $\sigma(X)$. Conversely, assume that we are given the conditional law $\mu$ of $Y$ given $\sigma(X)$.
By the Doob--Dynkin Lemma for random variables in complete separable metric spaces~\cite[Lemma~5]{Taraldsen}, a measurable random variable $(\Omega, \sigma(X) ) \to (\R, \sB)$ can be expressed as a measurable function of $X$, so $\mu_\cdot [B] : (\Omega, \sigma(X) ) \to (\R, \sB)$ generates a function $ \lambda_\cdot [B] : (E, \sE ) \to (\R, \sB) $ such that $\mu_\omega [B] = \lambda_{X(\omega)} [B]$. One then readily observes that such a $\lambda$ is the conditional law of $Y$ given $X$.

\subsection{\textbf{Existence and almost sure uniqueness}}

The regular conditional law of $Y$ given $X$, as defined above, exists and is unique in an almost sure sense:
for the existence, a regular conditional law of $Y$ taking values in a complete separble metric space, given any sub-sigma algebra of $\sF$, exists~\cite{Durrett}. As observed in the previous paragraph, the law of $Y$ given $X$ (or $\sigma(X)$) thus also exists. For the uniqueness, if $\lambda$ and $\tilde{\lambda}$ are two regular conditional laws of $Y$ given $X$, then for any fixed $B \in \sG$, we have that  $\lambda_X [B] = \tilde{\lambda}_X [B]$ almost surely. Since the sigma algebra $\sG$ can be generated by a countable collection of sets $B$, it follows that also $\lambda_X  = \tilde{\lambda}_X $ as measures on $(G, \sG)$, almost surely.


\subsection{\textbf{Conditional expectations determine the conditional law}}

It is well known that the expectations $\EX [f(X) ]$ of all continuous bounded functions $f(X)$ determine the law of a random variable $X$ on a metric space. For instance, this result implies that the weak limit of a weakly converging sequence of random variables is unique. We will next prove an analogous result for conditional laws, stating  that all conditional expectations of continuous bounded functions determine the conditional law of a random variable. This characterization, labelled (a) below, will be of key importance in this paper when identifying regular conditional laws in weak limits.

\begin{prop}
\label{prop: conditional expectation determine conditional law}
Let $(\Omega, \sF, \PR)$ be a probability space and $X: (\Omega, \sF) \to (E, \sE)$ and $Y:(\Omega, \sF) \to (G, \sG)$ be random variables taking values in compete separable metric spaces $E$ and $G$, respectively, with the Borel sigma algebras $\sE$ and $\sG$. Let $\lambda_x$ be probability measures on $(G, \sG)$ for all $x \in E$.
Then, the following are equivalent:
\begin{itemize}
\item[a)] for all bounded, non-negative, Lipschitz continuous functions $f: (G, \sG) \to (R, \sB)$, the function
\begin{align*}
F: x \mapsto \int_{y \in G} f( y) d \lambda_x (y)
\end{align*}
is measurable $(E, \sE ) \to (\R, \sB)$ and we have
\begin{align*}
\EX [f(Y) \; \vert \; \sigma(X)] = F(X)
\end{align*}
\item[b)]  $\lambda$ is the conditional law of $Y $ given $X$
\item[c)] for all $\PR$-integrable measurable functions $h: (E \times G, \sE \otimes \sG) \to (\R, \sB)$,  the function
\begin{align*}
H: x \mapsto \int_{y \in G} h(x, y) d \lambda_x (y)
\end{align*}
is measurable $(E, \sE ) \to (\R, \sB)$ and  we have
\begin{align*}
\EX [h(X,Y) \; \vert \; \sigma(X)] = H(X). 
\end{align*}
\end{itemize}
\end{prop}

\begin{proof}
We prove the implications (c) $\Rightarrow$ (b) $\Rightarrow$ (a) $\Rightarrow$ (c).


To see the implication (c) $\Rightarrow$ (b), by assumption, property (i) of regular conditional distributions holds. Properties (ii) and (iii) follow by taking $h (x, y) = \mathbb{I}_B (y)$ in (c), where $B \in \sG$, in which case $H(x) = \lambda_x [B]$.


For the implication (b) $\Rightarrow$ (a), study the class of bounded measurable functions $f: (G, \sG) \to (\R, \sB)$ for which the corresponding function $F: (E , \sE ) \to (\R, \sB)$ 
\begin{align*}
F: x \mapsto \int_{y \in G} f(y) d \lambda_x (y)
\end{align*}
is measurable and
\begin{align*}
\EX [\mathbb{I}_A (X) f (Y)] = \EX [\mathbb{I}_A (X) F (X)]
\end{align*}
for all $A \in \sE$.
By assumption,  for all $B \in \sG$, we have $ \mathbb{I}_B (y)$ belongs to this class of functions. By the ``standard machine'' of integration theory, one then readily finds that all bounded, non-negative, continuous functions of $y$ belong to this class.

For the implication (a) $\Rightarrow$ (c), take first an open set $B \subset G$ and a sequence of functions $f_n: (G, \sG) \to (\R, \sB)$ given by
$f_n (y) = 1- \max \{1 - n \cdot d(y, G \setminus B), 0 \}$, 
so that $f_n$ are non-negative, bounded, and Lipschitz continuous each, and they increase to $ \mathbb{I}_B $ pointwise, i.e., $f_n(y) \incrto \mathbb{I}_B (y) $ for all $y \in G$. Denote
\begin{align*}
F_n (x) = \int_{y \in G} f_n( y) d \lambda_x (y). 
\end{align*}
By monotone convergence, $F_n(x) \incrto \lambda_x [B]$ for all $x \in E$. As an increasing limit of measurable functions, $x \mapsto \lambda_x [B]$ is thus measurable. Finally, one deduces that $\EX [\mathbb{I}_B (Y) \; \vert \; \sigma (X)] = \lambda_X [B] $ by starting from
\begin{align*}
\EX [f_n (Y) \mathbb{I}_A (X)] =  \EX [F_n (X) \mathbb{I}_A (X)]
\end{align*}
which holds by assumption for all $A \in \sE$, and the using monotone convergence on both sides.


Next, study the class $\mathcal{H}$ of bounded measurable functions $h: (E \times G, \sE \otimes \sG) \to (\R, \sB)$ for which
\begin{align*}
H: x \mapsto \int_{y \in G} h(x, y) d \lambda_x (y)
\end{align*}
is measurable $(E , \sE ) \to (\R, \sB)$ and
\begin{align*}
\EX [h (X, Y) \; \vert \; \sigma (X)]= H(X). 
\end{align*}
By the previous paragraph, for all open sets $B \subset G$, we have $ \mathbb{I}_B (y) \in \mathcal{H}$. The same then also holds for all closed $B \subset G$. Thus, taking $A \subset E$ closed and $B \subset G$ closed, we also have $\mathbb{I}_A (x) \mathbb{I}_B (y) \in \mathcal{H}$. Next, it is a standard exercise in integration theory to show that $\mathcal{H}$ is a monotone class of functions, and that products of closed sets $A \times B \subset E \times G$ are a $\pi$-system that generate the sigma algebra $\sE \otimes \sG$. Thus, by the monotone class theorem, $\mathcal{H}$ contains all bounded measurable functions $f: (E \times G, \sE \otimes \sG) \to (\R, \sB)$. This can be extended to all integrable functions by the ``standard machine''. This shows that (a) $\Rightarrow$ (c)
\end{proof}

\bigskip{}

\section{Continuous modifications of exit times}
\label{app: continuous stopping times}

\begin{figure}
\includegraphics[width=0.4\textwidth]{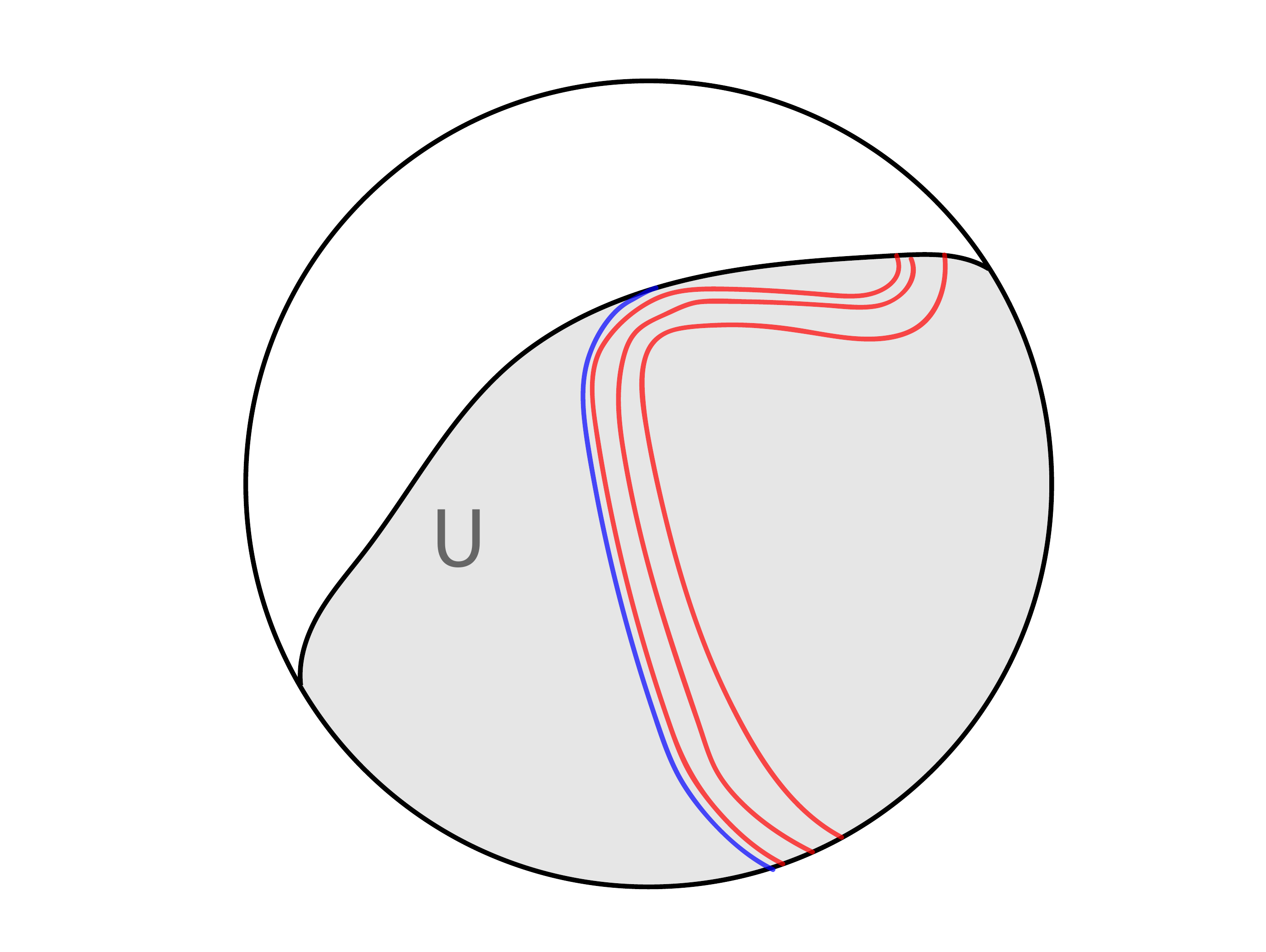}
\caption{\label{fig: discont st}
The exit time of a localization neighbourhood $U$ is discontinuous both in the topology of curves and driving functions. A schemitic illustration of a sequence of curves in $\overline{\UnitD}$, converging in both topologies, for which the exit times do not converge.
}
\end{figure}

The exit time of a localization neighbourhood is not a continuous function of the curves in the topologies that we consider in this paper, see Figure~\ref{fig: discont st}.
We define here the continuous modifications of the exit times of localization neighbourhoods. These stopping times will also be conformally invariant as precised shortly. Let us thus consider them in $\bH$. Let $U_1, \ldots, U_{2N}$ be bounded localization neighbourhoods of the local multiple SLE in $\bH$. Let $\tilde{U}_1, \ldots, \tilde{U}_{2N}$ be ``strictly larger'' localization neighbourhoods in the sense that $d(U_i, \bH \setminus \tilde{U}_i) > 0$ for all $i$, but so that the latter are also valid localization neighbourhoods for the local multiple SLE. Such $\tilde{U}_1, \ldots, \tilde{U}_{2N}$ always exist, and the conformal invariance mentioned above holds assuming that localization neighbourhoods actually come as such pairs $U_1, \ldots, U_{2N},\tilde{U}_1, \ldots, \tilde{U}_{2N} $. The continuous modification $\tau_i$ of the exit time of $U_i$ will then be between the exit times $T_i$ and $S_i$ of $U_i$ and $\tilde{U}_i$, respectively, i.e., $T_i < \tau_i < S_i$.

Now, the topological quadrilateral $\tilde{U}_i \setminus U_i$ (real-line segments being two opposite sides) can be conformally mapped to a rectangle $(0, 1) \times (0, L)$, with a unique $L>0$. Fix the reference point $w_i \in \tilde{U}_i \setminus U_i$ corresponding to the center point of this rectangle.

Denote by $h^{(i)}_t$ the value at $w_i$ of the following harmonic function on $\bH \setminus K^{(i)}_t$ (where $K^{(i)}_t$ is the hull growing from $U_i$): it takes boundary values $0$ on $\R \setminus K^{(i)}_t$ and inside $U_i$; for boundary points $z \in K^{(i)}_t $ with $z \in \tilde{U_i} \setminus U_i$, it takes the boundary values given by the corresponding $x$-coordinate in the rectangle $(0, 1) \times (0, L)$; and for boundary points  $z \in K^{(i)}_t \setminus \tilde{U}_i$ it takes the boundary value one.
Clearly, $h^{(i)}_t = 0$ if $K^{(i)}_t \subset U_i$, and $h^{(i)}_t$ is increasing in $t$. On the other hand, an easy Brownian motion argument gives a lower bound $h_{low}^{(i)}$ to the value $h^{(i)}_{S_i}$ at the exit time $S_i$ of $\tilde{U}_i$. Let $H^{(i)}_t$ be the harmonic measure of $K^{(i)}_t$ in $\bH$, as seen from $w_i$, so $H^{(i)}_t$ is strictly increasing in $t$. We define the continuous modification $\tau_i$ of the exit time $T_i$ to be the first time when the product $h^{(i)}_t H^{(i)}_t$ reaches the level $h^{(i)}_{low} H^{(i)}_0$.

The following result guarantees the continuity of the modified exit times. Let $\DetDrFcn$ be the driving function of a Loewner chain starting inside $U_i$.
Let $\DetDrFcn^{(n)} $ be a seuqence of driving functions, such that $\hat{\DetDrFcn}^{(n)}_s \to \hat{\DetDrFcn}_s$ uniformly over $s \in [0, S]$. Denote the continuous exit times of $U_i$ by the respective hulls by  $\tau^{(n)}$ and $\tau$.

\begin{lem}
\label{lem: cts stopping}
In the setup described above, the stopping times and stopped driving functions converge, $\tau^{(n)} \to \tau$ in $\R$ and $\hat{\DetDrFcn}^{(n)}_{t \wedge \tau_n} \to \hat{\DetDrFcn}_{t \wedge \tau}$ in $\ctsfcns$.
\end{lem}

\begin{proof}
Denote by $K^{(n)}_t$ and $K_t$ the growing hulls corresponding to the driving functions $\DetDrFcn^{(n)}_t $ and $\DetDrFcn_t $, respectively.
Combining \cite[Lemma~5.1]{Kemppainen-SLE_book} and Lemma~\ref{lem: Cara iff conv of confmaps}, we see that for any fixed time $s \in [0, S]$, $\bH \setminus K^{(n)}_t$ converge to $\bH \setminus K_t$ in the Carath\'{e}odory sense. This convergence also holds, if we equip the planar sets with two marked prime ends, the leftmost and rightmost real points of the hulls $ K^{(n)}_t$ or $ K_t$. We are considering harmonic functions $h^{(i)}_t$ and $H^{(i)}_t$ in these domains, whose the boundary data has discontinuities at two points, the leftmost and rightmost real points, and is otherwise bounded and continuous. Such harmonic functions in domains with two marked boundary points are Carath\'{e}odory stable --- a proof is similar to~\cite[Theorem~3.12]{CS-discrete_complex_analysis_on_isoradial} addressing harmonic measure. Thus, for any fixed $t$, the values of the harmonic functions $h^{(i)}_t$ and $H^{(i)}_t$ corresponding to $\bH \setminus K^{(n)}_t$ converge to those of $\bH \setminus K_t$. The convergence of hitting times of the product $h^{(i)}_t H^{(i)}_t$ follows as it is strictly increasing in $t$ for each growth process. This suffices to show that $\tau_n \to \tau$ in $\R$.

To deduce that $\hat{\DetDrFcn}^{(n)}_{t \wedge \tau_n} \to \hat{\DetDrFcn}_{t \wedge \tau}$ in $\ctsfcns$, note that  we at least have $\hat{\hat{\DetDrFcn}}^{(n)}_{s \wedge \tau} \to \hat{\DetDrFcn}_{s \wedge \tau}$ in $\ctsfcns$. The rest follows by combining the convergence $\tau_n \to \tau$ in $\R$ and the uniform continuity of the continuous function $\hat{\DetDrFcn}_s$ on the compact interval $[0, S] \supset [0, \tau_n]$.
\end{proof}

We can naturally also define the stopping time $\tau$ for chordal curves in $\UnitD$. The argument to prove the continuity is the identical.

\bigskip{}

\section{Deterministic Loewner equation and conformal maps}
\label{app: deterministic LE}

\addtocontents{toc}{\setcounter{tocdepth}{1}}

In this appendix, we provide some analysis of the deterministic Loewner equation under conformal maps.

Consider the following setup. Let $\DetDrFcn^{(n)}$ be a sequence of continuous (deterministic) driving functions, and assume that they converge in the space $\ctsfcns$ of continuous functions, $\DetDrFcn^{(n)} \to \DetDrFcn$. These functions generate by Loewner's equation some growing hulls $K^{(n)}_t$ (resp. $K_t$), starting their growth from $\DetDrFcn^{(n)}_{t = 0}$ (resp. $\DetDrFcn_{t = 0}$). Let $\LoeNbhd$ be a localization neighbourhood (a hull) of these starting points, and let $\tilde{\LoeNbhd}$ be a larger one, so that $d(\LoeNbhd, \bH \setminus \tilde{\LoeNbhd} ) > 0$. Let $\confmapSH_n$ be conformal maps from $\tilde{\LoeNbhd}$ to some subset of $\bH$, such that the real line segment of  $\bdry \tilde{\LoeNbhd}$ maps to real line under $\confmapSH_n$. (For instance, $\confmapSH_n$ could be conformal (M\"{o}bius) maps $ \bH \to \bH$ or mapping-out functions of some hull disjoint from $\tilde{\LoeNbhd}$.) Assume that $\confmapSH_n$ converge to another conformal map, $\confmapSH_n \to \confmapSH$, uniformly over the compact set $\tilde{\LoeNbhd}$. The main task of this appendix is to show that the Loewner driving functions of the hulls $\confmapSH_n (K_t^{(n)})$ converge to that of $\confmapSH (K_t)$.

Recall first that some growing family of hulls can be described by a Loewner equation driven by a continuous function if and only if satisfies the local growth property~(see~\cite{Kemppainen-SLE_book} for details). This characterization readily implies that the conformal images $\confmapSH_n (K^{(n)}_t)$ (resp. $\confmapSH (K_t)$) of interest here can be described by a Loewner driving function $\hat{\DetDrFcn}^{(n)}_s$ (resp. $\hat{\DetDrFcn}_s$), at least up to the time $\tilde{S}_n$ of exiting $\confmapSH_n (\tilde{\LoeNbhd})$ (resp. $\tilde{S}$ of exiting $\confmapSH (\tilde{\LoeNbhd})$). (We denote by $s$ the capacity parametrization after the conformal maps.) Note that these times depend on $n$, so it is way more convenient to observe that at the exit time $S$ of the smaller localization neighbourhood $\confmapSH ({\LoeNbhd})$ by $\confmapSH (K_t)$, also the hulls $\confmapSH_n (K^{(n)}_t)$ stay inside their larger neighbourhoods $\confmapSH_n (\tilde{\LoeNbhd})$, for all $n$ large enough. (This follows since $\bH \setminus K^{(n)}_t$ converge to $\bH \setminus K_t$ in the Carath\'{e}odory sense, for all fixed $t$.) Thus, we will study the driving processes up to the time $S$.

The main result of this appendix is the following. 

\begin{prop}
\label{lem: ctty of driving functions - pain lemma}
In the setup and notation above, if $\DetDrFcn^{(n)} \to \DetDrFcn$ in $\ctsfcns$, i.e., uniformly over compacts, then $\hat{\DetDrFcn}^{(n)}_s \to \hat{\DetDrFcn}_s$ uniformly over $s \in [0, S]$.
\end{prop}

The statement above is certainly not surprising.
However, note that differences compared to typical references addressing Loewner equation and conformal maps, e.g.,~\cite[Section~4.6]{Lawler-SLE_book}, are that we cannot apply It\^{o} calculus, and that the conformal maps $\confmapSH_n$ depend on $n$. With some effort, similar ideas can be used to prove this proposition. We have chosen not to include the proof of this proposition in this version of the paper. 




\subsection{\textbf{Some consequences}}


Let us list some consequences of Proposition~\ref{lem: ctty of driving functions - pain lemma}. Continue in the notation introduced before the statement of that lemma.

Take localization neighbourhoods $\LoeNbhd \subset \LoeNbhd_2 \subset \tilde{\LoeNbhd}$, such that $d(\LoeNbhd, \bH \setminus \LoeNbhd_2 ), d(\LoeNbhd_2, \bH \setminus \tilde{\LoeNbhd} ) > 0$. Denote the exit times of $\confmapSH(\LoeNbhd)$ and $\confmapSH(\LoeNbhd_2)$ by 
the growing hulls $\hat{K}_s$, with $\hat{K}_{s(t)} = \confmapSH(K_t)$,
by $S$ and $S_2$ respectively, and assume that the continuous  modification of the exit time $\sigma$ of $\confmapSH ( \LoeNbhd )$ is chosen using the pair of neighbourhoods  $\confmapSH(\LoeNbhd)$ and $\confmapSH(\LoeNbhd_2)$, so that $S < \sigma < S_2$. Denote the similar continuous exit times by $\hat{K}^{(n)}_s$ by $\sigma^{(n)}$. Let $\DetDrFcn^{(n)} \to \DetDrFcn$ and $\confmapSH_n \to \confmapSH$ as before.

\begin{cor}
\label{cor: coordinate change with converging maps}
In the notation above, $\sigma^{(n)} \to \sigma $ in $\R$ and $\hat{V}^{(n)}_{s \wedge \sigma^{(n)}} \to \hat{V}_{s \wedge \sigma}$ in $ \ctsfcns $.
\end{cor}

\begin{proof}
Apply Proposition~\ref{lem: ctty of driving functions - pain lemma} and with the localization neighbourhoods $\LoeNbhd_2 \subset \tilde{\LoeNbhd}$. This implies that  $\hat{V}^{(n)}_{s} \to \hat{V}_{s} $ uniformly over $s \in [0, S_2]$. Combining with Lemma~\ref{lem: cts stopping} proves the claim.
\end{proof}

\begin{rem}
\label{rem: exit times of neighbourhoods before or after conf map}
Above $\sigma$ and $\sigma^{(n)}$ are the continuous exit times of $\confmapSH(\LoeNbhd)$, defined via the same neighbourhoods $\confmapSH ( \LoeNbhd )$ and $\confmapSH ( \LoeNbhd_2 )$. The statement of Corollary~\ref{cor: coordinate change with converging maps} also holds if we in stead define $\sigma^{(n)}$ to be the continuous exit times of $\confmapSH_n(\LoeNbhd)$, with the neighbourhoods $\confmapSH_n ( \LoeNbhd )$ and $\confmapSH_n ( \LoeNbhd_2 )$. 
\end{rem}

%

Taking $\confmapSH_n = \confmapSH$ for all $n$, we obtain an important special case of Corollary~\ref{cor: coordinate change with converging maps}:

\begin{cor}
\label{cor: coordinate change of LE is cts fcn of dr fcns}
In the notation above, $\sigma \in \R$ and $\hat{V}_{s \wedge \sigma} \in \ctsfcns $ are continuous functions of $V \in \ctsfcns$.
\end{cor}

Denote by $\tau$ the continuous modification of the exit time of $\LoeNbhd$, with the pair $\LoeNbhd_2$. Recall that the continuous exit times are conformally invariant, so $s(\tau) = \sigma$. Corollary~\ref{cor: coordinate change of LE is cts fcn of dr fcns} now implies a stronger tool.

\begin{cor}
\label{cor: coordinate change is a bijection}
Equip the space of Loewner driving functions, in $t$ and $s$, stopped at exit times $\tau$ and $\sigma$, with the topology of $\ctsfcns$.
The mapping from $V_{t \wedge \tau} $ to $\hat{V}_{s \wedge \sigma}$ is a continuous bijection and its inverse is continuous.
\end{cor}

\begin{proof}
By Corollary~\ref{cor: coordinate change of LE is cts fcn of dr fcns},  $V_{t \wedge \tau} $ maps continuously to $\hat{V}_{s \wedge \sigma} $, the driving function of its conformal images, and  $V_{s \wedge \sigma}$ to $\hat{V}_{t \wedge \tau}$, the driving function of the conformal preimages. The bijectivity follows.
\end{proof}

Consider now another special case of Corollary~\ref{cor: coordinate change with converging maps}: take $\DetDrFcn^{(n)} = \DetDrFcn$ for all $n$ so that comparing $\hat{\DetDrFcn}^{(n)}$ and $\hat{\DetDrFcn}$ means comparing the effect of the different conformal maps. The below corollary shows that this effect is small \emph{uniformly} over the choice of $\DetDrFcn$:

\begin{cor}
\label{cor: LE coordinate changes with converging maps converge uniformly over compacts}
Take $\DetDrFcn^{(n)} = \DetDrFcn$ for all $n$, and fix the sequence of conformal maps $\confmapSH_n \to \confmapSH$. Then, we have $d_{\ctsfcns} (\hat{\DetDrFcn}^{(n)}_{s \wedge \sigma^{(n)}}, \hat{\DetDrFcn}_{s \wedge \sigma}) \to 0$ as $n \to \infty$, uniformly over $\DetDrFcn \in C$, for any compact set $C$  in the space of continous functions $\ctsfcns$.
\end{cor}

\begin{proof} For notational reasons, let us equip growing hulls, stopped at the continuous modification of the exit time of $\LoeNbhd$, with the topology of their driving functions. For instance, we deonte $d_{\ctsfcns} (\hat{\DetDrFcn}^{(n)}_{s \wedge \sigma_n}, \hat{\DetDrFcn}_{s \wedge \sigma}) = d ( \confmapSH_n (K_t),  \confmapSH (K_t) )$.

Assume now for a contradiction that for some $\delta$ and infinitely many values of $n$, there exist $\DetDrFcn^{(n)} \in C$ such that
\begin{align*}
d ( \confmapSH_n (K^{(n)}_t) , \confmapSH (K^{(n)}_t)) > \delta,
\end{align*}
where $K^{(n)}_t$ are the growing hulls generated by $\DetDrFcn^{(n)}$
By compactness, we may extract a subsequence (which we suppress in notation) so that $\DetDrFcn^{(n)}$ converge, $\DetDrFcn^{(n)} \to \DetDrFcn$ in $\ctsfcns$. Let $K_t$ be the growing hulls corresponding to $\DetDrFcn$. Now, compute
\begin{align*}
d ( \confmapSH_n (K^{(n)}_t) , \confmapSH (K^{(n)}_t))
\le 
d ( \confmapSH_n (K^{(n)}_t) , \confmapSH (K_t))
+ 
d ( \confmapSH (K_t) , \confmapSH (K^{(n)}_t)).
\end{align*}
By Corollary~\ref{cor: coordinate change with converging maps}, both terms on the right-hand side above converge to $0$ as $n \to \infty$. This is a contradiction, proving the claim.
\end{proof}

\begin{rem}
\label{rem: exit times of neighbourhoods before or after conf map 2}
Also Corollary~\ref{cor: LE coordinate changes with converging maps converge uniformly over compacts} holds with the alternative choice of stopping times $\sigma^{(n)}$ in Remark~\ref{rem: exit times of neighbourhoods before or after conf map}.
\end{rem}

%
%

%
%

\bibliographystyle{annotate}

\newcommand{\etalchar}[1]{$^{#1}$}

\end{document}